\documentclass{cpamart1}     %% Input Standard CPAM class.
\received{Month 200X}       %\revised{date of revision}
\volume{000}
\startingpage{1}                      

\authorheadline{C.L. Fefferman, S. Fliss and M.I. Weinstein}
\titleheadline{Discrete honeycombs, rational edges and edge states}

\usepackage{graphicx}
\usepackage{caption}
\usepackage{subcaption}
\usepackage{bm}
\usepackage{hyperref} 
\usepackage[mathscr]{euscript}
\usepackage{arydshln}
\usepackage{enumerate}
\usepackage{xcolor}
\usepackage{enumitem}
\usepackage{dsfont,textcomp}
\usepackage{float}
\usepackage{multirow}
\usepackage{pgf,tikz}
\usepackage{ulem,cancel}
\usepackage{csquotes}
\usepackage{stackengine,scalerel}

\usepackage{mathptmx}
\newtheorem{theorem}{Theorem}[section]
\newtheorem{lemma}[theorem]{Lemma}

\newtheorem{proposition}[theorem]{Proposition}
\newtheorem{example}[theorem]{Example}
\theoremstyle{definition}
\newtheorem{definition}[theorem]{Definition}
\theoremstyle{remark}
\newtheorem{remark}[theorem]{Remark}

\usetikzlibrary{shapes}
\usetikzlibrary{plotmarks}
\usetikzlibrary{shapes.misc}
\tikzset{cross/.style={cross out, draw=black, minimum size=5*(#1-\pgflinewidth), inner sep=0pt, outer sep=0pt},
cross/.default={3pt}}
 
\newcommand{\HH}{\mathbb{H}}
\newcommand{\R}{\mathbb{R}}
\newcommand{\C}{\mathbb{C}}
\newcommand{\Z}{\mathbb{Z}}
\newcommand{\N}{\mathbb{N}}

\newcommand{\ttm}{\tilde{m}}
\newcommand{\ttn}{\tilde{n}}

\newcommand{\Hbulk}{H_{\rm bulk}}
\newcommand{\Hk}{H_\sharp(k)}
\newcommand{\HTB}{H}
\newcommand{\kp}{k}
\newcommand{\nA}{n^A_{\rm min}}
\newcommand{\nB}{n^B_{\rm min}}
\newcommand{\none}{{\nu_1}}

\newcommand{\nthree}{{\nu_3}}

\newcommand{\hsigma}{{\hat{\sigma}}}

\newcommand{\ac}{{\arccos}}
\newcommand{\nbase}{{n_{\rm base}}}
\newcommand{\pcount}{{\mathfrak{p}(\kpar)}}
\newcommand{\qcount}{{\mathfrak{q}(\kpar)}}
\newcommand{\be}{{\bf e}}

\newcommand{\bv}{{\bf v}}

\newcommand{\tm}{{\tilde{m}}}
\newcommand{\tn}{{\tilde{n}}}

\newcommand{\kpar}{{k}}
\newcommand{\kperp}{{k_{\perp}}}

\newcommand{\D}{\partial}

\newcommand{\nit}{\noindent}
\newcommand{\nn}{\nonumber}

\newcommand\ocirc[1]{\ensurestackMath{\stackon[0.8pt]{#1}{\mkern1mu\scriptstyle\circ}}}
\begin{document}                        %% Standard LaTeX command

%%      -----------------------------------------------------------------------
%%      -------------------------------- TITLE -----------------------------
%%      -----------------------------------------------------------------------

\title{Discrete honeycombs, rational edges and edge states}

%%      -----------------------------------------------------------------------------
%%      ------------------------------- AUTHORS -----------------------------
%%      ----------------------------------------------------------------------------
\author{Charles L. Fefferman}{Princeton University, Princeton, NJ, USA}
% EXAMPLE: \author{Bart Simpson}{Universit� Paris-Sorbonne (Paris IV)}
% Uncomment and fill in the following lines as needed
\author{Sonia Fliss}{POEMS, CNRS, Inria, ENSTA Paris, Institut Polytechnique de Paris, Palaiseau, France}
\author{Michael I. Weinstein}{Columbia University, New York, NY, USA}
%\author{*** FOURTH AUTHOR'S NAME ***}{*** FOURTH AUTHOR'S AFFILIATION WHEN ARTICLE WAS WRITTEN ***}
%\author{*** FIFTH AUTHOR'S NAME ***}{*** FIFTH AUTHOR'S AFFILIATION WHEN ARTICLE WAS WRITTEN ***}
% Add additional names and affiliations as necessary using above format
%%      ---------------------------------------------------------------------
%%      --------------------------- DEDICATION  (OPTIONAL)------------------- 
%%      ---------------------------------------------------------------------

%       Uncomment the following line to insert a dedication.

%\dedication{ *** DEDICATION *** }        %% Enter dedication between braces.

%%      ------------------------------------------------------------------------------------
%%      --------------------------- ABSTRACT (OPTIONAL)----------------------
%%      ------------------------------------------------------------------------------------

%% ***** UNCOMMENT THE FOLLOWING TO INSERT AN ABSTRACT *****

\begin{abstract}
  Consider the tight binding model of graphene, sharply terminated along an edge ${\bf l}$ parallel to a direction of translational symmetry of the underlying period lattice. We classify such edges ${\bf l}$ into those of "zigzag type" and those of "armchair type", generalizing the classical zigzag and armchair edges. We prove that zero energy/flat band edge states arise for edges of zigzag type, but never for those of armchair type. We exhibit explicit formulas for flat band edge states when they exist. We produce strong evidence for the existence of dispersive (non flat) edge state curves of nonzero energy for most ${\bf l}$.
\end{abstract}

% With AMS-LaTeX, \maketitle follows the abstract
\maketitle   

%%      ---------------------------------------------------------------------
%%      ------------------- TABLE OF CONTENTS (OPTIONAL) --------------------
%%      ---------------------------------------------------------------------

%% ***** IF YOUR PAPER IS OVER 40 PAGES AND YOU WISH TO HAVE A TABLE
%% ***** OF CONTENTS, PLEASE UNCOMMENT THE FOLLOWING LINE

 \tableofcontents

%%      ---------------------------------------------------------------------
%%      ---------------------------- BODY OF PAPER --------------------------
%%      ---------------------------------------------------------------------

\section{Introduction}\label{intro}

 Graphene is a two-dimensional material, consisting of a single atomic layer of carbon atoms centered on a honeycomb lattice, $\HH$, and which extends to the macroscale. It exhibits remarkable electronic  properties, related to the energy spectrum around the Fermi (Dirac) energy, which is well-described by the tight-binding  Hamiltonian \cite{Wallace:47,geim2007rise,RMP-Graphene:09,Katsnelson:12}, a discrete Hamiltonian which acts on $l^2(\HH)$.
This tight-binding model has a band structure consisting of two dispersion surfaces
 which conically touch at {\it Dirac points}; see \cite{FW:12,FW:14}. These spectral characteristics play an important role in the novel conductivity  properties of bulk graphene, and its behavior as a topological insulator in the presence of a magnetic field.
  The relationship between the underlying continuum single electron Schroedinger equation for graphene and the tight-binding limit is investigated in detail in \cite{FLW-CPAM:17}.
For a general discussion of tight binding models and their relationship to the underlying continuum PDEs see, for example, \cite{Dimassi-Sjoestrand:99,Helffer-Sjoestrand:84,SW:21}.

A phenomenon of great interest in Materials Science, and in particular for graphene, is the propagation of energy along a line-defect or {\it edge}. Some references to the extensive literature on edge states are provided in Section \ref{sec:related}.  In this article, an edge is taken to be a sharp termination of the honeycomb lattice along a straight line, ${\bf l}$.
We let  $\HH_\sharp$ denote the set of nodes of $\HH$ that lie in a closed half space on one side  of  ${\bf l}$; a graphene  half-space interfaced with a vacuum. We introduce the (nearest neighbor) edge Hamiltonian $H_\sharp$, which acts on vectors in $l^2(\HH_\sharp)$; the set of square summable vectors, $(\psi_\omega)_{\omega\in\HH}$  with the property that $\psi_\omega=0$ for $\omega\in\HH\setminus\HH_\sharp$; see \eqref{Hshp} for a precise definition.

 Consider the case where the line ${\bf l}$ is in {a} direction of translation invariance of $\HH_\sharp$, {\it i.e.} in the direction of a triangular lattice vector. We call such an interface a {\it rational edge}. In this case, $H_\sharp$ is translation invariant parallel to {\bf l} and we let $\kpar\in\R/2\pi\Z$ denote the 
 associated {\it parallel quasimomentum}. $H_\sharp$ can be decomposed into the independent action of  Hamiltonians $ H_{\sharp,\kpar}$ in $ l^2_\kpar$ ($0\le\kpar<2\pi$), the space of vectors $\psi=(\psi_\omega)_{\omega\in\HH_\sharp}$,  which (i) under translation by a minimal period vector along ${\bf l}$ give $e^{ik}\psi$, and (ii)  decay to zero  as the distance of $\omega$ to ${\bf l}$ tends (within the bulk) to infinity. See equation \eqref{Hk-def} for the definition of an operator $H_\sharp(k)$ which is trivially equivalent to $ H_{\sharp,\kpar}$.
 For $\kpar\in[0,2\pi]$, 
\begin{align*}
&\textrm{we say that $(\psi,E)$ is a $\kpar-$ {\it pseudo-periodic edge state} if $H_{\sharp,\kpar}\psi = E\psi$ with $\psi\in l^2_\kpar$.}
\end{align*}
Thus, $E$ is a $\kpar-$ pseudo-periodic edge state eigenvalue if $E$ is in the $l^2_\kpar$ point spectrum of $H_{\sharp,\kpar}$. This paper studies the dependence of the spectrum $H_{\sharp,\kpar}$
 on the edge and on $\kpar$. 

If there are $\kpar-$ pseudo-periodic states with energy $E(\kpar)$ for all $\kpar$ in some subinterval of $[0,2\pi]$, then a continuous superposition 
of these edge states is an {\it edge wave-packet}, which is localized along and transverse to the edge, and whose large time evolution  is determined by the properties of the edge state energy curve $\kpar\mapsto E(\kpar)$.  It is therefore of interest
 to determine the subsets of $[0,2\pi]$ for which edge states exist, and the properties of the corresponding edge state curves. This question has been previously investigated for the best known edge-orientations: the classical {\it zigzag edge} (ordinary and bearded) and the {\it armchair edge}; see Figure \ref{fig:classical_edges} and references cited in Section \ref{sec:related}. 

  In particular, for the classical zigzag edge $H_{\sharp,\kpar}$ is known to have  $E=0$ energy  edge states:\\
 (a) for all $\kpar\in(2\pi/3,4\pi/3)$ in the case of  ordinary zigzag edges, and\\
  (b) for all $\kpar\in[0,2\pi]\setminus[2\pi/3,4\pi/3]$ in the case of bearded zigzag edges.\\
In the zigzag-edge case,  the Hamiltonian $H_\sharp$ is said to have a {\it flat band} of edge states.   In contrast, in the classical armchair case, $H_{\sharp,\kpar}$ does not support zero energy edge states for any $\kpar\in[0,2\pi]$. 
 Furthermore, neither the classical zigzag edge nor the classical armchair edge supports \underline{non-zero} energy edge states; see \cite{Dresselhaus-etal:96,Graf-Porta:13,FW:20};
  see Figure \ref{fig:edge_states_classical}.   
   The classical zigzag edge state flat band spectra are given by well-known calculations presented, for example, in \cite{Dresselhaus-etal:96,Graf-Porta:13,FW:20}; see also Remark \ref{classical-es}.
 
\nit {\it  In this paper we present results on the spectrum of $H_\sharp$ acting in $l^2(\HH_\sharp)$ for arbitrary rational edges. In particular, we give a complete analysis of the existence and non-existence of zero energy / flat band edge states for arbitrary rational edges. We also present strong evidence that for general zigzag-type and armchair-type edges,  there are non-zero energy (dispersive) edge states curves; {see also \cite{Jaskolski:11}}.
}

 \begin{remark}[Flat bands]\label{flatband}
Consider a wave-packet constructed via continuous superposition of edge states within a zero energy flat band. This wave-packet  will not transport because the group velocity, $E^\prime(\kpar)$, vanishes and it will {neither spread nor decay to zero} because the curvature of the dispersion relation, $E^{\prime\prime}(\kpar)$, vanishes. That spatial concentration without dispersion or transport leads, in the condensed matter physics setting, to enhancement of electron-electron interactions. Additionally, a flat or nearly flat band implies a very high density of states,  which has implications for light-matter interactions \cite{CT:04}.
  \end{remark}
 
  \subsection{Summary of results}\label{sec:outline-summ}
     Let the equilateral triangular lattice be given by $\Lambda=\Z\ocirc{\bv}_1\oplus\Z\ocirc{\bv}_2$ with 
$\ocirc{\bv}_1 = \begin{pmatrix} \frac{\sqrt3}{2} { } & \frac12\end{pmatrix}^\top$ and 
$\ocirc{\bv}_2 = \begin{pmatrix} \frac{\sqrt3}{2} { } & -\frac12\end{pmatrix}^\top$.   
  {  The honeycomb lattice, {\bf $\HH$}, is the union of two interpenetrating translates of  $\Lambda$. In Figure \ref{fig:notation}
   the two triangular sublattices are represented as $A-$ sites (blue) and $B-$ sites (red).
     A {\it rational edge} is specified by a line, ${\bf l}$, in $\R^2$ in the direction 
   \[ \bv_1=a_{11}\ocirc{\bv}_1
    + a_{12}\ocirc{\bv}_2,\]
     where $a_{11}$ and $a_{12}$ are relatively prime integers. We consider the terminated structure $\HH_\sharp$ consisting
     of all vertices in $\HH$ which are in a closed half-plane on one side of ${\bf l}$; see Figure \ref{fig:notation_edges31}. 
     The row of $A-$ sites in $\HH_\sharp$ which is closest to ${\bf l}$ is the set of {\it frontier $A-$ sites} and the row of $B-$ sites in $\HH_\sharp$ which is closest to ${\bf l}$ is the set of {\it frontier $B-$ sites}.  We denote by  $D_A$ and $D_B$ the distance from any frontier $A-$ site, respectively $B-$ site, to
      the line  {\bf l}.    }
     %
%     \footnote{\textcolor{blue}{Proposed replacement for Figure \ref{fig:notation}: a 2 panel figure with  the bulk honeycomb (current Figure \ref{fig:notation}) in a left panel and, in a right panel, a figure showing a truncated   honeycomb along an edge (like in Figure \ref{fig:notation_edges31} but showing only the line {\bf l} with no other labels.}\sfcomment{Please see Figure \ref{fig:theedges31} and new Figure \ref{fig:notation_edges31} }}
     %

    \begin{enumerate}  

 \item     {
  {\it Zigzag-type and Armchair-type edges, Section \ref{sec:rat-edge}:} There are two general classes of edges: zigzag-type (ZZ) and armchair-type (AC). Armchair-type edges are those for which $D_A=D_B$.  Zigzag-type edges are those for which $D_A\ne D_B$ .
These geometric conditions correspond to the following arithmetic conditions: an edge is of AC-type if and only if  $a_{11}-a_{12}\equiv 0\ {\rm mod}\ 3$ and an edge is of ZZ-type if and only if $a_{11}-a_{12}\equiv \pm1\ {\rm mod}\ 3$; see
Definition \ref{def:zz-ac}  in Section \ref{sec:rat-edge} and Proposition \ref{DADB} {for the proof of this correspondence}.
 }

{ Further, for any zigzag edge we have one of the following two cases:
 \begin{align*}
 &\textrm{{\it Balanced} zigzag edge:}\quad  |D_A-D_B| =  \frac13\frac{\sqrt 3}{2}|\bv_1|^{-1},\\
 &\textrm{{\it Unbalanced} zigzag edge:}\quad  |D_A-D_B|= \frac23\frac{\sqrt 3}{2}|\bv_1|^{-1};
 \end{align*}
see Definition \ref{def:ZZAC} and Proposition \ref{DADB}.
The classical armchair, and classical zigzag edges balanced (aka ordinary) and unbalanced (aka bearded) are displayed in Figure \ref{fig:classical_edges}.
} 
\item {{\it Existence and non-existence of zero energy / flat band edge states, Section \ref{sec:0energy}:}\ 
Our main result is:} 

\begin{theorem} \label{main-intro}
{
Assume $\kpar\in[0,2\pi]$. 
\begin{enumerate}
\item For any armchair edge ($D_A=D_B$), there are no zero-energy edge states.
\item All zigzag edges ($D_A\ne D_B$) support a ``flat band'' of zero-energy edge states for $\kpar$ varying in a proper quasi-momentum subset of $(0,2\pi)$. These states are supported exclusively on the $A-$ sites  
 of $\HH_\sharp$  ($A-$site ES) or on $B-sites$ of $\HH_\sharp$ ($B-$site ES). A complete classification is given in the following table:
\end{enumerate}
}
\begin{table}[htbp]
\begin{center}
\begin{tabular}{ |c|c|c|c|c| } 
\hline
 {\ } & {\footnotesize  $D_A<D_B$}&  
 {\footnotesize  $D_B<D_A$}\\
\hline
{\footnotesize Balanced}  & {\footnotesize A-site ES for $\kp\in(\frac{2\pi}{3},\frac{4\pi}{3})$} & {\footnotesize B-site ES for $\kp\in(\frac{2\pi}{3},\frac{4\pi}{3})$} \\ 
\hline
{\footnotesize Unbalanced} &  {\footnotesize {A-site ES for $\kp\in[0,2\pi]\setminus[\frac{2\pi}{3},\frac{4\pi}{3}]$}} & 
{\footnotesize {B-site ES for $\kp\in[0,2\pi]\setminus[\frac{2\pi}{3},\frac{4\pi}{3}]$}}\\
\hline
\end{tabular}
\caption{{\small{$0-$energy / flat band edge states for all rational zigzag edge geometries}}}\label{tab:table-intro}
\end{center}\end{table}
\end{theorem}
\medskip  

 Figures  \ref{fig:edge_states_classical} and \ref{fig:edge_states1} display $l^2_\kpar$ spectra of $H_{\sharp,\kpar}$ vs. $\kpar$ for several choices of rational edges.
In each panel,  the intersection of the vertical slice, corresponding to a fixed $k$, with the blue regions is the $l^2_\kpar$ spectrum of $H_{\sharp,\kpar}$.
{ The center and right panels of Figures \ref{fig:edge_states_classical} and \ref{fig:edge_states1}   display edge spectra for different choices of zigzag-type edge. 
Each of these panels shows a zero energy flat band over a proper subset of $[0,2\pi]$. 
 Part (a) of Theorem \ref{main-intro} states that for the AC-type edges, there are no zero energy edge states. The left panels of Figures \ref{fig:edge_states_classical} and  \ref{fig:edge_states1} show the edge spectra of  armchair edges. }
\item {\it Representation formulae for zero energy / flat band edge states (Section \ref{formulae}):} For zigzag type edges, and for $\kpar$ varying in the relevant subintervals of $[0,2\pi]$, we  present Fourier and rational function representation formulae for the zero energy / flat band edge states. {See Theorem \ref{0en-reps}.}
\item {\it Dispersive edge states (Section \ref{numerics}):} Through careful numerical computations of eigenvalue problems  for $H_{\sharp,\kpar}$,  we present strong evidence for the existence of non-zero energy dispersive edge state curves. 

In particular, our numerical investigations  strongly suggest:\\
 (i) Except for the classical zigzag and armchair edges, there exist dispersive (non-flat) edge state curves, which bifurcate from zero energy at  $\kpar={2\pi}/{3}$ or $4\pi/3$  for zigzag-type
edges and from zero energy at  $\kpar=0$ or $2\pi$ for the armchair-type case. This phenomenon is displayed in Figure  \ref{fig:edge_states1}.\\
 (ii) For a sequence of edges defined by $\bv_1=a_{11}\ocirc{\bv}_1 + \ocirc{\bv}_2$ ($a_{12}=1$), we find that as $a_{11}$ increases, the number of curves bifurcating from these points increases; see Figures \ref{fig:edge_states_armchair} and \ref{fig:edge_states_zigzag}. However for the sequence of edges defined by $\bv_1=a_{11}\ocirc{\bv}_1 + a_{12}\ocirc{\bv}_2$ where $(a_{11},a_{12})$ are consecutive Fibonacci numbers, we find no evidence for an increasing number of such curves; see Figure \ref{fig:edge_states_Fib}.~\\
  (iii) { Among the edges investigated, the classical armchair edge and its rotations appear to be the only edges for which there are no edge states at all. All others investigated appear to have some edge states: either zero-energy flat bands or dispersive non-zero energy curves; see Figure~\ref{fig:edge_states_classical}.}
\end{enumerate}

     \begin{figure}[htbp]
    	\begin{center}
 		\begin{tikzpicture}
			
 			\begin{scope}[shift={(-7.0,0)},scale=0.7,transform shape]	
 				\node at (0,0) {\includegraphics[height=5cm]{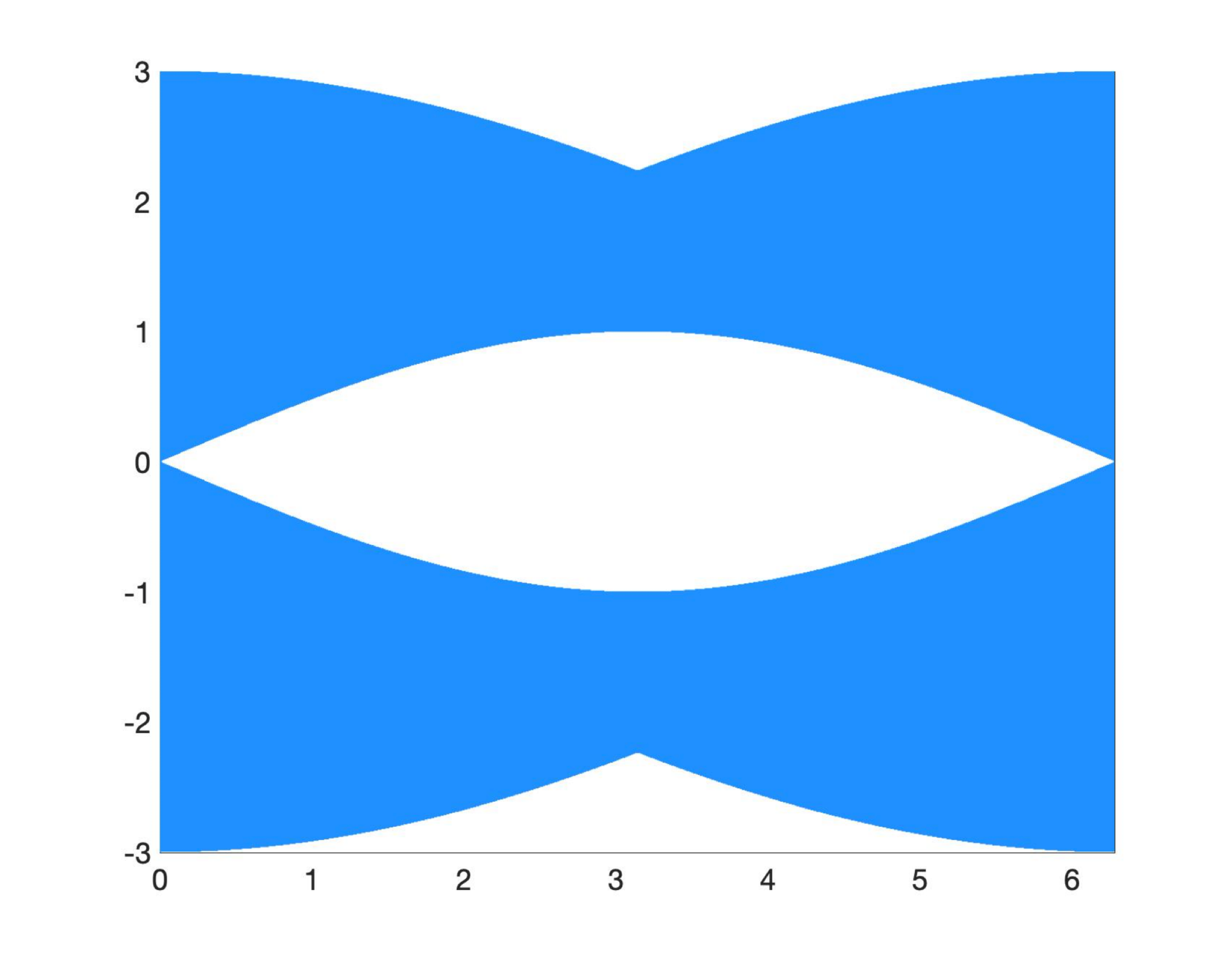}};
 				\draw (0.,2.5) node{$a_{11}=1,\;a_{12}=1$};
				\draw (-3,0) node{$E$};
				\draw (0,-2.5) node{$k$};
 			\end{scope}

 			\begin{scope}[shift={(-3,0)},scale=0.7,transform shape]	
 				\node at (0,0) {\includegraphics[height=5cm]{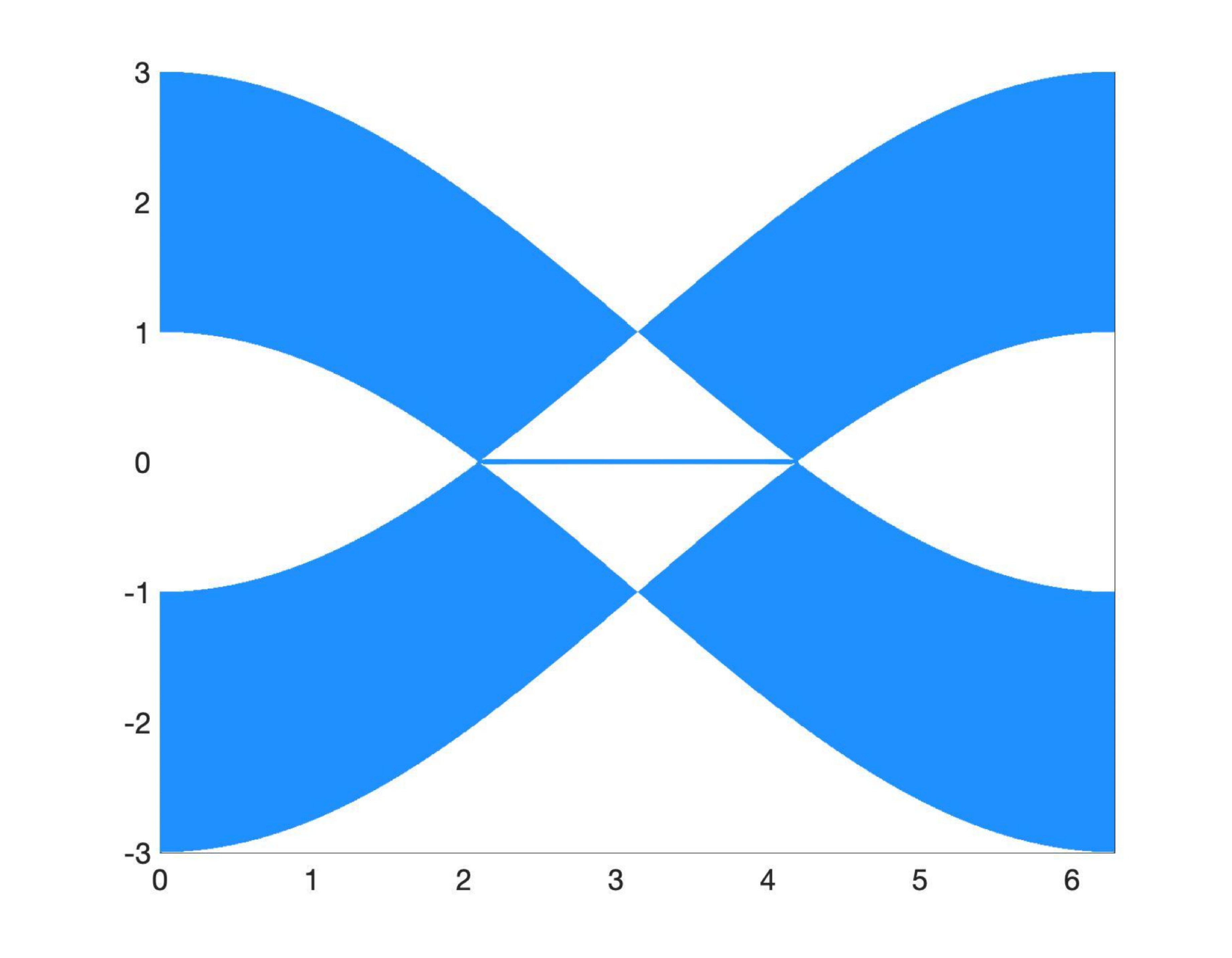}};
 				\draw (0.,2.5) node{$a_{11}=1,\;a_{12}=-1$ (balanced)};
				\draw (0,-2.5) node{$k$};
 			\end{scope}
			
 			\begin{scope}[shift={(1.0,0)},scale=0.7,transform shape]	
 				\node at (0,0) {\includegraphics[height=5cm]{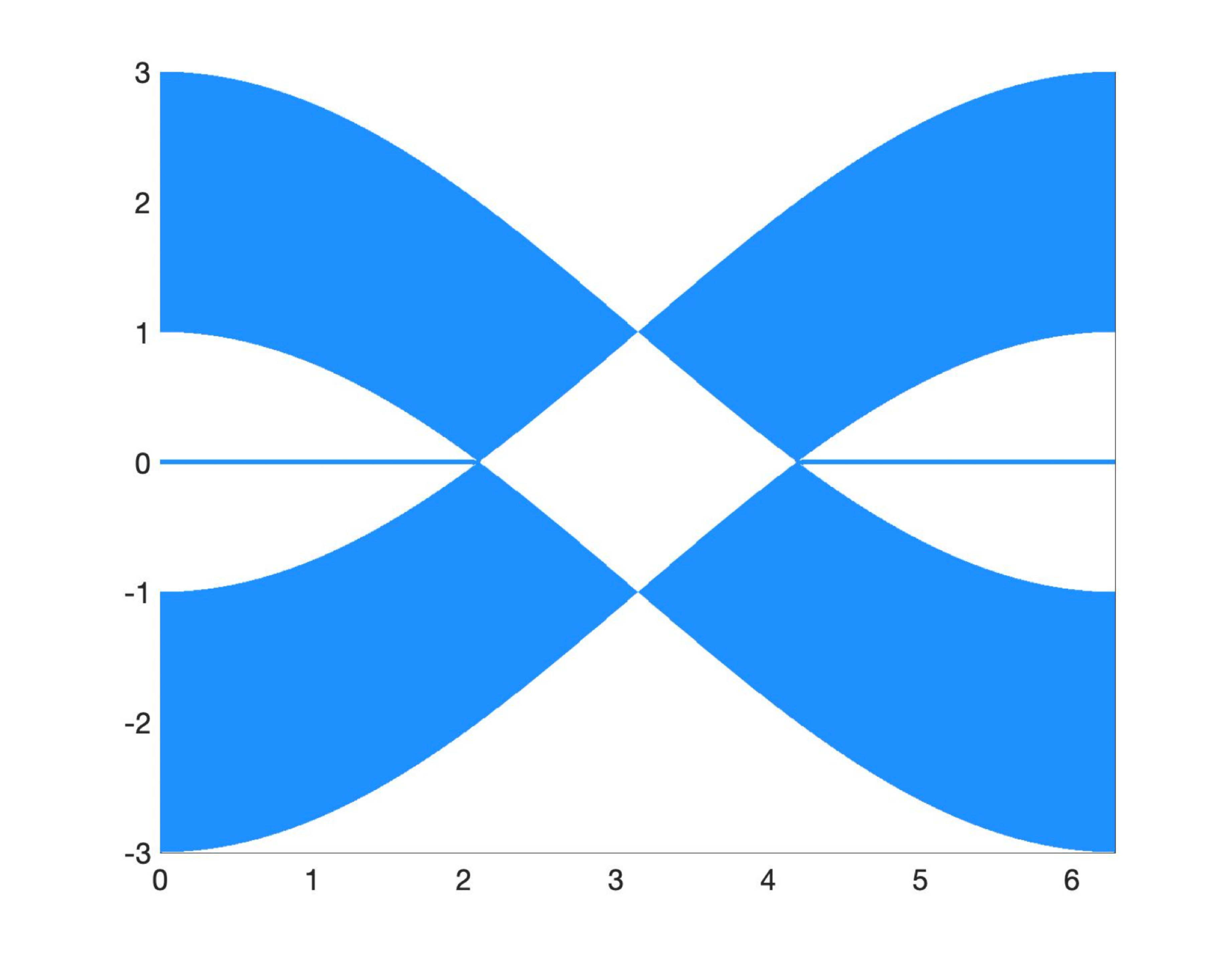}};
 				\draw (0.,2.5) node{$a_{11}=1,\;a_{12}=-1$ (unbalanced)};
				\draw (0,-2.5) node{$k$};
 			\end{scope}
			
 		\end{tikzpicture}
 		\caption{$l^2_\kpar$ spectrum of $H_{\sharp,\kpar}$ versus $\kpar$ for several choices of edges: (i) $(a_{11},a_{12})=(1,1)$, the classical armchair edge, (ii)
		$(a_{11},a_{12})=(1,-1)$, the classical zigzag balanced (aka ordinary) edge, (iii) $(a_{11},a_{12})=(1,-1)$, the classical zigzag unbalanced (aka bearded) edge.}
 		\label{fig:edge_states_classical}
 	\end{center}
 \end{figure}
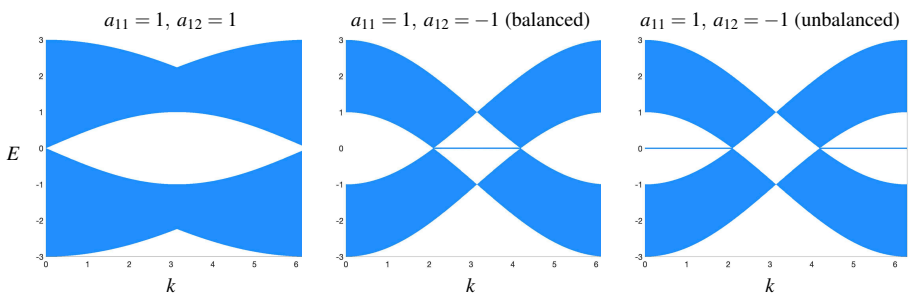 
 
     \begin{figure}[htbp]
    	\begin{center}
 		\begin{tikzpicture}
			
 			\begin{scope}[shift={(-7.0,-1)},scale=0.7,transform shape]	
 				\node at (0,0) {\includegraphics[height=5cm]{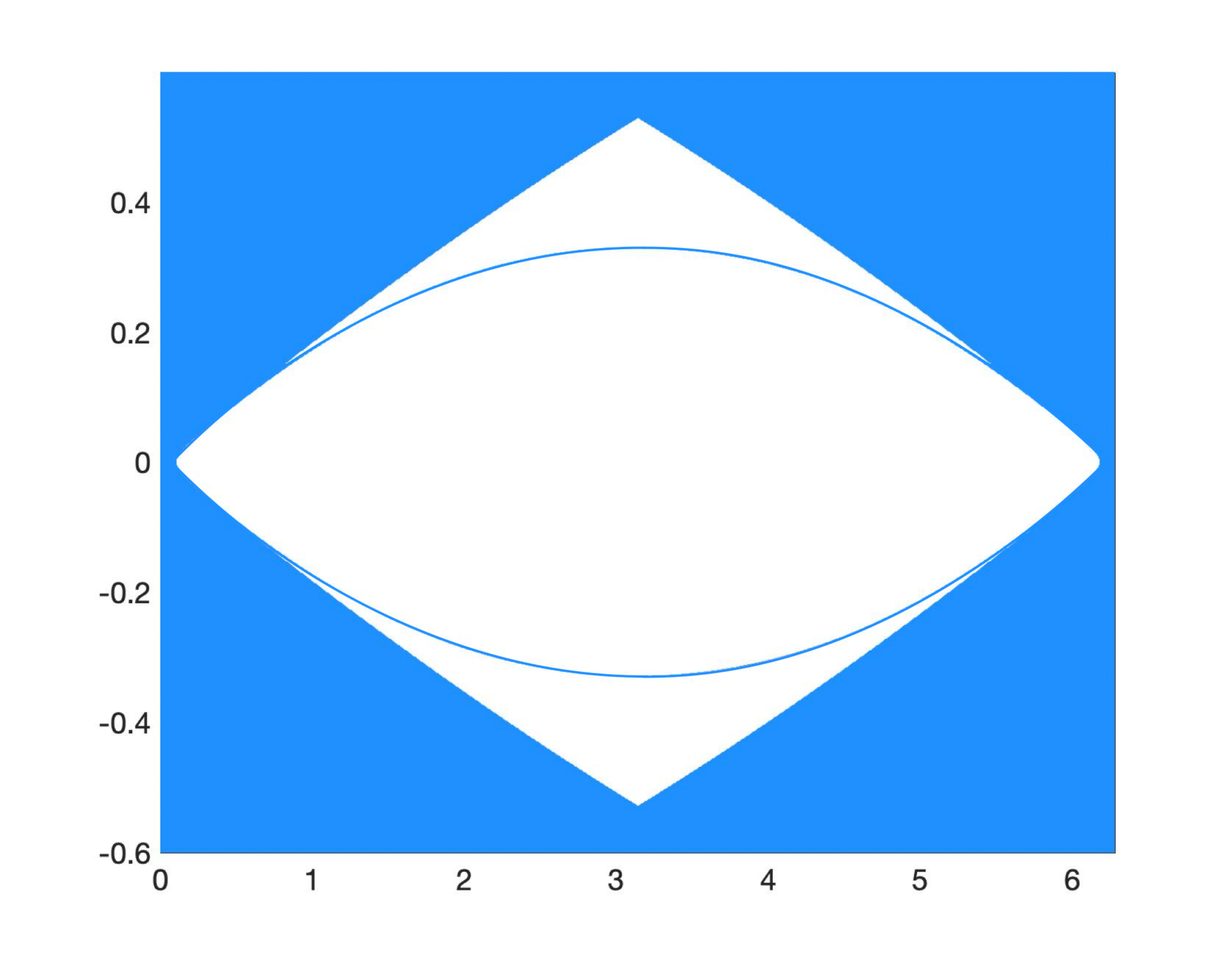}};
 				\draw (0.,2.5) node{$a_{11}=4,\;a_{12}=1$};
				\draw (-3,0) node{$E$};
				\draw (0,-2.5) node{$k$};
 			\end{scope}

 			\begin{scope}[shift={(-3,-1)},scale=0.7,transform shape]	
 				\node at (0,0) {\includegraphics[height=5cm]{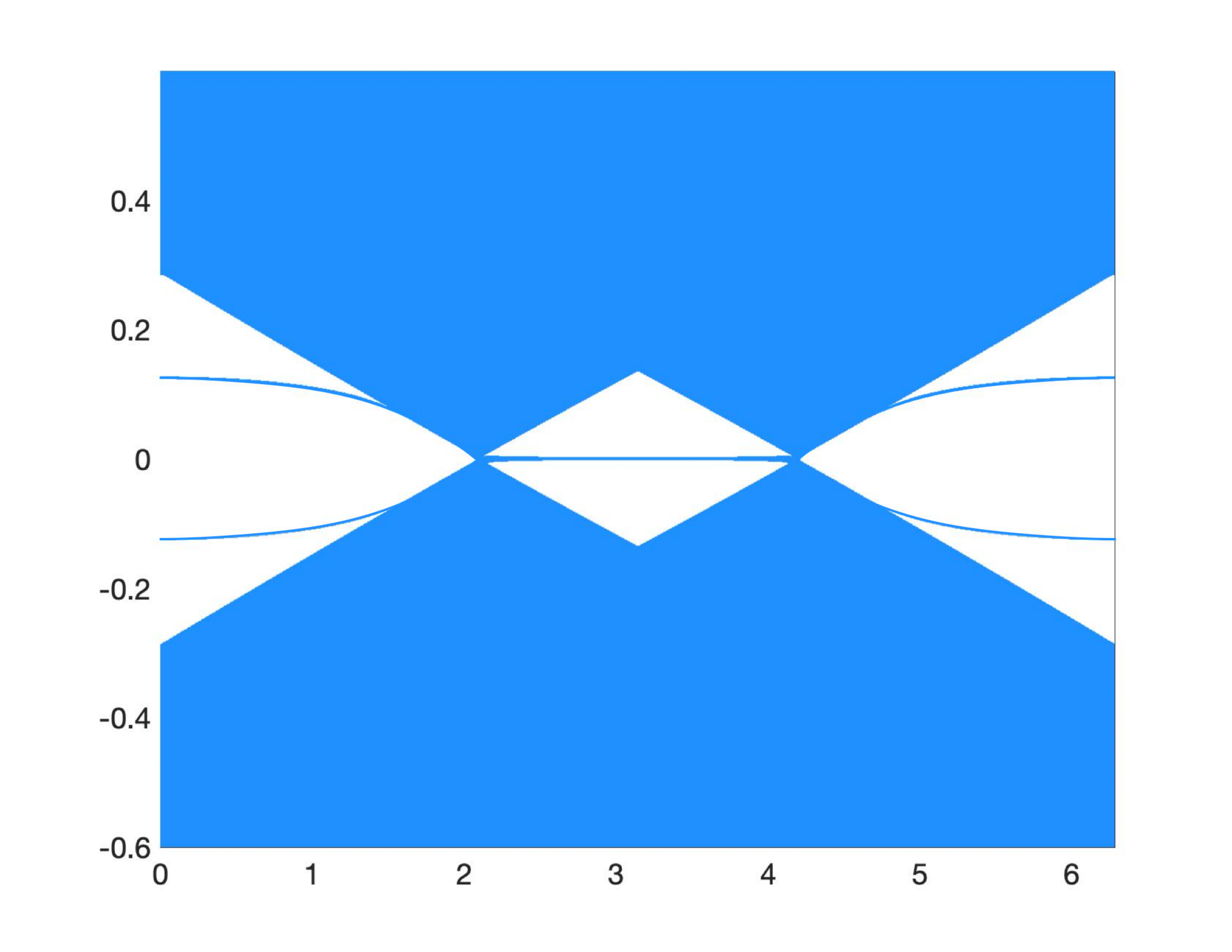}};
 				\draw (0.,2.5) node{$a_{11}=6,\;a_{12}=1$ (balanced)};
				\draw (0,-2.5) node{$k$};
 			\end{scope}
			
 			\begin{scope}[shift={(1,-1)},scale=0.7,transform shape]	
 				\node at (0,0) {\includegraphics[height=5cm]{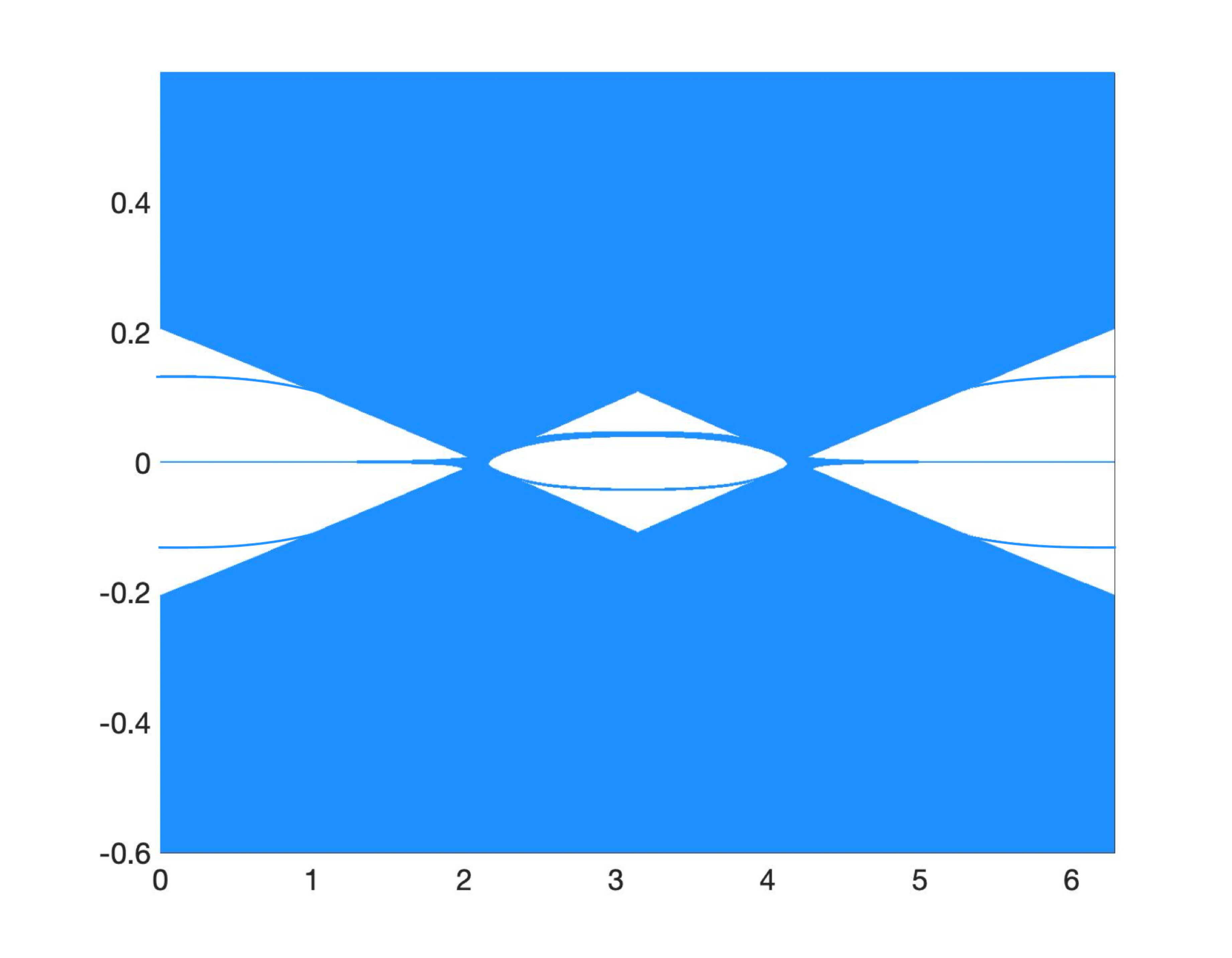}};
 				\draw (0.,2.5) node{$a_{11}=8,\;a_{12}=1$ (unbalanced)};
				\draw (0,-2.5) node{$k$};
 			\end{scope}
 		\end{tikzpicture}
 		\caption{$l^2_\kpar$ spectrum of $H_{\sharp,\kpar}$ versus $\kpar$ for several choices of edges: (i) $(a_{11},a_{12})=(4,1)$, an AC-type edge, (ii) $(a_{11},a_{12})=(6,1)$, a ZZ-type {balanced} edge and (iii) $(a_{11},a_{12})=(8,1)$, a ZZ-type {unbalanced} edge. }
 		\label{fig:edge_states1}
 	\end{center}
 \end{figure}
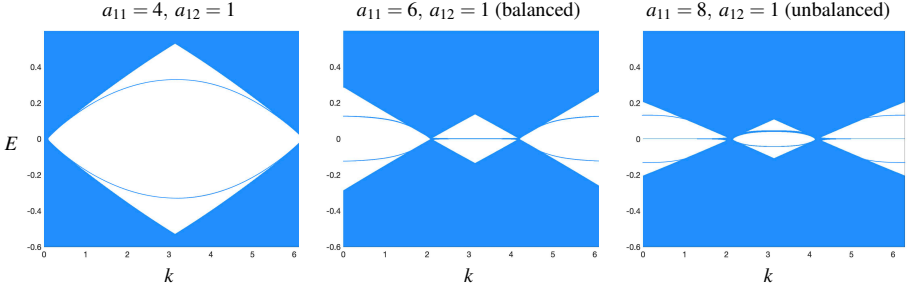

\subsection{Relation to previous work and some open questions} \label{sec:related}
The tight binding model on a honeycomb lattice plays a central role in the modeling of graphene and related materials; see, for example, 
\cite{RMP-Graphene:09,Katsnelson:12}. It was first recognized in 
\cite{Dresselhaus-etal:96,Fujita-etal:96}
that the existence of edge states depends on the shape of the edge.
The tight binding edge Hamiltonian has most commonly
been studied for the classical zigzag and armchair edges; see Figure  \ref{fig:edge_states_classical}. In \cite{Graf-Porta:13}  it is proved that the classical armchair edge supports no edge modes (zero or nonzero energy). 
Our rigorous analytical results on edge states for general rational edges, outlined in Section \ref{sec:outline-summ}, appear to be new.

 There are studies in the physics literature of rational edges \cite{Akhmerov-Beenakker:08,delplace2011zak,Jaskolski:11}. The definitions of edges used in these works differ. Let us now describe these classes of edges, and contrast them with the class of edges studied in this article. Recall that our edges are boundaries of structures $\HH_\sharp\subset\HH$, comprised of all honeycomb vertices in a closed half-space determined by a line parallel to $\bv_1$, where $\bv_1$ is any vector in the triangular lattice $\Lambda$. Here, we shall refer to such edges as {\it half-space termination edges}.

The notion of {\it minimal edge} was  introduced in \cite{Akhmerov-Beenakker:08}. 
Minimal edges have the following properties: 
\begin{itemize}
\item the structure is periodic with period vector
 $\bv_1=a_{11}\ocirc{\bv}_1+a_{12}\ocirc{\bv}_2$, where $a_{11}, a_{12}>0$, 
 \item  no site of $\HH_\sharp$ has two nearest neighbors
in $\HH\setminus\HH_\sharp$, 
\item no site of $\HH\setminus\HH_\sharp$ has two nearest neighbors
in $\HH_\sharp$, 
\item within a period, there are precisely $a_{11}+a_{12}$ frontier sites, i.e. sites of $\HH_\sharp$ with neighbors in $\HH\setminus\HH_\sharp$.
\end{itemize}
It is suggested in \cite{Akhmerov-Beenakker:08} that such minimal edge structures are energetically preferred.
 In general, a minimal edge need not be of the half-space termination type studied here.

 The class of {\it modified edges}, arising from the periodic attachment of atoms and bonds to minimal edge  atoms at frontier sites of $\HH_\sharp$, is studied in \cite{Jaskolski:11}.  The edges studied here may be either minimal or modified. 
 
  In \cite{delplace2011zak}, edges which arise from a periodic pattern of displacements of a selected dimer (pair of nearest neighbor sites) are studied, with period  vector $\bv_1=a_{11}\ocirc{\bv}_1+a_{12}\ocirc{\bv}_2$.  In the case where  $a_{11}, a_{12}>0$, this class of edges is asserted to be precisely the class  of minimal edges, as defined in  \cite{Akhmerov-Beenakker:08}.
 There is overlap between our class of half-space termination edges and those
discussed in \cite{delplace2011zak}, but neither class includes the other.
  
We now compare our results with those of  \cite{Akhmerov-Beenakker:08,delplace2011zak,Jaskolski:11}.
 The main goal of \cite{Akhmerov-Beenakker:08} is to derive continuum boundary conditions for  an effective Dirac operator, associated with a minimal rational edge.
Toward this goal, they consider the tight binding model for parallel quasimomentum $\kpar\approx0$.
The article \cite{delplace2011zak} postulates a {\it bulk-edge correspondence}: for a fixed edge, the dimension of the subspace of zero energy edge states is equal to the winding number of the Zak phase along a one-dimensional Brillouin zone determined by the edge orientation. 
The authors of \cite{delplace2011zak}  apply this approach to obtain an expression, derived previously in \cite{Akhmerov-Beenakker:08}, for the density of edge states.
%
%\footnote{\textcolor{blue}{
%Should we be making a rigorous statement about the density of states in our setting?}}
%
 The reader should note that the results of \cite{delplace2011zak} are displayed in terms of a scaled (edge-dependent) parallel quasi-momentum range, while the range of parallel quasimomenta in the present article is fixed to be $[0,2\pi]$. 
There appears to be agreement between our rigorous results and 
the  results in \cite{delplace2011zak}  for those edges in the overlap of our studies. 
To our knowledge, no previous articles
 rigorously address, for a general class of rational edges, the questions of: which parallel quasimomentum ranges support zero energy edge states; when they exist, whether they are supported on $A-$ or $B-$ sublattice sites;
  or explicit formulas for zero energy edge states when they exist. 

Numerical studies in \cite{Jaskolski:11} indicate that a flat band, for a minimal structure, can give rise to non-zero energy edge state curves when additional sites and bonds are attached to form a modified structure.
  Our numerical investigations give strong evidence  that non-zero energy edge state curves arise in minimal structures themselves.

Many natural open problems arise.  
 (a) Are there states which are bounded and oscillatory parallel to an irrational edge
 and which decay into the bulk? Related to this question is the article \cite{Gontier:21}, which demonstrates that the edge spectrum for a rationally terminated continuum
periodic Schroedinger operator (with Dirichlet boundary conditions) has a band-gap spectrum, while for an irrational termination the gaps are filled with ``edge spectrum''.
 (b) Can one realize an edge state for irrational termination as the limit of a sequence of edge state wave-packets
  (superpositions of edge states) of rationally terminated structures?
  (c)  Do all edge state curves emerge from and terminate in a band crossing?
 (d) Explain the following numerical observations: Along certain sequences of rational edges, for which $|a_{11}|+|a_{12}|$ increases, there is an increasing number of dispersive edge state curves  which bifurcate from the band crossings. However along other sequences we do not see evidence of this effect. (e) For an irrational edge, understand the long term dynamics of a state initially concentrated near the edge.  
 (f) Investigate analogous questions in other tight-binding models, such as the Harper model for an electron on a two-dimensional lattice in the presence of a  constant perpendicular  magnetic field, {\it e.g.} \cite{AMS:90,OA:01,AJ:09,Avila-etal:13} or models of multilayer structures, such as twisted bilayer graphene,  {\it e.g.} \cite{BM:11,TKV:18,BEWZ:20,WL:21}. Parallel questions for  quantum graph models \cite{Kuchment-Post:07,BK:13} would also be of interest.

\subsection{Structure of the paper}
 In Section \ref{sec:setup} we present a mathematical framework for studying the edge state eigenvalue problem for an arbitrary rational edge. \\
  In Section \ref{gen_spectrum} we discuss general properties of the spectrum 
 of the edge Hamiltonian, $H_{\sharp,k}$, e.g. essential spectrum and symmetry properties.\\
 { In Section \ref{sec:0energy} we prove (modulo technical results established later) our main results on the existence of zero-energy / flat band edge states:
   Theorems \ref{th:zigzag}, \ref{th:armchair} and Theorem \ref{main-intro} as their consequence.}
   {
   The zero energy edge state eigenvalue problem in the bulk is a system of decoupled $A-$ site and $B-$ site difference equations with complex $\kpar-$ dependent coefficients. The construction of zero energy edge states requires us to understand the complex roots of two  {\it honeycomb edge polynomials} (associated with $A-$ site edge states and $B-$ site edge states), which are related by a symmetry.}
    We prove that, depending on the parallel quasi-momentum $\kpar$, the number of roots in the open unit disc is either  (a) \underline{less or equal to} or (b) \underline{one more than}   the number of linear homogeneous algebraic  boundary conditions required for a general linear combination of decaying solutions of the difference equation to be an edge state.  In case (b), the algebraic system has one free parameter, which generates a one-dimensional   edge state eigenspace for the relevant value of $\kpar$. This computation is made in Section \ref{sec:HEPs}. 
   \\
    In Section \ref{formulae} we obtain both Fourier and rational function representations of the zero energy flat band edge states, in all cases where they exist. \\
Section \ref{numerics} provides an analytical framework for studying the edge state eigenvalue problem for arbitrary rational edges and general energies, $E$. In contrast to the case of zero energy edge states, the edge state eigenvalue problem for general energies does not decouple into separate difference equations on $A-$ and $B-$ sites. Hence, the relevant polynomial associated with the coupled system of $A-$ and $B-$ site difference equations is twice the degree of the polynomials arising for the uncoupled problem, and is less accessible to direct analysis.
 We present the results of numerical investigations showing non-zero energy dispersive edge states and zero energy / flat bands for representative non-classical zigzag-type edges and non-zero energy dispersive edge states for representative non-classical armchair-type edges.  \\
Finally, the appendices contain detailed computations used in the body of the paper.
  \bigskip
  
  \subsection{Notation}\label{notation}
  \begin{enumerate}
  \item $\Lambda$ equilateral triangular lattice;
  \item $\HH$ honeycomb lattice; $\HH_\sharp$ terminated honeycomb lattice;
  \item $H_{\rm bulk}$ and $H_\sharp$ bulk and edge nearest neighbor tight-binding operators;
  \item $\mathcal{B}(X)$, the space of bounded linear operators on a Banach space $X$;
  \item Pauli matrices: \begin{equation} \sigma_1=\begin{pmatrix} 0&1\\1&0\end{pmatrix},\quad  \sigma_2=\begin{pmatrix} 0&-i\\ i&0\end{pmatrix},\quad
  \sigma_3=\begin{pmatrix} 1&0\\ 0&-1\end{pmatrix}\label{pauli123}\end{equation}
  \item $\mathds{1}_S$ denotes the characteristic function of the condition $S$ .
  \end{enumerate}
% {  \nit{\bf Acknowledgements:}  This research was initiated
%    at a working group on "Irrational edges" at the American Institute of Mathematics (AIM) Workshop on the {\it Mathematics of Topological Insulators}, December 7-11, 2020, which was supported by the  American Institute of Mathematics, the US National Science Foundation, the Simons Foundation and Columbia University.
%   C.L.F. was supported in part by National Science Foundation grant DMS-1700180.
% M.I.W. was supported in part by National Science Foundation grants DMS-1620418 and DMS-1908657 as well as Simons Foundation Math + X Investigator Award \#376319. We warmly thank the participants of the AIM working group, as well as Pierre Delplace, David Gontier and Mikael Rechtsman for very stimulating discussions. 
% }
\section{Mathematical framework}\label{sec:setup}

\subsection{The honeycomb structure}\label{sec:Honeycomb}
We introduce the {\it equilateral triangular lattice} 
\[\Lambda=\Z\ocirc{\bv}_1\oplus\Z\ocirc{\bv}_2,\] with 
\begin{equation}
\ocirc{\bv}_1 = \begin{pmatrix} \frac{\sqrt3}{2}\\ { } \\ \frac12\end{pmatrix}\quad {\rm and}\quad 
\ocirc{\bv}_2 = \begin{pmatrix} \frac{\sqrt3}{2}\\ { }\\ -\frac12\end{pmatrix};
\label{tri-lattice}\end{equation}
see Figure \ref{fig:notation}.
\begin{figure}[htbp]
	\begin{center}
	 \begin{tikzpicture}
	 % \draw[->]({1/sqrt(3)/2},-1/2)--({1/sqrt(3)/2+sqrt(3)/2},0) node[yshift=-3,right] {$\ocirc{\bv}_1$};
% 	 \draw[->]({1/sqrt(3)/2},-1/2)--({1/sqrt(3)/2+sqrt(3)/2},-1) node[yshift=-3,right] {$\ocirc{\bv}_2$};
	 \foreach \n in {-3,-2,-1,0,1,2} 
 	\foreach \m in {-2,-1,0,1,2,3} {
	 \draw ({(\n+\m)*sqrt(3)/2},{(\n-\m)/2}) node{\color{blue} \scalebox{0.8}{\textbullet}};
	 \draw ({1/sqrt(3)+(\n+\m)*sqrt(3)/2},{(\n-\m)/2}) node{\color{red} \scalebox{0.8}{\textbullet}};
	 }
	 \foreach \n in {-2,-1,0,1,2} {
	 \draw ({(4+\n)*sqrt(3)/2},{(\n-4)/2}) node{\color{blue} \scalebox{0.8}{\textbullet}};
	 \draw ({1/sqrt(3)+(\n+4)*sqrt(3)/2},{(\n-4)/2}) node{\color{red} \scalebox{0.8}{\textbullet}};
	 }
	 \foreach \n in {-3,-2,-1,0,1} {
	 \draw ({(-3+\n)*sqrt(3)/2},{(\n+3)/2}) node{\color{blue} \scalebox{0.8}{\textbullet}};
	 \draw ({1/sqrt(3)+(\n-3)*sqrt(3)/2},{(\n+3)/2}) node{\color{red} \scalebox{0.8}{\textbullet}};
	 }
  	\foreach \m in {-2,-1,0,1,2} {
 	 \draw ({(-4+\m)*sqrt(3)/2},{(-4-\m)/2}) node{\color{blue} \scalebox{0.8}{\textbullet}};
 	 \draw ({1/sqrt(3)+(-4+\m)*sqrt(3)/2},{(-4-\m)/2}) node{\color{red} \scalebox{0.8}{\textbullet}};
 	 }
   	\foreach \m in {-1,0,1,2,3} {
  	 \draw ({(3+\m)*sqrt(3)/2},{(3-\m)/2}) node{\color{blue} \scalebox{0.8}{\textbullet}};
  	 \draw ({1/sqrt(3)+(3+\m)*sqrt(3)/2},{(3-\m)/2}) node{\color{red} \scalebox{0.8}{\textbullet}};
  	 }
	 \draw[ultra thick,->](0,0)--({sqrt(3)/2},0.5) node[yshift=1.5,xshift=-2,left] {$\ocirc{\bv}_1$};
	 \draw[ultra thick,->](0,0)--({sqrt(3)/2},-0.5) node[yshift=-0.5,xshift=-2,left] {$\ocirc{\bv}_2$};
	 \node at (-0.3,0) {\color{blue}$\ocirc{\bv}_A $};\node at (1,0) {\color{red}$\ocirc{\bv}_B $};
	 \draw[ultra thick,->]({-3*sqrt(3)/2},-1/2)--({1/sqrt(3)-3*sqrt(3)/2},-1/2) node[yshift=-4,right] {$\be^1$};
	 \draw[ultra thick,->]({-3*sqrt(3)/2},-1/2)--({-1/sqrt(3)/2-3*sqrt(3)/2},0) node[yshift=-2,right] {$\be^2$};
	 \draw[ultra thick,->]({-3*sqrt(3)/2},-1/2)--({-1/sqrt(3)/2-3*sqrt(3)/2},-1) node[yshift=-2,right] {$\be^3$};
	 \end{tikzpicture}
	 \caption{{Bulk honeycomb structure, $\HH$, comprised of $A-$ site (blue) and $B-$ site (red) triangular sublattices.} }
	 \label{fig:notation}
 	\end{center}
\end{figure}
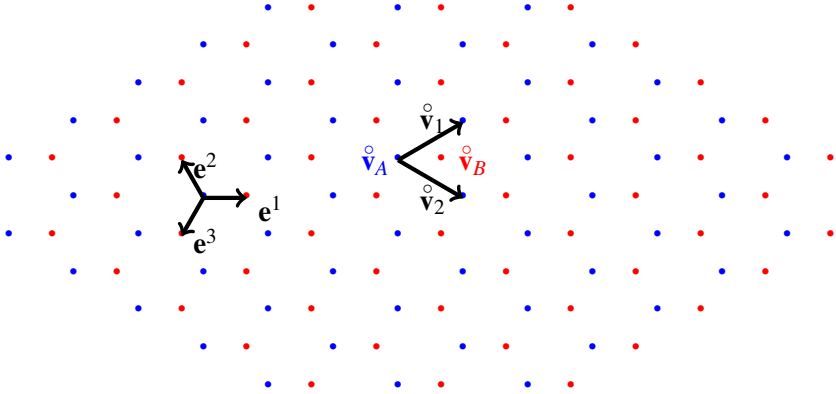
The {\it honeycomb structure} is the set  $\HH= \HH_A\cup \HH_B\subset\R^2$, where
\begin{align*}
\HH_A &= \ocirc{\bv}_A + \Lambda\quad {\rm and}\quad \HH_B = \ocirc{\bv}_B + \Lambda\label{OmegaAB},
\end{align*}
with 
\begin{equation*}
\ocirc{\bv}_A = 0 \quad {\rm and}\quad \ocirc{\bv}_B = \frac13\left(\ocirc{\bv}_1+\ocirc{\bv}_2\right)
\label{vAvB-def}
\end{equation*}
The points of $\HH_A$ and $\HH_B$ are called {\it A-points} and {\it B-points}, respectively. We write $\omega, \tilde\omega, \omega^\prime$ {\it etc.} to denote points of the honeycomb structure.

A given $A-$point $\omega$ has as its three nearest neighbors  in $\HH$ the three $B-$points
\begin{equation}
 \omega+\be^\nu,\quad \nu = 1,2,3\quad [\textrm{nearest neighbor B-points to $\omega\in\HH_A$]},
 \label{nn-Apts} %(4.5)
 \end{equation}
where
\begin{equation}
e^1 = \frac13\left(\ocirc{\bv}_1+\ocirc{\bv}_2\right),\quad e^2 = \frac13\left(\ocirc{\bv}_1-2\ocirc{\bv}_2\right),\quad 
e^3 = \frac13\left(-2\ocirc{\bv}_1+\ocirc{\bv}_2\right).
\label{e-nu}\end{equation}

Similarly, a given 
$B-$point $\omega$ has as its three nearest neighbors in $\HH$ the three $A-$points
\begin{equation}
 \omega-\be^\nu,\quad \nu = 1,2,3\ \quad [\textrm{nearest neighbor A-points to $\omega\in\HH_B$]};
 \label{nn-Bpts}%(5.5)
 \end{equation}
see Figure \ref{fig:notation}.

\subsection{The bulk Hamiltonian}\label{sec:BulkHam}
 We introduce the {\it bulk Hamiltonian}, $\Hbulk$, acting on `` wave functions '' 
\[ \psi = \left(\psi_\omega\right)_{\omega\in\HH} \in l^2(\HH).\]
It is defined by the formulas
\begin{subequations}
\begin{align}
\left(\Hbulk\psi\right)_\omega = \sum_{\nu=1}^3 \psi_{\omega+\be^\nu}\quad {\rm if}\quad \omega\in\HH_A
\label{Hbulk_A}\\ %(6)
\left(\Hbulk\psi\right)_\omega = \sum_{\nu=1}^3 \psi_{\omega-\be^\nu}\quad {\rm if}\quad \omega\in\HH_B
\label{Hbulk_B}  %(7)
\end{align}
\end{subequations}
Thus, $\left(\Hbulk\psi\right)_\omega$ is equal to the sum of $\psi(\omega^\prime)$ over the three nearest neighbors $\omega^\prime$ of $\omega$ in the honeycomb $\HH$.

\subsection{The sharply terminated honeycomb, $\HH_\sharp$, and the Hamiltonian $H_\sharp$}\label{sec:sharp_term}
Let ${\bf l}$ denote a line in $\R^2$, and $\HH_\sharp\subset\HH$ be the subset of points of $\HH$ that lie in {the closed half-space of $\R^2$ lying on one side of ${\bf l}$}. We refer to the line ${\bf l}$ as {\it the edge}. The Hamiltonian for the sharply terminated honeycomb will act on wave functions $\left(\psi_\omega\right)_{\omega\in\HH}$, such that $\psi_\omega=0$  for all
 $\omega\notin\HH_\sharp$. We shall abuse notation and denote the space of all such vectors $l^2(\HH_\sharp)$.
 
 The Hamiltonian, $H_\sharp$, acting on vectors $\left(\psi_\omega\right)_{\omega\in\HH}\in l^2(\HH_\sharp)$ is defined by the formulas
 \begin{subequations}
 \begin{align}
\left( H_\sharp\psi\right)_\omega & = \sum_{\nu=1}^3 \psi_{\omega+\be^\nu}\quad
 \textrm{if $\omega\in\HH_\sharp\cap\HH_A$} % (8)
\label{Hshp8} 
\\   \left( H_\sharp\psi\right)_\omega & = 0\quad
 \textrm{if $\omega\in\HH_A\setminus\HH_\sharp$} %(9)
\label{Hshp9} 
\\   \left( H_\sharp\psi\right)_\omega & = \sum_{\nu=1}^3 \psi_{\omega-\be^\nu}\quad
 \textrm{if $\omega\in\HH_\sharp\cap\HH_B$} %(10)
 \label{Hshp10}
 \\  
 \left( H_\sharp\psi\right)_\omega & = 0\quad \textrm{if $\omega\in\HH_B\setminus\HH_\sharp$} 
 \label{Hshp11}%(11)
 \end{align}
 \label{Hshp}
 \end{subequations}
 {The Hamiltonian $H_\sharp$ in  \eqref{Hshp} arises as the limit of the $1-$ electron model in 
  the strong binding regime; see, for example, \cite{Dimassi-Sjoestrand:99,Helffer-Sjoestrand:84,SW:21}.
 The operator $H_\sharp$ is bounded and self-adjoint  on $l^2(\HH_\sharp)$;
  see Section \ref{sec:fibers} and relation \eqref{Hconj}.}
 
 \subsection{Rational edges}\label{sec:rat-edge}
 We  call the edge ${\bf l}$ a {\it rational edge} if it is parallel to a non-zero vector 
$a_{11}\ocirc{\bv}_1 + a_{12}\ocirc{\bv}_2$  in the triangular lattice {where $a_{11}$ and $a_{12}$ are two integers}. This paper confines itself to the case of rational edges.
 
The above integers $a_{11}$ and $a_{12}$ may be taken to be relatively prime, in which case there exist
  integers $a_{21}$, $a_{22}$ (not unique) such that
  \[ \det\begin{pmatrix} a_{11} & a_{12}\\ a_{21} & a_{22}\end{pmatrix}\ = \ \pm 1.\]
We next introduce
\begin{subequations}
 \begin{align}
 \bv_1 &= a_{11} \ocirc{\bv}_1 + a_{12} \ocirc{\bv}_2,\label{v1}\\
 \bv_2 &= a_{21} \ocirc{\bv}_1 + a_{22} \ocirc{\bv}_2.\label{v2}
 \end{align}
 \label{vjs}
 \end{subequations}

 The vector $\bv_1$ is parallel to the edge ${\bf l}$, and the vector $\bv_2$ is transverse to  ${\bf l}$. After possibly changing $(a_{21},a_{22})$ to $(-a_{21},-a_{22})$ we can always take $\bv_2$ to point into the terminated bulk, $\HH_\sharp$. Then, after possibly changing $(a_{11},a_{12})$ to $(-a_{11},-a_{12})$ we can achieve the conditions
 \begin{subequations}
\begin{align}
 &\det\begin{pmatrix} a_{11} & a_{12}\\ a_{21} & a_{22}\end{pmatrix}\ = \  +1,
 \label{det1}\\
& \textrm{$\bv_1$ is parallel to the edge, ${\bf l}$,} \label{v1-par}\\
& \textrm{$\bv_2$ points into $\HH_\sharp$}. \label{v2-trans}
 \end{align}\label{Acond}
 \end{subequations}
 The integer vector $(a_{11},a_{12})$ is uniquely specified by conditions \eqref{vjs}, \eqref{Acond} for a fixed 
$\HH_\sharp$, while the vector $(a_{21},a_{22})$  is uniquely specified modulo translates by integer multiples of $(a_{11},a_{12})$. Our results are independent of the ambiguity in the choice of $(a_{21},a_{22})${; see Appendix \ref{app:v_2} for more details}. 

\noindent {\bf Convention:} {\it From now on we fix $(a_{ij})$ as in \eqref{vjs}, \eqref{Acond}.}
By \eqref{det1}, we have 
\begin{equation}\label{eq:change_basis}
    \begin{array}{l}
    \ocirc{\bv}_1 = a_{22} \bv_1 - a_{12} \bv_2,\\
    \ocirc{\bv}_2 = -a_{21} \bv_1 + a_{11} \bv_2.
    \end{array}
\end{equation}
Hence, 
\begin{equation}
\textrm{the set $\{\bv_1,\bv_2\}$ is a basis for the triangular lattice; $\Lambda=\Z\bv_1+\Z\bv_2 = \Z\ocirc{\bv}_1+\Z\ocirc{\bv}_2$.}\label{basis}\end{equation}

Because $\HH_\sharp$ consists of all points of $\HH$ that lie on one side of the line ${\bf l}$, we know from \eqref{vjs}, \eqref{Acond}  that 
\begin{equation}
\HH_\sharp\ =\HH\cap \Big\{ x_1\bv_1+x_2\bv_2;\;\quad x_1,x_2\in\R;\quad x_2\ge\beta \Big\}, 
\label{HHsharp} %(17)
\end{equation}
for some real number $\beta$, {with the line {\bf l} (the edge)  given by:
 \begin{equation}
  \{x_1\bv_1+x_2\bv_2\ :\ x_2=\beta,\  x_1\in\R\}.
  \label{theedge}
  \end{equation}
 }

In view of \eqref{basis} we prepare to re-express the honeycomb structure, $\HH$, the terminated honeycomb structure, $\HH_\sharp$, and the Hamiltonians $\Hbulk$ and $H_\sharp$ in terms of the lattice basis $\{\bv_1,\bv_2\}$. 
The detailed calculations are presented in Appendix \ref{sec:vs_circ_to_vs}.

Let us partition the plane $\R^2$ into the parallelograms defined for all $m,n\in\Z$ by
\begin{equation}
\Gamma(m,n) = \Big\{ \bv=x_1\bv_1+x_2\bv_2\ :\ x_1\in\left(m-\frac12,m+\frac12\right],\ x_2\in\left(n-\frac12,n+\frac12\right] \Big\}.
\label{Gamma-mn}
\end{equation}
We show in Appendix \ref{vAvB-app} that each parallelogram $\Gamma(m,n)$ contains the single $A-$point  $\bv_A+m\bv_1+n\bv_2$ and the single $B-$point 
 $\bv_B+m\bv_1+n\bv_2$, {where $\bv_A, \bv_B\in\HH $ lie in the parallelogram $\Gamma(0,0)$. }
{
Introduce $k_1, s_1, k_2, s_2$ such that
 \begin{subequations}
 \begin{align}
 a_{22}-a_{21} & = 3k_1 + s_1\quad {\rm with}\quad k_1\in\Z,\quad s_1\in\{-1,0,1\},\\
 a_{11}-a_{12} & = 3k_2 + s_2\quad {\rm with}\quad k_2\in\Z,\quad s_2\in\{-1,0,1\}.
 \end{align}
 \label{k1k2s1s2}
 \end{subequations}
 The points $\bv_A$ and $\bv_B$ are given by
 \begin{subequations}
 \begin{align}
 \bv_A&=0\in \Gamma(0,0), \label{vA-1}\\ %(21)
 \bv_B&=\frac13\left( s_1\bv_1 + s_2\bv_2 \right)\in\Gamma(0,0).\label{vB-1} %(22)
 \end{align}
 \label{vAvB-1} %(21)
 \end{subequations}
 }
{See Figures \ref{fig:theedges31} and \ref{fig:notation_edges31} for an example ($a_{11}=3,a_{12}=1$).}

Given an $A-$point in $\Gamma(m,n)$: $\bv_A+m\bv_1+n\bv_2 = \bv_B+m\bv_1+n\bv_2+(\bv_A-\bv_B)$, by \eqref{nn-Apts}, its $3$ nearest
 neighbor $B-$points are $\bv_B+m\bv_1+n\bv_2+(\bv_A-\bv_B) + \be^\nu$, $\nu=1,2,3$. And, given any $B-$point
in $\Gamma(m,n)$: $\bv_B+m\bv_1+n\bv_2 = \bv_A+m\bv_1+n\bv_2-(\bv_A-\bv_B)$, by \eqref{nn-Bpts}, its $3$ nearest
 neighbor $A-$points are $\bv_A+m\bv_1+n\bv_2-(\bv_A-\bv_B) - \be^\nu$, $\nu=1,2,3$. 
To represent these nearest neighbor points with respect to the basis $\{\bv_1,\bv_2\}$, we  introduce the integers $\ttm_\nu$ and $\ttn_\nu$ such that 
\begin{equation}
e^\nu+ (\bv_A-\bv_B) = \ttm_\nu \bv_1 +\ttn_\nu\bv_2\quad \textrm{for}\quad \nu=1,2,3.
\label{tntm}\end{equation}
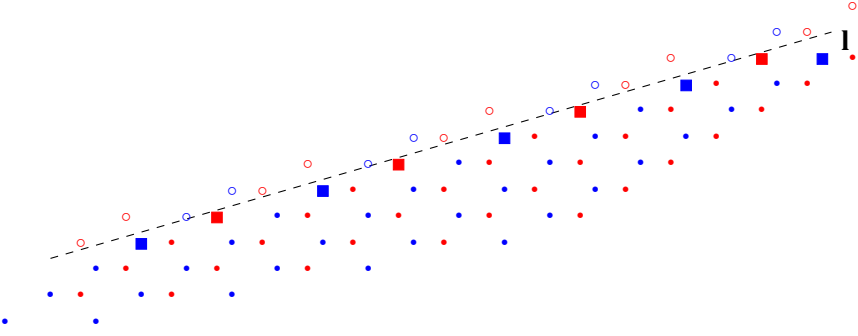
\begin{figure}[htbp]
	\begin{center}
	 	 \begin{tikzpicture}[scale=0.7]
			 \draw[-,dashed]({-2*sqrt(3)*1.5},{-0.85*1.5})--({4*sqrt(3)*1.4},{2.15*1.4}) node[yshift=-3,right] {${\bf l}$};
			  % \draw[-,dashed,red] ({-3/4*sqrt(3)},-1/4)--({-5/4*sqrt(3)},-3/4)--({3/4*sqrt(3)},1/4)--({5/4*sqrt(3)},3/4)--({-3/4*sqrt(3)},-1/4) node[yshift=10,left] {$\Gamma(0,0)$};
			   % \draw[-,dashed,red] ({-3/4*sqrt(3)+(2)*sqrt(3)-3*sqrt(3)/2},-1/4+1-3/2)--({-5/4*sqrt(3)+(2)*sqrt(3)-3*sqrt(3)/2},-3/4+1-3/2)--({3/4*sqrt(3)+(2)*sqrt(3)-3*sqrt(3)/2},1/4+1-3/2)--({5/4*sqrt(3)+(2)*sqrt(3)-3*sqrt(3)/2},3/4+1-3/2)--({-3/4*sqrt(3)+(2)*sqrt(3)-3*sqrt(3)/2},-1/4+1-3/2) node[yshift=-8,xshift=14,right] {$\Gamma(1,3)$};
 		 % \draw[ultra thick,->] ({-3/4*sqrt(3)+(2)*sqrt(3)-3*sqrt(3)/2},-1/4+1-3/2)--({-5/4*sqrt(3)+(2)*sqrt(3)-3*sqrt(3)/2},-3/4+1-3/2)  node[yshift=-10,xshift=2,left] {$\bv_2$};
%  		 \draw[ultra thick,->] ({-3/4*sqrt(3)+(2)*sqrt(3)-3*sqrt(3)/2},-1/4+1-3/2)-- ({5/4*sqrt(3)+(2)*sqrt(3)-3*sqrt(3)/2},3/4+1-3/2)  node[yshift=-10,xshift=5,left] {$\bv_1$};
		 % \draw[ultra thick,->]({sqrt(3)*5/2},0.5)--({sqrt(3)*6/2},1) node[yshift=1.5,xshift=-2,left] {$\ocirc{\bv}_1$};
% 		 \draw[ultra thick,->]({sqrt(3)*5/2},0.5)--({sqrt(3)*6/2},0) node[yshift=-0.5,xshift=-2,left] {$\ocirc{\bv}_2$};
 	 	 \foreach \n in {-1} 
 	  	\foreach \m in {1,2,3} {
 	 	 \draw ({(2*\n)*sqrt(3)-\m*sqrt(3)/2},{\n-(\m)/2}) node{\color{blue} \scalebox{0.7}{\textbullet}};
 		 %\draw ({1/sqrt(3)+(2*\n+1)*sqrt(3)/2},{(-2*\m+1)/2}) node{\color{red} $\bullet$};
 		 \draw ({5/6*sqrt(3)+(2*\n)*sqrt(3)-\m*sqrt(3)/2},{\n-(\m)/2+1/2}) node{\color{red} \scalebox{0.7}{\textbullet}};
 		 %\draw ({(2*\n+1)*sqrt(3)/2},{(-2*\m+1)/2}) node{\color{blue} $\bullet$};
 	 	 }	
  	 	 \foreach \n in {-1} 
  	  	\foreach \m in {0} {
		 \node at ({(2*\n)*sqrt(3)-\m*sqrt(3)/2},{\n-(\m)/2}) {\color{blue}\pgfuseplotmark{square*}} ; 
		  \node at ({5/6*sqrt(3)+(2*\n)*sqrt(3)-\m*sqrt(3)/2},{\n-(\m)/2+1/2}) {\color{red}\pgfuseplotmark{square*}};
		 }
		 \foreach \n in {0} 
	  	\foreach \m in {1,2,3,4,5} {
	 	 \draw ({(2*\n)*sqrt(3)-\m*sqrt(3)/2},{\n-(\m)/2}) node{\color{blue} \scalebox{0.7}{\textbullet}};
		 %\draw ({1/sqrt(3)+(2*\n+1)*sqrt(3)/2},{(-2*\m+1)/2}) node{\color{red} $\bullet$};
		 \draw ({5/6*sqrt(3)+(2*\n)*sqrt(3)-\m*sqrt(3)/2},{\n-(\m)/2+1/2}) node{\color{red} \scalebox{0.7}{\textbullet}};
		 %\draw ({(2*\n+1)*sqrt(3)/2},{(-2*\m+1)/2}) node{\color{blue} $\bullet$};
	 	 }
 		 \foreach \n in {0} 
 	  	\foreach \m in {0} {
 	 	 \node at ({(2*\n)*sqrt(3)-\m*sqrt(3)/2},{\n-(\m)/2}) {\color{blue}\pgfuseplotmark{square*}} ;
 		 %\draw ({1/sqrt(3)+(2*\n+1)*sqrt(3)/2},{(-2*\m+1)/2}) node{\color{red} $\bullet$};
 		 \node at ({5/6*sqrt(3)+(2*\n)*sqrt(3)-\m*sqrt(3)/2},{\n-(\m)/2+1/2}) {\color{red}\pgfuseplotmark{square*}} ;
 		 %\draw ({(2*\n+1)*sqrt(3)/2},{(-2*\m+1)/2}) node{\color{blue} $\bullet$};
 	 	 }
 		 \foreach \n in {1} 
 	  	\foreach \m in {1,2,3,4,5,6} {
 	 	 \draw ({(2*\n)*sqrt(3)-\m*sqrt(3)/2},{\n-(\m)/2}) node{\color{blue} \scalebox{0.7}{\textbullet}};
 		 %\draw ({1/sqrt(3)+(2*\n+1)*sqrt(3)/2},{(-2*\m+1)/2}) node{\color{red} $\bullet$};
 		 \draw ({5/6*sqrt(3)+(2*\n)*sqrt(3)-\m*sqrt(3)/2},{\n-(\m)/2+1/2}) node{\color{red} \scalebox{0.7}{\textbullet}};
 		 %\draw ({(2*\n+1)*sqrt(3)/2},{(-2*\m+1)/2}) node{\color{blue} $\bullet$};
 	 	 }
  		 \foreach \n in {1} 
  	  	\foreach \m in {0} {
  	 	 \node at ({(2*\n)*sqrt(3)-\m*sqrt(3)/2},{\n-(\m)/2}) {\color{blue}\pgfuseplotmark{square*}} ;
  		 %\draw ({1/sqrt(3)+(2*\n+1)*sqrt(3)/2},{(-2*\m+1)/2}) node{\color{red} $\bullet$};
  		 \node at ({5/6*sqrt(3)+(2*\n)*sqrt(3)-\m*sqrt(3)/2},{\n-(\m)/2+1/2}) {\color{red}\pgfuseplotmark{square*}} ;
  		 %\draw ({(2*\n+1)*sqrt(3)/2},{(-2*\m+1)/2}) node{\color{blue} $\bullet$};
  	 	 }
 	 	 \foreach \n in {2} 
 	  	\foreach \m in {1,2,3,4,5,6,7} {
 	 	 \draw ({(2*\n)*sqrt(3)-\m*sqrt(3)/2},{\n-(\m)/2}) node{\color{blue} \scalebox{0.7}{\textbullet}};
 		 %\draw ({1/sqrt(3)+(2*\n+1)*sqrt(3)/2},{(-2*\m+1)/2}) node{\color{red} $\bullet$};
 		 \draw ({5/6*sqrt(3)+(2*\n)*sqrt(3)-\m*sqrt(3)/2},{\n-(\m)/2+1/2}) node{\color{red} \scalebox{0.7}{\textbullet}};
 		 %\draw ({(2*\n+1)*sqrt(3)/2},{(-2*\m+1)/2}) node{\color{blue} $\bullet$};
 	 	 }
  		 \foreach \n in {2} 
  	  	\foreach \m in {0} {
  	 	 \node at ({(2*\n)*sqrt(3)-\m*sqrt(3)/2},{\n-(\m)/2}) {\color{blue}\pgfuseplotmark{square*}} ;
  		 %\draw ({1/sqrt(3)+(2*\n+1)*sqrt(3)/2},{(-2*\m+1)/2}) node{\color{red} $\bullet$};
  		 \node at ({5/6*sqrt(3)+(2*\n)*sqrt(3)-\m*sqrt(3)/2},{\n-(\m)/2+1/2}) {\color{red}\pgfuseplotmark{square*}} ;
  		 %\draw ({(2*\n+1)*sqrt(3)/2},{(-2*\m+1)/2}) node{\color{blue} $\bullet$};
  	 	 }
  	 	 \foreach \n in {3} 
  	  	\foreach \m in {2,3,4,5,6,7,8} {
  	 	 \draw ({(2*\n)*sqrt(3)-\m*sqrt(3)/2},{\n-(\m)/2}) node{\color{blue} \scalebox{0.7}{\textbullet}};
  		 %\draw ({1/sqrt(3)+(2*\n+1)*sqrt(3)/2},{(-2*\m+1)/2}) node{\color{red} $\bullet$};
  		 \draw ({5/6*sqrt(3)+(2*\n)*sqrt(3)-\m*sqrt(3)/2},{\n-(\m)/2+1/2}) node{\color{red} \scalebox{0.7}{\textbullet}};
  		 %\draw ({(2*\n+1)*sqrt(3)/2},{(-2*\m+1)/2}) node{\color{blue} $\bullet$};
  	 	 }
   		 \foreach \n in {3} 
   	  	\foreach \m in {1} {
   	 	 \node at ({(2*\n)*sqrt(3)-\m*sqrt(3)/2},{\n-(\m)/2}) {\color{blue}\pgfuseplotmark{square*}} ;
   		 %\draw ({1/sqrt(3)+(2*\n+1)*sqrt(3)/2},{(-2*\m+1)/2}) node{\color{red} $\bullet$};
   		 % \node at ({5/6*sqrt(3)+(2*\n)*sqrt(3)-\m*sqrt(3)/2},{\n-(\m)/2+1/2}) {\color{red}\pgfuseplotmark{square*}} ;
   		 }
  	 	 \foreach \n in {-2} 
  	  	\foreach \m in {-1,-2} {
  		 \draw ({5/6*sqrt(3)+(2*\n)*sqrt(3)-\m*sqrt(3)/2},{\n-(\m)/2+1/2}) node{\color{red} \scalebox{0.7}{$\circ$}};
  		 %\draw ({(2*\n+1)*sqrt(3)/2},{(-2*\m+1)/2}) node{\color{blue} $\bullet$};
  	 	 }
 	 	 \foreach \n in {-1} 
 	  	\foreach \m in {-1,-2} {
 	 	 \draw ({(2*\n)*sqrt(3)-\m*sqrt(3)/2},{\n-(\m)/2}) node{\color{blue} \scalebox{0.7}{$\circ$}};
 		 %\draw ({1/sqrt(3)+(2*\n+1)*sqrt(3)/2},{(-2*\m+1)/2}) node{\color{red} $\bullet$};
 		 \draw ({5/6*sqrt(3)+(2*\n)*sqrt(3)-\m*sqrt(3)/2},{\n-(\m)/2+1/2}) node{\color{red} \scalebox{0.7}{$\circ$}};
 		 %\draw ({(2*\n+1)*sqrt(3)/2},{(-2*\m+1)/2}) node{\color{blue} $\bullet$};
 	 	 }
  	 	 \foreach \n in {0} 
  	  	\foreach \m in {-1,-2} {
  	 	 \draw ({(2*\n)*sqrt(3)-\m*sqrt(3)/2},{\n-(\m)/2}) node{\color{blue} \scalebox{0.7}{$\circ$}};
  		 %\draw ({1/sqrt(3)+(2*\n+1)*sqrt(3)/2},{(-2*\m+1)/2}) node{\color{red} $\bullet$};
  		 \draw ({5/6*sqrt(3)+(2*\n)*sqrt(3)-\m*sqrt(3)/2},{\n-(\m)/2+1/2}) node{\color{red} \scalebox{0.7}{$\circ$}};
  		 %\draw ({(2*\n+1)*sqrt(3)/2},{(-2*\m+1)/2}) node{\color{blue} $\bullet$};
  	 	 }
  	 	 \foreach \n in {1} 
  	  	\foreach \m in {-1,-2} {
  	 	 \draw ({(2*\n)*sqrt(3)-\m*sqrt(3)/2},{\n-(\m)/2}) node{\color{blue} \scalebox{0.7}{$\circ$}};
  		 %\draw ({1/sqrt(3)+(2*\n+1)*sqrt(3)/2},{(-2*\m+1)/2}) node{\color{red} $\bullet$};
  		 \draw ({5/6*sqrt(3)+(2*\n)*sqrt(3)-\m*sqrt(3)/2},{\n-(\m)/2+1/2}) node{\color{red}\scalebox{0.7}{$\circ$}};
  		 %\draw ({(2*\n+1)*sqrt(3)/2},{(-2*\m+1)/2}) node{\color{blue} $\bullet$};
  	 	 }
   	 	 \foreach \n in {2} 
   	  	\foreach \m in {-1,-2} {
   	 	 \draw ({(2*\n)*sqrt(3)-\m*sqrt(3)/2},{\n-(\m)/2}) node{\color{blue} \scalebox{0.7}{$\circ$}};
   		 %\draw ({1/sqrt(3)+(2*\n+1)*sqrt(3)/2},{(-2*\m+1)/2}) node{\color{red} $\bullet$};
   		 \draw ({5/6*sqrt(3)+(2*\n)*sqrt(3)-\m*sqrt(3)/2},{\n-(\m)/2+1/2}) node{\color{red}\scalebox{0.7}{$\circ$}};
   		 %\draw ({(2*\n+1)*sqrt(3)/2},{(-2*\m+1)/2}) node{\color{blue} $\bullet$};
   	 	 }
	 	 \end{tikzpicture}
	 \caption{The edge defined by $(a_{11},a_{12})=(3,1)$. 
	 {$A-$ and $B-$ sublattice vertices of $\HH_\sharp$ are indicated by colored circles and colored squares. Colored squares indicate the frontier sites, as defined in the introduction. Empty circles indicate vertices outside $\HH_\sharp$, which are nearest neighbors to vertices in $\HH_\sharp$.  }  }
	 \label{fig:theedges31}
 	\end{center}
\end{figure}

\begin{figure}[htbp]
	\begin{center}
	 	 \begin{tikzpicture}[scale=0.7]
			 \draw[-,dashed]({-2*sqrt(3)*1.5},{-0.85*1.5})--({4*sqrt(3)*1.4},{2.15*1.4}) node[yshift=-3,right] {${\bf l}$};

			 \draw (0,0) node[yshift=7,xshift=5,left]{$\bv_A$};  \draw ({5/6*sqrt(3)},{1/2}) node[yshift=7,xshift=5,left]{$\bv_B$};
			  % \draw[-,dashed,red] ({-3/4*sqrt(3)},-1/4)--({-5/4*sqrt(3)},-3/4)--({3/4*sqrt(3)},1/4)--({5/4*sqrt(3)},3/4)--({-3/4*sqrt(3)},-1/4) node[yshift=10,left] {$\Gamma(0,0)$};
			   \draw[gray!20,fill=gray!20] ({-3/4*sqrt(3)+(2)*sqrt(3)-3*sqrt(3)/2},-1/4+1-3/2)--({-5/4*sqrt(3)+(2)*sqrt(3)-3*sqrt(3)/2},-3/4+1-3/2)--({3/4*sqrt(3)+(2)*sqrt(3)-3*sqrt(3)/2},1/4+1-3/2)--({5/4*sqrt(3)+(2)*sqrt(3)-3*sqrt(3)/2},3/4+1-3/2)--({-3/4*sqrt(3)+(2)*sqrt(3)-3*sqrt(3)/2},-1/4+1-3/2) node[yshift=-8,xshift=14,right,black] {$\Gamma(1,3)$};
 		 % \draw[ultra thick,->] ({-3/4*sqrt(3)+(2)*sqrt(3)-3*sqrt(3)/2},-1/4+1-3/2)--({-5/4*sqrt(3)+(2)*sqrt(3)-3*sqrt(3)/2},-3/4+1-3/2)  node[xshift=2,left] {$\bv_2$};
		  \draw[ultra thick,->] ({-3/4*sqrt(3)+3*sqrt(3)/4},-1/4+1/4)--({-5/4*sqrt(3)+3*sqrt(3)/4},-3/4+1/4)  node[xshift=2,left] {$\bv_2$};
		  \draw[ultra thick,->] ({-3/4*sqrt(3)+3*sqrt(3)/4},-1/4+1/4)-- ({5/4*sqrt(3)++3*sqrt(3)/4},3/4+1/4)  node[yshift=-10,xshift=5,left] {$\bv_1$};
 		 % \draw[ultra thick,->] ({-3/4*sqrt(3)+(2)*sqrt(3)-3*sqrt(3)/2},-1/4+1-3/2)-- ({5/4*sqrt(3)+(2)*sqrt(3)-3*sqrt(3)/2},3/4+1-3/2)  node[yshift=-10,xshift=5,left] {$\bv_1$};
		 \draw[ultra thick,->]({sqrt(3)*5/2},0.5)--({sqrt(3)*6/2},1) node[yshift=1.5,xshift=-2,left] {$\ocirc{\bv}_1$};
		 \draw[ultra thick,->]({sqrt(3)*5/2},0.5)--({sqrt(3)*6/2},0) node[yshift=-0.5,xshift=-2,left] {$\ocirc{\bv}_2$};
 	 	 \foreach \n in {-1} 
 	  	\foreach \m in {0,1,2,3} {
 	 	 \draw ({(2*\n)*sqrt(3)-\m*sqrt(3)/2},{\n-(\m)/2}) node{\color{blue} \scalebox{0.7}{\textbullet}};
 		 %\draw ({1/sqrt(3)+(2*\n+1)*sqrt(3)/2},{(-2*\m+1)/2}) node{\color{red} $\bullet$};
 		 \draw ({5/6*sqrt(3)+(2*\n)*sqrt(3)-\m*sqrt(3)/2},{\n-(\m)/2+1/2}) node{\color{red} \scalebox{0.7}{\textbullet}};
 		 %\draw ({(2*\n+1)*sqrt(3)/2},{(-2*\m+1)/2}) node{\color{blue} $\bullet$};
 	 	 }	 	 
		 \foreach \n in {0} 
	  	\foreach \m in {0,1,2,3,4,5} {
	 	 \draw ({(2*\n)*sqrt(3)-\m*sqrt(3)/2},{\n-(\m)/2}) node{\color{blue} \scalebox{0.7}{\textbullet}};
		 %\draw ({1/sqrt(3)+(2*\n+1)*sqrt(3)/2},{(-2*\m+1)/2}) node{\color{red} $\bullet$};
		 \draw ({5/6*sqrt(3)+(2*\n)*sqrt(3)-\m*sqrt(3)/2},{\n-(\m)/2+1/2}) node{\color{red} \scalebox{0.7}{\textbullet}};
		 %\draw ({(2*\n+1)*sqrt(3)/2},{(-2*\m+1)/2}) node{\color{blue} $\bullet$};
	 	 }
 		 \foreach \n in {1} 
 	  	\foreach \m in {0,1,2,3,4,5,6} {
 	 	 \draw ({(2*\n)*sqrt(3)-\m*sqrt(3)/2},{\n-(\m)/2}) node{\color{blue} \scalebox{0.7}{\textbullet}};
 		 %\draw ({1/sqrt(3)+(2*\n+1)*sqrt(3)/2},{(-2*\m+1)/2}) node{\color{red} $\bullet$};
 		 \draw ({5/6*sqrt(3)+(2*\n)*sqrt(3)-\m*sqrt(3)/2},{\n-(\m)/2+1/2}) node{\color{red} \scalebox{0.7}{\textbullet}};
 		 %\draw ({(2*\n+1)*sqrt(3)/2},{(-2*\m+1)/2}) node{\color{blue} $\bullet$};
 	 	 }
 	 	 \foreach \n in {2} 
 	  	\foreach \m in {0,1,2,3,4,5,6,7} {
 	 	 \draw ({(2*\n)*sqrt(3)-\m*sqrt(3)/2},{\n-(\m)/2}) node{\color{blue} \scalebox{0.7}{\textbullet}};
 		 %\draw ({1/sqrt(3)+(2*\n+1)*sqrt(3)/2},{(-2*\m+1)/2}) node{\color{red} $\bullet$};
 		 \draw ({5/6*sqrt(3)+(2*\n)*sqrt(3)-\m*sqrt(3)/2},{\n-(\m)/2+1/2}) node{\color{red} \scalebox{0.7}{\textbullet}};
 		 %\draw ({(2*\n+1)*sqrt(3)/2},{(-2*\m+1)/2}) node{\color{blue} $\bullet$};
 	 	 }
  	 	 \foreach \n in {3} 
  	  	\foreach \m in {2,3,4,5,6,7,8} {
  	 	 \draw ({(2*\n)*sqrt(3)-\m*sqrt(3)/2},{\n-(\m)/2}) node{\color{blue} \scalebox{0.7}{\textbullet}};
  		 %\draw ({1/sqrt(3)+(2*\n+1)*sqrt(3)/2},{(-2*\m+1)/2}) node{\color{red} $\bullet$};
  		 \draw ({5/6*sqrt(3)+(2*\n)*sqrt(3)-\m*sqrt(3)/2},{\n-(\m)/2+1/2}) node{\color{red} \scalebox{0.7}{\textbullet}};
  		 %\draw ({(2*\n+1)*sqrt(3)/2},{(-2*\m+1)/2}) node{\color{blue} $\bullet$};
  	 	 }
 	 	 \end{tikzpicture}
	 \caption{Notation for the edge defined by $(a_{11},a_{12})=(3,1)$: the vectors $\bv_1$ and $\bv_2$, the points $\bv_A$ and $\bv_B$, and the parallelogram $\Gamma(1,3)$.}
	 \label{fig:notation_edges31}
 	\end{center}
\end{figure}
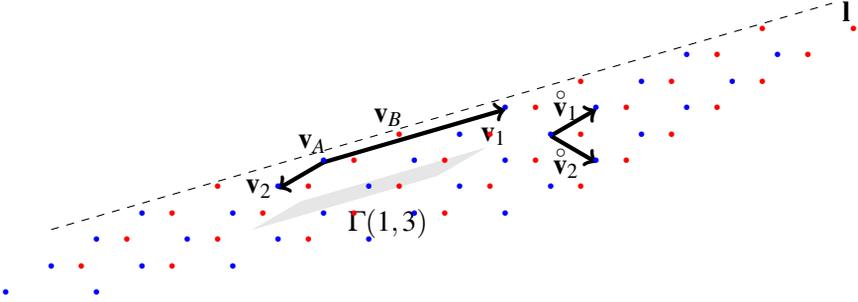

In Appendix \ref{vAvB-mn-app} we show that
\begin{subequations}
\begin{align}
\ttm_1&=k_1,\quad \ttn_1=k_2\\
\ttm_2&=k_1+a_{21},\quad \ttn_2=k_2-a_{11}\\
\ttm_3&=k_1-a_{22},\quad \ttn_3=k_2+a_{12},
\end{align}
\label{ttn-ttm}
\end{subequations}
with $(a_{ij})$ as in \eqref{vjs}, \eqref{Acond}, and $k_1$, $k_2$ as in \eqref{k1k2s1s2}.
~\\\\
We conclude that
\begin{align}
&\textrm{the three nearest neighbors to the $A-$point in $\Gamma(m,n)$:
 $\bv_A+m\bv_1+n\bv_2$}\nn\\
& \textrm{ are the three $B-$ points}\ 
\left(\bv_B+m\bv_1+n\bv_2\right) + \ttm_\nu\bv_1 + \ttn_\nu \bv_2\quad (\nu=1,2,3)\ , %(27)
\label{nn-Bpts1}\end{align}
while 
\begin{align}
&\textrm{the three nearest neighbors to the $B-$point  in $\Gamma(m,n)$:
 $\bv_B+m\bv_1+n\bv_2$}\nn\\
& \textrm{ are the three $A-$ points}\ 
\left(\bv_A+m\bv_1+n\bv_2\right) - \left( \ttm_\nu\bv_1 + \ttn_\nu \bv_2 \right)\quad (\nu=1,2,3)\ . %(28)
\label{nn-Apts1}\end{align}
\begin{definition}[Zigzag-type and armchair-type edges]\label{def:zz-ac}
Let $s_2\in\{-1,0,1\}$ be defined as $s_2:= a_{11}-a_{12}\ \text{mod}\ 3$; see \eqref{k1k2s1s2}. An edge is of
 {\it zigzag-type} if $s_2=\pm1$, and of {\it armchair-type} if $ s_2=0$ .
 \end{definition}
 
 {This definition agrees with the geometrical definition of zigzag and armchair edges given in the introduction; see Proposition \ref{DADB} below. }
\begin{example}[Classical Zigzag and Armchair Edges]\label{classical}
 The classical zigzag edges, displayed in panels (a) and (b) of  Figure \ref{fig:classical_edges}, are in the direction 
$\bv_1=\ocirc{\bv}_1-\ocirc{\bv}_2$. Thus,  $(a_{11},a_{12})=(1,-1)$ and therefore $s_2=-1$. {We take $\bv_2=\ocirc{\bv}_1$}. 
 Equivalent edges are obtained by (counterclockwise) rotation by $2\pi/3$, giving $(a_{11},a_{12})=(-1,0)$ ( $s_2=1$), and by $4\pi/3$ giving $(a_{11},a_{12})=(0,1)$ ( $s_2=1$). The  classical armchair edge displayed in panel (c) is in the direction $\bv_1=\ocirc{\bv}_1+\ocirc{\bv}_2 $, for which $(a_{11},a_{12})=(1,1)$, and hence $s_2=0$.{ We take $\bv_2=\ocirc{\bv}_2$}.  Equivalent edges, obtained via  counterclockwise rotations by $2\pi/3$ and $4\pi/3$, give configuration parameters  $(a_{11},a_{12})=(2,-1)$ and $(a_{11},a_{12})=(-1,2)$, respectively.  See, for example,  \cite{delplace2011zak,Graf-Porta:13,mong2011edge,Dresselhaus-etal:96} for  previous studies of classical zigzag edges. 
\end{example}

\begin{figure}[htbp]
	\begin{center}
		\begin{subfigure}{0.45\textwidth}
	 \begin{tikzpicture}
		 \draw[ultra thick,->]({sqrt(3)*3/2},-0.5)--({sqrt(3)*4/2},0) node[yshift=1.5,xshift=-2,left] {$\ocirc{\bv}_1$};
		 \draw[ultra thick,->]({sqrt(3)*3/2},-0.5)--({sqrt(3)*4/2},-1) node[yshift=-0.5,xshift=-2,left] {$\ocirc{\bv}_2$};
		 \draw[-,dashed](0.3,-2)--(0.3,2) node[yshift=-3,right] {${\bf l}$};
		 \draw[ultra thick,->] ({1/sqrt(3)},0)--({1/sqrt(3)},1)  node[yshift=-10,xshift=2,left] {$\bv_1$};
		 \draw[ultra thick,->] ({1/sqrt(3)},0)--({1/sqrt(3)+sqrt(3)/2},1/2)  node[yshift=7,xshift=5,left] {$\bv_2$}; 
		 \draw[gray!20,fill=gray!20] ({sqrt(3)/4},-5/4)--({3*sqrt(3)/4},-3/4)--({3*sqrt(3)/4},1/4)--({sqrt(3)/4},-1/4)--({sqrt(3)/4},-5/4) node [black] {$\Gamma(0,0)$};
 	 \foreach \m in {-1,0,1,2} 
  	\foreach \n in {0,1,2} {
 	 \draw ({1/sqrt(3)+(2*\n)*sqrt(3)/2},{(-2*\m)/2}) node{\color{red} \scalebox{0.7}{\textbullet}};
	 \draw ({1/sqrt(3)+(2*\n+1)*sqrt(3)/2},{(-2*\m+1)/2}) node{\color{red} \scalebox{0.7}{\textbullet}};
	 \draw ({(2*\n+2)*sqrt(3)/2},{(-2*\m)/2}) node{\color{blue} \scalebox{0.7}{\textbullet}};
	 \draw ({(2*\n+1)*sqrt(3)/2},{(-2*\m+1)/2}) node{\color{blue} \scalebox{0.7}{\textbullet}};
 	 }
  	 \foreach \m in {-1,0,1,2} 
   	\foreach \n in {-1} {
  	 %\draw ({1/sqrt(3)+(2*\n)*sqrt(3)/2},{(-2*\m)/2}) node{\color{red} \scalebox{0.7}{$\circ$}};
 	 %\draw ({1/sqrt(3)+(2*\n+1)*sqrt(3)/2},{(-2*\m+1)/2}) node{\color{red} \scalebox{0.7}{$\circ$}};
 	 \draw ({(2*\n+2)*sqrt(3)/2},{(-2*\m)/2}) node{\color{blue} \scalebox{0.7}{$\circ$}};
 	 %\draw ({(2*\n+1)*sqrt(3)/2},{(-2*\m+1)/2}) node{\color{blue} \scalebox{0.7}{$\circ$}};
  	 }
	 \end{tikzpicture} \caption{Balanced ZZ-edge /
		 $D_A>D_B$:  $(a_{11},a_{12})=(1,-1), (a_{21},a_{22})=(1,0)$}\end{subfigure}\hspace{1cm}
 		\begin{subfigure}{0.45\textwidth}
 	 \begin{tikzpicture}
		  \draw[gray!20,fill=gray!20] ({sqrt(3)/4},-5/4)--({3*sqrt(3)/4},-3/4)--({3*sqrt(3)/4},1/4)--({sqrt(3)/4},-1/4)--({sqrt(3)/4},-5/4) node [black] {$\Gamma(0,0)$};
 		 \draw[-,dashed](0.7,-2)--(0.7,2) node[yshift=-3,right] {${\bf l}$};
		 \draw[ultra thick,->] ({1/sqrt(3)+sqrt(3)/2},-1/2)--({1/sqrt(3)+sqrt(3)/2},1/2)  node[yshift=-10,xshift=2,left] {$\bv_1$};
		 \draw[ultra thick,->] ({1/sqrt(3)+sqrt(3)/2},-1/2)--({1/sqrt(3)+sqrt(3)},0)  node[yshift=-10,xshift=5,left] {$\bv_2$};
  	 \foreach \m in {-1,0,1,2} 
   	\foreach \n in {0,1,2} {
  	 \draw ({1/sqrt(3)+(2*(\n+1))*sqrt(3)/2},{(-2*\m)/2}) node{\color{red} \scalebox{0.7}{\textbullet}};
 	 \draw ({1/sqrt(3)+(2*\n+1)*sqrt(3)/2},{(-2*\m+1)/2}) node{\color{red} \scalebox{0.7}{\textbullet}};
 	 \draw ({(2*\n+2)*sqrt(3)/2},{(-2*\m)/2}) node{\color{blue} \scalebox{0.7}{\textbullet}};
 	 \draw ({(2*\n+1)*sqrt(3)/2},{(-2*\m+1)/2}) node{\color{blue} \scalebox{0.7}{\textbullet}};
  	 }
   	 \foreach \m in {-1,0,1,2} 
    	\foreach \n in {-1} {
   	 \draw ({1/sqrt(3)+(2*(\n+1))*sqrt(3)/2},{(-2*\m)/2}) node{\color{red} \scalebox{0.7}{$\circ$}};
  	 %\draw ({1/sqrt(3)+(2*\n+1)*sqrt(3)/2},{(-2*\m+1)/2}) node{\color{red} \scalebox{0.7}{$\circ$}};
  	 %\draw ({(2*\n+2)*sqrt(3)/2},{(-2*\m)/2}) node{\color{blue} \scalebox{0.7}{$\circ$}};
  	 %\draw ({(2*\n+1)*sqrt(3)/2},{(-2*\m+1)/2}) node{\color{blue} \scalebox{0.7}{$\circ$}};
   	 }
 	 \end{tikzpicture}\caption{Unbalanced ZZ-edge / $D_A<D_B$: 
		$(a_{11},a_{12})=(1,-1), (a_{21},a_{22})=(1,0)$}\end{subfigure}\vspace{0.5cm}
	 		\begin{subfigure}{0.7\textwidth}
	 	\begin{center}\begin{tikzpicture}
			 \draw[gray!20,fill=gray!20] ({-sqrt(3)/4},3/4)--({3*sqrt(3)/4},3/4)--({5*sqrt(3)/4},1/4)--({sqrt(3)/4},1/4)--({-sqrt(3)/4},3/4) node [black] {$\Gamma(0,0)$};
			 \draw[-,dashed](-0.5,0.7)--(5,0.7) node[yshift=-3,right] {${\bf l}$};
		 \draw[ultra thick,->] ({1/sqrt(3)},0)--({1/sqrt(3)+sqrt(3)},0)  node[yshift=-10,xshift=2,left] {$\bv_1$};
		 \draw[ultra thick,->] ({1/sqrt(3)},0)--({1/sqrt(3)+sqrt(3)/2},-1/2)  node[yshift=-10,xshift=5,left] {$\bv_2$};	  
	 	 \foreach \n in {0,1,2} 
	  	\foreach \m in {0,1,2,3} {
	 	 \draw ({1/sqrt(3)+(2*\n)*sqrt(3)/2},{(-2*\m)/2}) node{\color{red} \scalebox{0.7}{\textbullet}};
		 \draw ({1/sqrt(3)+(2*\n+1)*sqrt(3)/2},{(-2*\m+1)/2}) node{\color{red} \scalebox{0.7}{\textbullet}};
		 \draw ({(2*\n)*sqrt(3)/2},{(-2*\m)/2}) node{\color{blue} \scalebox{0.7}{\textbullet}};
		 \draw ({(2*\n+1)*sqrt(3)/2},{(-2*\m+1)/2}) node{\color{blue} \scalebox{0.7}{\textbullet}};
	 	 }
 	 	 \foreach \n in {0,1,2} 
 	  	\foreach \m in {-1} {
 	 	 \draw ({1/sqrt(3)+(2*\n)*sqrt(3)/2},{(-2*\m)/2})  node{\color{red} \scalebox{0.7}{$\circ$}};
 		 %\draw ({1/sqrt(3)+(2*\n+1)*sqrt(3)/2},{(-2*\m+1)/2})  node{\color{red} \scalebox{0.7}{$\circ$}};
 		 \draw ({(2*\n)*sqrt(3)/2},{(-2*\m)/2}) node{\color{blue} \scalebox{0.7}{$\circ$}};
 		 %\draw ({(2*\n+1)*sqrt(3)/2},{(-2*\m+1)/2}) node{\color{blue} \scalebox{0.7}{$\circ$}};
 	 	 }
	 	 \end{tikzpicture}\end{center}\caption{AC-edge / $D_A=D_B$: $(a_{11},a_{12})=(1,1), (a_{21},a_{22})=(0,1)$}\end{subfigure}
 \caption{{$\HH_\sharp$ for the classical ZZ and AC edges, balanced and unbalanced; see Definition \ref{def:ZZAC}. $A-$ sites are in blue and $B-$ sites are in red. The edge, ${\bf l}$,  is indicated with a dashed line. The parallelogram $\Gamma(0,0)$ is shaded gray; see Remark \ref{p-gram}.\ }
	 }
	 \label{fig:classical_edges}
 	\end{center}
\end{figure}
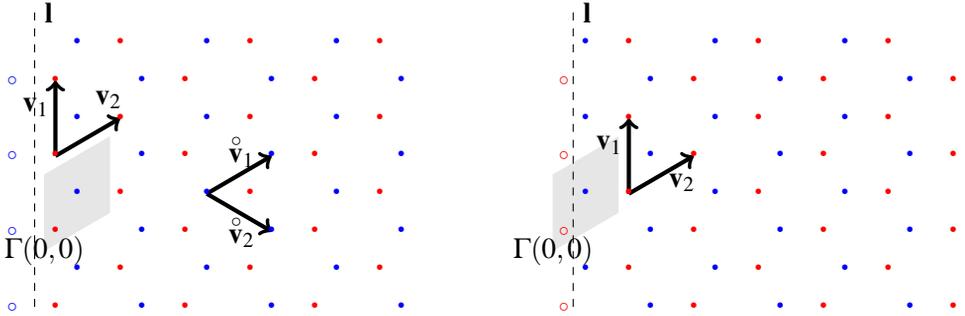
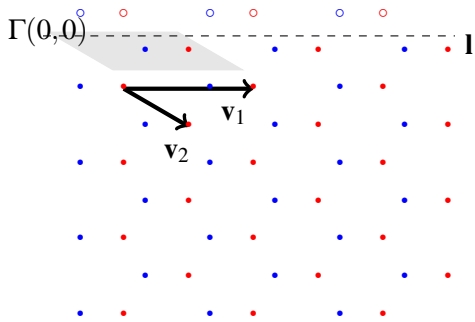

\nit{\bf Convention:} {\it We find it very convenient to order the pairs $(\ttm_\nu,\ttn_\nu),\ \nu=1,2,3$ according
to {increasing} $\ttn_\nu$.  Of course this makes sense only if $\ttn_1$, $\ttn_2$ and $\ttn_3$ are all distinct. 
In Appendix \ref{distinct-n} 
 we verify that $\ttn_\nu$ are indeed all distinct except for the case of the classical zigzag edges discussed in 
  Example  \ref{classical}. We exclude the classical zigzag case in the coming general analysis of rational edges and will analyze the classical zigzag edge case separately.
\begin{equation}
\begin{aligned}
&\textrm{We denote by $(n_1,n_2,n_3)$ the permutation of $(\ttn_1,\ttn_2,\ttn_3)$ for which:}\ \  n_1<n_2< n_3;\\
&\textrm{we denote by $(m_1,m_2,m_3)$ the same permutation applied to $(\ttm_1,\ttm_2,\ttm_3)$.}
\end{aligned}
 \label{perm-n}
\end{equation} 
In  Appendix \ref{a-few-ineq} it is proved that 
\begin{equation}\label{eq:n_nu_ineq} n_1<0<n_3.\end{equation}
}

\subsection{Bulk and edge Hamiltonians with respect to the basis $\{\bv_1,\bv_2\}$}\label{H_wrtv1v2}
The above observations and calculations allow us to rewrite the bulk Hamiltonian $\Hbulk$ and the edge Hamiltonian, $H_\sharp$,  in terms of the lattice basis
 $\{\bv_1,\bv_2\}$. %MG13
 
 \subsubsection{$\Hbulk$ in the basis $\{\bv_1,\bv_2\}$} \label{Hbk_wrtv1v2}
  We write the wave functions $\left(\psi_\omega\right)_{\omega\in\HH}$, with respect to the basis $\{\bv_1,\bv_2\}$
   in the form 
   \begin{equation}
    \psi=\left(\psi^A(m,n),\psi^B(m,n)\right)_{m,n\in\Z}\ \in\ l^2(\Z\times\Z;\C^2),
 \label{psi-mn}   \end{equation} %(29)
 where
   \begin{align*}
   \psi^A(m,n)&= \psi_\omega\quad \textrm{for}\quad \omega=\bv_A+m\bv_1+n\bv_2\in \HH_A,\quad {\rm and}\nn\\
   \psi^B(m,n)&= \psi_\omega\quad \textrm{for}\quad \omega=\bv_B+m\bv_1+n\bv_2\in \HH_B.\nn
  \end{align*} 
  
  Thanks to the formulas \eqref{nn-Bpts1} and \eqref{nn-Apts1}  for nearest neighbor points, the action of our bulk Hamiltonian \eqref{Hbulk_A}, \eqref{Hbulk_B} on the vector \eqref{psi-mn} is given, for $m, n\in\Z$, by
  \begin{subequations}\label{Hpsi}
  \begin{align}
  \left(\Hbulk\psi\right)^A(m,n) &= \sum_{\nu=1}^3\psi^B(m+m_\nu,n+n_\nu),\label{HpsiA-mn}\\ %(30)
   \left(\Hbulk\psi\right)^B(m,n) &= \sum_{\nu=1}^3\psi^A(m-m_\nu,n-n_\nu) .\label{HpsiB-mn}   %(31)
   \end{align}
   \end{subequations}
   
\subsubsection{$H_\sharp$ in the basis $\{\bv_1,\bv_2\}$} \label{Hbk_wrtv1v2_2}
 We want to carry out an analogous reformulation for the Hamiltonian $H_\sharp$,  \eqref{Hshp}, associated with the
 sharply terminated honeycomb structure, $\HH_\sharp$. To do so, we recall \eqref{HHsharp}, \eqref{vAvB-1} and we deduce, using \eqref{HHsharp}, that
\begin{align*}
&\textrm{the $A-$point  $\bv_A+m\bv_1+n\bv_2 \in\HH_\sharp$    if and only if $n\ge\nA$, and }\\
 &\textrm{the $B-$point  $\bv_B+m\bv_1+n\bv_2\in\HH_\sharp$   if and only if $n\ge\nB$ ,}
 \end{align*}
 where
  \begin{subequations}
  \label{nAnB}
  \begin{align}
  \nA &= \textrm{the least integer $\ge \beta$}\label{nA-def}\\ %(32)
  \nB &= \textrm{the least integer $\ge \beta-\frac13 s_2$}.\label{nB-def} %(33)
\end{align}
\end{subequations}
Note from \eqref{nA-def}, \eqref{nB-def} that 
  \begin{subequations}
  \begin{align}
&  \textrm{$\nA=\nB$ if $s_2=0$}\\
& \textrm{$\nA=\nB$ or $\nA=\nB+1$ if  $s_2=1$}\\
 & \textrm{$\nA=\nB$ or $\nA=\nB-1$ if  $s_2=-1$}.
 \end{align}
 \label{nAnBs2} %(34), (35), (36)
 \end{subequations}
 \begin{definition}\label{def:ZZAC}
 We will say that our terminated structure is {\it balanced} if $\nA=\nB$.  Otherwise, we say that the terminated structure is {\it unbalanced}; see Figure \ref{fig:classical_edges} for some examples.
 \end{definition}
 
 \begin{remark}\label{p-gram}
	 Recall that each parallelogram $\Gamma(m,n)$ contains one $A-$point and one $B-$point. If $\HH_\sharp$ is balanced then in each $\Gamma(m,n)$, either both of those points or neither will lie in $\HH_\sharp$. An unbalanced $\HH_\sharp$ gives rise to $\Gamma(m,n)$ containing an $A-$point and not a $B-$ point ($\nA<\nB$) or a $B-$point and not an $A-$ point ($\nA>\nB$); see Figure \ref{fig:classical_edges}.
	 \end{remark}
 Note, by Definition \ref{def:ZZAC} and \eqref{nAnBs2} that zigzag-type edges can be either balanced or unbalanced
 but that armchair edges are all of balanced type. 
 
 Recall that $D_A$ (resp. $D_B$) denotes  the distance from the edge ${\bf l}$ to the closest $A-$ points (resp. $B-$ points) of the terminated structure $\HH_\sharp$ (also called frontier points in the introduction and in Figure \ref{fig:theedges31}).
   {
  Thanks to \eqref{nAnBs2}, the  following proposition shows that the notions of armchair edge, balanced zigzag edge
  and unbalanced zigzag edge given in Definitions \ref{def:zz-ac} and \ref{def:ZZAC} agree with the corresponding definitions of those concepts in the introduction.
  \begin{proposition}\label{DADB}
  With $D_A$ and $D_B$ as defined  in the introduction, we have
   \begin{equation}
     D_B-D_A = \frac{\sqrt3}{2} |\bv_1|^{-1} \left( \frac13 s_2 + \nB-\nA\right).
     \label{***3}\end{equation}
  \end{proposition}
  }
\begin{proof}
 { Recall that the line ${\bf l}$ is given by
    \[ \{x_1\bv_1+x_2\bv_2\ :\ x_2=\beta, x_1\in\R\}\quad \textrm{for some $\beta\in\R$}.\]
    We introduce the unit vector $\be^\perp$ which is perpendicular to $\bv_1$ and pointing into the bulk $\HH_\sharp$.
    Let
     \begin{equation*} 
     \mu=\be^\perp\cdot \bv_2. \label{**1}
    \end{equation*}
    Note that $\mu>0$ since both $\bv_2$ and $\be^\perp$ both point into the bulk.    } 
  {
  The (signed) distance from a point $\omega=x_1\bv_1+x_2\bv_2$ to the edge {\bf l}
     is equal to $(x_2-\beta)\mu$.}
{  Recall that the $A-$ points (respectively, $B-$ points) of 
 $\HH_\sharp$ are  the points
     $\bv_A+m\bv_1+n\bv_2$ (respectively, $\bv_B+m\bv_1+n\bv_2$) with $m, n\in\Z$
      and $n\ge\nA$ (respectively, $n\ge\nB$). Here, $\bv_A=(0,0)$ and 
       $\bv_B=\frac13(s_1\bv_1+s_2\bv_2)$; see Section \ref{H_wrtv1v2} and
       \eqref{vAvB-1}. It follows that 
     $D_A=(\nA-\beta)\mu$, while 
     $D_B=(\frac13 s_2+\nB -\beta)\mu$. Hence,
     \begin{equation}
     D_B-D_A = \left( \frac13 s_2 + \nB-\nA\right)\mu.
     \label{**2}\end{equation}
     We next compute $\mu$ using \eqref{tri-lattice}, \eqref{vjs} and \eqref{det1}. Writing $\wedge$ for the wedge product of two vectors,
      we have $|\bv_1\wedge\bv_2|=|\ocirc{\bv}_1\wedge \ocirc{\bv}_2|=\frac{\sqrt3}{2}$. 
      Note now that $\bv_2=z\bv_1+\mu\be^\perp$ for some $z\in\R$.
       Therefore, since $\mu>0$
      \[ \frac{\sqrt3}{2}= |\bv_1\wedge\bv_2| = |\bv_1\wedge \mu\be^\perp| = 
      %\mu |\bv_1|\left( \frac{\bv_1}{|\bv_1|}\wedge \frac{\bv_1^\perp}{|\bv_1^\perp|} \right) =
       \mu |\bv_1|. \]
      Substituting into \eqref{**2} yields \eqref{***3}.
     This completes the proof of Proposition \ref{DADB}.
   }  
\end{proof}

 Now we know the nearest neighbors to any given point in $\HH$, and we know exactly which points 
  $\bv^A+m\bv_1+n\bv_2$ and $\bv^B+m\bv_1+n\bv_2$ lie in the terminated structure $\HH_\sharp$.
  Therefore, we can rewrite the Hamiltonian, $H_\sharp$, given by \eqref{Hshp} in terms of the $\{\bv_1,\bv_2\}$ basis.
  \begin{proposition} The Hamiltonian $H_\sharp$ acts on vectors: 
  \begin{equation}
  \psi(m,n)  = \left(\psi^A(m,n),\psi^B(m,n) \right)_{m,n\in \Z} \ \in\ l^2(\Z\times\Z;\C^2), \label{psi-mn1}%(37)
  \end{equation}
  subject to the restrictions
     \begin{equation}
   \begin{aligned}
   \psi^A(m,n) &= 0 \quad {\rm for}\quad n<\nA,\\%\label{nltnA}\\ %(38)
   \psi^B(m,n) &= 0 \quad {\rm for}\quad n<\nB.\\%\label{nltnB}    %(39)
   \end{aligned}
   \label{nltnAnB}
   \end{equation}
   The action of $H_\sharp$ on such vectors $\psi$ is given by 
   \begin{equation*}
   \begin{aligned}
   \left(H_\sharp\psi\right)^A(m,n) &= \sum_{\nu=1}^3 \psi^B(m+m_\nu,n+n_\nu)\quad {\rm if}\quad n\ge \nA, %(40)
    \\ 
    \left(H_\sharp\psi\right)^A(m,n) &= 0 \quad {\rm if}\quad n< \nA, %(41)
    \\  
     \left(H_\sharp\psi\right)^B(m,n) &= \sum_{\nu=1}^3 \psi^A(m-m_\nu,n-n_\nu)\quad {\rm if}\quad n\ge \nB, %(42)
    \\ 
     \left(H_\sharp\psi\right)^B(m,n) &= 0 \quad {\rm if}\quad n< \nB. %(42)
     \end{aligned}
     \label{Hsharp-ef}
   \end{equation*}
   $H_\sharp$ is a bounded self-adjoint operator acting on $\psi$ of the form \eqref{psi-mn1}, \eqref{nltnAnB}.
  \end{proposition}

  \subsubsection{Decomposition of $H_\sharp$ into fiber Hamiltonians,  $\Hk$}\label{sec:fibers}
  
  The Hamiltonian $H_\sharp$ is invariant under translations by the edge direction vector $\bv_1$. Consequently our Hilbert space of wave functions \eqref{psi-mn1}, \eqref{nltnAnB} decomposes into a direct integral over {\it parallel quasi-momenta}
 of spaces $l^2_{\kp}$, with $\kp\in \R/2\pi\Z$. Here, $l^2_{\kp}$ denotes the space of {\it $\kp$ pseudoperiodic} wave functions of the form
 \begin{align}\label{eq:k-QP}
 \psi^A(m,n) &= e^{i\kp m}\psi^A(n), \\ % (44)
 \psi^B(m,n) &= e^{i\kp m}\psi^B(n),  % (45)
 \end{align}
 where
  \begin{equation}
 \psi= \left(\psi(n)\right)_{n\in\Z}= \left(\psi^A(n),\psi^B(n)\right)_{n\in\Z} \in l^2(\Z;\C^2) \label{psi-mn2}
 %(46)
 \end{equation}
satisfies
 \begin{equation} \psi^A(n) = 0 \quad {\rm for}\quad n<\nA\quad {\rm and}\quad  \psi^B(n) = 0 \quad {\rm for}\quad n<\nB.
 %(47)
\label{psiABout} \end{equation}
The Hamiltonian $H_\sharp$, accordingly decomposes as a direct integral  \cite{RS4} over $\kp\in\R/2\pi\Z$ of the Hamiltonian
 $\Hk$
 acting on wave functions \eqref{psi-mn2}, \eqref{psiABout}, given by the formulas
 \begin{equation}
   \begin{aligned}
   \left(\Hk\psi\right)^A(n) &= \sum_{\nu=1}^3 e^{i\kp m_\nu} \psi^B(n+n_\nu)\quad {\rm if}\quad n\ge \nA %(48)
    \\ 
    \left(\Hk\psi\right)^A(n) &= 0 \quad {\rm if}\quad n< \nA %(49)
    \\  
     \left(\Hk\psi\right)^B(n) &= \sum_{\nu=1}^3 e^{-i\kp m_\nu}\psi^A(n-n_\nu)\quad {\rm if}\quad n\ge \nB %(50)
    \\ 
     \left(\Hk\psi\right)^B(n) &= 0 \quad {\rm if}\quad n< \nB. %(51)
     \end{aligned}
     \label{Hk-def}
   \end{equation}
   Each $\Hk$ is a self-adjoint operator acting on the Hilbert space given by \eqref{psi-mn2}, \eqref{psiABout}. The operator $\Hk$ is equivalent  to the operator $H_{\sharp,k}$ mentioned in the introduction, by an obvious identification of the Hilbert spaces on which those operators act.

\section{Spectrum of $\Hk$}\label{gen_spectrum}
{
We want to understand, for each $\kp\in[0,2\pi]$, the spectrum of $\Hk$ acting in $l^2(\Z;\C^2)$ . We shall, in particular, investigate the existence of edge states which correspond to  eigenvalues of the following spectral problem for the operator $\Hk$:\\
 $E$ is an eigenvalue of $\Hk$ if there exists a non-trivial solution in $l^2(\Z;\C^2)$ of the equations
\begin{subequations}
\begin{align}
E\psi^A(n) &= \sum_{\nu=1}^3 e^{im_\nu k} \psi^B(n+n_\nu)\quad {\rm for}\quad n\ge\nA\label{psiAevp}\\ %(52)
E\psi^B(n) &= \sum_{\nu=1}^3 e^{-im_\nu k} \psi^A(n-n_\nu)\quad {\rm for}\quad n\ge\nB\label{psiBevp}\\ %(53)
\psi^A(n) &= 0 \quad {\rm for}\quad n<\nA \label{psiA-BC}\\ %(53)
\psi^B(n) &= 0 \quad {\rm for}\quad n<\nB. \label{psiB-BC}%(54)
\end{align}
\label{es-evp}
\end{subequations}
\begin{definition}\label{def:es}
A nonzero $l^2(\Z;\C^2)$ solution of the eigenvalue problem \eqref{es-evp} is called an {\it edge state} for the quasimomentum $k\in[0,2\pi]$.
\end{definition}
}{
In the remainder of this section we  study general properties of the spectrum of $\Hk$ for $\kp\in[0,2\pi]$, in particular its essential spectrum. In Section \ref{sec:0energy} we carry out a comprehensive study of zero energy ($E=0$) edge states. In Section \ref{numerics} we present a numerical study of the edge state eigenvalue problem
 for general energies, $E$. 
}
\subsection{Essential spectrum of $\HTB_\sharp(k)$}\label{ess-spec}

For a fixed $\kp\in[0,2\pi]$, the essential spectrum of $\Hk$ can be computed from the spectrum of $\Hbulk$ defined in \eqref{Hpsi}. Since the bulk Hamiltonian, $\Hbulk$ is invariant under translations by the vector $\bv_1$ it  decomposes as a direct integral over $\kp\in\R/2\pi\Z$ of the Hamiltonian
 $\Hbulk(\kp)$, defined by
 \begin{equation*}
   \begin{aligned}
   \left(\Hbulk(k)\psi\right)^A(n) &= \sum_{\nu=1}^3 e^{i\kp m_\nu} \psi^B(n+n_\nu)\quad n\in\Z\\ %(48)
     \left(\Hbulk(k)\psi\right)^B(n) &= \sum_{\nu=1}^3 e^{-i\kp m_\nu}\psi^A(n-n_\nu)\quad n\in\Z .%(50)
     \end{aligned}
     \label{Hbulk_k-def}
   \end{equation*}
One can show, by a Weyl sequence argument,
 that the essential spectrum of $\Hk$ is equal to that of $\Hbulk(\kp)$, and consists of those values of $E$ for which  there exists  a non-trivial solution $\psi$, of sub-exponential growth, of the following system of equations:
 \begin{align}
 \sum_{\nu=1}^3 e^{im_\nu\kpar}\psi^B(n+n_\nu)  &=  E\psi^A(n),\qquad n\in\Z, \label{HA-bulk}\\
 \sum_{\nu=1}^3 e^{-im_\nu\kpar}\psi^A(n-n_\nu)  &= E\psi^B(n),\qquad n\in\Z.\label{HB-bulk}
 \end{align}
This coincides with  \eqref{psiAevp}-\eqref{psiBevp} for  $n\gg\max\{\nA,\nB\}$. 
 
 Fix $E\in\R$ and $\kpar\in[0,2\pi]$.  It is natural to consider particular solutions of \eqref{HA-bulk}, \eqref{HB-bulk} of the form: $\psi(n)=\zeta^n\xi\ $, where $(\zeta,\xi)\in\C\times\C^2$ satisfies the  linear system \ 
\begin{equation} 
\left(\ \mathscr{P}_\kpar(\zeta) - E I\ \right)\begin{pmatrix}\xi^A\\ \xi^B\end{pmatrix} = \begin{pmatrix}0\\ 0\end{pmatrix},\quad \xi=\begin{pmatrix}\xi^A\\ \xi^B \end{pmatrix}.
\label{Tkpar}\end{equation}
Here,
 \begin{align}
 \mathscr{P}_\kpar(\zeta) := \begin{pmatrix} 
0 &P_+(\zeta,\kpar) \\
P_-(\zeta,\kpar)&0
\end{pmatrix}\quad\text{and}\quad P_\pm(\zeta,\kpar) := \sum_{j=1}^3 e^{\pm i\kpar m_j}\zeta^{\pm n_j}.\label{P+-}
 \end{align}

 When there is no ambiguity, we shall frequently suppress the dependence of $P_\pm$ on $\kpar$.
 Note that 
 \[ P_-(\zeta) = \overline{P_+(1/\overline\zeta)},\quad \text{and}\quad |\zeta|=1\ \ \textrm{implies}\ \ P_-(\zeta)=\overline{P_+(\zeta)}.\]
 
\begin{proposition} \label{ess_spec}
 For $k\in[0,2\pi]$, we have
$$ E\in {\rm spec}_{\rm ess}(\Hk)\quad\iff\quad\exists k_\perp\in[0,2\pi],\quad E^2=|h(k_\perp,\kpar)|^2$$
where $h(\kperp,\kpar) := \sum_{\nu=1}^3 e^{i\kpar m_\nu}\ e^{i\kperp n_\nu}$. Furthermore, if $E\in {\rm spec}_{\rm ess}(\Hk)$, then a non-trivial solution $\psi$ of (\ref{HA-bulk}-\ref{HB-bulk}) is given by $\psi(n)=e^{in\kperp }\xi $ where $\xi\in\mathbb{C}^2-\{0\}$ solves
\[
\Hbulk(\kperp,\kpar)\xi  = E\xi ,\]
where
\begin{equation} \Hbulk(\kperp,\kpar)\xi := \begin{pmatrix} 0 & h(\kperp,\kpar) \\ h^*(\kperp,\kpar) & 0 \end{pmatrix}\xi. 
\label{Hbulk-def}\end{equation}
Here, $z^*$ denotes the complex conjugate of $z$.
 \end{proposition}
\begin{proof}[Proof of Proposition \ref{ess_spec}] Since $\Hbulk(\kp)$  is invariant under translations by the vector $\bv_2$, it decomposes as a direct integral over $\kperp\in\R/2\pi\Z$ of the matrix $\Hbulk(\kperp,\kpar)$. If $\psi\in l^2(\Z)$, 
\begin{equation*} \label{eq:Hpar_TFB}
	[\Hbulk(\kpar)\psi](n) = \frac{1}{2\pi} \int_0^{2\pi} e^{i\kperp n} \Hbulk(\kperp,\kpar) \hat{\psi}(\kperp) d\kperp,
	\end{equation*}
where $\hat{\psi}(\kperp)\in \C^2$ is the discrete Fourier transform of $\psi$. The essential spectrum of $\Hbulk(\kp)$  is nothing but the union of the spectra of $\Hbulk(\kperp,\kpar)$ for $\kperp\in[0,2\pi]$.

It is easy to see that for each $k\in[0,2\pi]$, $\kperp\mapsto h(\kperp,\kpar)$ is a $2\pi-$ periodic continuous function. From Proposition \ref{ess_spec}, we deduce that for a fixed $\kpar\in[0,2\pi]$, the range of the maps 
\begin{equation*} \kperp\mapsto \pm|h(\kperp,\kpar)| \label{2bands}\end{equation*}
sweeps out the essential spectrum of $\Hbulk(\kp)$  and therefore that of $\Hk$. 
\end{proof} 
~\\\\
	Let us define $X_\text{bulk}:=l^2(\Z,\C^2)$ and $X_\sharp$ to be the subspace of elements of $X_\text{bulk}$ whose components $\psi^A(n)$ and  $\psi^B(n)$ vanish, respectively, for $n<\nA$ and $n<\nB$. Let $\mathfrak{i}$ denote the inclusion map from  $X_\sharp$ to $X_\text{bulk}$. Then, $H_\sharp(\kpar)$ and $\Hbulk(\kpar)$ are 
    related by
     \begin{equation} H_\sharp(\kpar)=\mathfrak{i}^*H_{\rm bulk}(\kpar)\mathfrak{i},\label{Hconj}
     \end{equation}
     thus making explicit the self-adjointness of $H_\sharp(\kpar)$ and showing that
  {
     \begin{equation} \| H_\sharp(\kpar)\|_{\mathcal{B}(X_\sharp)} =
      \| H_{\rm bulk}(\kpar)\|_{\mathcal{B}(X_{\rm bulk})}.\label{Hnorms}
      \end{equation} 
      } 

  {    Note that we can write $\Hbulk(\kperp,\kpar)$, defined in \eqref{Hbulk-def}, 
as
\[ \Hbulk(\kperp,\kpar) = \sigma_1\ \Re(h(\kperp,\kpar)) - \sigma_2\ \Im(h(\kperp,\kpar)) ,\]
where $\sigma_j$, $j=1,2,3$,  denote the Pauli matrices; see \eqref{pauli123}. Since  $\sigma_3$ anti-commutes with $\sigma_1$ and $\sigma_2$, we deduce that
$\sigma_3\Hbulk(\kperp,\kpar) = - \Hbulk(\kperp,\kpar)\sigma_3$. 
 Defining $\Sigma_3 $ as acting on a set of $A-$ and $B-$ amplitudes by transforming $\psi^A(n)$ to $\psi^A(n)$
 and $\psi^B(n)$ to $-\psi^B(n)$, we have that  $\Sigma_3$ commutes with 
$\mathfrak{i}^*$ and $\mathfrak{i}$. 
%Hence, $\Sigma_3H_\sharp=-H_\sharp\Sigma_3$ and we have the symmetry of the full spectrum of $H_\sharp$. 
%Hence, if we introduce $\Sigma_3: l^2\to l^2$, $\Sigma_3\psi =  \{\sigma_3\psi_n\}_n,$ from \eqref{eq:Hpar_TFB} 
Hence,  $\Sigma_3\Hbulk(\kpar)=-\Hbulk(\kpar)\Sigma_3$ and via \eqref{Hconj}  $\Sigma_3H_\sharp(\kpar)=-H_\sharp(\kpar) \Sigma_3$. This implies the symmetry about 0-energy of the full spectrum of $\Hbulk(k)$ and relates the modes  associated with  any $\pm E$  in the spectrum of  $\Hbulk(k)$.
}

{The previous discussion together with  the observation
 $h(\kperp,\kpar)=h^*(2\pi-\kperp,2\pi-\kpar)$, for all $\kperp,\kpar\in[0,2\pi]$,  implies the following:
 \begin{proposition}[Symmetry of the spectrum of $H_\sharp(\kpar)$]\label{lem:symmetry}
For $\kpar\in[0,2\pi]$, we have
\[
	E\in {\rm spec}(\Hk)\quad\iff\quad-E\in {\rm spec}(\Hk)
\]
and
\[
	E\in {\rm spec}(\Hk)\quad\iff\quad E\in {\rm spec}(H_\sharp(2\pi-k)).
\]
\end{proposition}
}

\subsection{Bands and gaps in the spectrum of $\HTB_\sharp(\kpar)$ as $\kpar$ varies}

 Since for any $\kpar\in[0,2\pi]$, $\kperp\mapsto\pm|h(\kperp,\kpar)|$ sweeps out symmetric intervals of the essential spectrum of $\HTB_\sharp(\kpar)$, we have
 \[
 	{\rm spec}_{\rm ess}(\HTB_\sharp(\kpar))=\Big[-\max_{\kperp}|h(\kperp,\kpar)|,-\min_{\kperp}|h(\kperp,\kpar)|\Big]\;\bigcup\;\Big[\min_{\kperp}|h(\kperp,\kpar)|,\max_{\kperp}|h(\kperp,\kpar)|\Big].
 \]
 It is easy to see that $|h(\kperp,\kpar)|\leq 3$ for all $(\kperp,\kpar)$ and $|h(\kperp,\kpar)|= 3$ if and only if $(\kperp,\kpar)=(0\,\text{mod}\,2\pi,0\,\text{mod}\,2\pi)$. This implies that ${\rm spec}_{\rm ess}(\Hk)\subset[-3,3]$. It appears that the boundary curves of the essential spectrum are monotone away from high symmetry points but we have not proven this.
 % FOOTNOTE MOVED INTO THE TEXT FOR CPAM SUBMISSION \footnote{
%{It appears that the boundary curves of the essential spectrum are monotone away from high symmetry points but we have not proven this. }}
It is of interest to determine for which values of $\kpar\in[0,2\pi]$ is {the spectrum of the }Hamiltonian $\Hbulk(\kpar)$ {not} gapped near 0. By continuity this is equivalent to determining the values of $\kpar$ for which zero is in the essential spectrum of $\HTB_\sharp(\kpar)$. 
 
 \begin{proposition}\label{nogap}
 $E=0$ is in the spectrum of $\Hbulk(\kpar)$ (that is,  $\Hbulk(\kpar)$ is not gapped) if and only if $\kpar$ is given by
 \begin{align*}
(i)\ \ \hat{k} &=0\ \textrm{or}\ 2\pi,\quad \textrm{if}\quad  a_{11}-a_{12}=0\quad \textrm{mod}\ 3
\quad \textrm{(AC case)}\\
(ii)\ \ \hat{k}  &= \frac{2\pi}{3}\ \textrm{or}\ \frac{4\pi}{3},\quad \textrm{if}\quad  a_{11}-a_{12}=\pm 1\quad \textrm{mod}\ 3\quad \textrm{(ZZ case)}.
\end{align*}
 \end{proposition}
 
 \begin{figure}[htbp]
\centering
\includegraphics[width=.8\textwidth]{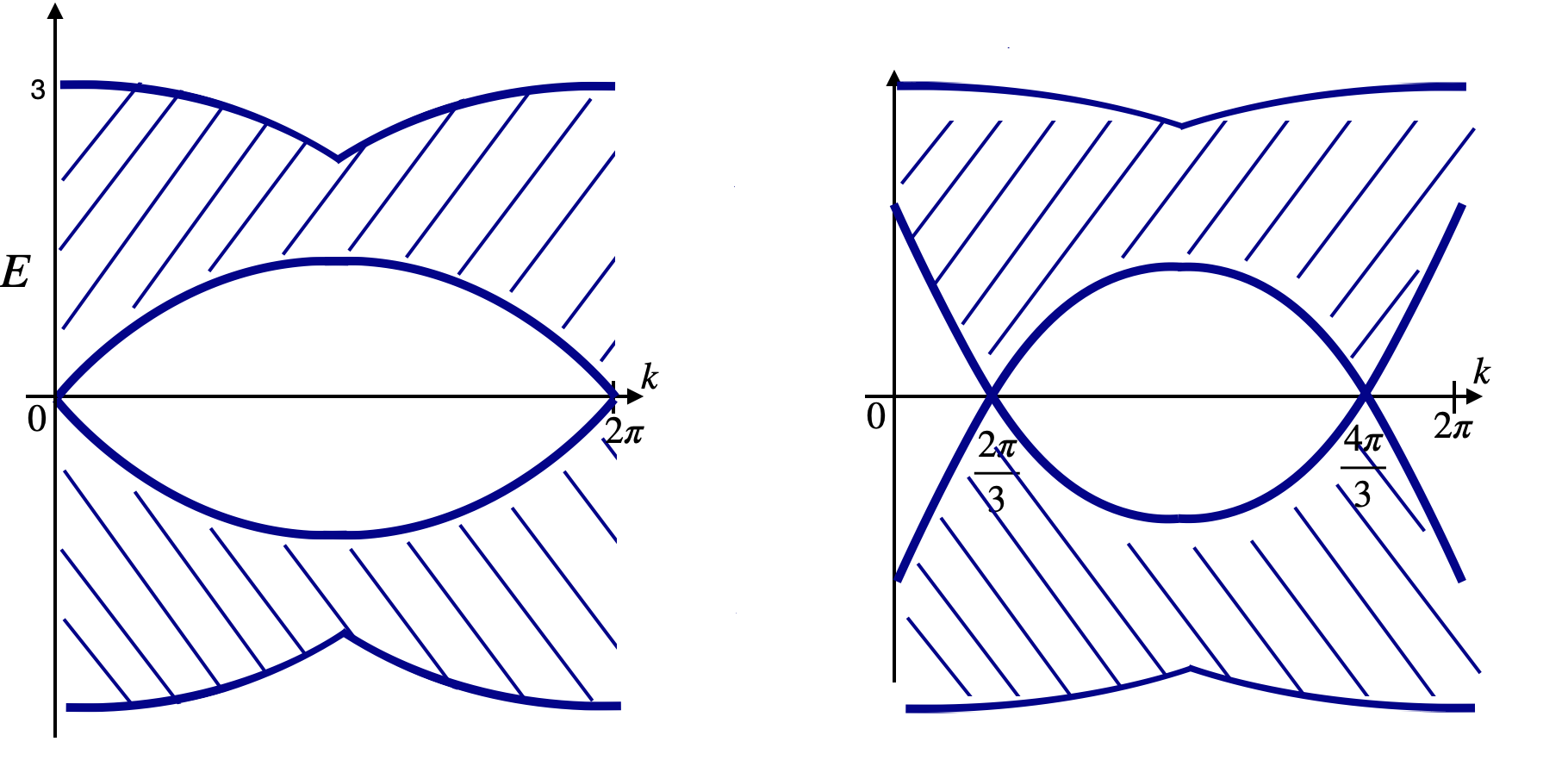}
\caption{Shaded regions indicate essential spectrum of $\HTB_{\rm bulk}(\kpar)$ and $H_\sharp(\kpar)$ for different types of rational edges;
armchair type edge (left) and zigzag type edge (right). 
}
 \label{fig:ess-spec}
\end{figure}
 \begin{figure}[htbp]
 	\begin{center}
 		\begin{tikzpicture}
			
 			\begin{scope}[shift={(-5.0,0)},scale=0.7,transform shape]	
 				\node at (0,0) {\includegraphics[height=5cm,width=6.5cm]{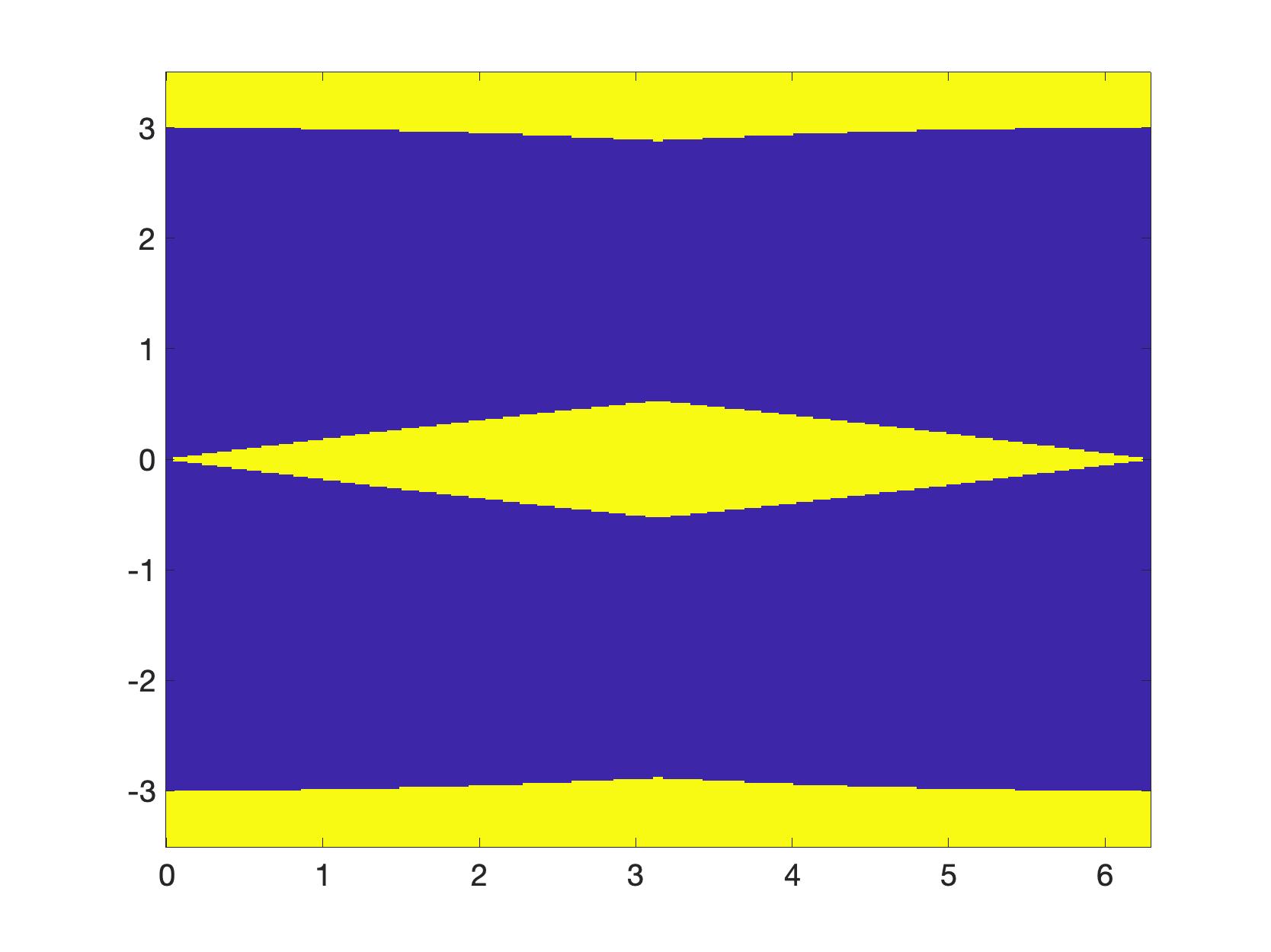}};
 				\draw (0.,2.5) node{$a_{11}=4,\;a_{12}=1$};
				\draw (-3,0) node{$E$};
				\draw (0,-2.5) node{$k$};
 			\end{scope}
			
 			\begin{scope}[shift={(-1,0)},scale=0.7,transform shape]	
 				\node at (0,0) {\includegraphics[height=5cm,width=6.5cm]{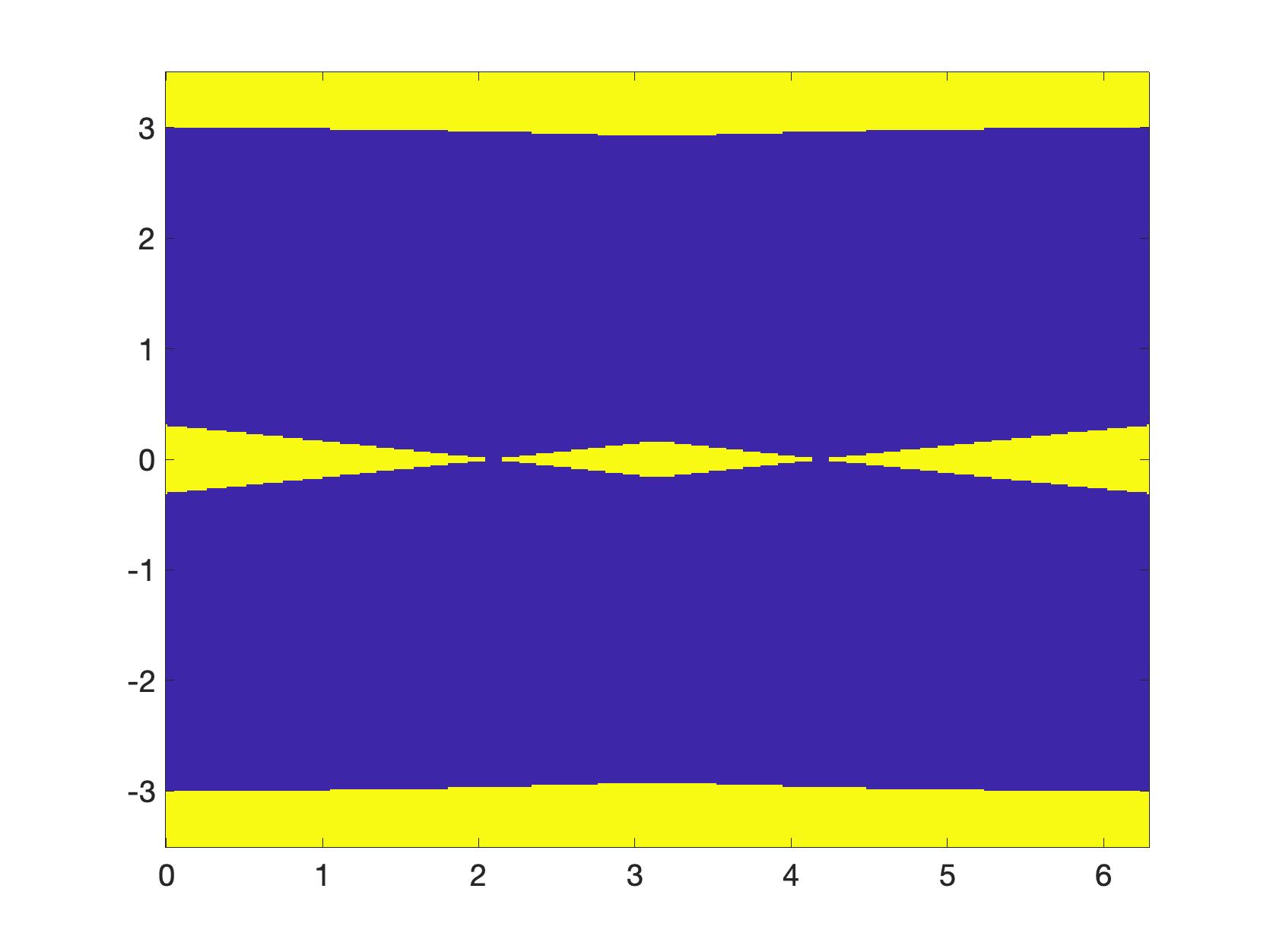}};
 				\draw (0.,2.5) node{$a_{11}=5,\;a_{12}=1$};
				\draw (0,-2.5) node{$k$};
 			\end{scope}
			
 			\begin{scope}[shift={(3.0,0)},scale=0.7,transform shape]	
 				\node at (0,0) {\includegraphics[height=5cm,width=6.5cm]{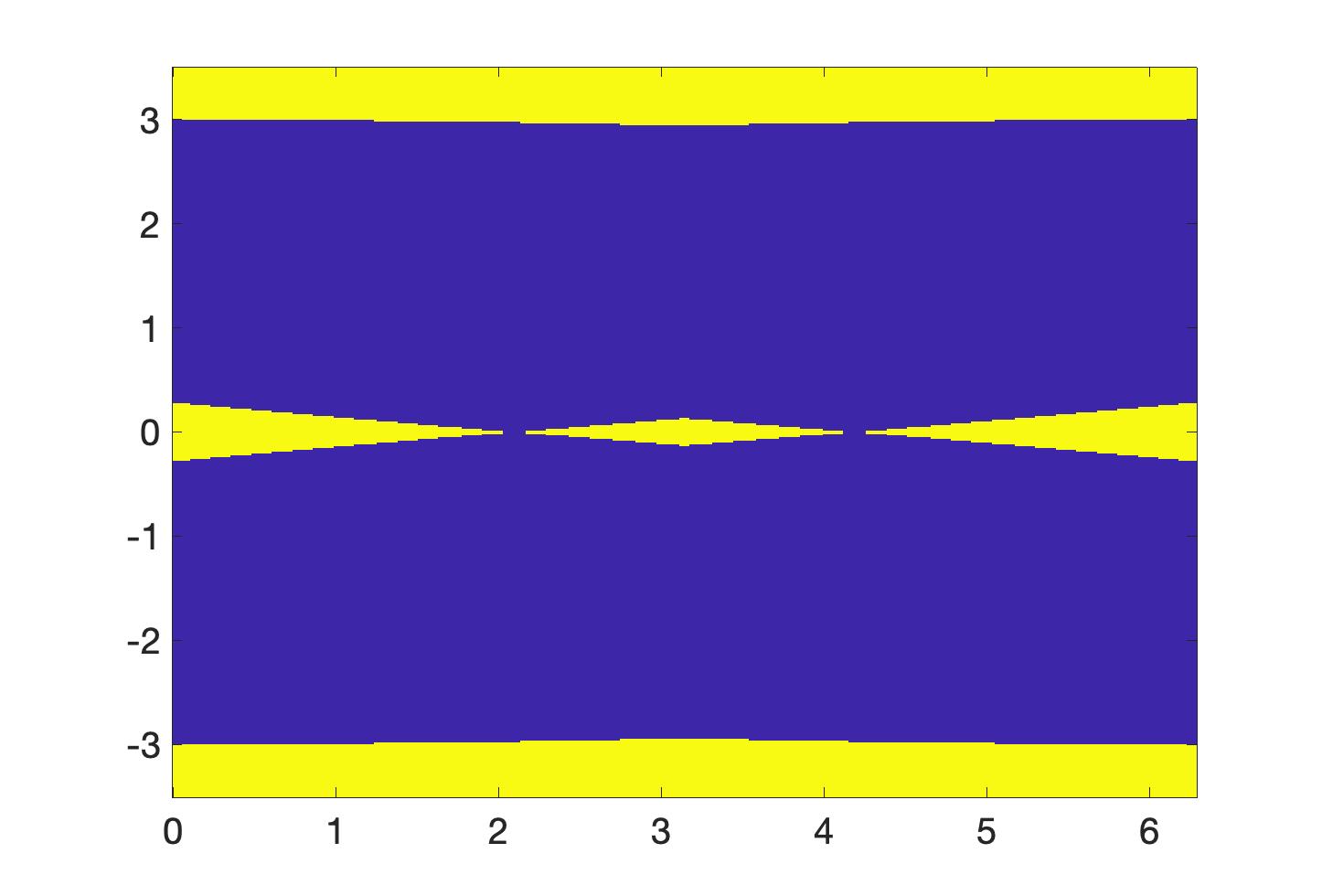}};
 				\draw (0.,2.5) node{$a_{11}=6,\;a_{12}=1$};
				\draw (0,-2.5) node{$k$};
 			\end{scope}
 		\end{tikzpicture}
 		\caption{For each $k\in[0,2\pi]$, we indicate in blue (resp. in yellow) the essential spectrum (resp. the resolvent set or the gaps) of $\Hbulk(k)$ in the interval $E\in(-3.5,3.5)$.}
 		\label{fig:essential_spectrum_large}
 	\end{center}
 \end{figure}
 \begin{figure}[htbp]
 	\begin{center}
 		\begin{tikzpicture}
			
 			\begin{scope}[shift={(-5.0,0)},scale=0.7,transform shape]	
 				\node at (0,0) {\includegraphics[height=5cm,width=6.5cm]{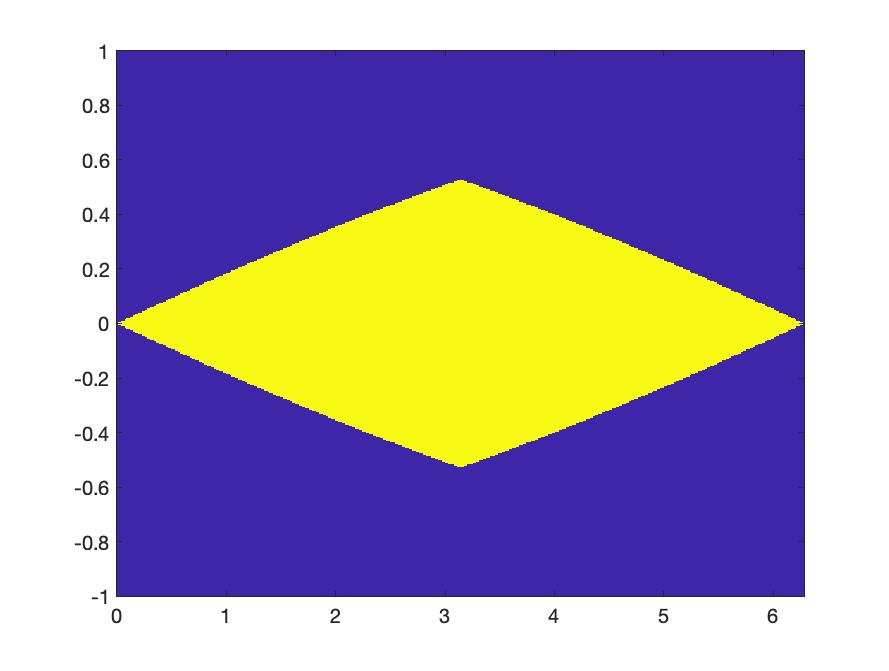}};
 				\draw (0.,2.5) node{$a_{11}=4,\;a_{12}=1$};
				\draw (-3,0) node{$E$};
				\draw (0,-2.5) node{$k$};
 			\end{scope}
			
 			\begin{scope}[shift={(-1,0)},scale=0.7,transform shape]	
 				\node at (0,0) {\includegraphics[height=5cm,width=6.5cm]{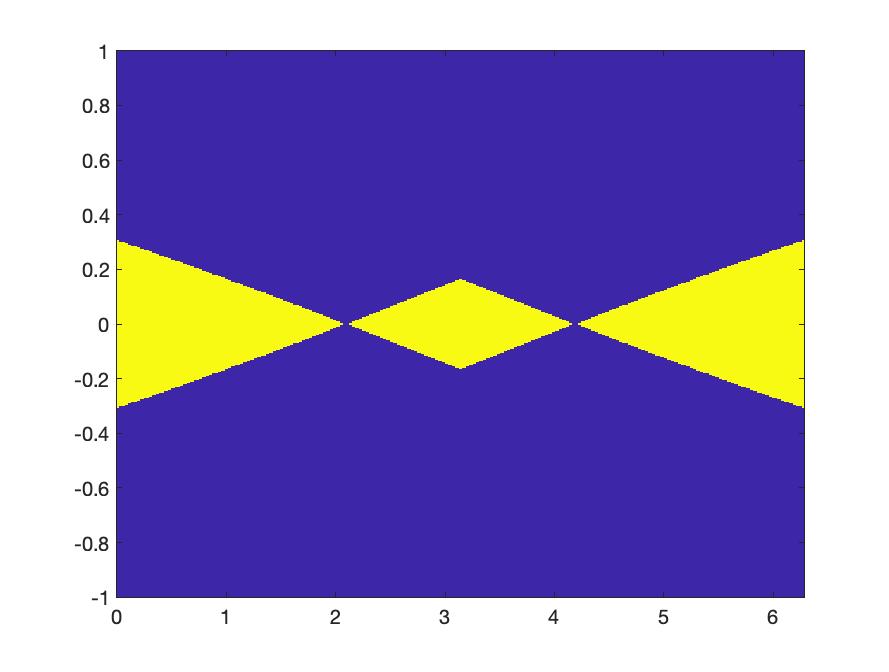}};
 				\draw (0.,2.5) node{$a_{11}=5,\;a_{12}=1$};
				\draw (0,-2.5) node{$k$};
 			\end{scope}
			
 			\begin{scope}[shift={(3.0,0)},scale=0.7,transform shape]	
 				\node at (0,0) {\includegraphics[height=5cm,width=6.5cm]{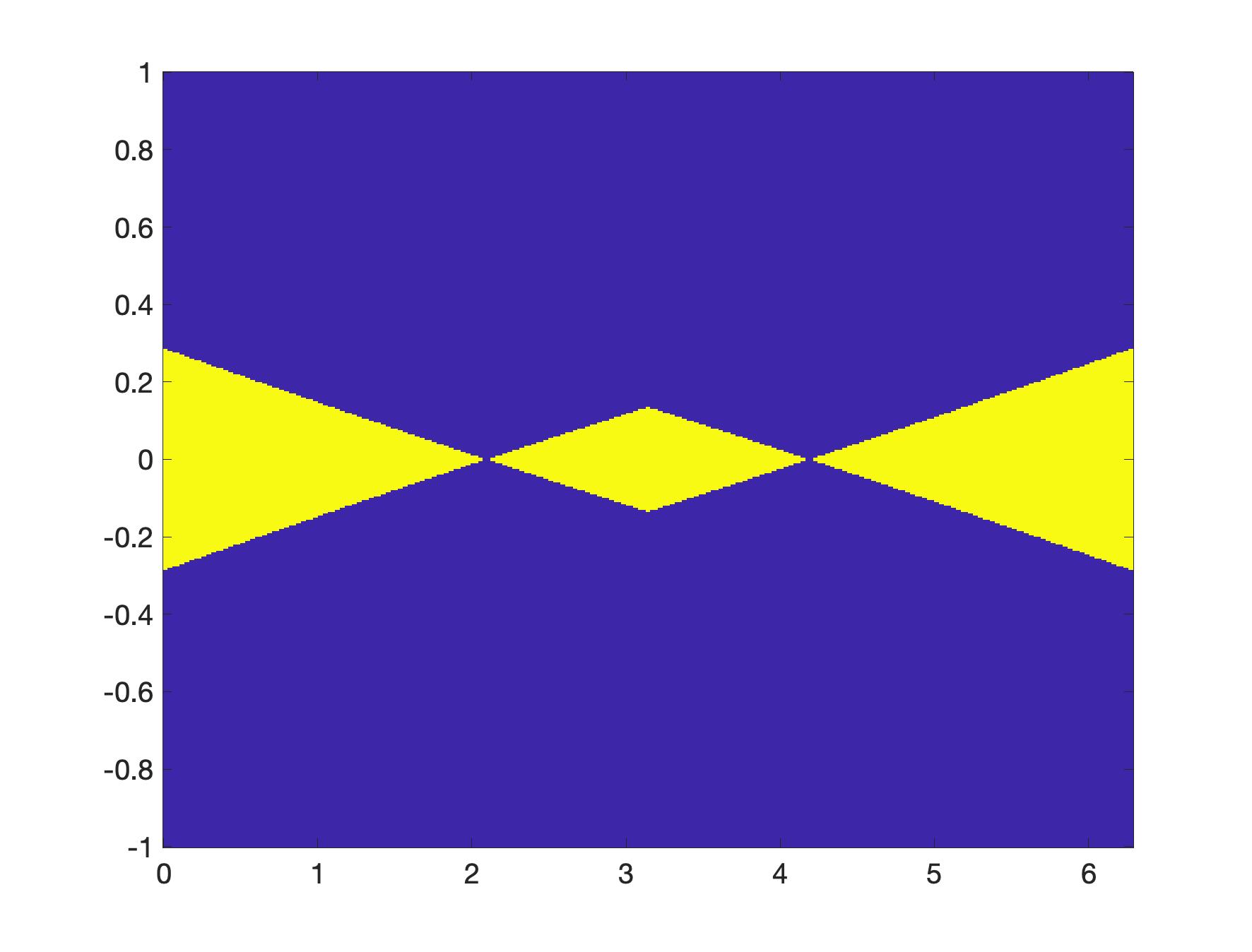}};
 				\draw (0.,2.5) node{$a_{11}=6,\;a_{12}=1$};
				\draw (0,-2.5) node{$k$};
 			\end{scope}
 		\end{tikzpicture}
 		\caption{For each $k\in[0,2\pi]$, we indicate in blue (resp. in yellow) the essential spectrum (resp. the resolvent set or the gaps) of $\Hbulk(k)$ in the interval $E\in(-1,1)$.}
 		\label{fig:essential_spectrum}
 	\end{center}
 \end{figure}
 \begin{remark}\label{ou_ptspec} By \eqref{Hnorms}, any discrete spectrum of 
	 $H_\sharp(\kpar)$ must lie: either in the bounded component of the complement (with respect to $\R$) of 
	  the essential spectrum (bounded {light} regions) or be embedded in the essential spectrum.
	  { We have not investigated the existence or non-existence of embedded eigenvalues for general rational edges. Note however that 
	  for the case of the classical (ordinary) zigzag edge, $E=+1$ and $E=-1$ are  eigenvalues of $H_\sharp(\pi)$ of infinite multiplicity that are embedded in the essential spectrum; see \cite[Theorem 2.2]{FW:20}.}
	 \end{remark}
\begin{proof}[Proof of Proposition \ref{nogap}] {By Proposition \ref{ess_spec},} we must determine those values of $\kpar\in[0,2\pi]$ for which 
there exists $\kperp$ with $h(\kpar,\kperp)=0$. We have
\begin{equation*}
h(\kpar,\kperp)=0 \;\iff\;  1 + e^{i(\beta_1\kpar+\gamma_1\kperp)} + e^{i(\beta_2\kpar +\gamma_2\kperp)} =0, \label{3pts}
\end{equation*}
where
\begin{align}
\beta_1=m_2-m_1,\quad \beta_2=m_3-m_1\nn\\
\gamma_1=n_2-n_1,\quad \gamma_2=n_3-n_1.\label{bg-def}\end{align}
We shall later make use of the additional notations:
 \[ \beta=\left(\begin{matrix}\beta_1\\ \beta_2\end{matrix}\right),\quad\text{and}\quad \gamma=\left(\begin{matrix}\gamma_1\\ \gamma_2\end{matrix}\right).\]
The centroid of the three points $1,\,e^{i(\beta_1\kpar+\gamma_1\kperp)},\,e^{i(\beta_2\kpar +\gamma_2\kperp)}$ on the unit circle is at the origin if and only if the three points are located at vertices of an equilateral triangle. Since one of the three points is fixed at $1$, there are two cases: $\hsigma=1$ or $\hsigma=-1$, with 
\begin{align*}
\beta_1\kpar + \gamma_1\kperp &=  \hsigma \frac{2\pi}{3}\quad \textrm{mod}\ 2\pi\\
\beta_2 \kpar + \gamma_2 \kperp &=  -\hsigma \frac{2\pi}{3}\quad \textrm{mod}\ 2\pi\ .
\end{align*}
Solving for $\kpar$ and $\kperp$, we obtain the two cases:
\begin{align}
{k}  &=  \frac{2\pi}{3} \hsigma \left( \gamma_1+\gamma_2\right)\det[\beta\ \gamma]\quad \textrm{mod}\ 2\pi
\label{h_kpar} \\
{k}_\perp  &=  -\frac{2\pi}{3} \hsigma \left(\beta_1+\beta_2\right)\det[\beta\ \gamma] \quad \textrm{mod}\ 2\pi,
 \label{h_kperp}\end{align}
where $\det[\beta\ \gamma]=\pm1;\quad \textrm{see  \eqref{det1-app}. }$
Finally, from \eqref{bg-def} we have 
 \begin{equation}\label{eq:sum_gamma} \gamma_1 + \gamma_2= (n_1+n_2+n_3)-3n_1 \underset{\textrm{\eqref{perm-n}} }{=}(\ttn_1+\ttn_2 +\ttn_3)-3n_1\underset{\textrm{\eqref{k1k2s1s2}} }{=}-s_2-3n_1.
 \end{equation}
 and similarly 
\begin{equation}\label{eq:sum_beta}
 	\beta_1+\beta_2=-s_1-3m_1
 \end{equation}
  %, $(n_1+n_2+n_3)=\ttn_1+\ttn_2+\ttn_3$.\eqref{k1k2s1s2} and \eqref{ttn-ttm}, we  obtain that  $ \gamma_1 + \gamma_2=-s_2-3n_1$. 
 Equations \eqref{det1} and \eqref{k1k2s1s2} imply that  $s_1$ and $s_2$ cannot be both zero. Hence, \eqref{eq:sum_gamma} and \eqref{eq:sum_beta} yield {
  \begin{equation}
  	\label{eq:ess_spec_arm}
	\begin{array}{l}
	\text{For}\; a_{11}-a_{12}=0\, \text{mod}\, 3\; (s_2=0)\\
	\text{$E=0$ is in the essential spectrum of 
 $\HTB_\sharp(\kpar)$ if and only if ${k} =0$  mod $2\pi$.}\\
 \displaystyle\text{In this case, there are two distinct } {k}_\perp \text{ arising from \eqref{h_kperp}, namely }{k}_\perp=\frac{2\pi}{3}\text{ or }\frac{4\pi}{3}.
 \end{array}
  \end{equation}
   \begin{equation}
   	\label{eq:ess_spec_zig}
   	\begin{array}{l}
   	\text{For}\; a_{11}-a_{12}=\pm 1\, \text{mod}\, 3\; (s_2=\pm 1)\\
   	\displaystyle\text{$E=0$ is in the essential spectrum of 
    $\HTB_\sharp(\kpar)$ if and only if } {\kp} =\frac{2\pi}3\text{ or }{\kp} =\frac{4\pi}3\;\text{mod}\; 2\pi.\\
    \displaystyle\text{Each of those $k$'s gives rise to a single $k_\perp$ via \eqref{h_kperp}}.
    \end{array}
   \end{equation}
  }

\end{proof}

\subsection{The wedge of the edge  }\label{woe}
An analytical theory of the dispersive edge states documented numerically in Section \ref{numerical} is work in progress.  The numerical simulations of Section \ref{numerical} show that in a neighborhood of a bifurcation point, $(\kpar,E)=(\hat{\kp},0)$, a dispersive edge state curve emanates into a locally wedge shaped region. It is of interest to understand the dependence of the wedge-opening on the edge-direction $\bv_1=a_{11}\ocirc{\bv}_1+a_{12}\ocirc{\bv}_2$. 
\begin{proposition}[Wedge of the Edge]\label{woe-prop} Assume that $(\kpar,E)=(\hat\kpar,0)$ is a crossing of bands
in the essential spectrum of $\HTB_\sharp(\kpar)$ . Hence, for the zigzag type edges ($s_2=\pm1$) $\hat\kpar\in\{2\pi/3,4\pi/3\}$ and for armchair type edges $\hat\kpar\in\{0,2\pi\}$.
For $\kappa$ small, let   $\kappa\mapsto\eta_\pm(\hat\kpar+\kappa)$ locally define  the upper and lower bounding curves of the essential spectrum of $\HTB_\sharp$ near $(\kpar,E)=(\hat\kpar,0)$. Then,  
 \begin{equation}
\quad \eta_\pm(\hat\kpar+\kappa) = \left( \pm\frac{\sqrt{3}}{2} \cdot \frac{1}{|\bv_1|} + o(1) \right) |\kappa|,\quad \textrm{as $\kappa\rightarrow 0$,}
  \label{m-wedge1}\end{equation}
  where $|\bv_1|=\left(a_{11}^2 + a_{11}a_{12}+ a_{12}^2\right)^{\frac12}$,
 see Figure \ref{fig:ess-spec_wedge}. 
\end{proposition} 

We prove Proposition \ref{woe-prop} in Appendix \ref{app:woe-prop}.
 \begin{figure}[htbp]
\centering
\includegraphics[width=.8\textwidth]{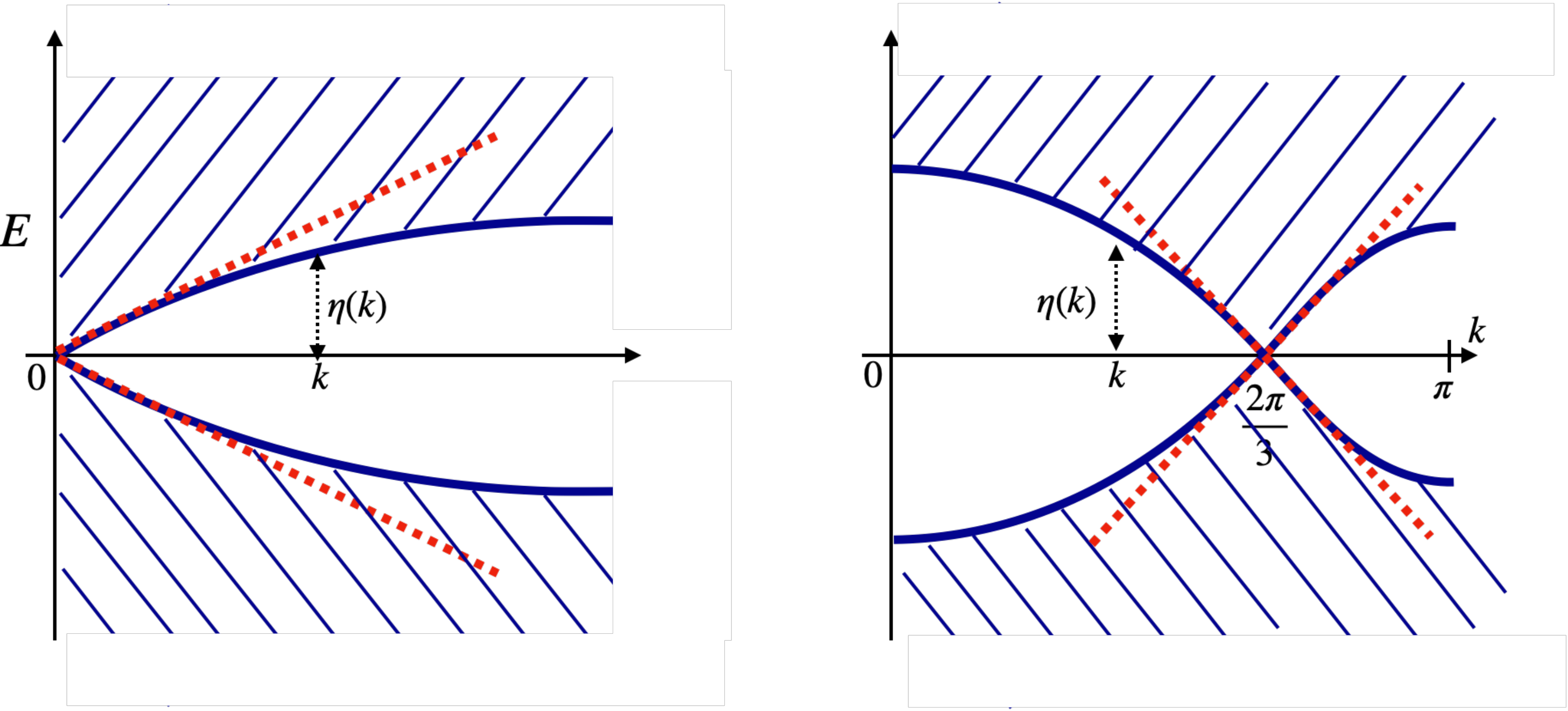}
\caption{Schematic for Edge of the Wedge Proposition \ref{woe-prop}. Blue regions are essential spectrum and red dotted lines have slopes of magnitude $\alpha=  \frac{\sqrt{3}}{2}|\bv_1|^{-1}=\frac{\sqrt{3}}{2} \left(a_{11}^2 + a_{11}a_{12}+ a_{12}^2\right)^{-1/2}$: 
armchair type edge (left) and zigzag type edge (right).}
\label{fig:ess-spec_wedge}
\end{figure}
\section{Zero energy / flat band edge states}\label{sec:0energy}

In this section we present a complete classification of zero energy edge states  ($E=0$) according to edge-type
 and range of parallel quasimomentum $\kp$.  {The main results of this section are Theorems \ref{th:zigzag}
  and \ref{th:armchair}, which are shown to imply Theorem \ref{main-intro}, stated in the introduction.}
 
From \eqref{es-evp} we note that the wave function amplitudes at the $A-$sites and $B-$sites are completely decoupled in the case of zero energy edge states. This leads to the decoupled eigenvalue problems: \medskip

\noindent{\bf $B-$site edge state eigenvalue problem:} Find $\psi^B\in l^2(\Z;\C^2)$ such that
\begin{subequations}
\begin{align}
 \sum_{\nu=1}^3 e^{i m_\nu k} \psi^B(n+n_\nu) &=0 \quad {\rm for}\quad n\ge\nA\label{psiAevp-0}\\ %(60)
 \psi^B(n) &= 0 \quad {\rm for}\quad n<\nB %
 \end{align}
 \label{Bsite-states}
 \end{subequations}
 and the 
 
 \noindent{\bf $A-$site edge state eigenvalue problem:} Find $\psi^A\in l^2(\Z;\C^2)$ such that
 \begin{subequations}
 \begin{align}
\sum_{\nu=1}^3 e^{-i m_\nu k} \psi^A(n-n_\nu) &=0 \quad {\rm for}\quad n\ge\nB\label{psiBevp-0}\\ %(61)
\psi^A(n) &= 0 \quad {\rm for}\quad n<\nA \ .%
\end{align}
 \label{Asite-states}
\end{subequations}

 We shall see, for zigzag edges, there are subintervals, $I$, of $\kp\in[0,2\pi]$ over which the zero energy edge state eigenvalue problem \eqref{Bsite-states} or \eqref{Asite-states} has non-trivial solutions. Since the dispersion curves $\kp\in I\mapsto E(\kp)\equiv0$ are constant, these intervals of zero energy eigenstates are called {\it flat bands} or {\it zero energy flat bands}. 

{Observe that} non-zero solutions of \eqref{Bsite-states}
 give rise to zero-energy edge states that live only at the $B-$sites ({\it i.e.} $\psi^A(n)=0$ for all $n$), while non-zero solutions of \eqref{Asite-states}
 give rise to zero-energy edge states that live only at the $A-$sites ({\it i.e.} $\psi^B(n)=0$ for all $n$).

%\subsection{Solving the zero energy edge state eigenvalue problem}\label{solve-0energy}

{Any solution of \eqref{psiAevp-0} is connected to solutions $\zeta$ of the {\it indicial equation} 
\begin{equation}
P_+(\zeta,\kp) = 0;
\label{B-indicial}
\end{equation}
where $P_+$ is defined in \eqref{P+-}.}
Recalling the ordering $n_1<n_2<n_3$ we rewrite \eqref{B-indicial} as the polynomial equation:
\begin{equation}
p_+(\zeta,\kp)=0 \quad\text{where}\;p_+(\zeta,\kp):= 1 + e^{i(m_2-m_1)\kp} \zeta^{n_2-n_1} + e^{i(m_3-m_1)\kp} \zeta^{n_3-n_1} .
\label{B-poly} %(62)
\end{equation}
{Similarly, any solution of  \eqref{psiBevp-0} is connected to solutions $\zeta$ of the indicial equation
  \begin{equation}
P_-(\zeta,\kp) = 0;
\label{A-indicial}
\end{equation}
where $P_-$ is defined in \eqref{P+-}, and the roots of the polynomial equation
\begin{equation}
p_-(\zeta,\kp)=0\quad\text{where}\;  p_-(\zeta,\kp):= 1 + e^{i(m_3-m_2)\kp} \zeta^{n_3-n_2} + e^{i(m_3-m_1)\kp} \zeta^{n_3-n_1} .
\label{A-poly} %(63)
\end{equation}}
The polynomials $p_\pm(\zeta;\kp)$ both have degree $n_3-n_1$.
We next review the connection between the polynomial equations \eqref{B-poly}, \eqref{A-poly} and  the discrete boundary value problems  \eqref{Bsite-states}, \eqref{Asite-states}.
We focus initially on the properties of \eqref{B-poly} as a simple transformation enables us to map conclusions about
\eqref{B-poly} to conclusions about   \eqref{A-poly}.\medskip

By Proposition \ref{nogap}, $0$ is in the essential spectrum of $\HTB_\sharp(\kpar)$ for $k=0$ or $2\pi$ for the armchair-type edges and $k=2\pi/3$ or $4\pi/3$ for zigzag type edges. By Proposition \ref{ess_spec}, this means that, except for these values, equations \eqref{B-poly} and \eqref{A-poly} have no roots (simple or multiple) on the unit circle. 
{Moreover, in Lemma \ref{lem:simple_roots} below we shall prove that any roots of \eqref{B-poly} and \eqref{A-poly} are simple.}

\subsection{Solving for edge states that live on $B-$ sites}\label{sec:Bsite-es}
%
%Assume that $\kpar\notin\{0,2\pi/3,4\pi/3,2\pi\}$. 
Introduce $\pcount$, the number of roots of \eqref{B-poly} located in the open unit disc in $\C$:
\begin{equation}\label{eq:number_roots}
 \pcount = \#\{ |\zeta|<1 : p_+(\zeta,\kp)=0 \} \quad \textrm{for $\kpar\in[0,2\pi]$}.\end{equation}
Let $\zeta_1,\dots,\zeta_\pcount$ denote the $\pcount$ roots  inside the unit circle, and let $\zeta_{\pcount+1},\dots, \zeta_{n_3-n_1}$ denote the roots that lie outside the unit circle. Then, since the roots are all distinct (see Lemma \ref{lem:simple_roots}), the general solution of the equation \eqref{psiAevp-0}, which can be rewritten equivalently as
\begin{equation}
\sum_{\nu=1}^3 e^{i m_\nu k} \psi^B(n+n_\nu-n_1) = 0 \quad {\rm for}\quad n\ge\nA+n_1\ ,
\label{psiB0-eqn}% (64)
\end{equation}
is of the form
\begin{subequations}
\begin{align}
& \psi^B(n) = \sum_{j=1}^{n_3-n_1}A_j\zeta_j^n\quad {\rm for}\quad n\ge \nA+n_1.
 \label{psiBlinco}\\
&\textrm{\quad with $\psi^B(n)$ arbitrary for $n<\nA+n_1$ .}
\end{align}
\label{psi-gen}
\end{subequations}
Here, $A_1,\dots, A_{n_3-n_1}$ are arbitrary coefficients. Note that for $n\gg\nB$, where the terms of \eqref{psiB0-eqn} only sample
 sites within the bulk, we have the expression for $\psi^B(n)$ in \eqref{psiBlinco}.  Since $n_1<n_2<n_3$, equation \eqref{psiB0-eqn} expresses $\psi^B(n)$  in terms of its values at sites with indices larger than $n$. Therefore, we can  use \eqref{psiB0-eqn}  to recur
  down from the bulk to obtain \eqref{psiB0-eqn}  over the full range $n\ge\nA+n_1$.
% FOOTNOTE INTEGRATED TO THE TEXT FOT THE CPAM SUBMISSION  \footnote{\label{footnote:largen} For $n\gg\nB$, where the terms of \eqref{psiB0-eqn} only sample
% sites within the bulk, we have the expression for $\psi^B(n)$ in \eqref{psiBlinco}.  Since $n_1<n_2<n_3$, equation \eqref{psiB0-eqn} expresses $\psi^B(n)$  in terms of its values at sites with indices larger than $n$. Therefore, we can  use \eqref{psiB0-eqn}  to recur
%  down from the bulk to obtain \eqref{psiB0-eqn}  over the full range $n\ge\nA+n_1$.}
 %
%
 Since $|\zeta_j|<1$ for $j=1,\dots, \pcount$ and $|\zeta_j|>1$ for $j>\pcount$, it follows that the $l^2$ solutions of \eqref{psiB0-eqn} are given by
\begin{subequations}
\begin{align}
& \psi^B(n) = \sum_{j=1}^\pcount A_j\zeta_j^n\quad {\rm for}\quad n\ge \nA+n_1,\\
&\textrm{\quad $\left(\psi^B(n)\right)_{n<\nA+n_1}$ is an arbitrary vector in $l^2$.}
\end{align}
\label{psi-gen0}
\end{subequations}Note that since
 the $(n_3-n_1)$ degree polynomial in \eqref{B-poly} depends on $n_\nu$ and $k$, the number $\pcount$ depends also on these quantitites. 
In Appendix \ref{a-few-ineq} we check that $\nA+n_1\le \nB$. We conclude that the $l^2$ solutions of equations
 \eqref{Bsite-states} are precisely the vectors
 \begin{equation} %(65)
 \psi^B(n) =
  \begin{cases} 
  \sum_{j=1}^\pcount A_j\zeta_j^n & {\rm for}\quad n\ge\nA+n_1\\
0 &  {\rm for}\quad  n<\nA+n_1
\end{cases}\ ,
\label{psiB-form}
\end{equation}
such that $\psi^B(n)=0$ for $n<\nB$. Hence, the coefficients $A_1,\dots, A_\pcount$ are subject to the constraints
\begin{equation}
 \sum_{j=1}^\pcount A_j\zeta_j^n = 0 \quad {\rm for}\quad \nA+n_1\le n<\nB. %(66)
\label{0outside-B} \end{equation}
The existence of $B-$site edge states for a given value of $\kpar$, then reduces to a comparison of the number of free constants, $\pcount$,  and the number of independent equations in the linear homogeneous 
algebraic system \eqref{0outside-B}.

Our task is now to determine the number, $\pcount$,  of roots of \eqref{B-poly} in the open unit disk which varies:
   (i) with the type of edge   and (ii) as a function of the parallel quasi-momentum, $\kp$.
The following proposition, proved in Section \ref{sec:HEPs}, is the key. \\
 Recall that $s_2=a_{11}-a_{12}\ {\rm mod}\ 3$,  $s_2\in\{-1,0,1\}$.
  \begin{proposition}\label{prop:p-value}  Let $\pcount$ be given by \eqref{eq:number_roots}. Suppose that $\kpar\notin\{0,2\pi\}$ if $s_2=0$ and that $\kpar\notin\{2\pi/3,4\pi/3\}$ if $s_2=\pm1$. Then,
  \begin{equation}
 \pcount = -n_1 - s_2 \mathds{1}_{\kp\in(\frac{2\pi}3,\frac{4\pi}3)}\label{p-value}
  \end{equation} %(69)
  \end{proposition}
 From Proposition \ref{prop:p-value} we deduce:
\begin{proposition}[Conditions for $B-$site edge states]\label{Bstate-conds1}
 Let $\kpar$ be as in Proposition \ref{prop:p-value}. 
 Consider the eigenvalue problem  \eqref{Bsite-states} governing edge states which are supported on the 
 $B-$ sites of $\HH_\sharp$. Then,
 \begin{align}
& \textrm{If $s_2\mathds{1}_{\kp\in(\frac{2\pi}3,\frac{4\pi}3)}\ge \nA-\nB $, then there are no nonzero $l^2$ solutions of \eqref{Bsite-states}.} 
 \label{no-Bstates1} \\
\nonumber\end{align} %(70)
\begin{align}
& \textrm{If $s_2\mathds{1}_{\kp\in(\frac{2\pi}3,\frac{4\pi}3)}< \nA-\nB $, then the space of $l^2$ solutions of \eqref{Bsite-states},}\nonumber \\
 &\textrm{has dimension {$(\nA-\nB)-s_2\mathds{1}_{\kp\in(\frac{2\pi}3,\frac{4\pi}3)} $} 
 }.
\label{yes-Bstates1}
\end{align}%(71)
\end{proposition}
\begin{proof}[Proof of Proposition \ref{Bstate-conds1}]
The $l^2$ solutions of equations
 \eqref{psiAevp-0} are given by \eqref{psiB-form} where the coefficients $A_1,\dots, A_\pcount$ are subject to the constraints \eqref{0outside-B}.
Equations \eqref{0outside-B} is a system of $\nB-\nA-n_1$ equations in the $\pcount$ unknowns $A_1,\dots,A_\pcount$.
 Thanks to the non-vanishing of the relevant Vandermonde determinants, we deduce that
if $\pcount\le \nB-\nA-n_1$ which is equivalent to $s_2\mathds{1}_{\kp\in(\frac{2\pi}3,\frac{4\pi}3)}\ge \nA-\nB $ by Proposition \ref{prop:p-value}, then there are no nonzero $l^2$ solutions of \eqref{Bsite-states}
and
if $\pcount> \nB-\nA-n_1$ which is equivalent to $s_2\mathds{1}_{\kp\in(\frac{2\pi}3,\frac{4\pi}3)}< \nA-\nB $ by Proposition \ref{prop:p-value}, then the space of $l^2$ solutions of \eqref{Bsite-states},
 has dimension $\pcount-(\nB-\nA-n_1)=(\nA-\nB)-s_2\mathds{1}_{\kp\in(\frac{2\pi}3,\frac{4\pi}3)}$. 
\end{proof}
\subsection{Solving for edge states which live on $A-$sites}\label{sec:Asite-es}
A completely analogous discussion applies to the zero energy edge state eigenvalue problem for states which live
on the $A-$sites, \eqref{Asite-states}, and the associated polynomial \eqref{A-poly}.  Let
\begin{equation}\label{eq:number_roots_A}
\qcount= \#\{|\zeta|<1 : p_-(\zeta,\kp)=0\}\quad \textrm{for $\kpar\in[0,2\pi]$,} 
\end{equation}
and denote by $\zeta_1,\dots,\zeta_{_{\qcount}}$ the solutions of $p_-(\zeta,\kp)=0$ that are inside the open unit disc. 
Then, the $l^2$ solutions of equations
 \eqref{psiBevp-0} are precisely the vectors
 \begin{equation} %(6)
 \psi^A(n) =
  \begin{cases} 
  \sum_{j=1}^\qcount A_j\zeta_j^n & {\rm for}\quad n\ge\nB-n_3\\
0 &  {\rm for}\quad  n<\nB-n_3
\end{cases}\ ,
\label{psiA-form}
\end{equation}
where the coefficients $A_1,\dots, A_\qcount$ are subject to the constraints
\begin{equation}
 \sum_{j=1}^\qcount  A_j\zeta_j^n = 0 \quad {\rm for}\quad \nB-n_3\le n<\nA, %(6)
\label{0outside-A} \end{equation}
In Appendix \ref{a-few-ineq} we check that $\nB-n_3\le\nA$.

The number, $\qcount$, of roots of \eqref{A-poly} in the open unit disk, can be computed from $\pcount$, which was evaluated in Proposition \ref{prop:p-value}. Recall that $s_2=a_{11}-a_{12}\,\text{mod}\,3$,  $s_2\in\{-1,0,1\}$.
 \begin{proposition}\label{prop:q-value}
 Let $\qcount$ be given by \eqref{eq:number_roots_A}.	 Suppose that $\kpar\notin\{0,2\pi\}$ if $s_2=0$ and that $\kpar\notin\{2\pi/3,4\pi/3\}$ if $s_2=\pm1$. Then,
  \begin{equation}
  \qcount = n_3 + s_2 \mathds{1}_{\kp\in(\frac{2\pi}3,\frac{4\pi}3)},\quad k\in[0,2\pi].
 \label{q-value} \end{equation} %(72)
  \end{proposition}
 \begin{proof}[Proof of Proposition \ref{prop:q-value}]
Suppose $k\notin\{0,2\pi\}$, for the case $s_2=0$, and $k\notin\{2\pi/3,4\pi/3\}$, for the case  $s_2=\pm1$. Then,  $P_+(\zeta,\kp)$ and $P_-(\zeta,\kp)$ have no zeros
  on the unit circle. By the relation $P_-(\zeta,\kp)=\overline{P_+(\bar\zeta^{-1},\kp)}$ and Proposition \ref{prop:p-value},  the number of zeros of 
 $P_-(\zeta,\kp)$ (and therefore of  $p_-(\zeta,\kp)$) outside the closed unit disc is equal to $\pcount = -n_1 - s_2 \mathds{1}_{\kp\in(\frac{2\pi}3,\frac{4\pi}3)}$. Since the polynomial $p_-(\zeta,\kp)$ has $n_3-n_1$ complex roots, 
  we have that the number of roots of $p_-(\zeta,\kp)$ inside the open unit disc is $\qcount=(n_3-n_1)-\pcount % =
%    (n_3-n_1)-(-n_1 - s_2 \mathds{1}_{\kp\in(\frac{2\pi}3,\frac{4\pi}3)})
=n_3+s_2 \mathds{1}_{\kp\in(\frac{2\pi}3,\frac{4\pi}3)}$.
 \end{proof}

Equations \eqref{0outside-A} form a system of $\nA-\nB+n_3$ equations for the $\qcount$ unknowns $A_1,\dots, A_\qcount$. We conclude as in Proposition \ref{Bstate-conds1} that 
if $\qcount\le  \nA-\nB+n_3$, then there are no non-zero $l^2$ solutions of \eqref{Asite-states}, and if 
 $\qcount>  \nA-\nB+n_3$, then the space of $l^2$ solutions of \eqref{Asite-states} is $\qcount-(\nA-\nB+n_3)$ dimensional, which, by Proposition \ref{prop:q-value}, is equivalent to the following:  
 \begin{proposition}[Conditions for $A-$ site edge states]\label{Astate-conds1}
 Let $\kpar$ be as in Proposition \ref{prop:q-value}.
 Consider the eigenvalue problem  \eqref{Asite-states} governing edge states which are supported on the 
 $A-$ sites of $\HH_\sharp$. Then,
 \begin{align}
& \textrm{If $s_2\mathds{1}_{\kp\in(\frac{2\pi}3,\frac{4\pi}3)}\le \nA-\nB$, then there are no nonzero $l^2$ solutions of \eqref{Asite-states}.} 
 \label{no-Astates1} \\
\nonumber\end{align} %(70)
\begin{align}
& \textrm{ If $s_2\mathds{1}_{\kp\in(\frac{2\pi}3,\frac{4\pi}3)}> \nA-\nB $, then the space of $l^2$ solutions of \eqref{Asite-states},}\nonumber \\
 &\textrm{has dimension {$s_2\mathds{1}_{\kp\in(\frac{2\pi}3,\frac{4\pi}3)} -(\nA-\nB)$.}}
\label{yes-Astates1}
\end{align}%(71)
\end{proposition}
{So far we have not settled {the exceptional values of $k$ for $E=0$ ( $k=2\pi/3$ and $4\pi/3$ for zigzag-like edges and $k=0$ and $2\pi$ for armchair-like edges)}. Theorem \ref{th:exceptional} below shows that there are no edge states in those cases.}
\subsection{Conclusions on the zero energy edge state eigenvalue problem}
Proposition \ref{Bstate-conds1} and Proposition \ref{Astate-conds1} can be used to determine
precisely for which edge-types (balanced zigzag,  unbalanced zigzag, armchair) and for which ranges of parallel quasimomentum, $\kp\in[0,2\pi]$, zero energy (flat band) edge states, $\psi=(\psi^A,\psi^B)$, do / do not exist, and whether they are supported on $B-$sites; $\psi=(0,\psi^B)$, or on $A-$ sites, $\psi=(\psi^A,0)$.

{
We shall apply Proposition \ref{Bstate-conds1} and Proposition \ref{Astate-conds1} to deduce Theorems \ref{th:zigzag} and \ref{th:armchair} below. The results  are  summarized in the Table \ref{tab:table}.  At the end of this section we deduce Theorem \ref{main-intro} from 
Theorems \ref{th:zigzag} and  \ref{th:armchair}. }

\begin{table}[htbp]
\begin{center}
\begin{tabular}{ |c|c|c|c|c| } 
\hline
 {\ } & {\footnotesize armchair-type} & {\footnotesize  zigzag-type }&  
 {\footnotesize zigzag-type}\\
\hline
 {\ } & {\footnotesize  $a_{11}-a_{12}\equiv0\ {\rm mod}\ 3$} &
 {\footnotesize  $a_{11}-a_{12}\equiv+1\ {\rm mod}\ 3$} & {\footnotesize  $a_{11}-a_{12}\equiv-1\ {\rm mod}\ 3$}  
 \\
\hline
%\multirow{3}{4em}{Multiple row} & cell2 & cell3 \\ 
{\footnotesize balanced edge} &{\footnotesize No ES} & {\footnotesize A-site ES for $\kp\in I$} & {\footnotesize B-site ES for $\kp\in I $}   \\ 
\hline
{\footnotesize unbalanced edge} & {------ } & {\footnotesize {B-site ES for $\kp\in[0,2\pi]\setminus\overline{I}$}}
 & {\footnotesize {A-site ES for $\kp\in[0,2\pi]\setminus\overline{I}$}}  
\\
\hline
\end{tabular}
\caption{{\small{Complete description of all zero energy / flat band edge states; Theorems \ref{th:zigzag} and \ref{th:armchair}. {Note that $I=(\frac{2\pi}{3},\frac{4\pi}{3})$}}}}\label{tab:table}
\end{center}\end{table}

\nit With reference to this table, we make a few clarifying remarks. There are no unbalanced armchair (AC) edges.    The expressions  A-site ES and B-site ES refer to edge states (ES) supported exclusively on $A-$sites, respectively on $B-$sites. 
Whenever zero energy edge states exist, the corresponding $l^2_\kpar$ eigenspace has dimension equal to one. This table concerns only zero energy edge states. {Finally, note that this result is independent of the choice of $(a_{21},a_{22})$; see Section \ref{sec:rat-edge} for the possible choices.}

\begin{theorem}[Existence of zero-energy edge states for zigzag-type edges.]\label{th:zigzag}
	Suppose that $s_2=\pm 1$ where $s_2=a_{11}-a_{12}\;\text{mod}\ 3$.
	\begin{itemize}
	\item Suppose the terminated structure is balanced ($\nA=\nB$). Then, for $\kp\in(2\pi/3,4\pi/3)$ there is a one-dimensional space of zero-energy  edge states for $H_\sharp(\kp)$. These edge states are supported on the $B-$sites if $s_2=-1$, and on the $A-$sites if $s_2=+1$. There are no zero energy edge states for $\kp\in[0,2\pi]\setminus[2\pi/3,4\pi/3]$.
%		\item if the terminated structure is balanced, $0$ is in the spectrum of $H_\#(k)$ iff $k\in (\frac{2\pi}{3},\frac{4\pi}{3})$ and the associated subspace is one dimensional composed of edge states supported on $B-$sites if $s_2=-1$ and supported on $A-$sites if $s_2=1$.
		\item Suppose the terminated structure is unbalanced ($\nA\neq\nB$). Then, for $\kp\in[0,2\pi]\setminus[2\pi/3,4\pi/3]$ there is a one-dimensional space of zero-energy edge states for $H_\sharp(\kp)$. 
		 These edge states are supported on the $A-$sites if $s_2=-1$, and on the $B-$sites if $s_2=+1$. There are no zero energy edge states for $\kp\in(2\pi/3,4\pi/3)$.
%		if the terminated structure is unbalanced, $0$ is in the spectrum of $H_\#(k)$ iff $k\in (0,\pi)\setminus[\frac{2\pi}{3},\frac{4\pi}{3}]$ and the associated subspace is one dimensional composed of edge states supported on $A-$sites if $s_2=-1$ and supported on $B-$sites if $s_2=1$.
	\end{itemize}
\end{theorem}
\begin{proof}
	If the terminated structure is balanced, then $\nA=\nB$. If $s_2=-1$, Proposition \ref{Bstate-conds1}-\eqref{yes-Bstates1} implies that if $k\in (\frac{2\pi}{3},\frac{4\pi}{3})$, the space of $l^2$ solutions of \eqref{Bsite-states} (leading to zero energy / flat band edge states supported  on $B-$sites) has dimension 1. Proposition \ref{Bstate-conds1}-\eqref{no-Bstates1} implies that if $k\notin (\frac{2\pi}{3},\frac{4\pi}{3})$ there are no $l^2$ solutions of \eqref{Bsite-states} and Proposition \ref{Astate-conds1}-\eqref{no-Astates1}  implies that there are no $l^2$ solutions of \eqref{Asite-states} for any $k$.  If $s_2=1$, Proposition \ref{Astate-conds1}-\eqref{yes-Astates1} implies that if $k\in (\frac{2\pi}{3},\frac{4\pi}{3})$, the space of $l^2$ solutions of \eqref{Asite-states} (leading to zero energy / flat band edge states supported on $A-$sites) has dimension 1. Proposition \ref{Astate-conds1}-\eqref{no-Astates1} implies that if $k\notin (\frac{2\pi}{3},\frac{4\pi}{3})$ there are no $l^2$ solutions of \eqref{Asite-states} and Proposition \ref{Bstate-conds1}-\eqref{no-Bstates1}  implies that there are no $l^2$ solutions of \eqref{Bsite-states} for any $k$. \\
	If the terminated structure is unbalanced, then $\nA-\nB=s_2=\pm1$ (see \eqref{nAnBs2}). If $s_2=-1$, Proposition \ref{Astate-conds1}-\eqref{yes-Astates1} implies that if $k\notin (\frac{2\pi}{3},\frac{4\pi}{3})$, the space of $l^2$ solutions of \eqref{Asite-states} (leading to $0-$ energy flat band edge states supported on $A-$sites) has dimension 1. Proposition \ref{Astate-conds1}-\eqref{no-Astates1} implies that if $k\in (\frac{2\pi}{3},\frac{4\pi}{3})$ there are no $l^2$ solutions of \eqref{Asite-states} and Proposition \ref{Bstate-conds1}-\eqref{no-Bstates1}  implies that there are no $l^2$ solutions of \eqref{Bsite-states} for all $k$.  If $s_2=1$, Proposition \ref{Bstate-conds1}-\eqref{yes-Bstates1} implies that if $k\notin (\frac{2\pi}{3},\frac{4\pi}{3})$, the space of $l^2$ solutions of \eqref{Bsite-states} (leading to $0-$ energy flat band edge states supported on $B-$sites) has dimension 1. Proposition \ref{Bstate-conds1}-\eqref{no-Bstates1} implies that if $k\in (\frac{2\pi}{3},\frac{4\pi}{3})$ there are no $l^2$ solutions of \eqref{Bsite-states} and Proposition \ref{Astate-conds1}-\eqref{no-Astates1}  implies that there are no $l^2$ solutions of \eqref{Asite-states} for all $k$.
	\end{proof}

\begin{theorem}[Absence of zero-energy edge states for armchair-type edges.]\label{th:armchair}
	Suppose that $s_2=0$ where $s_2=a_{11}-a_{12}\;\text{mod}\ 3$. Then, for all $k\in (0,2\pi)$,  $0$ is not in the {point} spectrum of $H_\sharp(k)$.
\end{theorem}
\begin{proof}
	In that case, the terminated structure is always balanced, {\it i.e.} $\nA=\nB$. As $s_2=0$, Proposition \ref{Bstate-conds1}-\eqref{no-Bstates1} implies  there are no $l^2$ solutions of \eqref{Bsite-states} for all $k$ and Proposition \ref{Astate-conds1}-\eqref{no-Astates1} implies  there are no $l^2$ solutions of \eqref{Asite-states} for all $k$.
	\end{proof}
{We return to the case of the exceptional quasimomenta $k=2\pi/3$ and $4\pi/3$ for zigzag edges and $k=0$ and $2\pi$ for armchair edges.
\begin{theorem}\label{th:exceptional}
	There are no zero energy edge states at exceptional quasi-momenta.
	\end{theorem}
	\begin{proof}
		Let $\triangle$ be the open unit disc and $\overline{\triangle}$ its closure.
		Recall that
		\[
			\pcount= \#\{\zeta \in \triangle : p_+(\zeta,\kp)=0\}=\#\{\zeta \in \C\setminus\overline\triangle : p_-(\zeta,\kp)=0\}
		\]
		and 	
		\[
			\qcount= \#\{\zeta \in \triangle : p_-(\zeta,\kp)=0\}=\#\{\zeta \in \C\setminus\overline\triangle : p_+(\zeta,\kp)=0\}.
		\]	
		For an exceptional quasi-momentum $k_0$, there is a single zero $\zeta_0$ of $p_+(\cdot,\kp_0)$ on the unit circle if the edge is of zigzag-type. If the edge is of armchair type, there are two such $\zeta_0$ (see 
	\eqref{eq:ess_spec_arm} and \eqref{eq:ess_spec_zig}). For $k$ slightly above or slightly below $k_0$, each such $\zeta_0$ perturbs either into $\triangle$ or into $\C\setminus\triangle$, thus increasing $\pcount$ or $\qcount$ by 1; hence $\pcount=\mathfrak{p}(\kpar_0)+\#\{\zeta_0\text{ perturbing into }\triangle\}$, and $\qcount=\mathfrak{q}(\kpar_0)+\#\{\zeta_0\text{ perturbing into }\C\setminus\overline\triangle\}$. On the other hand, we know $\pcount$ and $\qcount$ from Propositions \ref{prop:p-value} and \ref{prop:q-value}. From the above remarks it follows that $\mathfrak{p}(\kpar_0)\leq -n_1$ and $\mathfrak{q}(\kpar_0)\leq n_3$ (we omit details). Since the relevant Vandermonde determinants are nonzero, it follows in turn that the linear systems \eqref{0outside-B}  and \eqref{0outside-A}  have no nontrivial solutions, as in the proof of Theorem \ref{th:zigzag}. Thus, there are no zero energy edge states for the exceptional quasi-momenta.
		\end{proof}}
{	\begin{remark}
		One can prove without Proposition \ref{Bstate-conds1} and Proposition \ref{Astate-conds1} that every edge of zigzag type gives rise to flat band edge states either for all $\kp\in(2\pi/3,4\pi/3)$ or for all $\kp\in(0,2\pi)\setminus[2\pi/3,4\pi/3]$. To see this, recall the functions $\pcount,\,\qcount$ from \eqref{eq:number_roots}
 and \eqref{eq:number_roots_A}. For fixed $\kp$, \eqref{psiB-form} and \eqref{0outside-B} show that a flat band edge forms if $\pcount>\nB-\nA-n_1$; while \eqref{psiA-form} and \eqref{0outside-A}	show that a flat band edge forms if $\qcount>\nA-\nB+n_3$. Since $\pcount+\qcount=n_3-n_1$, we conclude that an edge state with a quasimomentum $k$ exists unless
\begin{equation}
	\label{eq:condition_edgestate}
	\pcount=\nB-\nA-n_1\;\text{and}\; \qcount=\nA-\nB+n_3
\end{equation}
For an edge of zigzag type, \eqref{eq:condition_edgestate} must fail for all  $\kp\in(2\pi/3,4\pi/3)$ or for all $\kp\in[0,2\pi]\setminus[2\pi/3,4\pi/3]$.  That's because the functions $\kp\mapsto \pcount$ and $\kp\mapsto \qcount$ both experience jumps $\pm 1$ at $\kp=2\pi/3$ and at $\kp=4\pi/3$, as a zero of $p_+(\zeta,k)$ crosses the unit circle. To see that the zero crosses the unit circle, one Taylor expands $p_+(\zeta,k)$ to first order in $(\zeta,k)$ about the zeros {$(e^{\imath \hat{k}_\perp},\hat{k})$ where $\hat{k}$ and ${\hat{k}_\perp}$ are given by \eqref{h_kpar} and \eqref{h_kperp}.} We omit details and we leave it to the reader to see how the argument breaks down for an edge of armchair type.
	\end{remark}}
{
\begin{remark}[Classical zigzag edge states]\label{classical-es}
Recall that in our general analysis we have excluded the  classical zigzag edges; see the discussion around 
  \eqref{perm-n}. The analysis of zero energy / flat band edge states for the classical zigzag edges, balanced and unbalanced, is given, for example, in \cite{Dresselhaus-etal:96,Graf-Porta:13,FW:20}.  Without loss of generality, we 
  consider the balanced case ($\nA=\nB=0$) and the unbalanced case ($\nA=0$, $\nB=1$). It is easily derived that:
  \begin{align*}
 & \textrm{\it (i)\ Balanced (aka ordinary) $0-$ energy zigzag edge states exist}\\
 &\qquad   \textrm{\it $\iff\ {\kpar\in(2\pi/3,4\pi/3)}$ with one-dimensional eigenspace of $H_\sharp(\kpar)$ spanned by }\\  
 &\quad   \psi_n^A= \left[-(1+e^{ik})\right]^n,\ \quad  n\ge0,\quad {\rm and}\quad  \psi_n^A=0,\quad  n\le-1,\\
&\quad  \psi_n^B=0,\ \quad \textrm{for all}\quad n\in\Z,\\
&{\textrm{\it For $k=\pi$, we take instead $\psi_0^A=1,\;\psi_n^A=0$ for $n\neq 0$, 
 and $\psi_n^B=0,  \quad \textrm{for all}\quad n\in\Z$.}}\\
  & \textrm{\it (ii)\ Unbalanced (aka bearded) $0-$ energy  zigzag edge states exist}\\
  &\qquad \textrm{\it  $\iff\ \kpar\in[0,2\pi]\setminus[2\pi/3,4\pi/3]$ with one-dimensional eigenspace of $H_\sharp(\kpar)$ spanned by }\\ 
  &\quad  \psi_n^A=0,\quad n\in\Z\\
  &\quad   \psi_n^B= \left(\frac{-1}{1+e^{-ik}}\right)^n\,\ \quad n\ge0,\quad \psi_n^B=0,\ \quad n\le-1.
  \end{align*}
  \end{remark} 
  }
  The results of this section, together with Remark \ref{classical-es}, cover all rational edges and our results are summarized in Table \ref{tab:table}. 
  
Thanks to Proposition \ref{DADB} and equations \eqref{nAnBs2},
  Theorems \ref{th:zigzag} and \ref{th:armchair} imply Theorem  \ref{main-intro} of the introduction.

\section{Honeycomb edge polynomials}\label{sec:HEPs}
%\sfcomment{
Central to the flat band edge state classification of the previous section is Proposition \ref{prop:p-value} 
 which counts the number of roots  in the open unit disc, $\pcount$, of the polynomial $p_+(\zeta;\kp)$, defined in \eqref{B-poly}. Note that this result also allows us to deduce Proposition \ref{prop:q-value}  concerning $\qcount$, the number of roots of $p_-(\zeta;\kp)$, defined in \eqref{A-poly}, in the open unit disc. We make a general study of such polynomials in this section which, due to their origins, we refer to as {\it honeycomb edge polynomials} and apply our results to prove
   Proposition \ref{prop:p-value}.

We begin by counting, for any fixed $\kp\in[0,2\pi]$,  the roots, $\zeta$, of the polynomial equation
 \begin{equation}
  1 + e^{i\beta_1\kp}\zeta^{\gamma_1}+e^{i\beta_2\kp}\zeta^{\gamma_2}=0
 \label{hep-1} %(78)
 \end{equation}
 in the open unit disc. Here, we assume that $\beta_1, \beta_2, \gamma_1, \gamma_2$ are integers 
such that
  \begin{equation}
  0<\gamma_1<\gamma_2
  \label{gamma12} %(79)
  \end{equation}
  and
  \begin{equation}
  \det[\beta\;\gamma]= \pm1\quad\text{with}\;\beta=\left(\begin{matrix}\beta_1\\ \beta_2\end{matrix}\right)\;\text{and}\;\gamma=\left(\begin{matrix}\gamma_1\\ \gamma_2\end{matrix}\right)
  \label{det-sig} %(80)
  \end{equation}
In Appendix \ref{sec:p_pm} we check that the polynomial equation $p_+(\zeta,\kp)=0$  is of the form \eqref{hep-1}, \eqref{gamma12}, \eqref{det-sig} where 
\begin{subequations}
\label{eq:beta_gamma}
\begin{align}
\beta_1&=m_2-m_1,\quad \beta_2=m_3-m_1,\\
\gamma_1&=n_2-n_1,\quad \gamma_2=n_3-n_1.\end{align}
\end{subequations}

     To analyze \eqref{hep-1}, we write our unknown $\zeta$ in polar form:
     \begin{equation*}
     \zeta := \rho^{\frac{1}{\gamma_2}}\ e^{i\theta},
     \label{zeta-polar} %(81)
     \end{equation*}
     and set 
     \begin{equation}
     \kappa := \frac{\gamma_1}{\gamma_2} \in (0,1)\quad \textrm{(by \eqref{gamma12}).}
     \label{kappa-def} %(82)
     \end{equation}
     Equation \eqref{hep-1} then becomes
     \begin{equation}
     1+ \rho^\kappa\ e^{i(\beta_1\kp+\gamma_1\theta)} + \rho\ e^{i(\beta_2\kp+\gamma_2\theta)} = 0.
     \label{hep-2}
     \end{equation}
  We now outline our strategy for calculating 
  \begin{equation}\label{eq:rootsnum}
  	p(k)=\#\{(\rho,\theta)\in(0,1)\times[0,2\pi)\; \text{solving }\;\eqref{hep-2}\}
  \end{equation}
  as follows
	  \begin{enumerate}
	  	\item For given $\rho_1$, $\rho_2>0$, we consider the equation
 \begin{equation}
 1+\rho_1\ e^{i\phi_1}+\rho_2\ e^{i\phi_2} = 0,
 \label{hep-X}%(84)
 \end{equation}
 and look for the possible solutions $\phi_1, \phi_2\in\R/2\pi\Z$. Proposition \ref{hep-X-solve} states that, under some constraints on $\rho_1$ and $\rho_2$, the solutions $\phi_1, \phi_2$ are given by
 \[
 	\phi_1 = \hsigma\alpha_1(\rho_1,\rho_2)+ 2\pi\ell,\quad \phi_2=\hsigma\alpha_2(\rho_1,\rho_2)+2\pi\ell',\quad\ell,\ell'\in\mathbb{Z},
 \]
where $\hsigma=\pm 1$ and the functions $\alpha_1,\alpha_2$ are given in \eqref{alpha12}.
\item \ We deduce in Proposition \ref{prop:rhotheta} that $\zeta=\rho^{\frac{1}{\gamma_2}}e^{i\theta}$ is a solution of  \eqref{hep-1} in the open unit disc if and only if $\rho$ and $\theta$ satisfy:\\
(a) $\rho\in[\rho_{\rm critical},1)$ with $\rho_{\rm critical}$ defined by \eqref{rhoc2} below,\\
(b) \begin{equation}\label{eq:eq_rho}
	\alpha_1(\rho^\kappa,\rho)-\kappa \alpha_2(\rho^\kappa,\rho) = \frac{1}{\gamma_2}\left( \hsigma k \det[\beta\;\gamma]\; + 2\pi \ell\right) 
\end{equation}
for some $(\hsigma,\ell)\in\{\pm 1\}\times\mathbb{Z}$., and \\
(c) $\theta$ is computed from $\rho,\kp,l,\hsigma$ as in \eqref{kth-solv2b} below.
\item Lemma \ref{mono-lem} states that the function
	\begin{equation}
	    \rho\mapsto M(\rho):=\alpha_1(\rho^\kappa,\rho)-\kappa \alpha_2(\rho^\kappa,\rho)
	    \label{M-def}
	    \end{equation}
is strictly decreasing in $[\rho_\text{critical},1]$. Hence, the number of roots of \eqref{hep-1} (equivalently, \eqref{hep-2}) inside the open unit disc is equal to the number of pairs $(\hsigma,\ell)\in\{\pm 1\}\times\mathbb{Z}$ for which the right hand side of \eqref{eq:eq_rho} lies in the interval $[M(1),M(\rho_\text{critical})]$. This enables us to deduce the number of  $\rho$'s in $[\rho_{\rm critical},1)$ which solve \eqref{eq:eq_rho}, and therefore to deduce the number of $\zeta$'s inside the open unit disc solving \eqref{hep-1}. Along the way, we must also verify that all roots obtained in this manner are indeed distinct.
	  \end{enumerate}

	  We now embark on this strategy. The following proposition presents, for fixed $\rho_1$ and $\rho_2$ satisfying explicit constraints, the solutions of
  \eqref{hep-X}. The proof, via elementary algebra and trigonometry, is omitted.
  %given in Appendix \ref{sec:solv-hepX}.
  \begin{proposition}\label{hep-X-solve}
  Equation \eqref{hep-X} has a solution if and only if  $|\rho_1-\rho_2|\le1$ and $\rho_1+\rho_2\ge1$. All solutions of \eqref{hep-X} are given by
 \begin{equation}
 \phi_1 = \hsigma\alpha_1(\rho_1,\rho_2)\;{\rm mod}\ 2\pi,\quad \phi_2=\hsigma\alpha_2(\rho_1,\rho_2)\; {\rm mod}\ 2\pi,
 \label{phi12}
 \end{equation}
for  $\hsigma=+1$ or $-1$, where
 \begin{equation}
 \alpha_1(\rho_1,\rho_2) \equiv \cos^{-1}\left(\frac{\rho_2^2-\rho_1^2-1}{2\rho_1}\right),\quad 
\alpha_2(\rho_1,\rho_2) \equiv -\cos^{-1}\left(\frac{\rho_1^2-\rho_2^2-1}{2\rho_2}\right).
\label{alpha12}
\end{equation}
 In \eqref{alpha12} a branch of $\cos^{-1}$  is chosen to take values in $[0,\pi]$.
\end{proposition}
Note that the conditions imposed on $\rho_1$ and $\rho_2$ guarantee that the arguments of the arc cosines in \eqref{alpha12} lie 
  in the interval
 $[-1,1]$.
 
Applying Proposition \ref{hep-X-solve} to our equation \eqref{hep-X}, we deduce the following result.
\begin{proposition}\label{prop:rhotheta}
A complex $\zeta= \rho^{\frac{1}{\gamma_2}}\ e^{i\theta}$ is solution of \eqref{hep-1} in the open unit disc if and only if  
\begin{equation}\label{rhoc2}
 \rho\in[\rho_{\rm critical},1),\quad {\rm where}\;
 \rho_{\rm critical}\in(0,1),\quad \rho^\kappa_{\rm critical}+\rho_{\rm critical}=1
 \end{equation}
 and $(\rho,\theta)$ satisfy, for some $\ell\in\Z$:
  \begin{subequations}
  \begin{align}
&  \alpha_1(\rho^\kappa,\rho)-\kappa \alpha_2(\rho^\kappa,\rho) = \frac{1}{\gamma_2}\left( \det[\beta\;\gamma]\hsigma k + 2\pi \ell\right), \label{kth-solv2a}\\ %(93)
&  \theta = \frac{\hsigma}{\gamma_2} \alpha_2(\rho^\kappa,\rho) - \frac{\beta_2}{\gamma_2} \left(\kp + 2\pi\det[\beta\;\gamma]\hsigma \ell \right) \quad {\rm mod}\ 2\pi\Z,%(94)
 \label{kth-solv2b} \end{align}
 \label{kth-solv2}
  \end{subequations}
with $\hsigma=+1$ or $-1$.
\end{proposition}
\begin{proof}
By Proposition \ref{hep-X-solve}, equation \eqref{hep-X} has a solution if and only if
 \[ \textrm{$|\rho-\rho^\kappa|\le1$ and $\rho+\rho^\kappa\ge1$.}
 \]
These conditions hold 
 if and only if $\rho\in[\rho_{\rm critical},\rho_+]$
   where $\rho_{\rm critical}$ is defined in \eqref{rhoc2} and $\rho_+>1$ is such that $\rho_+-\rho_+^\kappa=1$. Since we are interested in roots, $\zeta$,  inside the open unit disc, we restrict our attention to $\rho\in[\rho_{\rm critical},1)$.
By Proposition \ref{hep-X-solve}, the solutions 
  $\zeta=\rho^{\frac{1}{\gamma_2}}e^{i\theta}$ of  \eqref{hep-1} inside the open unit disc correspond to $(\rho,\kp,\theta)$ with 
  $\rho\in[\rho_{\rm critical},1)$  and  
   $(\kp,\theta)$ satisfying  
 \begin{subequations}
 \begin{align}
 \beta_1\kp + \gamma_1\theta &= \hsigma \alpha_1(\rho^\kappa,\rho)\quad \textrm{mod}\ 2\pi\\
 \beta_2\kp + \gamma_2\theta &= \hsigma \alpha_2(\rho^\kappa,\rho)\quad \textrm{mod}\ 2\pi
 \end{align}
 \label{kth-sys} % (89)
 \end{subequations}
with either $\hsigma=+1$ or $\hsigma=-1$.
We invert the system \eqref{kth-sys} by using
  \[ \begin{pmatrix}\beta_1&\gamma_1\\ \beta_2&\gamma_2\end{pmatrix}^{-1}\ =\ \left(\det[\beta\;\gamma]\right)^{-1}\,\begin{pmatrix}\gamma_2&-\gamma_1\\ -\beta_2&\beta_1\end{pmatrix},\]
  where $\det[\beta\;\gamma]=\pm1 $ by our assumption \eqref{det-sig}.
  Hence, the solution of \eqref{kth-sys} is given by
  \begin{align}
  \begin{pmatrix} \kp\\ \theta\end{pmatrix} &=
 \det[\beta\;\gamma]\hsigma\ 
   \begin{pmatrix} 
  \gamma_2\ \alpha_1(\rho^\kappa,\rho) - \gamma_1\ \alpha_2(\rho^\kappa,\rho)\ +\ 2\pi \ell^\prime\\
 - \beta_2\ \alpha_1(\rho^\kappa,\rho) + \beta_1\ \alpha_2(\rho^\kappa,\rho)\ +\ 2\pi \ell^{\prime\prime}
  \end{pmatrix},
  \label{kth-solv1} %(92)
  \end{align}
  for arbitrary integers $\ell^\prime$ and $\ell^{\prime\prime}$.
  By elementary computation using \eqref{det-sig} and \eqref{kappa-def}, we may rewrite \eqref{kth-solv1} as \eqref{kth-solv2}. This completes the proof of Proposition \ref{prop:rhotheta}.\end{proof}
 Thanks to \eqref{kth-solv2a} and \eqref{kth-solv2b} we can find all $(\rho,\theta)$ such that $\zeta=\rho^{\frac{1}{\gamma_2}}e^{i\theta}$ is a root of \eqref{hep-1}, which lies inside the unit circle, as follows: 
      \begin{enumerate}
      \item[Step 1:] Given $\kp\in[0,2\pi]$, find all $(\hsigma,\ell)$ with $\hsigma\in\{-1,+1\}$ and $\ell\in\Z$,  such that
       equation \eqref{kth-solv2a} admits a solution $\rho\in[\rho_{\rm critical}, 1)$.
      \item[Step 2:] For $\kp, \hsigma, \ell$ and $\rho$ produced by Step 1, we compute $\theta$ from 
      \eqref{kth-solv2b}.
      \end{enumerate} 
The two following lemmata will be useful for the proof of the main result of this section which gives $p(k)$, defined in \eqref{eq:rootsnum}.
  \begin{lemma}\label{distinct}
For any fixed $\kp$ satisfying the constraint 
{$\kp\in[0,2\pi]\ \setminus\ \Big\{0,\frac{2\pi}{3}, \pi,\frac{4\pi}{3},2\pi\Big\}$}, any two distinct 
$(\hsigma,\ell)$ arising from Step 1 give rise to two distinct solutions of \eqref{hep-1}.
%In Proposition \ref{distinct} below, we will see that different $(\hsigma,l)\in\{-1,+1\}\times\Z$ give rise to different solutions of  \eqref{hep-1} arising from Step 1. 
  Hence, the roots $\zeta=\rho^{\frac{1}{\gamma_2}}\ e^{i\theta}$ of \eqref{hep-1}, which lie inside the open unit circle, are in one-to-one correspondence with the pairs $(\hsigma,\ell)$ obtained in Step 1. 
\end{lemma}
\begin{proof}
  Let $(\hsigma,\ell)$ and $(\hsigma^\prime,\ell^\prime)$ be two distinct pairs arising in Step 1. From 
  $(\hsigma,\ell)$  we produce $\rho$ by solving \eqref{kth-solv2a} in Step 1, and then from Step 2 we obtain  $\theta$ from \eqref{kth-solv2b}. We then set $\zeta=\rho^{\frac{1}{\gamma_2}} e^{i\theta}$. Similarly, from $(\hsigma^\prime,\ell^\prime)$, we produce $\rho^\prime$ by solving \eqref{kth-solv2a}, then obtain $\theta^\prime$ from \eqref{kth-solv2b}, and then set $\zeta^\prime=(\rho^\prime)^{\frac{1}{\gamma_2}} e^{i\theta^\prime}$. We must show that $\zeta\ne\zeta^\prime$. We distinguish two cases.
  ~\\
{\bf Case 1:\ $\boldsymbol\hsigma=\boldsymbol\hsigma^\prime$ and $\boldsymbol{\ell}\ne \boldsymbol{\ell}^\prime$.} Then, \eqref{kth-solv2a} implies that  $M(\rho)-M(\rho^\prime)=\left(2\pi/\gamma_2\right)(\ell-\ell^\prime)\ne0$, where $M$ is defined in \eqref{M-def}. Therefore,  $\rho\ne\rho^\prime$, and hence $\zeta\ne\zeta^\prime$.
 \medskip
~\\
{\bf Case 2:\ $\boldsymbol\hsigma\ne\boldsymbol\hsigma^\prime$} If $\zeta=\zeta^\prime$, then $\rho=\rho^\prime$, so from \eqref{kth-solv2a}
$ \det[\beta\;\gamma]\hsigma\kp+2\pi \ell = \det[\beta\;\gamma]\hsigma^\prime\kp+2\pi \ell^\prime$ and hence
 $ \det[\beta\;\gamma](\hsigma-\hsigma^\prime)\kp = 2\pi(\ell^\prime-\ell)$. Since $\hsigma\ne\hsigma^\prime$, 
  $\det[\beta\;\gamma](\hsigma-\hsigma^\prime)=\pm2$ and therefore $\kp\in\Z\pi$, which is ruled out by our hypothesis.
  \end{proof}
To count  all $(\hsigma,\ell)$ obtained in Step 1 above, we make use of a crucial property of 
  \eqref{kth-solv2a}.
    \begin{lemma}[Monotonicity Lemma]\label{mono-lem}
    Let $\kappa\in(0,1)$, and let $\rho_{\rm critical}\in(0,1)$ be the solution of the equation $\rho_{\rm critical}^\kappa+\rho_{\rm critical}=1$. Then, the function $M$ defined in \eqref{M-def}
    is strictly decreasing on $[\rho_{\rm critical},1]$.
  \end{lemma}
 For the proof of this lemma, see Appendix \ref{app:mono-lem}.
We shall next count the number of pairs $(\hsigma,l)$ arising from Step 1, apply the Monotonicity Lemma \ref{mono-lem} to deduce the number of possible $\rho$'s in $[\rho_{\rm critical},1]$ and finally establish the following result:
\begin{proposition}\label{root-count} Consider the polynomial equation \eqref{hep-1}:
 $1+e^{i\beta_1\kp}\zeta^{\gamma_1}+e^{i\beta_2\kp}\zeta^{\gamma_2}=0$, where 
 $\kp\in[0,2\pi]\setminus\{0,\frac{2\pi}{3},\frac{4\pi}{3},2\pi\}$, and $\beta_j,\gamma_j,\ (j=1,2)$ are integers satisfying the constraints:
  $0<\gamma_1<\gamma_2$ and $\det[\beta\ \gamma]=\pm1$. We express $\gamma_1+\gamma_2$ modulo $3$ as:
\begin{align*}
&\quad \gamma_1+\gamma_2=3\hat{k}+\hat{s},\quad  \textrm{with}\quad \hat{k}\in\mathbb{Z}\quad {\rm and}\quad   \hat{s}\in\{-1,0,+1\}; \end{align*}  
see \eqref{gamma12} and \eqref{det-sig}. 
Then,  $p(k)$, defined in \eqref{eq:rootsnum}, is given by
\begin{equation}
p(k) = \hat{k} + \hat{s} \mathds{1}_{\kp\in(\frac{2\pi}{3},\frac{4\pi}{3})}.
\label{p-value4}
\end{equation}
\end{proposition}
\begin{proof}
We begin noting that $M(\rho)$ can be explicitly evaluated at $\rho=\rho_{\rm critical}$ and $\rho=1$.
Directly from \eqref{alpha12} we have that
\begin{enumerate}
\item for $\rho=\rho_{\rm critical}$  (see \eqref{rhoc2}) we have $\alpha_1(\rho^\kappa,\rho)=\pi$, 
$\alpha_2(\rho^\kappa,\rho)=-\pi$ and hence,
\begin{equation}
M(\rho_{\rm critical}) = \pi(1+\kappa) = \frac{\pi(\gamma_1+\gamma_2)}{\gamma_2}.
\label{Mrho_crit}
\end{equation}
\item for $\rho=1$, we have  $\alpha_1(\rho^\kappa,\rho)=\frac{2\pi}{3}$ and 
$\alpha_2(\rho^\kappa,\rho)=-\frac{2\pi}{3}$ and hence,
\begin{equation}
M(1) = \frac{2\pi}{3}(1+\kappa) = \frac{2\pi}{3}\frac{(\gamma_1+\gamma_2)}{\gamma_2}.
\label{Mrho1}
\end{equation}
\end{enumerate}

From \eqref{Mrho_crit}, \eqref{Mrho1} and  Monotonicity Lemma \ref{mono-lem}, it follows that for any given $\kp$, $\hsigma$ and $\ell$ in Step 1, equation \eqref{kth-solv2a} has a unique solution 
 $\rho\in[\rho_{\rm critical},1)$ provided 
\[
   \frac{\det[\beta\;\gamma]\hsigma \kp + 2\pi \ell}{\gamma_2}\quad \textrm{lies in}\quad
\left(  \frac{2\pi}{3}  \frac{(\gamma_1+\gamma_2)}{\gamma_2}\ , \frac{\pi(\gamma_1+\gamma_2)}{\gamma_2}
\right]
 \]
 and \eqref{kth-solv2a} has no solutions  $\rho\in[\rho_{\rm critical},1)$ otherwise. Therefore, 
  for fixed $\kp\in(0,2\pi)$ and $\hsigma\in\{-1,+1\}$, the integers $\ell$ for which \eqref{kth-solv2a} admits
   a solution $\rho\in[\rho_{\rm critical},1)$ are precisely the integers $\ell$ for which
   \begin{equation}
  \det[\beta\;\gamma]\hsigma \kp + 2\pi \ell\ \in 
   \left(  \frac{2\pi}{3} (\gamma_1+\gamma_2)\ , \pi(\gamma_1+\gamma_2)
\right]
\label{rhosolve_l}
\end{equation}
 When \eqref{rhosolve_l} holds, the solution $\rho\in[\rho_{\rm critical},1)$ of \eqref{kth-solv2a} is unique.
  We can therefore count the number of distinct $(\hsigma,\ell)$ arising from Step 1 for fixed $\kp$. 
  
For simplicity, assume $\kpar\notin\{0,\pi,2\pi\}$ so that 
  \begin{equation*}\textrm{ $\frac{1}{2} (\gamma_1+\gamma_2) - \frac{\kp}{2\pi}\notin\Z$}.
  \label{k_bad}
  \end{equation*}
  The number of such $(\hsigma,\ell)$ is equal to
 % \begin{subequations}
 \begin{equation}
  \begin{aligned}
  p (k) &= \#\Big\{\ell\in\Z : k+2\pi \ell \in  \left(  \frac{2\pi}{3} (\gamma_1+\gamma_2)\ ,\ \pi(\gamma_1+\gamma_2)
\right]  \Big\}\\
&\quad + \#\Big\{\ell\in\Z : -k+2\pi \ell \in  \left(  \frac{2\pi}{3} (\gamma_1+\gamma_2)\ ,\ \pi(\gamma_1+\gamma_2)
\right]  \Big\}\\
&= \#\Big\{\ell\in\Z :  \ell \in  \left(  \frac{1}{3} (\gamma_1+\gamma_2)-\frac{k}{2\pi}\ ,\ \frac12(\gamma_1+\gamma_2)-\frac{k}{2\pi}
\right]  \Big\}\\
&\quad + \#\Big\{\tilde{\ell}\in\Z : \tilde{\ell} \in  \left[  -\frac12(\gamma_1+\gamma_2) -\frac{k}{2\pi}\ ,\
  -\frac{1}{3} (\gamma_1+\gamma_2) -\frac{k}{2\pi}
\right)  \Big\}\\
&= \#\Big\{\ell\in\Z :  \ell \in  \left(  \frac{1}{3} (\gamma_1+\gamma_2)-\frac{k}{2\pi}\ ,\ \frac12(\gamma_1+\gamma_2)-\frac{k}{2\pi})
\right]  \Big\}\\
&\quad + \#\Big\{ \hat{\ell}\in\Z : \hat{\ell} \in  \left[  \frac12(\gamma_1+\gamma_2) -\frac{k}{2\pi}\ ,\
  \frac{2}{3} (\gamma_1+\gamma_2) -\frac{k}{2\pi}\ 
\right)  \Big\}\\
&= \#\Big\{\ell\in\Z :  \ell \in  \left(  \frac{1}{3} (\gamma_1+\gamma_2)-\frac{k}{2\pi}\ ,\ \frac23(\gamma_1+\gamma_2)-\frac{k}{2\pi})
\right)  \Big\}.
 \label{num_ells}
 \end{aligned}
 \end{equation}
In proving \eqref{num_ells}  we have used symmetry and translation via the correpondences $\tilde{\ell}=-\ell$ 
 and $\hat{\ell}=\tilde{\ell}+(\gamma_1+\gamma_2)$. Now let 
 \begin{equation*}
 \kp\in[0,2\pi]\ \setminus\ \Big\{0,\frac{2\pi}{3}, \pi,\frac{4\pi}{3},2\pi\Big\} %(99)
 \label{k_bad2}
 \end{equation*}
Then, we deduce that $p(k)$, the number of $(\hsigma,\ell)$ arising from Step 1 is equal to
\begin{equation}
p(k) = \Big\lfloor \frac23(\gamma_1+\gamma_2)-\frac{k}{2\pi}\Big\rfloor - 
 \Big\lfloor  \frac{1}{3} (\gamma_1+\gamma_2)-\frac{k}{2\pi} \Big\rfloor,
\label{p-value2}
\end{equation}
where $\lfloor x \rfloor$ denotes the largest integer smaller than or equal to $x$.

By Lemma \ref{lem:simple_roots}, proved below, there are no multiple roots.
Let us verify finally that \eqref{p-value4} follows from \eqref{p-value2}. We write $\gamma_1+\gamma_2=3\hat{k}+{\hat{s}}$ with ${\hat{k}}\in\mathbb{Z}$ and $\hat{s}\in\{-1,0,+1\}$ so that \eqref{p-value2} yields
 \[
 p(k)   =\Big\lfloor 2\hat{\kp} + \frac23\hat{s}-\frac{\kp}{2\pi}\Big\rfloor - 
  \Big\lfloor \hat{\kp}+\frac13\hat{s}-\frac{\kp}{2\pi} \Big\rfloor= \hat{\kp} + \Big\lfloor \frac23\hat{s}-\frac{\kp}{2\pi}\Big\rfloor - 
  \Big\lfloor \frac13\hat{s}-\frac{\kp}{2\pi} \Big\rfloor
 \]
 and the result follows from the identity
 \[ \Big\lfloor \frac23\hat{s}-\xi\Big\rfloor - 
  \Big\lfloor \frac13\hat{s}-\xi \Big\rfloor  = \hat{s} \mathds{1}_{\xi\in(\frac{1}{3},\frac{2}{3})},\quad\text{for}\; \hat{s}\in\{-1,0,+1\},\quad \xi\in(0,1)\setminus\Big\{\frac13,\frac23\Big\}.\]
  Finally, note that although we excluded $k=\pi$ above, {Proposition \ref{root-count} also holds for $k=\pi$ because the roots  of \eqref{hep-1} depend continuously on $k$, and because we can show that there are no roots  on the unit circle for $k=\pi$ by an argument analogous to that used in the proof of Proposition \ref{nogap}.}
   
\end{proof}
  We now use Proposition \ref{root-count} to prove Proposition \ref{prop:p-value} concerning $\pcount$, the number of roots,
  $\zeta$, of 
\eqref{B-poly}, $p_+(\zeta,\kp)=0$, in the open unit circle. Suppose first that $\kpar\notin\{0,2\pi/3,4\pi/3,2\pi\}$. Since $n_1<n_2<n_3$, we see 
  that equation \eqref{B-poly} is of the form \eqref{hep-1} with 
$ \gamma_1=n_2-n_1  $ and $\gamma_2=n_3-n_1$, which satisfy \eqref{gamma12} and \eqref{det-sig} by Appendix \ref{sec:p_pm}. By \eqref{k1k2s1s2} and \eqref{ttn-ttm} we have
\[
\gamma_1+\gamma_2 = n_1+n_2+n_3-3n_1= \tn_1+\tn_2+\tn_3-3n_1=
3k_2 + (a_{12}-a_{11})-3n_1 = -s_2-3n_1.
\]
Therefore, $\hat{s}=-s_2$ and $\hat{k}=-n_1$. Substitution into the expression for $p(\kpar)$ in \eqref{p-value4} of Proposition \ref{root-count} yields the assertion of Proposition \ref{prop:p-value} concerning $\pcount$, the number of roots of 
$p_+(\zeta,\kp)$ inside the unit disc. Finally, the cases $\kpar\in\{0,2\pi\}$ (for zigzag edges) and $\kpar\in\{2\pi/3,4\pi/3\}$ (for armchair edges) can be dealt with using a continuity argument.

All details of the proof of Proposition \ref{prop:p-value} are now complete except that we must rule out multiple roots. We now address this point. 

 {
  \begin{lemma}\label{lem:simple_roots}
Let $\beta_1$, $\beta_2$, $\gamma_1$ and $\gamma_2$ be integers and such that $0<\gamma_1<\gamma_2$.
 Then for all $\kpar\in[0,2\pi]$,  equation \eqref{hep-1} has no multiple roots.
	\end{lemma}
	\begin{proof}
	First note that $\zeta=0$ is not a root. A multiple root must satisfy \eqref{hep-1} and 
	\begin{equation}\label{*2}
	\gamma_1e^{i\beta_1\kpar}\zeta^{\gamma_1-1} + \gamma_2e^{i\beta_2\kpar}\zeta^{\gamma_2-1} =0.
	\end{equation}
	Hence, 
	\[
	\zeta^{\gamma_2-\gamma_1} = -\frac{\gamma_1}{\gamma_2}e^{i(\beta_1-\beta_2)\kpar}
	\]
	implying, since $0<\gamma_1<\gamma_2$, that $|\zeta|<1$. On the other hand, \eqref{*2} together with \eqref{hep-1}
	 implies that 
	 \[
\left(1-\frac{\gamma_1}{\gamma_2}\right) e^{i\beta_1\kpar}\zeta^{\gamma_1}=-1
\]
implying, since $0<\gamma_1<\gamma_2$, that $|\zeta|>1$, a contradiction. 
	\end{proof}
	}

{We summarize the section in the following result.
  \begin{proposition}\label{lem:main-hep}
 \begin{enumerate}
 \item[(A)] For $0<\kappa<1$, define functions $f_\kappa, g_\kappa$ of one variable as follows.
 \begin{equation*}
 \begin{aligned}
& \textrm{$\bullet$\ Define $\rho_{\rm critical}\in (0,1)$ as the solution of $\rho_{\rm critical}^\kappa+\rho_{\rm critical}=1$.}\\
&\textrm{ For $t\in[\frac{2\pi}{3}(1+\kappa),\pi(1+\kappa)]$, define $g_\kappa(t)$ to be the unique solution} \\
&\textrm{ $\rho\in[\rho_{\rm critical},1]$ of the equation $
 M(\rho)\equiv \alpha_1(\rho^\kappa,\rho)-\kappa\alpha_2(\rho^\kappa,\rho)=t$.}\\
& \textrm{$\bullet$\ For $t\in[\frac{2\pi}{3}(1+\kappa),\pi(1+\kappa)]$, define $f_\kappa(t)=\alpha_2(\rho^\kappa,\rho)$ for $\rho=g_\kappa(t)$. }
\end{aligned}
\end{equation*}
\item[(B)] Now let $\beta_1, \beta_2, \gamma_1, \gamma_2$ be integers, with $0<\gamma_1<\gamma_2$ and $\det[\beta\  \gamma]\in\{-1,+1\}$. Set $\kappa = \gamma_1/\gamma_2\in(0,1)$. 
Assume $\kp\in[0,2\pi]\setminus\{0, 2\pi/3, 4\pi/3, 2\pi\}$. Then, 
\begin{equation*}
\begin{aligned}
&\textrm{ $\bullet$ The polynomial equation \eqref{hep-1} has only simple roots.}\\
&\textrm{  $\bullet$  Let $\gamma_1+\gamma_2 =  3\hat{k}+ \hat{s}$, where $\hat{k}\in\Z$ and 
$\hat{s}\in\{-1,0,+1\}$.} \ \ Then, \\
&\qquad\qquad  \pcount\ =\ \#\Big\{|\zeta|<1: \textrm{$\zeta$ is a root of  \eqref{hep-1} \Big\}} =
 \hat{k} + \hat{s}\mathds{1}_{\kp\in(\frac{2\pi}{3},\frac{4\pi}{3})}\\
&\textrm{ $\bullet$ Moreover, the roots of \eqref{hep-1} in the open unit disc are determined from }\\
&\textrm{ all the pairs $(\hsigma,l)\in\{-1,+1\}\times\Z$ such that}\\
&\textrm{ $t\equiv\frac{1}{\gamma_2}\left(\det[\beta\ \gamma]\ \hsigma \kp + 2\pi l \right)$ lies in $\left(\frac{2\pi}{3}(1+\kappa),\pi(1+\kappa)\right]$.}\\
&\textrm{by setting $\rho=g_\kappa(t),\ \theta=\frac{\hsigma}{\gamma_2}f_\kappa(t)-\beta_2\det[\beta\ \gamma]\hsigma\ t$ }
\\
&\textrm{and finally taking $\zeta=\rho^{\frac{1}{\gamma_2}}\ e^{i\theta}$.}\\
&\textrm{Furthermore, distinct $(\hsigma,l)$ give rise to distinct roots $\zeta$ of \eqref{hep-1}.}
\end{aligned}
\end{equation*}
 \end{enumerate}
 \end{proposition}
 }

\section{ Explicit formulas for zero energy (flat band) edge states}\label{formulae}
In this section we provide explicit formulas for the zero energy edge states when they exist. In Section \ref{sec:0energy}, we have found that
\begin{itemize}
	\item  When zero energy edge states living on $B-$ sites exist (see Theorem \ref{th:zigzag}), they are given by 
\begin{equation}
\psi^B(n) = \sum_{j=1}^p A_j{\zeta_j^n}\quad \textrm{if $n\ge\nA+n_1$},\qquad \psi^B(n) = 0\quad \textrm{otherwise}  %(108)
\label{psiB-sum}\end{equation} 
and the $A_1,\dots,A_p$ satisfy the system of $p-1$ equations:
\begin{equation}
\sum_{j=1}^p A_j{\zeta_j^n} = 0\quad  \textrm{for}\; \nA+n_1\le n<\nB.
\label{A-sys-Bsites}
\end{equation}  %(109)
Here, $\zeta_1,\dots,\zeta_p$ are distinct complex numbers in the open unit disc, roots of \eqref{B-poly} and the space of $(A_j)_{1\le j\le p}$ satisfying
 \eqref{A-sys-Bsites} is one-dimensional.
 \item When zero energy edge states living on $A-$ sites exist (see Theorem \ref{th:zigzag}), they are given by 
\begin{equation}
\psi^A(n) = \sum_{j=1}^q A_j{\zeta_j^n} \quad \textrm{if $n\ge\nB-n_3$},\qquad
\psi^A(n) = 0\quad \textrm{otherwise}  %(111)
\label{psiA-sum}
\end{equation} 
and the $A_1,\dots,A_q$ satisfy the system of $q-1$ equations:
\begin{equation}
\sum_{j=1}^q A_j{\zeta_j^n} = 0\quad  \textrm{for $\nB-n_3\le n<\nA$}.
\label{A-sys-Asites}
\end{equation}  %(112)
Note that the $\zeta_j$ and $A_j$ in \eqref{psiB-sum} differ from those in \eqref{psiA-sum};
 the $\zeta_1,\dots,\zeta_q$ are distinct complex numbers in the open unit disc; roots of \eqref{A-poly} and the space of $(A_j)_{1\le j\le q}$ satisfying
 \eqref{A-sys-Asites} is one-dimensional.
\end{itemize}
{In cases where zero edge states are supported exclusively on $A-$ sites or exclusively on $B-$ sites}  we have given distinct complex numbers $\zeta_1,\dots,\zeta_r$ in the open unit disc, and the zero energy state is given by the vector of amplitudes
 \begin{equation}
 \Psi(n) = \sum_{j=1}^r A_j \zeta_j^n\quad \textrm{if}\; n\geq \nbase,
 \label{Psin-gen}\end{equation} %(113)
 where the $A_1,\dots,A_r$ satisfy  {the $r-1$ homogeneous equations}:
 \begin{equation}
 \sum_{j=1}^r A_j\zeta_j^n = 0\quad {\rm for}\; \nbase\le n< \nbase+r-1
\label{A-sys-gen} \end{equation} %(114)
{Here, 
\begin{equation}
\nbase=\begin{cases} 
\nA+n_1 & \textrm{for $B-$ site edge states}\\
\nB-n_3 & \textrm{for $A-$ site edge states}
        \end{cases}
\end{equation}
and 
\begin{equation}
r=\begin{cases} 
\pcount & \textrm{for $B-$ site edge states; see Proposition \ref{prop:p-value} }\\
\qcount & \textrm{for $A-$ site edge states; see Proposition \ref{prop:q-value}.}
        \end{cases}
\end{equation}
}

Our goal here is to produce formulas for non-zero solutions of 
  \eqref{Psin-gen}, \eqref{A-sys-gen}.

The following result provides explicit representations of zero energy / flat band edge states.
\begin{theorem}[Representation formulae for $0-$ energy edge states] \label{0en-reps}
Let $\zeta_j(\kpar)$, with  $1\le j\le r$, denote the zeros of the relevant polynomial 
($p_+(\zeta,\kpar)$ or $p_-(\zeta,\kpar)$) inside the open unit disc. Assume that $a_{11}-a_{12}\ne0\ {\rm mod}\ 3$ so that there exist edge states.
 Finally, assume that $\kpar$ lies in an appropriate subinterval of $(0,2\pi)$ for which there are edge states; see Table \ref{tab:table}. Then, the edge states $\Psi(n)=\Psi(n,\kpar)$, given by \eqref{Psin-gen} and satisfying \eqref{A-sys-gen},  can be expressed as 
 \[ \Psi(n)= c \tilde\Psi(n),\]
 where $c$ is an arbitrary constant and $\tilde\Psi(n)$ is given by the three equivalent formulas 

  \begin{equation}%(118-alt)
  \tilde\Psi(n) = \sum_{\substack{l_1+\dots+l_r=n-\nbase+1 \\ l_1,\dots,l_r\ge1}}  \zeta_1^{l_1-1}\cdots\zeta_r^{l_r-1}
  \quad {\rm for}\; n\geq \nbase,
  \label{Psi_n-a}
  \end{equation}
 \begin{equation}
\tilde\Psi(n)\ =\ \frac{1}{2\pi}\int_0^{2\pi} e^{i(n-\nbase+1)\omega} \prod_{j=1}^r \left(e^{i\omega}-\zeta_j\right)^{-1}
 d\omega\quad {\rm for}\; n\geq \nbase\ , \label{es-fourier-a}
 \end{equation}
\begin{equation}
\tilde\Psi(n) =   \sum_{j=1}^r \frac{\zeta_j^{n-\nbase}}{\prod_{l\in\{1,\dots,r\}\setminus\{j\}}\left(\zeta_l-\zeta_j\right)}\quad {\rm for}\; n\geq \nbase\ .
\label{128-a}
\end{equation}
For the normalization, we have
 \begin{align}
\|\tilde\Psi\|_{l^2(\Z)}^2 &=  \sum_{j=1}^r \frac{\zeta_j^{r-1}}{1-|\zeta_j|^2} \
 \prod_{l\in\{1,\dots,r\}\setminus\{j\}} 
   \frac{1}{(\zeta_l-\zeta_j)(1-\overline{\zeta_j}\zeta_l)}\ .
 \label{norm-fou-a}  \end{align}
\end{theorem}
The above wave functions are independent of the choice of $(a_{21},a_{22})$; see Section \ref{sec:rat-edge} for the possible choices and Appendix \ref{app:v_2} for the proof of this independence. Note also that in the classical zigzag edges (for which $r=1$ and $\zeta_1=(1+e^{\imath k})^{\pm 1})$, Formula \eqref{Psi_n-a} recovers the standard formula for the zero-energy edge states (see Remark \ref{classical-es}).

{To prove Theorem \ref{0en-reps}, we first prove \eqref{Psi_n-a}, and then derive the other representations as consequences. We note that by replacing in \eqref{Psin-gen}, \eqref{A-sys-gen} the coefficients $A_1,\dots,A_p$ by $A_1\zeta_1^{\nbase-1},\dots,A_p\zeta_p^{\nbase-1}$, we may take $\nbase=1$.}
\begin{proof}[Proof of \eqref{Psi_n-a}]
 Our starting point is the identity
\begin{equation}
\textrm{Assume}\ \ \xi\ne\eta.\quad  \textrm{Then, for all}\ \  n\geq1,\quad \sum_{\substack{k_1+k_2=n \\ k_1,k_2\ge1}}
 \xi^{k_1}\eta^{k_2}\ =\   \frac{\xi}{\eta-\xi}\eta^n +  \frac{\eta}{\xi-\eta}\xi^n\ .
 \label{sum-id} %(115)
\end{equation}
  The expression  \eqref{Psi_n-a} for $\Psi(n)$ with $n_{\rm base}=1$ follows from \eqref{Psin-gen} and the following
   \begin{lemma}\label{lem:prod-sum}
  Let $\zeta_1,\dots,\zeta_r$ denote distinct complex numbers. Then, there exist $A_1,\dots,A_r\in\mathbb{C}$ such that 
  \begin{equation}
\textrm{For all}\ \  n\ge1,\quad \sum_{\substack{k_1+\dots+k_r=n \\ k_1,\dots,k_r\ge1}}  \zeta_1^{k_1}\cdots\zeta_r^{k_r}
   = \sum_{j=1}^r A_j \zeta_j^n.
  \label{p-sum} %(116)
  \end{equation}
  Note, in particular, that this expression vanishes for $1\le n<r$.
  \end{lemma}
%}

\nit {\it Proof of the Lemma \ref{lem:prod-sum}:}\ We proceed by induction on $r$.  The relation \eqref{p-sum} for $r=1$ holds with $A_1=1$. Now fix $r\ge2$ let  $\zeta_1,\dots,\zeta_r$ be distinct complex numbers.
   Assume \eqref{p-sum} holds for $r-1$ in place of $r$: since $\zeta_1,\dots,\zeta_{r-1}$ are distinct complex numbers, there exist $B_1,\ldots, B_{r-1}$ such that
   \begin{equation*}
\textrm{For all}\ \  p\geq 1,\quad \sum_{\substack{k_1+\dots+k_{r-1}=p \\ k_1,\dots,k_{r-1}\ge1}}  \zeta_1^{k_1}\cdots\zeta_r^{k_{r-1}}
   = \sum_{j=1}^{r-1} B_j \zeta_j^p.
  \label{p-sum1}
  \end{equation*}
  We obtain
   \begin{align*}
  \sum_{\substack{k_1+\dots+k_r=n \\ k_1,\dots,k_r\ge1}}  \zeta_1^{k_1}\cdots\zeta_r^{k_r}
   = \sum_{\substack{k_r+p=n \\ k_r,p\ge1}} 
 \left[  
    \sum_{\substack{k_1+\dots+k_{r-1}=p \\ k_1,\dots,k_{r-1}\ge1}}  \zeta_1^{k_1}\cdots\zeta_r^{k_{r-1}} 
    \right] \zeta_r^{k_r}= \sum_{\substack{k_r+p=n \\ k_r,p\ge1}}
     \left[\sum_{j=1}^{r-1} B_j \zeta_j^p \right]\zeta_r^{k_r} \\
    \qquad =\  \sum_{j=1}^{r-1} B_j \left[ \sum_{\substack{k_r+p=n \\ k_r,p\ge1}} \zeta_j^p \zeta_r^{k_r}\right] =  \sum_{j=1}^{r-1} \left[  B_j  \left( \frac{\zeta_r}{\zeta_j-\zeta_r} \right) \right] \zeta_j^n +
      \left[ \sum_{j=1}^{r-1}  B_j \left(\frac{\zeta_j}{\zeta_r-\zeta_j}\right)\right]\zeta_r^n
%  \label{p-sum_r}
  \end{align*}
  where we have used the identity \eqref{sum-id}  to obtain the last equality. This
   has the form \eqref{p-sum}, so the proof of Lemma \ref{lem:prod-sum} is complete and therewith that of \eqref{Psi_n-a}.
  \end{proof}

{
\begin{proof}[Proof of \eqref{es-fourier-a}]
We may easily use Fourier analysis to re-express \eqref{Psi_n-a}. Note that the expression for $\Psi(n)$, given by \eqref{Psi_n-a}, extends naturally to vanish
 for all $n<1$. We introduce $\hat\Psi(\theta)$,  the discrete Fourier transform of $\{\Psi(n)\}_{n\in\Z}$:
\begin{equation}
\hat{\Psi}(\theta) =  \sum_{n\ge1} \Psi(n)\ e^{-in\theta}\quad {\rm for}\quad \theta\in\R/2\pi\Z.
\label{dftPsi} %(119)
\end{equation}
Hence,
\begin{equation}
\begin{aligned}
\hat{\Psi}(\theta)  &= 
   \sum_{n\ge1} \sum_{\substack{k_1+\dots+k_n=n \\ k_1,\dots,k_r\ge1}}  
 (\zeta_1e^{-i\theta})^{k_1}\cdots (\zeta_re^{-i\theta})^{k_r} 
  = \prod_{j=1}^r \left[  \sum_{\kappa\ge1}(\zeta_je^{-i\theta})^{\kappa} \right] = \prod_{j=1}^r \left[  \frac{\zeta_je^{-i\theta}}{1-\zeta_je^{-i\theta}} \right]
\end{aligned}
\label{hatPsi-es}
\end{equation}
Note that the above formal manipulations are justified because the roots $\zeta_1,\dots,\zeta_r$ lie within the open unit disc. In particular, the sum \eqref{dftPsi} converges. From \eqref{Psi_n-a} and \eqref{hatPsi-es} we deduce, by inversion of the discrete Fourier transform,  \eqref{es-fourier-a} for the solution of \eqref{Psin-gen}, \eqref{A-sys-gen} with $\nbase=1$. 
\end{proof}
}

 \begin{proof}[Proof of \eqref{128-a}] The representation  \eqref{128-a} follows from \eqref{es-fourier-a} via a residue calculation.
   Changing variables: $z=e^{i\theta}$, $( iz)^{-1}dz = d\theta${, we obtain} for $n\ge1$:
   \begin{align*}
   \Psi(n) &= \frac{1}{2\pi i} \int_{|z|=1} z^{n-1}\prod_{j=1}^r\frac{\zeta_j}{z-\zeta_j} dz =
   \prod_{i=1}^r \zeta_i\times \sum_{j=1}^r \frac{\zeta_j^{n-1}}{\prod_{l\in\{1,\dots,r\}\setminus\{j\}}\left(\zeta_l-\zeta_j\right)}\ .
   \end{align*}
 \end{proof}

{
 \begin{proof}[Proof of \eqref{norm-fou-a}]
A similar residue calculation gives the normalization constant of the wave function.
   Using the expression in \eqref{hatPsi-es} we have
   \begin{align*}
  \|\Psi\|_{l^2(\N)}^2=  \sum_{n\ge1}|\Psi(n)|^2 &= \frac{1}{2\pi} \int_0^{2\pi} |\hat{\Psi}(\theta)|^2 d\theta = 
    \frac{1}{2\pi} \int_0^{2\pi}\ \prod_{j=1}^r \frac{\zeta_j }{e^{i\theta}-\zeta_j}\ 
 \frac{\overline{\zeta_j}}{e^{-i\theta}-\overline{\zeta_j}}\ d\theta
   \end{align*}
   So, 
     \begin{align*}
\|\Psi\|_{l^2(\N)}^2 &=  |\zeta_1|^2\cdots |\zeta_r|^2\times \frac{1}{2\pi i}  \int_{|z|=1} \ \prod_{j=1}^r 
   \frac{1}{(z-\zeta_j )(z^{-1}-\bar{\zeta_j}) }\   z^{-1} dz \nn \\
   &=  |\zeta_1|^2\cdots |\zeta_r|^2\times  \frac{1}{2\pi i}  \int_{|z|=1} \ z^{r-1}\ \prod_{j=1}^r 
   \frac{1}{(z-\zeta_j)(1-\bar{\zeta_j}z) }\   dz\nn \\
   &= |\zeta_1|^2\cdots |\zeta_r|^2 \times \sum_{j=1}^r \frac{\zeta_j^{r-1}}{1-|\zeta_j|^2} \ \prod_{l\in\{1,\dots,r\}\setminus\{j\}} 
   \frac{1}{(\zeta_l-\zeta_j)(1-\overline{\zeta_j}\zeta_l)}
 \label{norm-fou}  \end{align*}
   \end{proof}
   }

\section{Non-zero energy, dispersive edge states}\label{numerics}
{By the results of Section \ref{sec:0energy}, if $a_{11}-a_{12}=\pm1$ mod $3$ then there are flat band / edge states over the parallel-quasimomentum ranges: 
 \[\textrm{ $(2\pi/3,4\pi/3)$ (balanced zigzag edge)
   or   $[0,2\pi/3)\cup (4\pi/3,2\pi]$ (\textcolor{black}{unbalanced} zigzag cut);}\]
  see Theorem \ref{th:zigzag}.
  \\\\
   These flat edge state curves bifurcate from band-crossings, with origins in the {\it Dirac points} of $H_{\rm bulk}$,
   %FOOTNOTE MOVED INTO THE TEXT FOR THE CPAM VERSION \footnote{with origins in the {\it Dirac points} 
  % of $H_{\rm bulk}$ acting in $l^2(\HH)$.} 
   in the essential spectrum of $\HTB_\sharp$.
   Any wave-packet constructed via superposition of such flat band edge states will not transport and will not disperse since the group velocities vanish identically.  
 \\\\  
In this section, we investigate the existence of dispersive non-zero energy edge states. {In contrast, wave packets constructed via superposition of such edge states will transport, spread out and decay with time along the edge.}
We present strong numerical evidence for the existence of non-zero energy edge state curves for 
 edges of \underline{both} zigzag and armchair type. Moreover, our simulations indicate a strong dependence of the number of such bifurcation curves on the characteristics of the edge.
 }
         
 \subsection{Setup for the study of non-zero energy edge states}\label{setup}
{Motivated by Remark \ref{ou_ptspec}, we expect all eigenvalue curves $\kpar\mapsto E(\kpar)$ 
to lie in the region
\begin{equation}
 \mathscr{K}=\Big\{(\kpar,E) : E\in\textrm{bounded component of}\  
 \R\setminus\sigma_{\rm ess}(H_\sharp(\kpar)) \Big\}.
 \label{Kdef}
 \end{equation} 
 Throughout this section we restrict our search for dispersion curves
 to $(\kpar,E)$ varying in $ \mathscr{K}$.
 }
{Fix any $k\in[0,2\pi]$. We seek  $E\in\R\setminus \sigma_{\rm ess}(H_\sharp(\kpar))$ for which there exists a non-trivial solution in $l^2(\mathbb{Z},\mathbb{C}^2)$ of the equations (\ref{psiAevp},\ref{psiBevp},\ref{psiA-BC},\ref{psiB-BC}). For $n$ large and positive, the system (\ref{psiAevp}-\ref{psiBevp}) can be rewritten as
\begin{equation}
	\sum_{\nu=1}^3\begin{pmatrix}
		0&e^{im_\nu k}\\0& 0
		\end{pmatrix}\psi(n+n_\nu)+ \sum_{\nu=1}^3\begin{pmatrix}
		0& 0\\e^{-im_\nu k}& 0
		\end{pmatrix}\psi(n-n_\nu)-\begin{pmatrix}
		E& 0\\0& E
		\end{pmatrix}\psi(n)=0.
\label{mat-EVP}\end{equation}
As explained in Section \ref{ess-spec}, {we study solutions of eigenvalue problem for $H_\sharp(\kpar)$ (\ref{psiAevp}-\ref{psiBevp})} by starting with exponential solutions of the form $\psi(n)=\zeta^n\xi$, where $0\ne\xi \in$ Ker$( \mathscr{P}_\kpar(\zeta) - E I\ )$,
$\mathscr{P}_\kpar(\zeta)$ is defined in \eqref{P+-}, and $\zeta$ satisfies 
\begin{equation}
	\label{detTkpar}
	\text{det}\left(\ \mathscr{P}_\kpar(\zeta) - E I\ \right)=0.
\end{equation}
{Although the eigenvalues of $H_\sharp(\kpar)$ are real, we will examine \eqref{mat-EVP} and \eqref{detTkpar} for complex $E\in\C\setminus{\rm spec}_{\rm ess}\left(H_\sharp(\kpar)\right)$. This will allow us to later apply elementary complex function theory to check the accuracy of our numerical computations.
}

It is easy to see that
\begin{equation}   \left(\ \mathscr{P}_\kpar(\zeta) - E I\ \right)\xi = 0 \
\iff\ \left(\ \mathscr{P}_{\kpar}\left(\ 1 / \bar\zeta\ \right) - E I\ \right)\sigma_1\bar\xi = 0.\label{root-sym2-bis}\end{equation}
Therefore, 
\[\textrm{$\{\psi(n)\}=\{\zeta^n \xi\}$\ satisfies\ (\ref{psiAevp}-\ref{psiBevp})\  iff\  $\{\tilde\psi(n)\}=\{\left(\ 1 / \bar\zeta\ \right)^n \sigma_1\bar\xi\}$
 \ satisfies\  (\ref{psiAevp}-\ref{psiBevp}).}\]}
 Here, $\sigma_1$ is the standard Pauli matrix given in \eqref{pauli123}.

Noting that any roots of \eqref{detTkpar}
 must be non-zero, we 
multiply \eqref{detTkpar} by $e^{i\kpar (m_3-m_1)}\zeta^{n_3-n_1}$ and obtain an equivalent polynomial equation of degree $2(n_3-n_1)$ for the roots $\zeta$ of \eqref{detTkpar}:
\begin{align*}
q_\kpar(\zeta,E) &= 0,\qquad {\rm where}\nn\\
q_\kpar(\zeta,E)\equiv & \left( \sum_{j=1}^3 e^{i\kpar (m_j-m_1)}\zeta^{(n_j-n_1)}\right)\
\times \left(\sum_{j=1}^3 e^{-i\kpar  (m_j-m_3)}\zeta^{-(n_j-n_3)}\right) -\  e^{i\kpar (m_3-m_1)} \zeta^{n_3-n_1}E^2  .
\label{q0}\end{align*}
{
The sets 
 \begin{align*}
 J^+_\kpar(E) &\equiv \big\{\zeta: |\zeta|<1\ {\rm and}\ {\rm det}\left(\mathscr{P}_\kpar(\zeta)-E I\right)=0\ \big\},
 %\label{J+}
 \\
  J^-_\kpar(E) &\equiv \big\{\zeta: |\zeta|>1\ {\rm and}\ {\rm det}\left(\mathscr{P}_\kpar(\zeta)-E I\right)=0\ \big\}
 %\label{J-} 
\end{align*}
 generate bulk solutions $\zeta^n \xi$, which decay as $n\to+\infty$ or $n\to-\infty$, respectively.
\\\\
By Proposition \ref{ess_spec}, we show easily that if $E\in\C$ and $E\notin{\rm spec}_{\rm ess}(\HTB_\sharp(k))$, then each of the $2(n_3-n_1)$ roots of $q_\kpar(\zeta,E) = 0$ is
 either strictly inside
 or strictly outside the unit circle. 
 {
 {For simplicity of the presentation, }we make the following assumption {regarding $(\kpar,E)$}:
  \begin{equation}
   \textrm{{\bf Assumption:}\ The roots, $\zeta$,  of $ P_{E,\kpar}(\zeta):=\det\left(\ \mathscr{P}_\kpar(\zeta) - E I\ \right)=0$  are all simple.}\label{simple-roots}
   \end{equation}
   }
   Note that this assumption implies that the nullspace of $\mathscr{P}_\kpar(\zeta) - E I$ is of dimension 1}.\begin{remark}\label{rem:dim1}If the dimension of the nullspace of $\mathscr{P}_\kpar(\zeta) - E I$ is strictly larger than 1, then there exist $v_1$ and $v_2$, two linearly independent vectors in the the nullspace of $\mathscr{P}_\kpar(\zeta_0) - E I$. Then for $\zeta$ near $\zeta_0$, $(\mathscr{P}_\kpar(\zeta) - E I)v_1$ and $(\mathscr{P}_\kpar(\zeta) - E I)v_2$ are both $\mathcal{O}(\zeta-\zeta_0)$ and hence $ \det\left(\ \mathscr{P}_\kpar(\zeta) - E I\ \right)=\mathcal{O}(\zeta-\zeta_0)^2)$ which contradicts \eqref{simple-roots}.
   \end{remark}
  %FOOTNOTE MOVED INTO THE TEXT FOR THE CPAM SUBMISSION \footnote{\label{footnote:dim1}Indeed, if not, suppose $v_1$ and $v_2$ are linearly independent vectors in the the nullspace of $\mathscr{P}_\kpar(\zeta_0) - E I$ then for $\zeta$ near $\zeta_0$, $(\mathscr{P}_\kpar(\zeta) - E I)v_1$ and $(\mathscr{P}_\kpar(\zeta) - E I)v_2$ are both $\mathcal{O}(\zeta-\zeta_0)$ and hence $ \det\left(\ \mathscr{P}_\kpar(\zeta) - E I\ \right)=\mathcal{O}(\zeta-\zeta_0)^2)$.}.
   {The following study could be extended to the case where Assumption \eqref{simple-roots} is not satisfied; we note that in all of our simulations this assumption is satisfied.}
 
  By \eqref{root-sym2-bis}, we deduce that
  \begin{equation*}
   \#J^+_\kpar(E)=\#J^-_\kpar(E) = n_3-n_1.\label{Nroots}
   \end{equation*}
Let $\zeta_j(\kpar,E)$ for $ j=1,\dots,n_3-n_1$ denote the $n_3-n_1$ distinct simple roots of $ P_{E,\kpar}(\zeta)$ in $J^+_\kpar(E)$,
 and 
$\zeta_j(\kpar,E)$ for $ j=n_3-n_1+1,\dots,2(n_3-n_1)$ the $n_3-n_1$ distinct simple roots of $ P_{E,\kpar}(\zeta)$ in $J^-_\kpar(E)$
\\\\
For each root $\zeta_j(\kpar,E)$, we pick a nonzero vector $\xi_j(k,E)$ in the one dimensional nullspace of $\mathscr{P}_\kpar(\zeta_j) - E I$. For most of our discussion the particular choice of $\xi_j(k,E)$ will be irrelevant. 
{However, our approach to checking the accuracy of our numerical computations requires  $\xi_j(k,E)$ to depend analytically on $E$ for fixed $k$. In particular, we shall require that if $(\kpar,E_0)\in\mathscr{K}$  then  $\xi_j(k,E)$ varies analytically $E$ in a sufficiently small open complex disc centered at $E_0$.  It is straightforward to verify that if $P_{E,\kpar}\left(\zeta_j(\kpar,E)\right)=0$ then
 \begin{equation}
  \xi_j(\kpar,E) := 
  \begin{pmatrix} P_+(\zeta_j(\kpar,E),k)+E\\ P_-(\zeta_j(\kpar,E),k)+E\end{pmatrix},\quad
\label{e-vecs1a}\end{equation}
is in the nullspace of $\mathscr{P}_\kpar(\zeta_j) - E I$, 
and we make the following assumption on $ \xi_j(\kpar,E)$:
\begin{equation}
   \textrm{{\bf Assumption:}\  All the vectors $\xi_j(\kpar,E)$ given by \eqref{e-vecs1a}, where $j=1,\dots,(n_3-n_1)$,
    are non-zero.}\label{evecs-ne0}
\end{equation}
   This assumption holds for all $(\kpar,E)$ encountered in our numerical simulations. Perhaps it holds for all $(\kpar,E)\in\mathscr{K}$ for any rational edge; we have not investigated whether this is the case. One could dispense with Assumption \eqref{evecs-ne0}. If for $(k,E_0)\in\mathscr{K}$ and a particular $j$ the vector given by \eqref{e-vecs1a} vanishes, then for $E$ in a complex neighborhood of $E_0$ we can choose instead of \eqref{e-vecs1a} the vector given by $\begin{pmatrix} P_+(\zeta,k)\\ E\end{pmatrix}$ which varies analytically and can be shown to be nonzero.
%FOOTNOTE MOVED INTO THE TEXT FOR THE CPAM SUBMISSION \footnote{One could dispense with Assumption \eqref{evecs-ne0}. If for $(k,E_0)\in\mathscr{K}$ and a particular $j$ the vector given by \eqref{e-vecs1a} vanishes, then for $E$ in a complex neighborhood of $E_0$ we can choose instead of \eqref{e-vecs1a} the vector given by $\begin{pmatrix} P_+(\zeta,k)\\ E\end{pmatrix}$ which varies analytically and can be shown to be nonzero. }
 %   \textcolor{purple}{Indeed, suppose they do both vanish. Then, $E=0$ and  $P_+(\zeta,\kpar)=0=P_-(\zeta,\kpar)$;
% therefore,  $p_+(\zeta,\kpar)-p_+(\zeta,\kpar)=0$. It follows
%  that $|\zeta|=1$,  provided $n_3-n_2\ne n_2-n_1$; equality corresponds to the classical armchair case.
%  But $|\zeta|=1$ and $E=0$ contradicts our assumption that $E\notin \sigma_{\rm ess}(H_\sharp(\kpar))$. So at least one of the vectors  is non-zero. }
  Until Proposition \ref{prop:edge-1D}, our particular choice of the $\xi_j(\kpar,E)$ will play no role.
}
%%%
\medskip

{By assumption \eqref{simple-roots} and since the roots $\zeta_j$ are nonzero,} the general solution of \eqref{psiAevp}-\eqref{psiBevp} can be written as
   {
    \begin{equation*} 
  	\psi(n) = \sum_{j=1}^{2(n_3-n_1)} A_j\ [\zeta_j(\kpar,E)]^n\ \xi_j(\kpar,E), \quad \text{for $n$ large and positive,}
  \label{psi-bulk}\end{equation*}
  where $A_1,\dots, A_{2(n_3-n_1)}$  are $2(n_3-n_1)-$ complex parameters. Consequently,  an $\ell^2$ vector satisfying \eqref{psiAevp}-\eqref{psiBevp} for $n$ large {can be written as a linear combination of the  $n_3-n_1$  vectors $\xi_j(\kpar,E)$ corresponding 
  to roots $\zeta_j\in J_\kpar^+(E)$:}
\begin{equation} 
	\psi(n) = \sum_{j=1}^{n_3-n_1} A_j\ [\zeta_j(\kpar,E)]^n\ \xi_j(\kpar,E), \quad \text{for $n$ large and positive.}
\label{psi-bulk-dec-bis}\end{equation}
 } 
  
{Recall that an edge state of $\HTB_\sharp(\kpar)$ is $\ell^2$ and solves \eqref{psiAevp}-\eqref{psiBevp} and the boundary conditions \eqref{psiA-BC}-\eqref{psiB-BC}. 
{Let us now determine the range of  $n$ over which an edge state has the form \eqref{psi-bulk-dec-bis}, and the set of algebraic constraints on the coefficients $A_1,\dots, A_{(n_3-n_1)}$ implied by the boundary conditions.
Relation \eqref{psiAevp} is equivalent to
\begin{align}
\sum_{\nu=1}^3 e^{im_\nu\kpar}\psi^B(n-n_1+n_\nu)) = E \psi^A(n-n_1)
,\quad\text{for $n\ge \nA+n_1$}. \label{psiBn-bis}
\end{align}
And similarly, \eqref{psiBevp} is equivalent to
\begin{align}
\sum_{\nu=1}^3 e^{-im_\nu\kpar}\psi^A(n-n_\nu+n_3) = E \psi^B(n+n_3)
,\quad\text{for $ n\ge \nB-n_3$}.  \label{psiAn-bis}
\end{align}
%
% Since, $\nA\in\{0,1\}$, $-n_1\ge\nA$. Hence,  $\{n\ge \nA+n_3\}\supset\{n\ge-n_1+n_3\}$. And since
% $\nB\in\{0,1\}$, $n_3\ge\nB$ and hence
%  $\{n\ge \nB-n_1\}\supset\{n\ge n_3-n_1\}$. Thus
 Both equations in \eqref{psiBn-bis}-\eqref{psiAn-bis} 
apply to the range $n\ge \max( \nA+n_1,\nB-n_3)$.
}
 Consequently, $\psi$ can be written as \eqref{psi-bulk-dec-bis} for $n\ge \max( \nA+n_1,\nB-n_3)$; see Section \ref{sec:Bsite-es} for a similar reasoning.
 % REMOVED FOR THE CPAM SUBMISSION see Footnote \ref{footnote:largen}.}
{
{Let us first assume that $\nA+n_1\geq\nB-n_3$; we comment below 
on the case where $\nA+n_1<\nB-n_3$.}
Then,
the solutions of (\ref{psiAevp}-\ref{psiBevp}-\ref{psiA-BC}-\ref{psiB-BC} satisfy precisely
  \begin{align}
\psi^A(n) = \sum_{j=1}^{n_3-n_1} A_j\ [\zeta_j(\kpar,E)]^n\ \xi_j^A(\kpar,E),\quad&{\rm for\ }\; n\ge \nA+n_1 \label{psi-genA-bis}\\
\psi^B(n) = \sum_{j=1}^{n_3-n_1} A_j\ [\zeta_j(\kpar,E)]^n\ \xi_j^B(\kpar,E),\quad&{\rm for\ }\; n\ge \nA+n_1, \label{psi-genB-bis}\\
\sum_{\nu=1}^3 e^{-im_\nu\kpar}\psi^A(n-n_\nu+n_3) = E \psi^B(n+n_3)
,\quad&\text{for $ \nB-n_3\le n < \nA+n_1$}\label{BC-A-plus}\\
\psi^A(n)=0,\quad& {\rm for\ }\; n<\nA, \label{BC-A-bis}\\
\psi^B(n)=0,\quad& {\rm for\  }\; n<\nB. \label{BC-B-bis}
\end{align}
First, \eqref{psi-genA-bis},\eqref{BC-A-bis} and \eqref{psi-genB-bis},\eqref{BC-B-bis} imply that 
$A_1,\ldots, A_{n_3-n_1}$ are subject to the constraints
\begin{align}
\sum_{j=1}^{n_3-n_1} A_j\ [\zeta_j(\kpar,E)]^n\ \xi_j^A(\kpar,E)=0,\quad \text{for}\;\nA+n_1\le n< \nA \label{const_A_0}\\
\sum_{j=1}^{n_3-n_1} A_j\ [\zeta_j(\kpar,E)]^n\ \xi_j^B(\kpar,E)=0,\quad  \text{for}\;\nA+n_1\le n< \nB \label{const_B_0}
\end{align}

{
Let us now show that \eqref{BC-A-plus} is equivalent to the assertion that \eqref{const_A_0} holds also in the range $\nB-n_3\le n< \nA+n_1$. Let us focus first on the right hand side \eqref{BC-A-plus}.
 Since $ \nB-n_3\le n < \nA+n_1$, we have $n+n_3\geq \nB\geq \nA+n_1$. Therefore, 
 $\psi^B(n+n_3)$ can be re-expressed using \eqref{psi-genB-bis}. Turning to the three terms, $\psi^A(n-n_\nu+n_3)$, on the left hand side of \eqref{BC-A-plus}, we observe:
 concerning the term $\psi^A(n-n_1+n_3)$,  since $n-n_1+n_3\geq \nA >\nA+n_1$ (because $n_1<0$), it follows that $\psi^A(n-n_1+n_3)$ can be re-expressed using \eqref{psi-genA-bis}.
Concerning the term $\psi^A(n-n_2+n_3)$, we have that if $n-n_2+n_3\geq\nA+n_1$, then  $\psi^A(n-n_2+n_3)$ can be re-expressed via \eqref{psi-genA-bis}, and otherwise $\psi^A(n-n_2+n_3)=0$. And finally, concerning the term $\psi^A(n-n_3+n_3)=\psi^A(n)$, since $n<\nA$ we have $\psi^A(n)=0$. } These observations imply that the relations \eqref{BC-A-plus} can be rewritten as
\begin{multline*}
	 \sum_{j=1}^{n_3-n_1} A_j\left[e^{-im_1\kpar}\ [\zeta_j(\kpar,E)]^{n-n_1+n_3}\ \xi_j^A(\kpar,E)\right.\\+ \mathds{1}_{n-n_2+n_3\geq\nA+n_1} e^{-im_2\kpar}  [\zeta_j(\kpar,E)]^{n-n_2+n_3}\xi_j^A(\kpar,E)\\\left. - 
	  E  [\zeta_j(\kpar,E)]^{n+n_3}\ \xi_j^B(\kpar,E)\right]=0
	,\quad\text{for}\; \nB-n_3\le n < \nA+n_1
\end{multline*}
{
We may simplify this system by noting that  $[\mathscr{P}_\kpar(\zeta_j(\kpar,E)) - E I][\xi_j(\kpar,E)]=0$ for each $j$,
see \eqref{Tkpar}-\eqref{P+-}.
}
It follows that
\begin{multline*}
	 \sum_{j=1}^{n_3-n_1} A_j\left[-e^{-im_2\kpar}\mathds{1}_{_{n-n_2+n_3<\nA+n_1}}\ [\zeta_j(\kpar,E)]^{n-n_2+n_3}\ 
	 - e^{-im_3\kpar}  [\zeta_j(\kpar,E)]^{n}\right]\xi_j^A(\kpar,E) = 0 ,
	\\\quad\text{for $ \nB-n_3\le n < \nA+n_1$}
\end{multline*}
Since $n_2<n_3$, we have $ \nB-n_3\le n \le n-n_2+n_3<\nA+n_1$ whenever $\mathds{1}_{_{n-n_2+n_3<\nA+n_1}}=1$ so that this system can be rewritten as a triangular system
% FOOTNOTE MOVED INTO THE TEXT FOR THE CPAM SUBMISSION \footnote{Since $n_2<n_3$, we have $ \nB-n_3\le n \le n-n_2+n_3<\nA+n_1$ whenever $\mathds{1}_{_{n-n_2+n_3<\nA+n_1}}=1$ so that the system is indeed triangular.} 
in terms of $R_{n}:=  \sum_{j=1}^{n_3-n_1} A_j [\zeta_j(\kpar,E)]^{n}\xi_j^A(\kpar,E)$
\[
	e^{-i m_2\kpar} \mathds{1}_{_{n-n_2+n_3<\nA+n_1}} R_{n-n_2+n_3}+ e^{- im_3 k} R_{n}=0,\quad\text{for $ \nB-n_3\le n < \nA+n_1$}
\] 
whose only solution is $R_n=0$ for $ \nB-n_3\le n < \nA+n_1$.\\\\
 {Gathering \eqref{const_A_0}, \eqref{const_B_0} and this last result, we} conclude the edge state eigenvalue problem \eqref{psiAevp}-\eqref{psiBevp},
 \eqref{psiA-BC}- \eqref{psiB-BC} is satisfied  if and only if $A_1,\ldots, A_{n_3-n_1}$ are subject to the constraints
\begin{align}
\sum_{j=1}^{n_3-n_1} A_j\ [\zeta_j(\kpar,E)]^n\ \xi_j^A(\kpar,E)=0,\quad \text{for}\;\nB-n_3\le n< \nA \label{const_A}\\
\sum_{j=1}^{n_3-n_1} A_j\ [\zeta_j(\kpar,E)]^n\ \xi_j^B(\kpar,E)=0,\quad  \text{for}\;\nA+n_1\le n< \nB\ . \label{const_B}
\end{align}
A similar analysis in the case where  $\nA+n_1<\nB-n_3$ also leads to \eqref{const_A}-\eqref{const_B}.

Let us count the number of equations. Altogether we have:   $ (\nA-\nB+n_3)+(\nB-\nA-n_1)= n_3-n_1$ equations. Thus,  \eqref{const_A}-\eqref{const_B} is a linear homogeneous system of $n_3-n_1$ equations in  $n_3-n_1$ unknowns
 $\mathbb{A}:=(A_1,\dots,A_{n_3-n_1})^T$.}

{ We abbreviate this system of $n_3-n_1$ equations in $n_3-n_1$ unknowns by
 \begin{equation} \mathbb{M}(\kpar,E) \mathbb{A} = 0, 
\label{MA0}\end{equation}
we denote
its determinant
 \begin{align}
 \Delta(\kpar,E) := \det \mathbb{M}(\kpar,E)
 \label{Delta-def}
 \end{align}
 {Although $\mathbb{M}(\kpar,E)$ depends on the choice of the vectors $\xi_j(\kpar,E)$,
 we do not explicitly indicate this dependence.
 However, the zeros of  $E\mapsto\Delta(\kpar,E)$ do not depend on the choice of $n_3-n_1$ vectors, $\xi_j(\kpar,E)$.
 }
  \begin{proposition}\label{es-cond} Assume that  $E_0\in\C$ is not in the essential\ spectrum of $\HTB_\sharp(\kpar)$ and that 
 Assumption \eqref{simple-roots} is satisfied. Then, $E_0$ is an eigenvalue of $\HTB_\sharp(\kpar)$ if and only if
  $\Delta(\kpar,E_0)=0$. 
\\\\	  
 Furthermore, if   $(\kpar,E_0)$ is such that $\Delta(\kpar,E_0)=0$, 	 then the corresponding edge states $\psi\in l^2(\Z)$ are given by 
   \begin{align*}
 \psi^A(n) &= \sum_{j=1}^{n_\nthree-n_\none} A_j\ [\zeta_j(\kpar,E_0)]^n\ \xi_j^A(\kpar,E_0),\quad n\ge\nA \\
 \psi^B(n) &= \sum_{j=1}^{n_\nthree-n_\none} A_j\ [\zeta_j(\kpar,E_0)]^n\ \xi_j^B(\kpar,E_0),\quad n\ge\nB. \\
 \psi^A(n)&=0\quad {\rm for\ all}\quad n<\nA \\
 \psi^B(n)&=0\quad {\rm for\ all }\quad n<\nB 
 \end{align*}
 where $A=(A_1,\dots,A_{n_3-n_1})^\top$ is an arbitrary solution of \eqref{MA0}.
   \end{proposition}
{
Fix $(k,E_0)\in \mathscr{K}$ satisfying Assumptions \eqref{simple-roots} and \eqref{evecs-ne0}. As $E$ varies in a small complex disc $D$ about $E_0$, those assumptions hold also for $(k,E)$, and 
$E\notin{\rm spec}_{\rm ess}\left(H_\sharp(\kpar) \right)$. Therefore, the roots $\zeta_j(\kpar,E)$, the vectors $\xi_j(\kpar,E)$ given by \eqref{e-vecs1a}, the matrix $\mathbb{M}(\kpar,E)$ and its determinant $\Delta(\kpar,E)$ are all analytic functions of $E\in D$. We can easily deduce the following from the argument of Remark \ref{rem:dim1}.
% REMOVED FOR THE CPAM submission footnote \ref{footnote:dim1}.
   }
  \begin{proposition} \label{prop:edge-1D}  
  { Suppose that $(k,E_0)\in \mathscr{K}$, Assumptions \eqref{simple-roots} and \eqref{evecs-ne0} are satisfied and the vectors $\xi_j(\kpar,E_0)\in\C^2$ are chosen according to \eqref{e-vecs1a}.} Then $\partial_E\Delta(\kpar,E_0)$ is well defined and if $\Delta(\kpar,E_0)=0$ but $\partial_E\Delta(\kpar,E_0)\neq 0$ then the space of edge states is one-dimensional.
	\end{proposition}}
 
   \subsection{Numerical results on dispersive edge states
 }\label{numerical}
  We numerically search for $(\kpar,E)\in\mathscr{K}$ such that 
 $\Delta(\kpar,E)=0$.
To do so, given $(\kpar,E)\in\mathscr{K}$, we may carry out the following algorithm:
 \begin{itemize}
 	\item Compute the $2(n_3-n_1)$ roots $\zeta$ of $ P_{E,\kpar}(\zeta):= \det\left(\ \mathscr{P}_\kpar(\zeta) - E I\ \right)$. The roots are calculated by computing the eigenvalues of the associated companion matrix {We verify Assumption \ref{simple-roots}; the roots, $\zeta$, are distinct. In all of our computations, evaluation of the polynomial on computed roots gives a residual  of  order $10^{-13}$.}
	 See Figure~\ref{fig:roots_P} for some examples. 
	  The roots are ordered as explained in the previous section:  $\zeta_j(\kpar,E),\ j=1,\dots,n_3-n_1$ denote the $n_3-n_1$ roots which lie inside the unit circle (in red in Figure~\ref{fig:roots_P}) and
$\zeta_j(\kpar,E),\ j=n_3-n_1+1,\dots,2(n_3-n_1)$ denote the $n_3-n_1$ roots which lie outside (in blue in Figure~\ref{fig:roots_P}). 
	\item Deduce whether $E\in{\rm spec}_{\rm ess}(\HTB(\kpar))$ or not by using Proposition \ref{ess_spec}-(iii), see Figures \ref{fig:essential_spectrum_large} and \ref{fig:essential_spectrum}.
	\item Compute, for each root $\zeta(\kpar,E)$, 
	{the vector given by \eqref{e-vecs1a}.}
	\item Construct the matrix $\mathbb{M}(\kpar,E)$ appearing in \eqref{MA0} and compute its determinant $\Delta(\kpar,E)$.
 \end{itemize}
	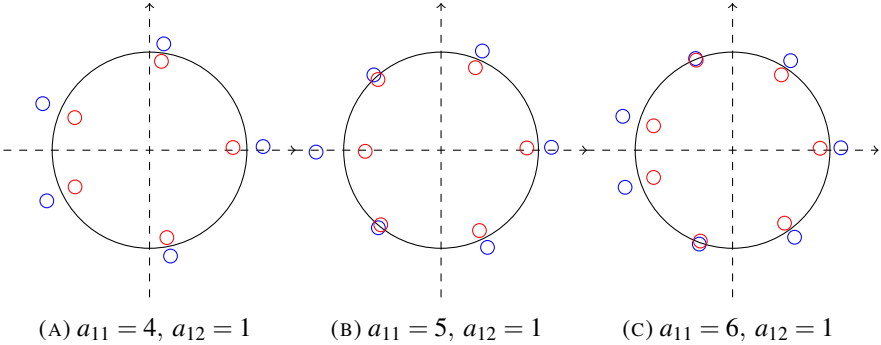
\begin{figure}[htbp]
	    \centering
   	     \begin{subfigure}{0.3\textwidth}
   	        \centering
   	        \begin{tikzpicture}[scale=1.3]
   				\draw[->,dashed](-1.5,0)--(1.5,0);
   				\draw[->,dashed](0,-1.5)--(0,1.5);
				\draw[blue] (1.1658,0.0370) circle (2pt);
				\draw[red] (0.8569,0.0272) circle (2pt);
				\draw[blue] (0.2155,- 1.0792) circle (2pt);
				\draw[red] (0.1779,- 0.8911) circle (2pt);
				\draw[blue] (0.1463,+ 1.0839) circle (2pt);
				\draw[red] (0.1223,+ 0.9061) circle (2pt);
				\draw[blue] (-1.0550,- 0.5163) circle (2pt);
				\draw[blue] (-1.0971, + 0.4745) circle (2pt);
				\draw[red] (-0.7647, - 0.3743) circle (2pt);
				 \draw[red] (-0.7679, + 0.3321) circle (2pt);				
   	        \draw (0,0) circle (1cm);
   	        \end{tikzpicture}
   	        \caption{$a_{11}=4,\;a_{12}=1$}\end{subfigure}
			\begin{subfigure}{0.3\textwidth}
		        \centering
		        \begin{tikzpicture}[scale=1.3]
					\draw[->,dashed](-1.5,0)--(1.5,0);
					\draw[->,dashed](0,-1.5)--(0,1.5);
		        \draw (0,0) circle (1cm);
				\draw[blue] (1.1336,0.0294) circle (2pt);
				\draw[red] (0.8815,0.0228) circle (2pt);
				\draw[blue] (0.4770,- 0.9909) circle (2pt);
				\draw[red] (0.3944,- 0.8193) circle (2pt);
				\draw[blue] (0.4239,+ 1.0111) circle (2pt);
				\draw[red] (0.3527,+ 0.8412) circle (2pt);
				\draw[blue] (-0.6439,- 0.7911) circle (2pt);
				\draw[red] (-0.6189,- 0.7604) circle (2pt);
				\draw[blue] (-1.2835,- 0.0191) circle (2pt);
				\draw[blue] (-0.6901,+ 0.7675) circle (2pt);
				\draw[red] (-0.6478,+ 0.7205) circle (2pt);
				\draw[red] (-0.7789,- 0.0116) circle (2pt);
		        \end{tikzpicture}
		        \caption{$a_{11}=5,\;a_{12}=1$}\end{subfigure}
		    \begin{subfigure}{0.3\textwidth} 
		        \centering
		        \begin{tikzpicture}[scale=1.3]
					\draw[->,dashed](-1.5,0)--(1.5,0);
					\draw[->,dashed](0,-1.5)--(0,1.5);
		        \draw (0,0) circle (1cm);
				\draw[blue] (1.1119,0.0243) circle (2pt);
				\draw[red] (0.8989,0.0197) circle (2pt);
				\draw[blue] (0.5963,0.9126) circle (2pt);
				\draw[red] (0.5017,0.7679) circle (2pt);
				\draw[blue] (0.6367,- 0.8873) circle (2pt);
				\draw[red] (0.5338,- 0.7439) circle (2pt);
				\draw[blue] (-0.3421,- 0.9574) circle (2pt);
				\draw[red] (-0.3310,- 0.9263) circle (2pt);
				\draw[blue] (-0.3801,+ 0.9346) circle (2pt);
				\draw[red] (-0.3734,+ 0.9181) circle (2pt);
				\draw[blue] (-1.1037,- 0.3783) circle (2pt);
				\draw[blue] (-1.1274,+ 0.3454) circle (2pt);
				\draw[red] (-0.8108,- 0.2779) circle (2pt);
				\draw[red] (-0.8109,+ 0.2484) circle (2pt);
		        \end{tikzpicture}
		        \caption{$a_{11}=6,\;a_{12}=1$}\end{subfigure}	   	     
			\caption{The roots of the polynomial $ P_{E,\kpar}(\zeta)$ for $E=0.1$ and $\kpar=3$ for three different edges}
			\label{fig:roots_P}
		\end{figure}
We make a {\it heat map}
 of the function $(\kpar,E)\mapsto\log |\Delta(\kpar,E)|$ over the $(N_\kpar+1)\times(N_E+1)$ grid of points:
\begin{equation}\label{eq:disc}
	\kpar\in\Big\{\ell\frac{2\pi}{N_k},\;\ell\in\{0,1,\ldots,N_k\}\Big\} \;\text{and}\; E\in\Big\{-E_\text{lim}+j\frac{2E_\text{lim}}{N_E},\;\ell\in\{0,1,\ldots,N_E\}\Big\},
\end{equation}
where $N_k, N_E, E_\text{lim}$ are specified for each simulation. In all figures, the dark areas correspond to the essential spectrum;
 see Figures \ref{fig:edge_states_classical}, \ref{fig:edge_states1}, \ref{fig:essential_spectrum_large}, \ref{fig:essential_spectrum}, \ref{fig:edge_states}, \ref{fig:edge_states_armchair}, \ref{fig:edge_states_zigzag}, \ref{fig:edge_states_Fib}. Outside the dark areas, in particular for $(\kpar,E)\in\mathscr{K}$ (see \eqref{Kdef}), we seek edge state curves by studying where $(\kpar,E)\mapsto\log |\Delta(\kpar,E)|$ takes on very large negative values. In Figure \ref{fig:edge_states}, we consider 4 different edges, three of zigzag-type and one of armchair-type, with numerical parameters 
 and $ N_k=N_E=1000$ and specified values of $E_\text{lim}$.  The function $E\mapsto|\Delta(\kpar,E)|$ appears to vanish at $E=0$ for $k\in (2\pi/3,4\pi/3)$, for all three choices of zigzag-type edges (subplots (B),(C) and (D)) %$(a_{11},a_{12})=(5,1)$ and $(a_{11},a_{12})=(6,1)$. 
 but not to vanish at $E=0$ for all $\kpar\in(0,2\pi)$ for the armchair-type edge (subplot (A)).
 This illustrates Proposition \ref{Bstate-conds1} and Proposition \ref{Astate-conds1}. Note that these plots, however, do not reveal whether the edge states live on  $B-$ sites or a $A-$sites of $H_\sharp$.
% FOOTNOTE MOVED INTO THE TEXT FOR THE CPAM Submission \footnote{These plots, however, do not reveal whether the edge states live on  $B-$ sites or a $A-$sites of $H_\sharp$.}. 
 Moreover $|\Delta(\kpar,E)|$ appears to vanish along other (non-flat) curves. In Figure \ref{fig:edge_states1} we present, for the 3 different edges, such curves along which $|\Delta(E,\kpar)|$ falls below some small threshold value. The curves have been darkened for clarity. %Such curves appear in Figure \ref{fig:edge_states1} 
 % FOOTNOTE MOVED INTO THE TEXT FOR THE CPAM Submission \footnote{In Figure \ref{fig:edge_states1} we present curves along which $|\Delta(E,\kpar)|$ falls below some small threshold value. The curves have been darkened for clarity.}.
  We suspect, and we shall confirm,  the existence of edge states for $(\kpar,E)$ varying along these curves. 
 \begin{figure}[htbp]
 	\begin{center}
      	     \begin{subfigure}{0.45\textwidth}
      	        \centering
 			 \begin{tikzpicture}[scale=0.9]
 				\node at (0,0) {\includegraphics[height=5cm]{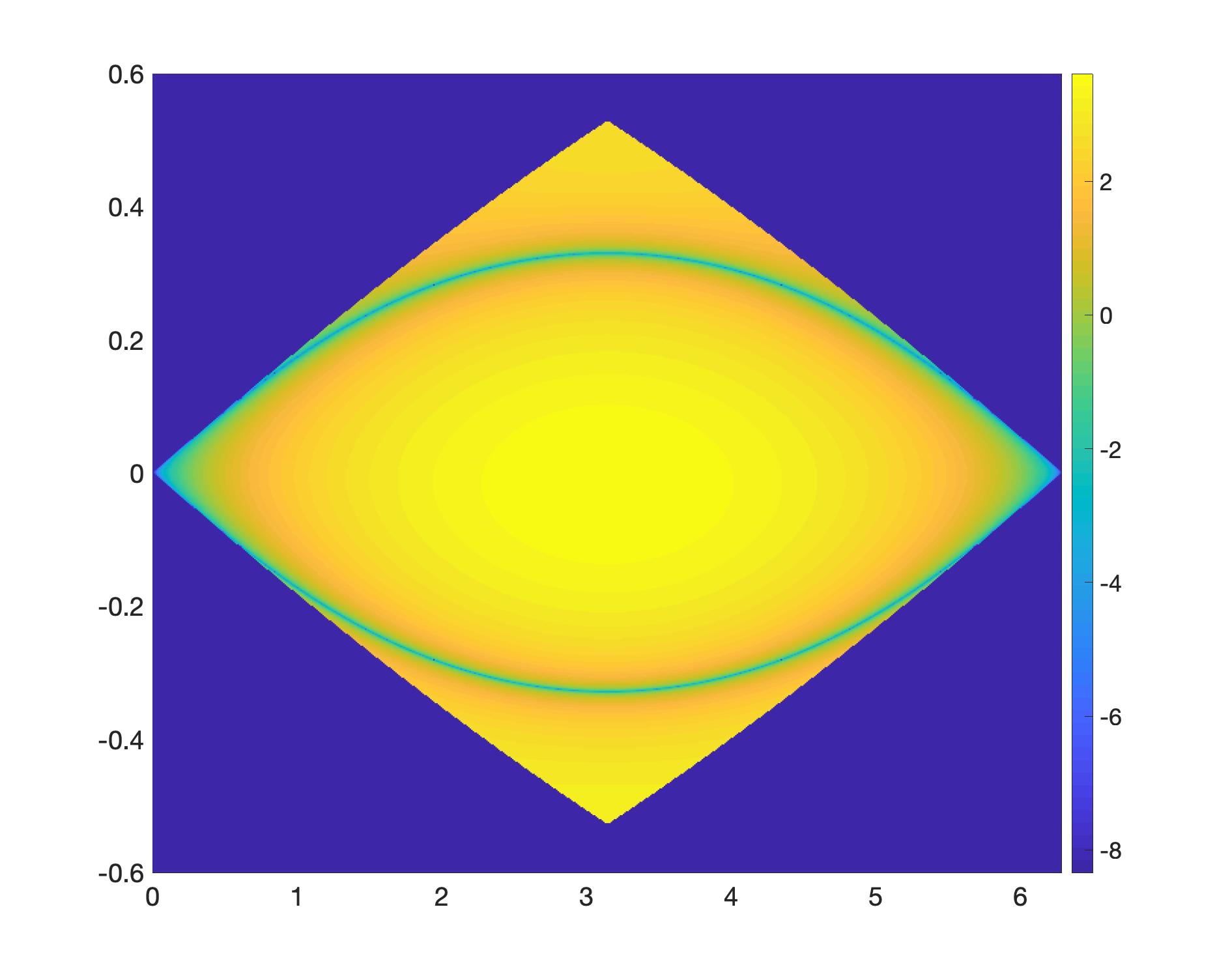}};
				\draw (-3,0) node{$E$};
				\draw (0,-2.5) node{$k$};
   	        \end{tikzpicture}
   	        \caption{$a_{11}=4,\,a_{12}=1,\,E_\text{lim}=0.6$ \textcolor{white}{(balanced)}}\end{subfigure}
     	     \begin{subfigure}{0.45\textwidth}
     	        \centering
			 \begin{tikzpicture}[scale=0.9]
 				\node at (0,0) {\includegraphics[height=5cm]{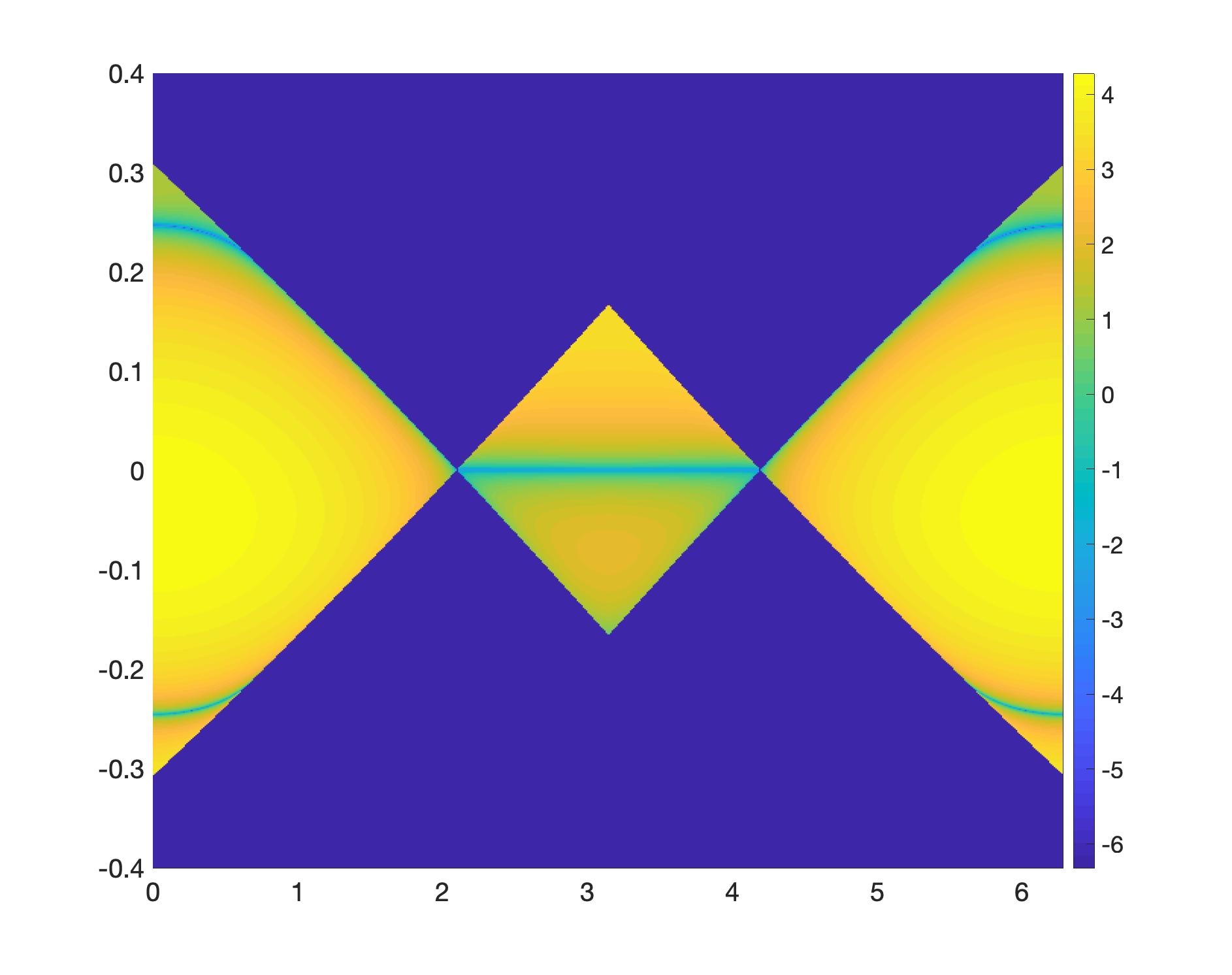}};
				\draw (0,-2.5) node{$k$};\draw (-3,0) node{$E$};
   	        \end{tikzpicture}
   	        \caption{$a_{11}=5,\,a_{12}=1,\,E_\text{lim}=0.4$ (balanced)}\end{subfigure}
			     	     \begin{subfigure}{0.45\textwidth}
     	        \centering
			 \begin{tikzpicture}[scale=0.9]
 				\node at (0,0) {\includegraphics[height=5cm]{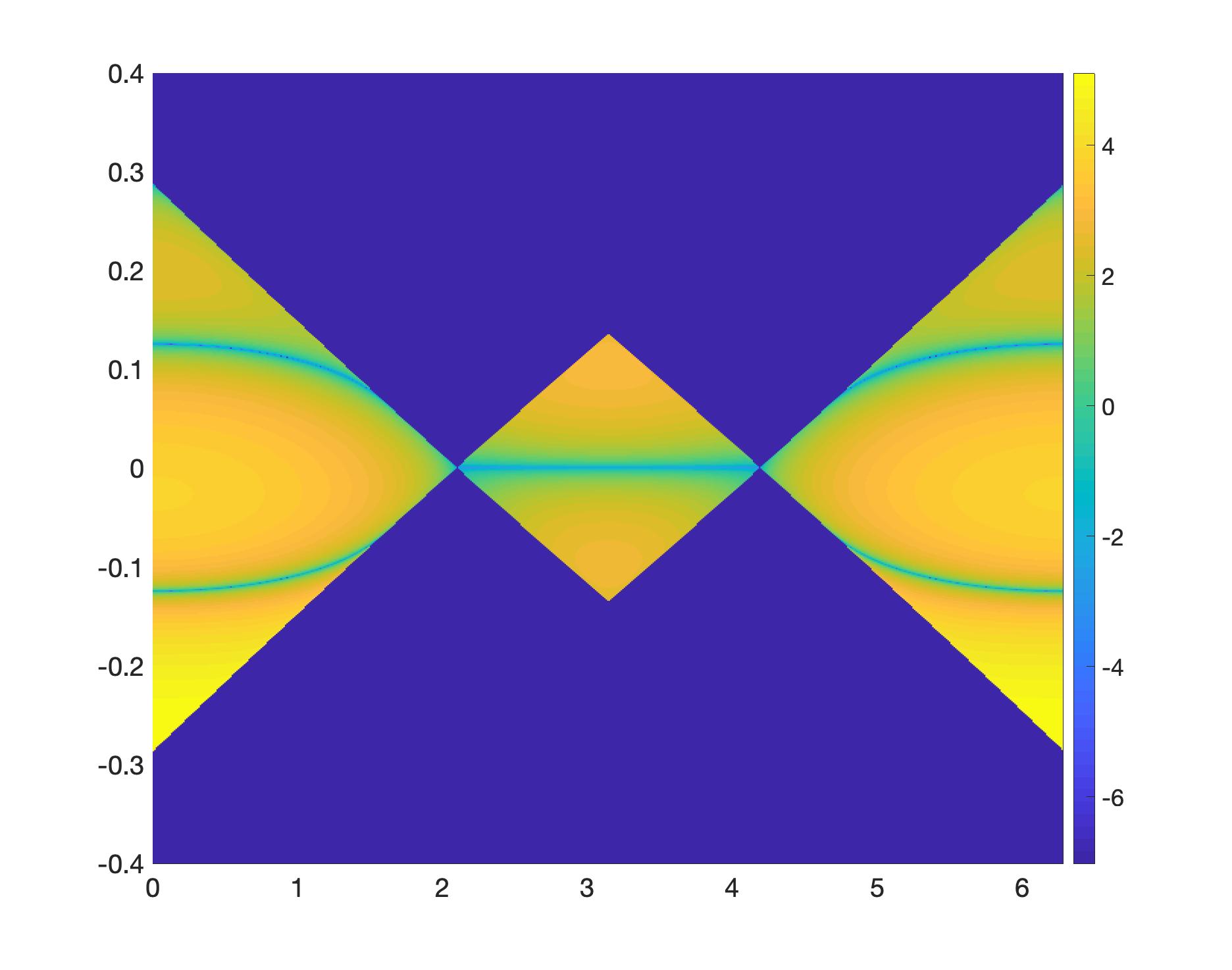}};
				\draw (0,-2.5) node{$k$};
   	        \end{tikzpicture}
   	        \caption{$a_{11}=6,\,a_{12}=1,\,E_\text{lim}=0.4$ (balanced)}\end{subfigure}
			    	     \begin{subfigure}{0.45\textwidth}
    	        \centering
		 \begin{tikzpicture}[scale=0.9]
 				\node at (0,0) {\includegraphics[height=5cm]{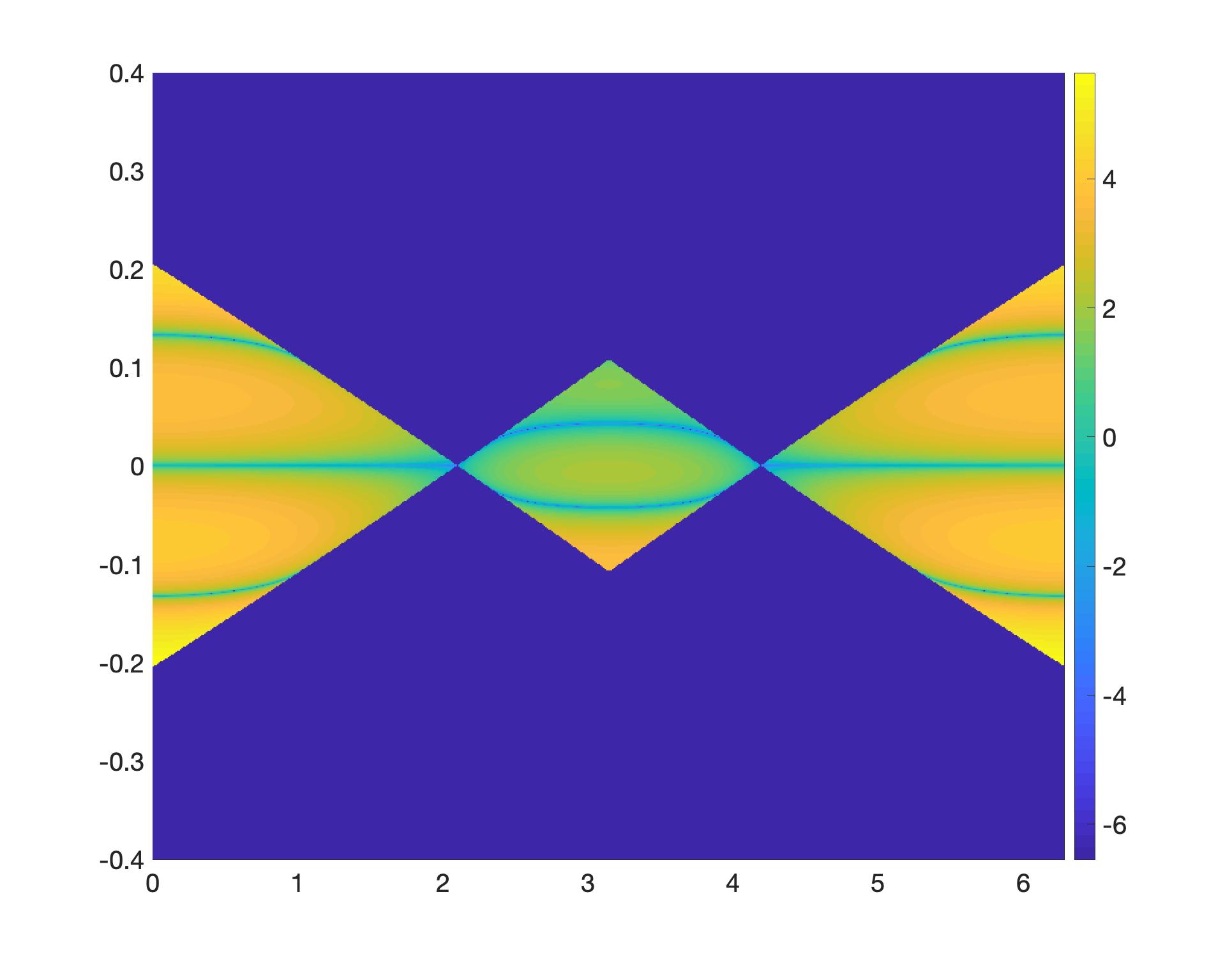}};
				\draw (0,-2.5) node{$k$};
   	        \end{tikzpicture}
   	        \caption{ {$a_{11}=8,\,a_{12}=1,\,E_\text{lim}=0.4$ (unbalanced)}}\end{subfigure}
 		\caption{Plots of $(\kpar,E)\mapsto\log |\Delta(\kpar,E)|$ for $k\in(0,2\pi)$ ($N_k=1000$ points) and $E\in(-E_\text{lim},E_\text{lim})$ ($N_E=1000$ points), see \eqref{eq:disc}. {Note that the color scale is different from one figure to another.}}
 		\label{fig:edge_states}
 	\end{center}
 \end{figure}
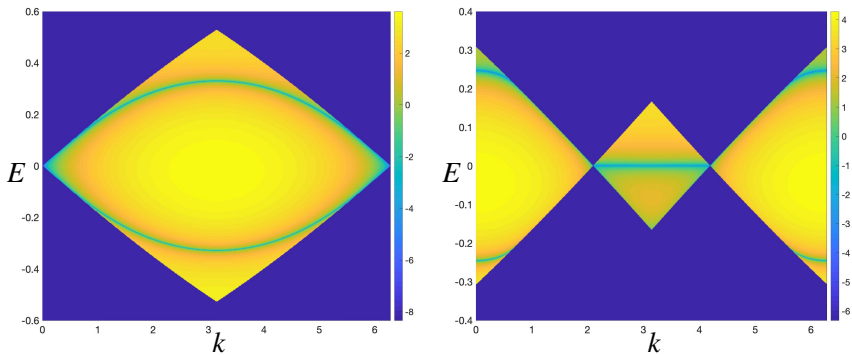
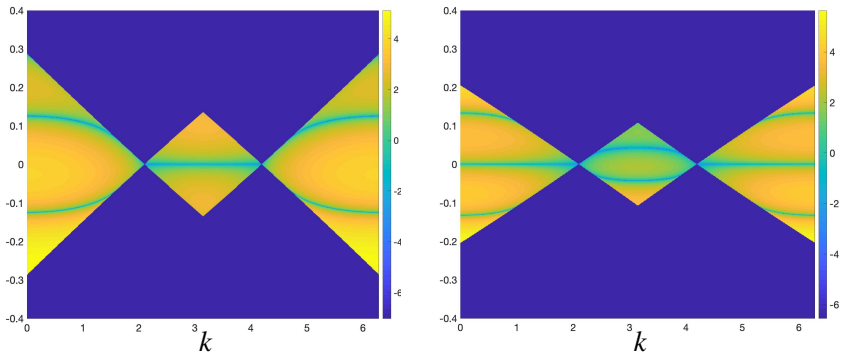
 
~\\\\ 
{To confirm existence of edge state near $(\kpar,E)=(k_0,E_0)$, we investigate $E\mapsto \Delta(k_0,E)$ in a neighborhood of $E_0$.  
{We first check that Assumptions \eqref{simple-roots} and \eqref{evecs-ne0}  hold and that $(\kpar_0,E_0)\in\mathscr{K}$.} From the discussion at the end of Section ~\ref{setup}, we know that for $E$ varying in some open complex neighborhood of $E_0$, the roots of \eqref{detTkpar} are still simple and do not lie on the unit circle. Hence, $E\mapsto \Delta(k_0,E)$ is well-defined and analytic in this neighborhood. We can then apply the algorithm presented above for any value $E$ in this neighborhood.} 

{For example, we find approximate zeros of the mapping $E\mapsto \Delta(k_0,E)$ for three cases:
 (a) $(a_{11},a_{12})=(4,1)$ and $(k_0,E_0)=(3,0.33)$, (b) $(a_{11},a_{12})=(5,1)$ and $(k_0,E_0)=(0.27,0.24)$ and (c) $(a_{11},a_{12})=(6,1)$ and $(k_0,E_0)=(1,0.11)$. }

{
We corroborate the existence of a zero near each $E_0$ by computing the winding number
 of the mapping $E\mapsto \Delta(k_0,E)$ along a sufficiently small circle about $E_0$.
 In Figure~\ref{fig:winding} we display, corresponding to each approximate zero $(k_0,E_0)$,
  the  image of a discretization of a small circle about $E_0$:
\begin{equation*}\label{eq:winding_discE}
	E\in \{E_0+r_Ee^{\imath j\frac{2\pi}{N_c}},\;j\in\{0,\ldots,N_c-1\}\},\quad (r_E=0.01, N_c=50).
\end{equation*}
  ~\\\\
For $E\in\{E_0+r_Ee^{\imath \theta_E},\;\theta_E\in[0,2\pi]\}$ and $r_E$ small enough, $\Delta(k_0,E)$ can be rewritten as $\Delta(k_0,E)=\rho_\Delta(\theta_E)e^{\imath \theta_\Delta(\theta_E)}$ with $\rho_\Delta$ and $\theta_\Delta$ continuous with respect to $\theta_E$. The winding number is then given by
\[
	W(k_0,E_0)=\frac{\theta_\Delta(2\pi)-\theta_\Delta(0)}{2\pi}.
\]
In all the  cases mentioned, the winding number is equal to $1$, implying that  there exists a simple root of $E\mapsto \Delta(k_0,E)$ near $E_0$. By Proposition \ref{prop:edge-1D}, the space of edge states {corresponding to a pair near} $(k_0,E_0)$ is one dimensional. We have checked the robustness of our winding number calculation by computing it for a range
 of sufficiently small radii, $r_E$; see Figure~\ref{fig:winding_big}. 
 }

 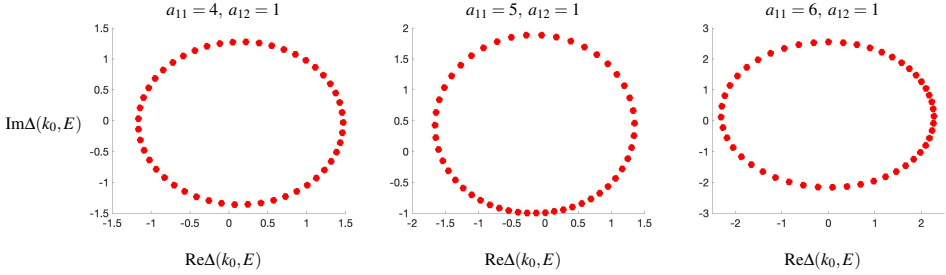
\begin{figure}[htbp]
 	\begin{center}
 		\begin{tikzpicture}		
 			\begin{scope}[shift={(-6.0,0)},scale=0.6,transform shape]	
 				\node at (0,0) {\includegraphics[height=5cm]{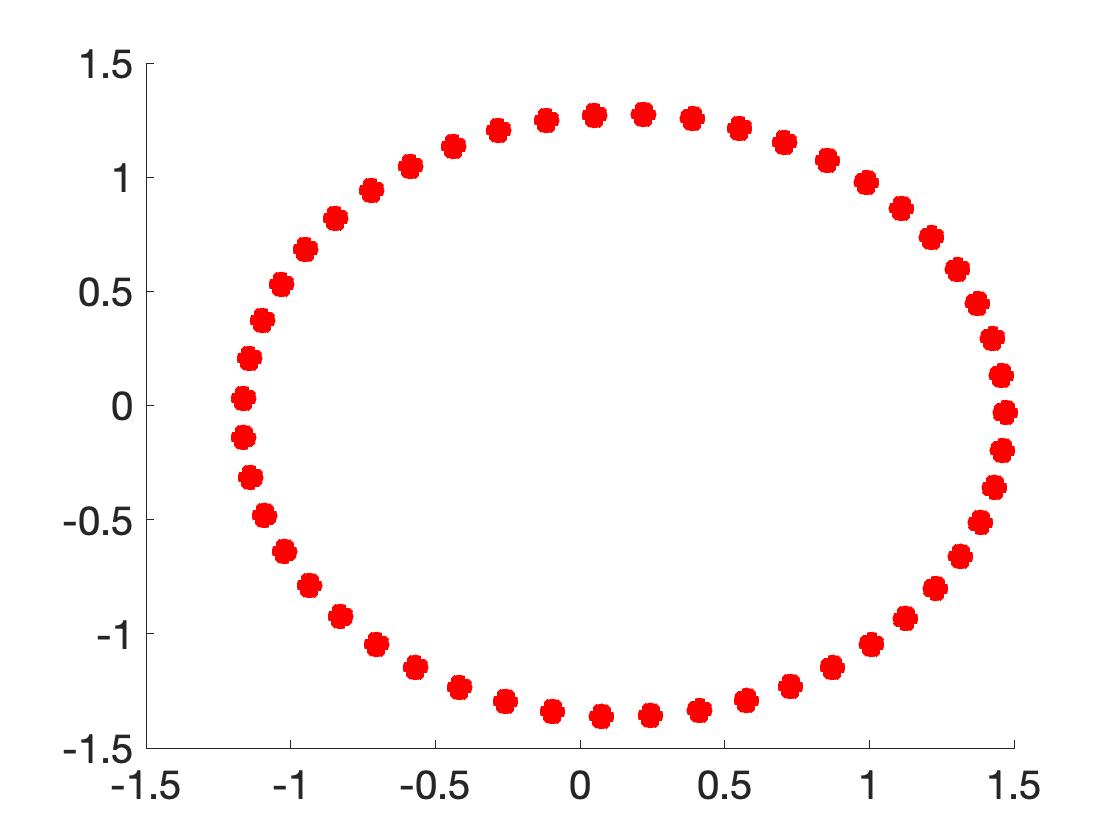}};
 				\draw (0.,2.5) node{$a_{11}=4,\;a_{12}=1$};
				\draw (-4,0) node{Im$\Delta({k}_0,E)$};
				\draw (0,-3) node{Re$\Delta({k}_0,E)$};
 			\end{scope}
			
 			\begin{scope}[shift={(-2,0)},scale=0.6,transform shape]	
 				\node at (0,0) {\includegraphics[height=5cm]{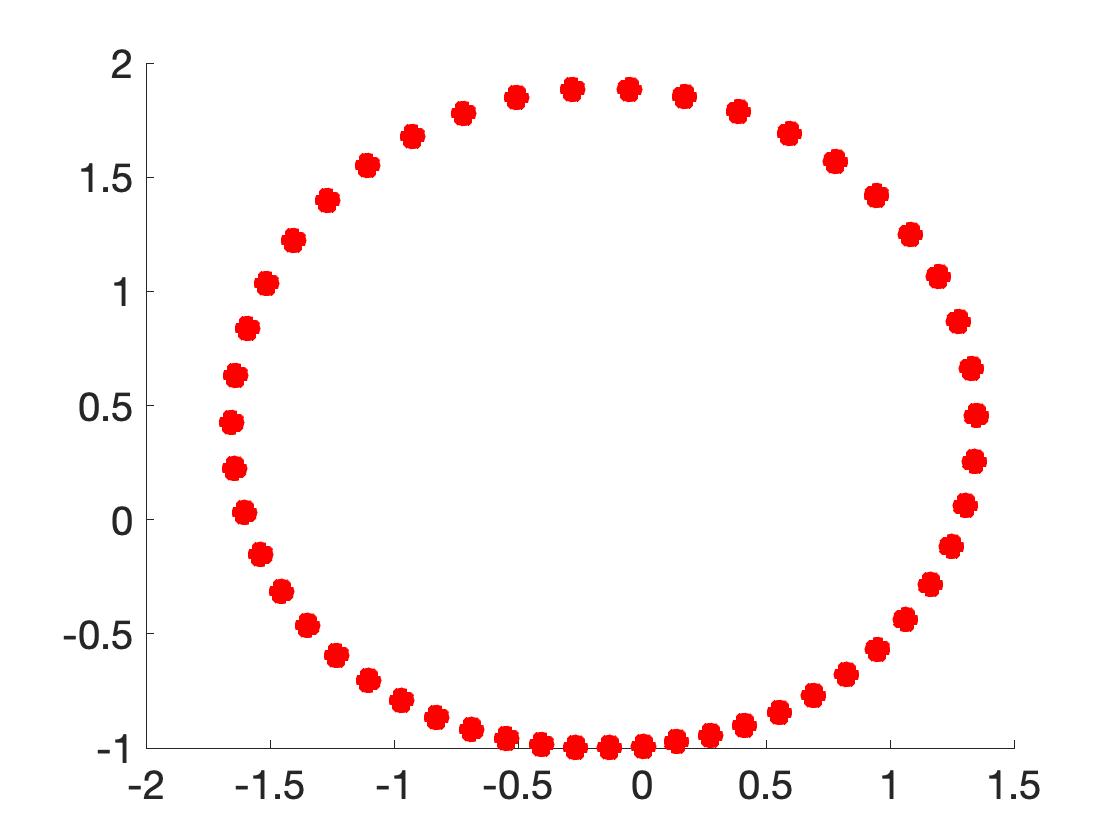}};
 				\draw (0.,2.5) node{$a_{11}=5,\;a_{12}=1$};
				\draw (0,-3) node{Re$\Delta({k}_0,E)$};
 			\end{scope}
			
 			\begin{scope}[shift={(2.0,0)},scale=0.6,transform shape]	
 				\node at (0,0) {\includegraphics[height=5cm]{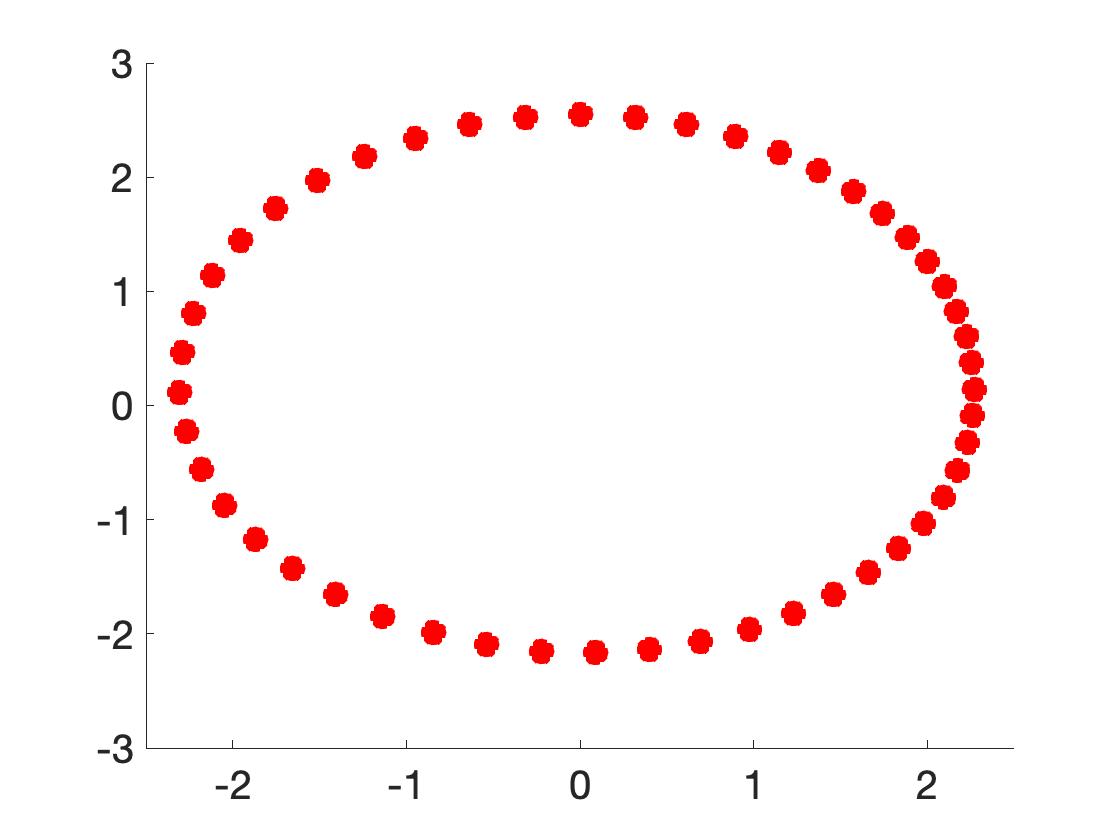}};
 				\draw (0.,2.5) node{$a_{11}=6,\;a_{12}=1$};
				\draw (0,-3) node{Re$\Delta({k}_0,E)$};
 			\end{scope}
 		\end{tikzpicture}
 		\caption{Plots of $\Delta({k}_0,E)$in the complex plane when $E$ takes the values specified in \eqref{eq:disc} with $r_E=0.01$ and $N_c=50$.}
 		\label{fig:winding}
 	\end{center}
 \end{figure}

 \begin{figure}[htbp]
 	\begin{center}
 		\begin{tikzpicture}
 			\begin{scope}[shift={(-6.0,0)},scale=0.6,transform shape]	
 				\node at (0,0) {\includegraphics[height=5cm]{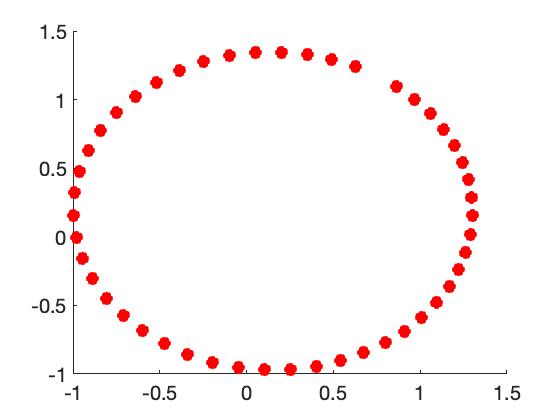}};
 				\draw (0.,2.5) node{$r_E=0.005$};
				\draw (-4,0) node{Im$\Delta({k}_0,E)$};
				\draw (0,-3) node{Re$\Delta({k}_0,E)$};
 			\end{scope}
			
 			\begin{scope}[shift={(-2,0)},scale=0.6,transform shape]	
 				\node at (0,0) {\includegraphics[height=5cm]{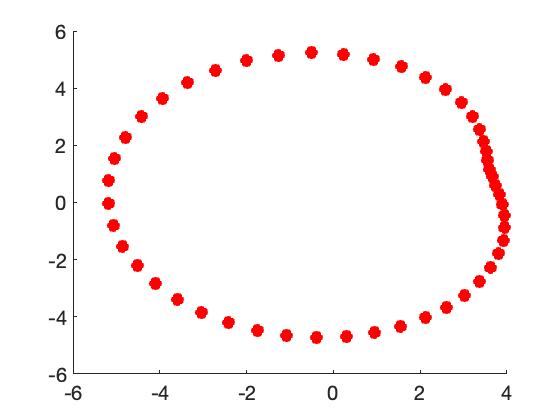}};
 				\draw (0.,2.5) node{$r_E=0.02$};
				\draw (0,-3) node{Re$\Delta({k}_0,E)$};
 			\end{scope}
			
 			\begin{scope}[shift={(2.0,0)},scale=0.6,transform shape]	
 				\node at (0,0) {\includegraphics[height=5cm]{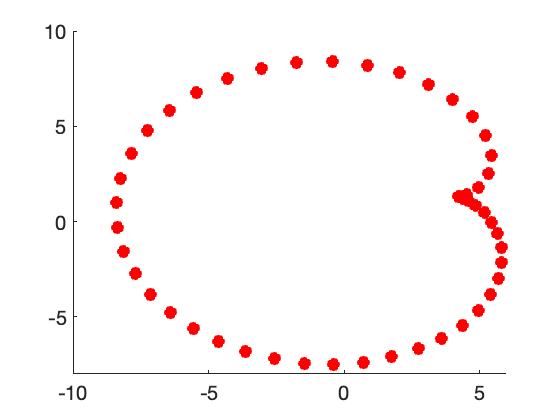}};
 				\draw (0.,2.5) node{$r_E=0.03$ };
				\draw (0,-3) node{Re$\Delta({k}_0,E)$};
 			\end{scope}
 		\end{tikzpicture}
 		\caption{For the edge $a_{11}=6,\;a_{12}=1$, plots of $\Delta({k}_0,E)$ in the complex plane when $E$ takes the values specified in \eqref{eq:disc} with $N_c=50$ and various values of $r_E$.}
 		\label{fig:winding_big}
 	\end{center}
 \end{figure}
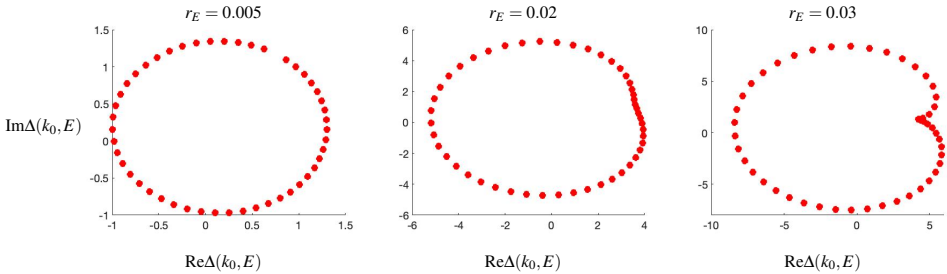
 {We consider now a sequence of armchair edges defined by $a_{12}=1$ as $a_{11}$ increases. We observe the presence of multiple dispersive (non-flat) edge state curves bifurcating from $(\kpar,E)=(0,0)$ (and from $(\kpar,E)=(2\pi,0)$), {we see} that the number of curves increases when $a_{11}$ increases; see Figure ~\ref{fig:edge_states_armchair}. Similar observations hold for ordinary zigzag-like edges; see Figure ~\ref{fig:edge_states_zigzag}. Note that as $a_{11}$ tends to infinity, the sequence of studied edges ($a_{11}$ increasing  and $a_{12}=1$) tends to the ordinary classical zigzag edge. Although the classical zigzag edge has a single dispersion curve, which is flat only over a limited range of $k$, the nearly flat dispersion curves in Figures ~\ref{fig:edge_states_armchair} and ~\ref{fig:edge_states_zigzag} extend over all $k\in(0,2\pi)$. Note however that the definition of quasimomentum depends on the edge.~\\\\

 Let us now consider a sequence of edges, defined by parameters $( a_{11}^{(n)},a_{12}^{(n)})$ given by
 \[
 	\forall n\geq 3,\quad a_{11}^{(n)}=F_n\quad\text{and}\quad a_{12}^{(n)}=F_{n+1},
 \]
where $(F_n)_n$ is the Fibonacci sequence, defined recursively via
 \[
 	F_1=F_2=1,\quad F_{n+2}=F_{n+1}+F_n, 
 \]
and giving an approximation of the golden ratio
 \[
 	\lim_{n\rightarrow+\infty}\frac{a_{12}^{(n)}}{a_{11}^{(n)}}= \lim_{n\rightarrow+\infty}\frac{F_{n+1}}{F_n} = \frac{1+\sqrt{5}}2.
 \] For this sequence, we do not observe an increasing number of dispersion curves; see Figure \ref{fig:edge_states_Fib}. Instead we see a single dispersion curve and its mirror image, again tending to a flat band as $n$ tends to infinity. Note that color scale is not the same from one plot to another, which reflects the sensitive dependence of the determinant $\Delta(k,E)$ on $E$. The existence of edge states is validated as previously by computing winding numbers near the energies on suspected dispersive curves.}
  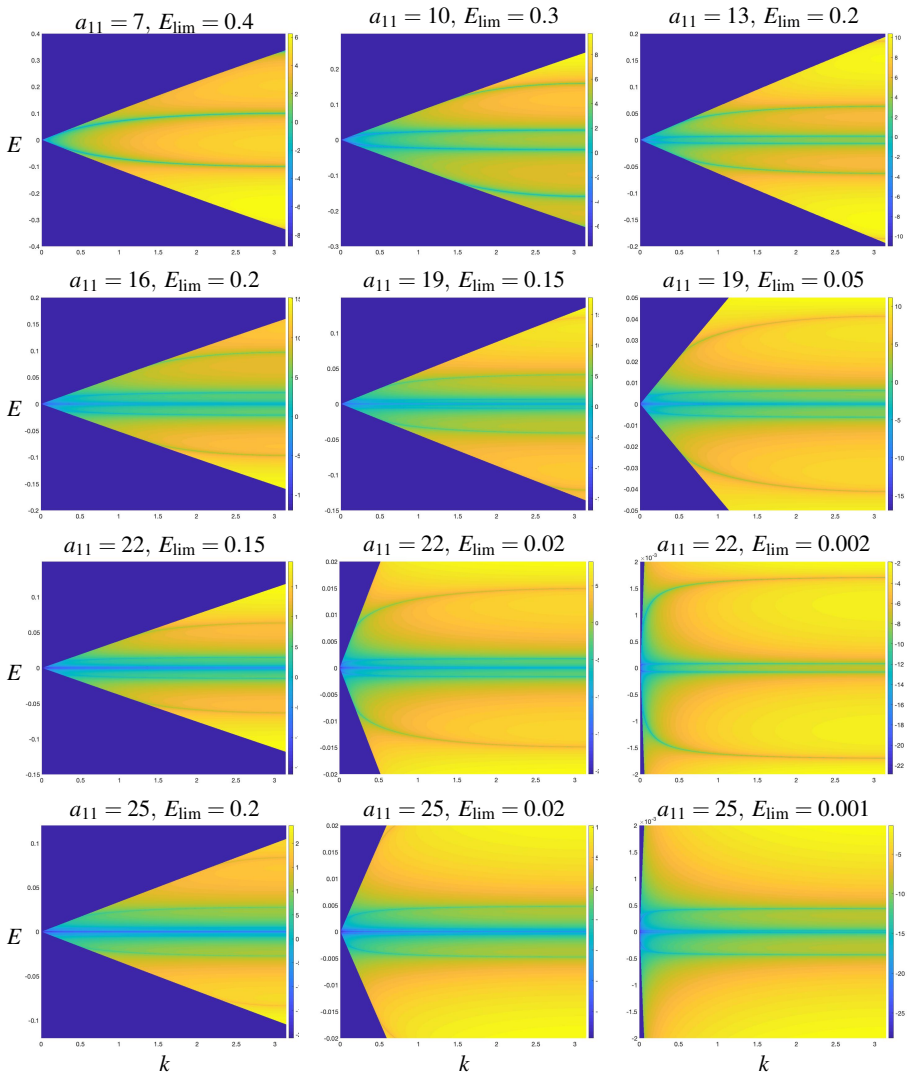
\begin{figure}[htbp]
  	\begin{center}
  		\begin{tikzpicture}
			
  			\begin{scope}[shift={(-4.0,5.5)},scale=0.8,transform shape]	
  				\node at (0,0) {\includegraphics[height=4.3cm]{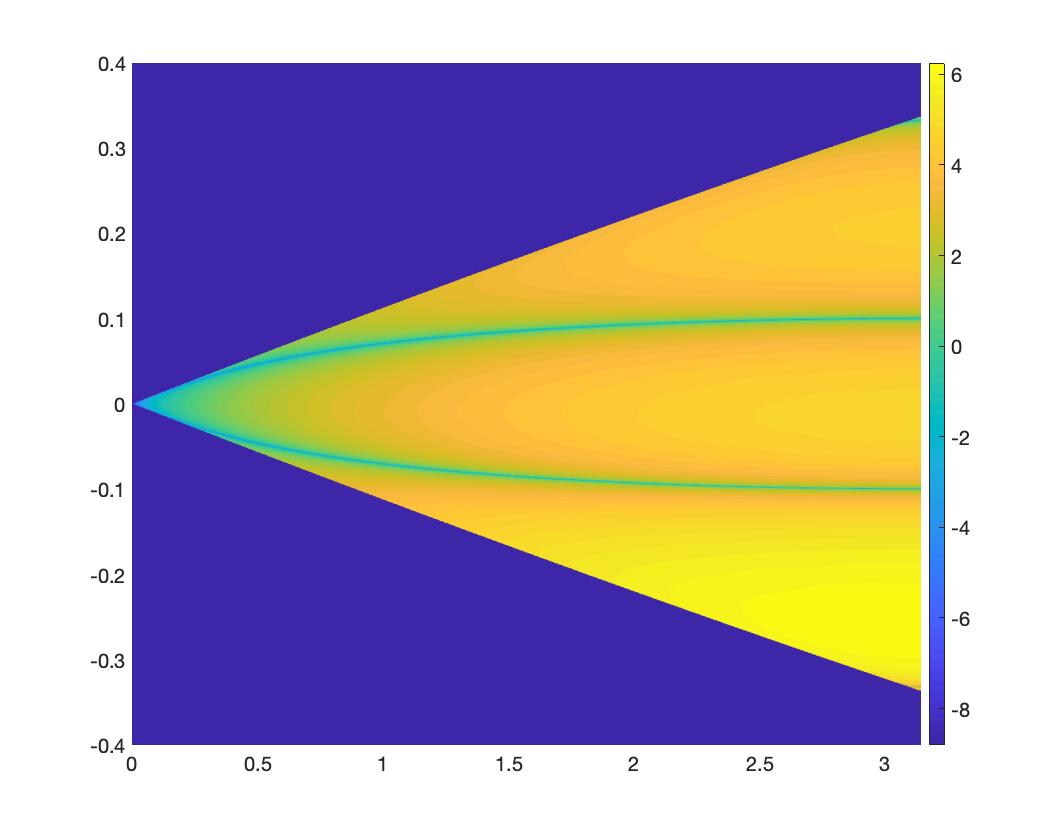}};
  				\draw (0.,2) node{$a_{11}=7,\;E_\text{lim}=0.4$};
 				\draw (-2.5,0) node{$E$};
 				%\draw (0,-2.5) node{$k$};
  			\end{scope}
			
  			\begin{scope}[shift={(0,5.5)},scale=0.8,transform shape]	
  				\node at (0,0) {\includegraphics[height=4.3cm]{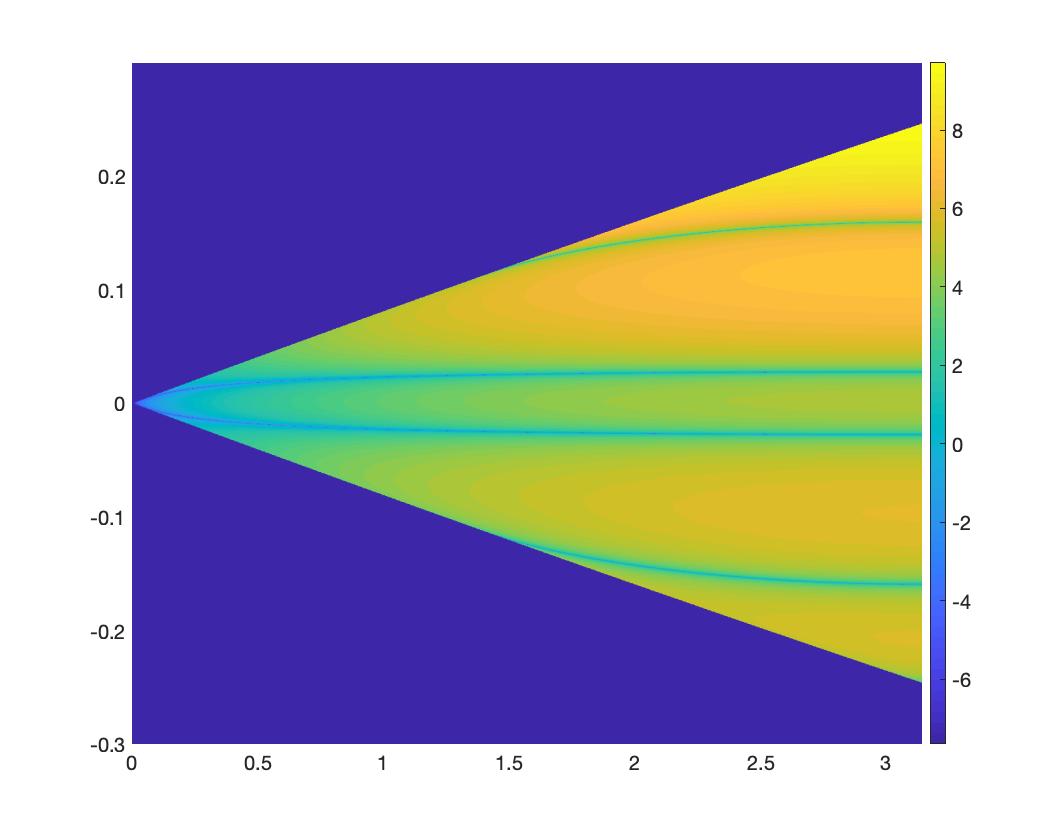}};
  				\draw (0.,2.1) node{$a_{11}=10,\;E_\text{lim}=0.3$};
 				%\draw (0,-2.5) node{$k$};
  			\end{scope}
			
  			\begin{scope}[shift={(4.0,5.5)},scale=0.8,transform shape]	
  				\node at (0,0) {\includegraphics[height=4.3cm]{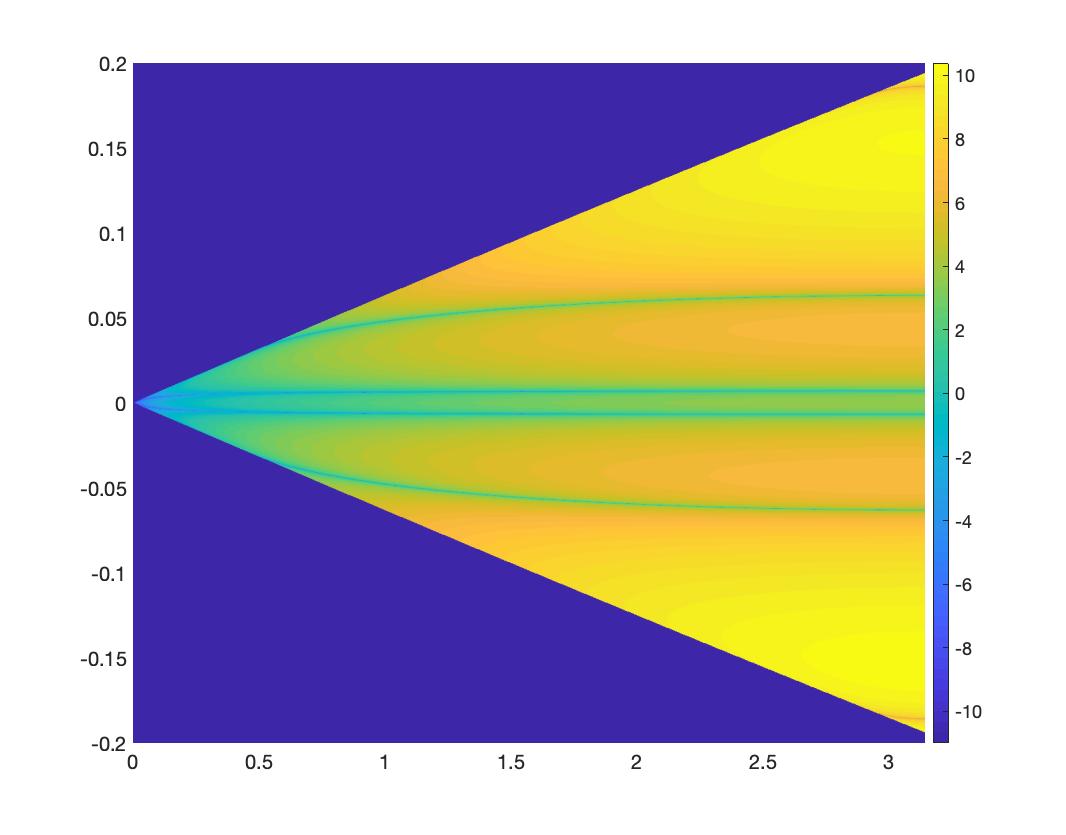}};
  				\draw (0.,2.1) node{$a_{11}=13,\;E_\text{lim}=0.2$};
 				%\draw (0,-2.5) node{$k$};
  			\end{scope}
			
  			\begin{scope}[shift={(-4.0,2)},scale=0.8,transform shape]
  				\node at (0,0) {\includegraphics[height=4.3cm]{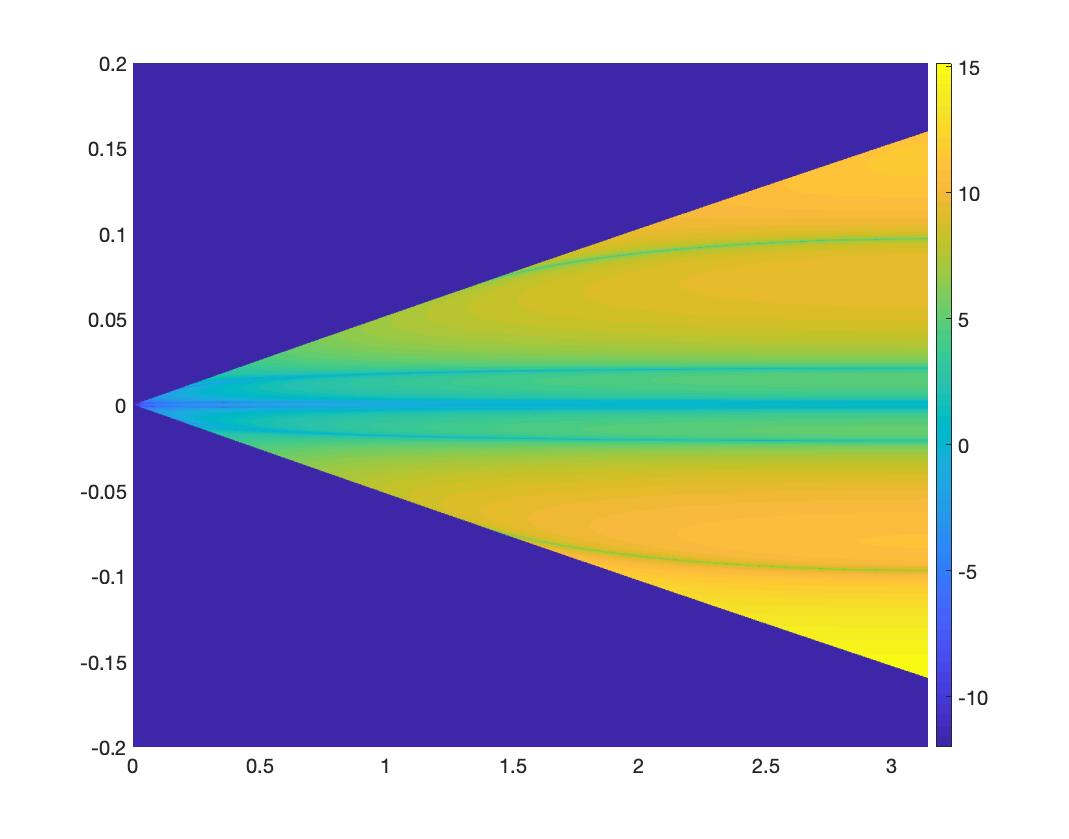}};
  				\draw (0.,2.1) node{$a_{11}=16,\;E_\text{lim}=0.2$};
 				%\draw (0,-2.5) node{$k$};
         \draw (-2.5,0) node{$E$};
  			\end{scope}
  			\begin{scope}[shift={(0,2)},scale=0.8,transform shape]
  				\node at (0,0) {\includegraphics[height=4.3cm]{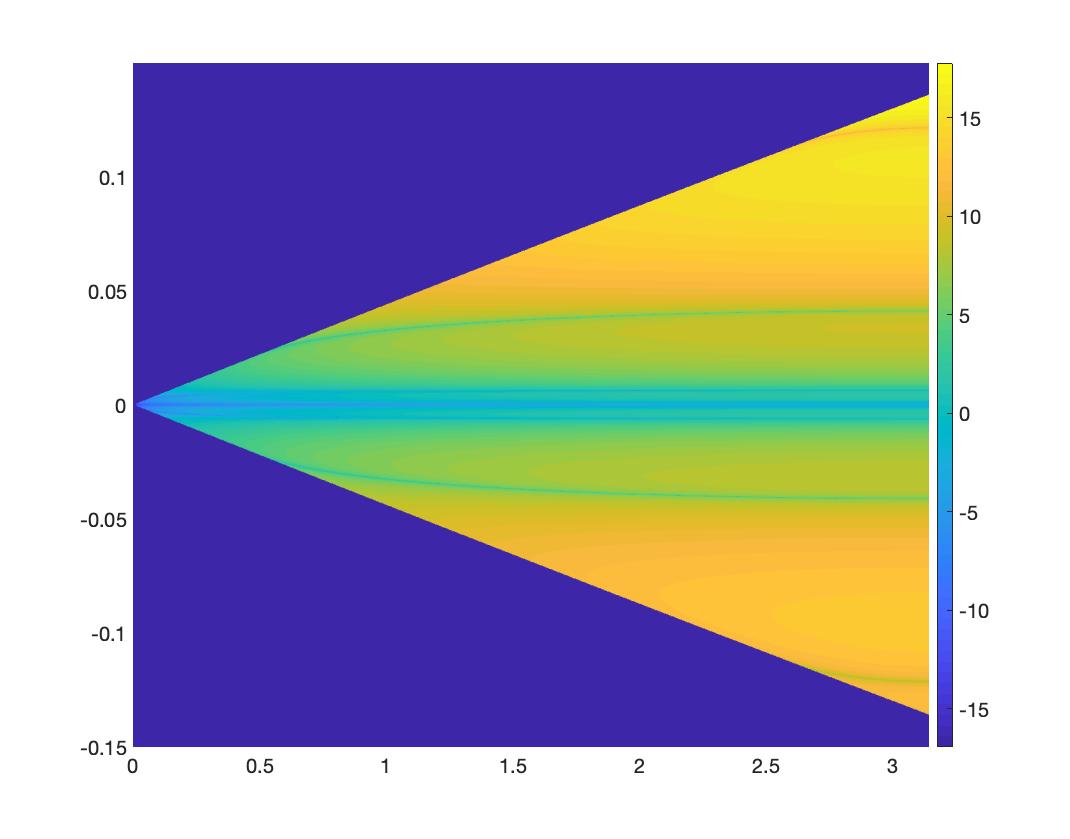}};
  				\draw (0.,2.1) node{$a_{11}=19,\;E_\text{lim}=0.15$};
 				%\draw (0,-2.5) node{$k$};
  			\end{scope}
  			\begin{scope}[shift={(4.0,2)},scale=0.8,transform shape]
  				\node at (0,0) {\includegraphics[height=4.3cm]{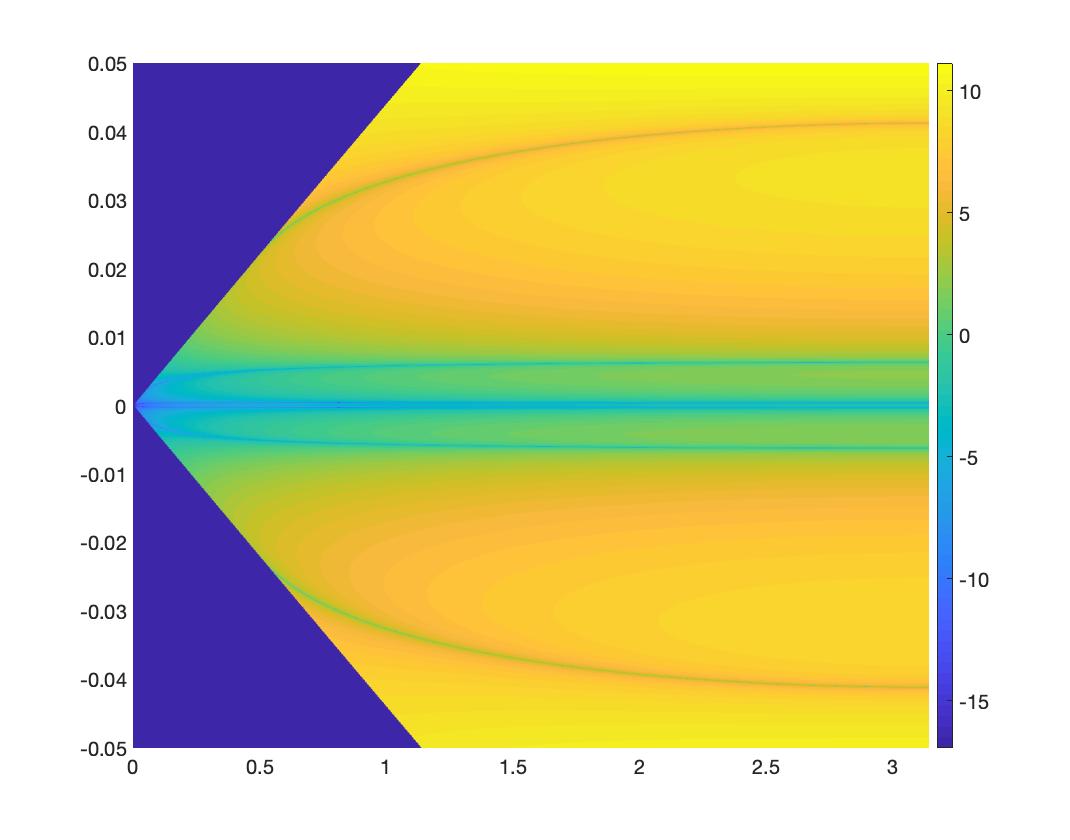}};
  				\draw (0.,2.1) node{$a_{11}=19,\;E_\text{lim}=0.05$};
 				%\draw (0,-2.5) node{$k$};
  			\end{scope}
		     			\begin{scope}[shift={(-4.0,-1.5)},scale=0.8,transform shape]
		     				\node at (0,0) {\includegraphics[height=4.3cm]{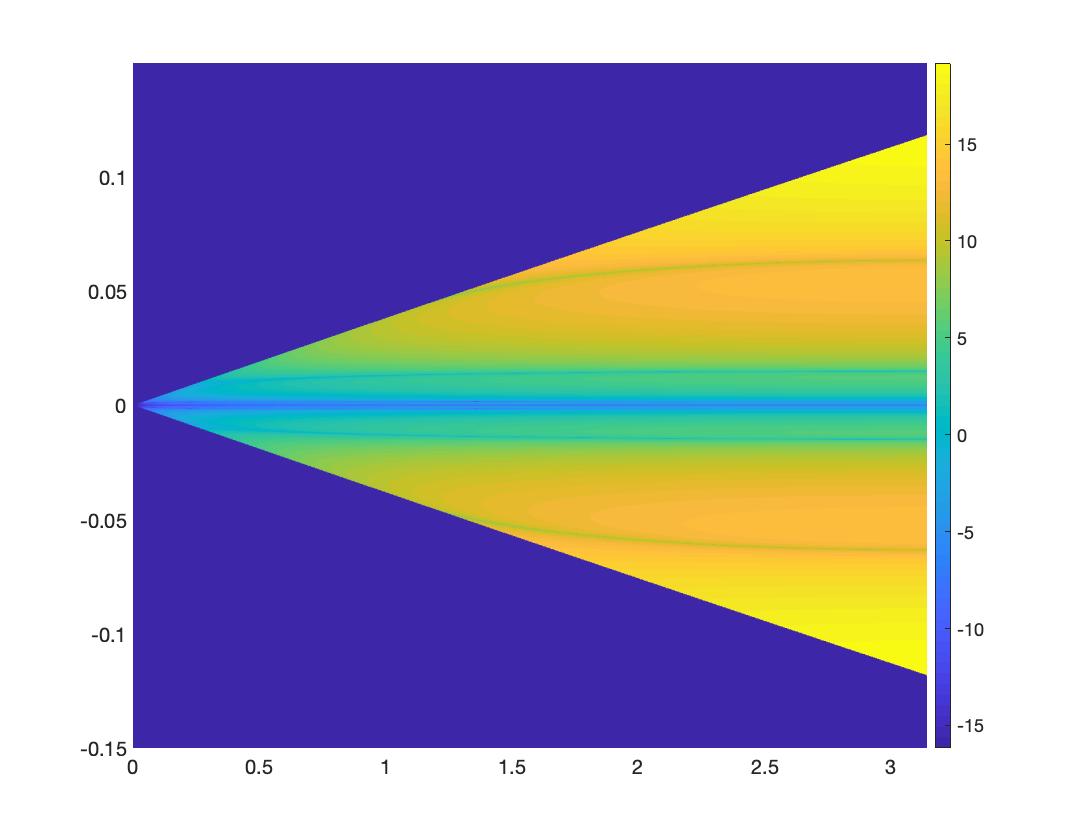}};
		     				\draw (0.,2.1) node{$a_{11}=22,\;E_\text{lim}=0.15$};
		    				%\draw (0,-2.5) node{$k$};
                \draw (-2.5,0) node{$E$};
		     			\end{scope}
		     			\begin{scope}[shift={(0.0,-1.5)},scale=0.8,transform shape]
		     				\node at (0,0) {\includegraphics[height=4.3cm]{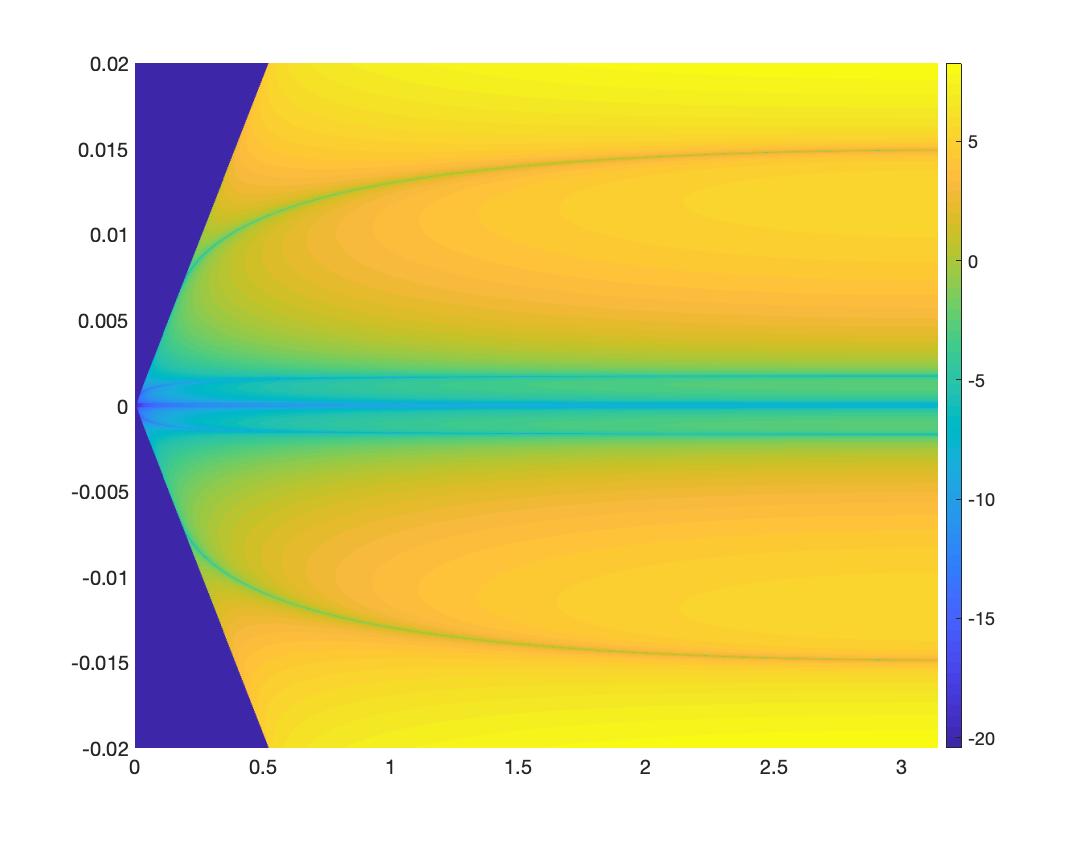}};
		     				\draw (0.,2.1) node{$a_{11}=22,\;E_\text{lim}=0.02$};
		    				%\draw (0,-2.5) node{$k$};
		     			\end{scope}
		     			\begin{scope}[shift={(4.0,-1.5)},scale=0.8,transform shape]
		     				\node at (0,0) {\includegraphics[height=4.3cm]{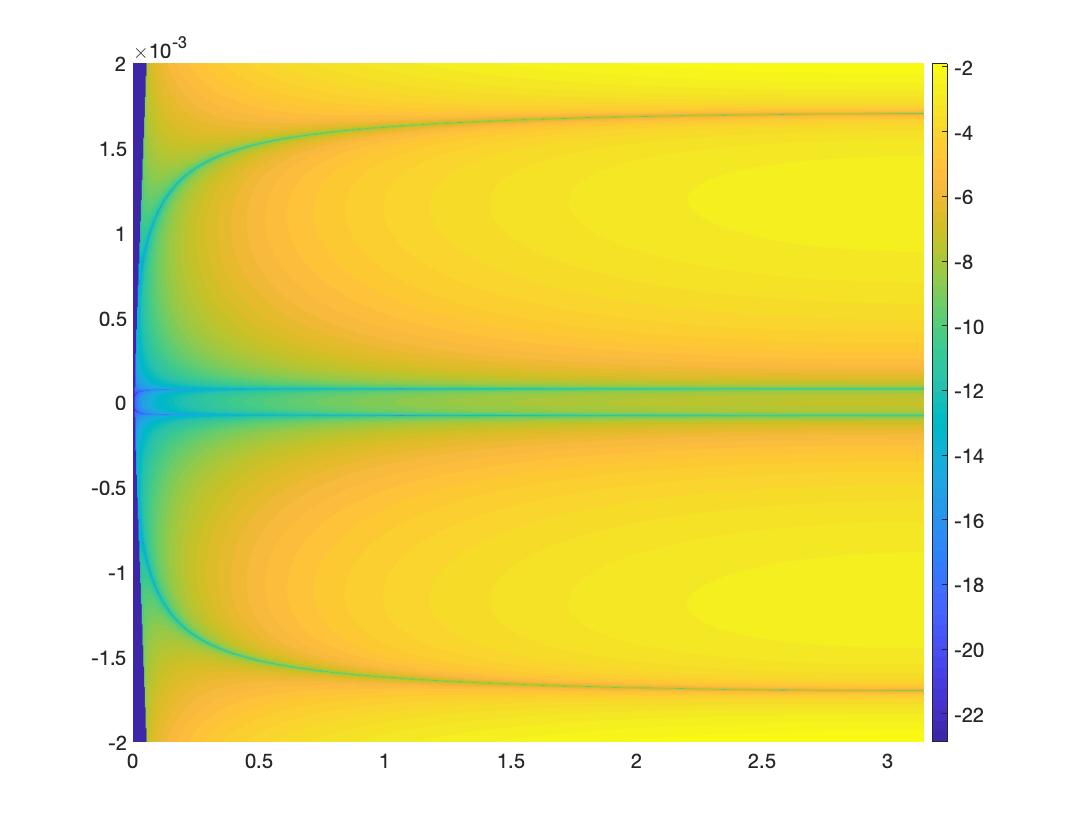}};
		     				\draw (0.,2.1) node{$a_{11}=22,\;E_\text{lim}=0.002$};
		    				%\draw (0,-2.5) node{$k$};
		     			\end{scope}
						
					     			\begin{scope}[shift={(-4.0,-5)},scale=0.8,transform shape]
					     				\node at (0,0) {\includegraphics[height=4.3cm]{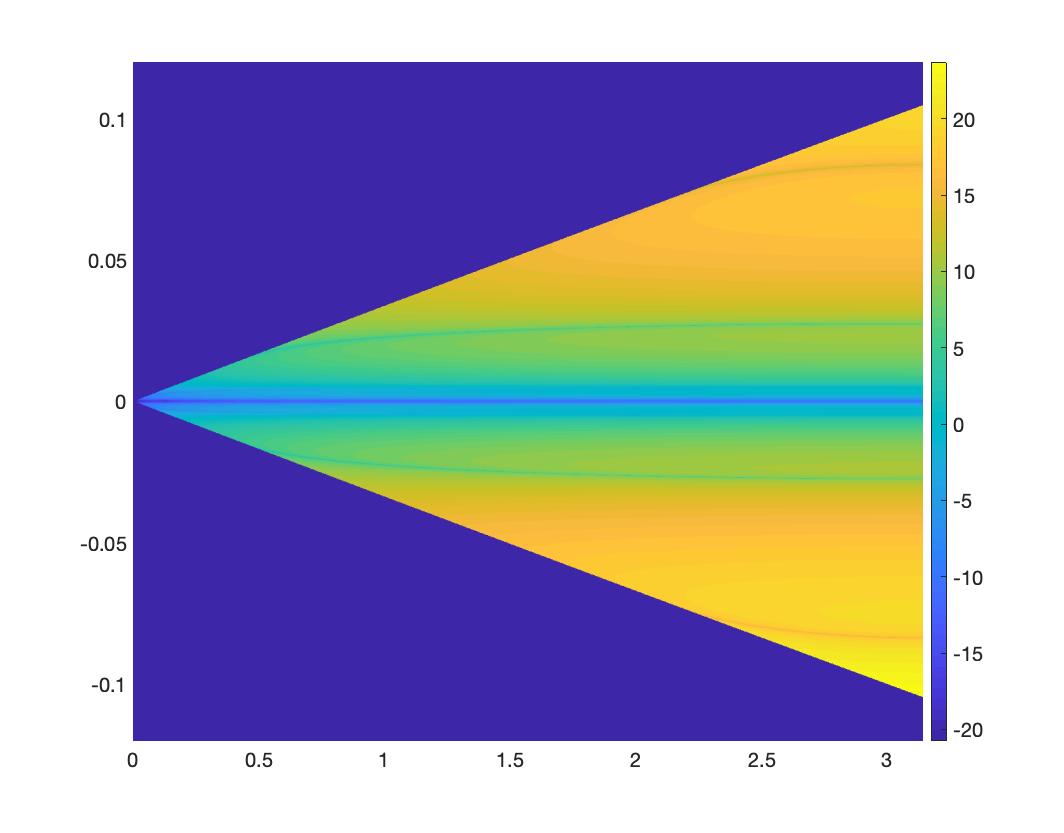}};
					     				\draw (0.,2.1) node{$a_{11}=25,\;E_\text{lim}=0.2$};
					    				\draw (0,-2.1) node{$k$};
                      \draw (-2.5,0) node{$E$};
					     			\end{scope}
					     			\begin{scope}[shift={(0.0,-5)},scale=0.8,transform shape]
					     				\node at (0,0) {\includegraphics[height=4.3cm]{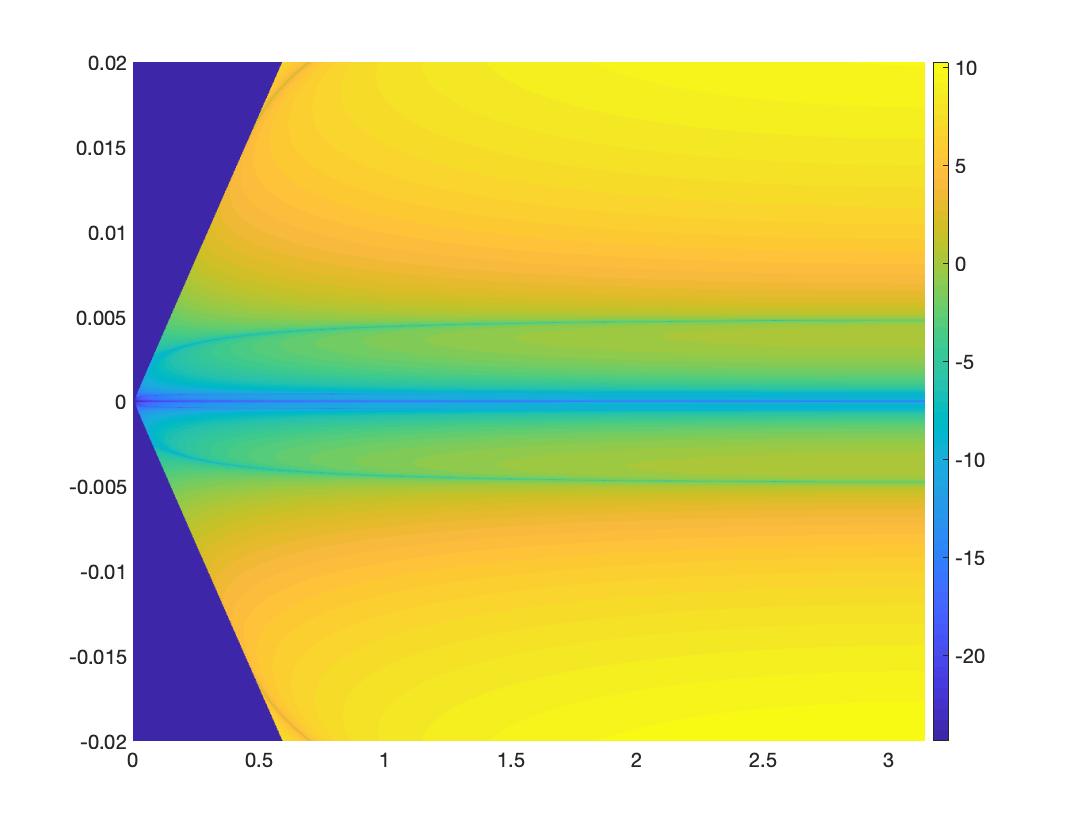}};
					     				\draw (0.,2.1) node{$a_{11}=25,\;E_\text{lim}=0.02$};
					    				\draw (0,-2.1) node{$k$};
					     			\end{scope}
					     			\begin{scope}[shift={(4.0,-5)},scale=0.8,transform shape]
					     				\node at (0,0) {\includegraphics[height=4.3cm]{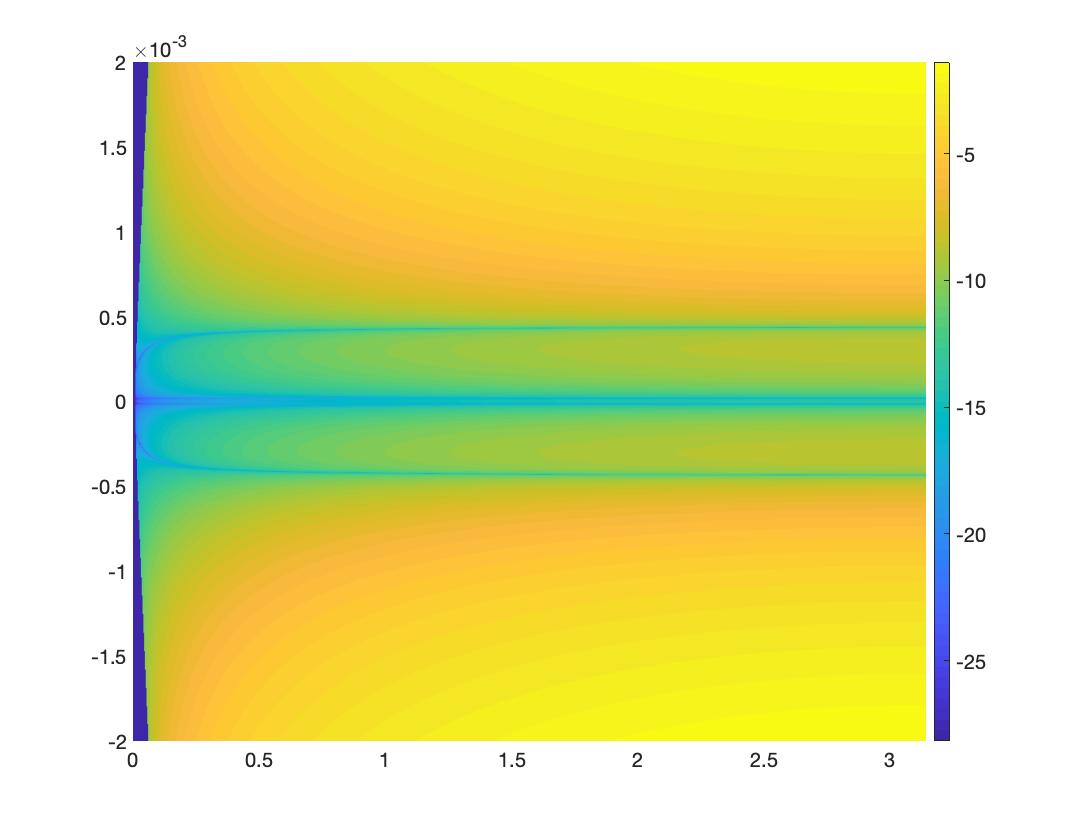}};
					     				\draw (0.,2.1) node{$a_{11}=25,\;E_\text{lim}=0.001$};
					    				\draw (0,-2.1) node{$k$};
					     			\end{scope}
  		\end{tikzpicture}
  		\caption{Plots of $(\kpar,E)\mapsto\log |\Delta(E,\kpar)|$ for $k\in(0,\pi)$ ($N_k=1000$ points) and $E\in(-E_\text{lim},E_\text{lim})$ ($N_E=1000$ points) for various armchair-like edges with $a_{12}=1$. }
  		\label{fig:edge_states_armchair}
  	\end{center}
  \end{figure}
   \begin{figure}[htbp]
   	\begin{center}
   		\begin{tikzpicture}
			
   			\begin{scope}[shift={(-4.0,6)},scale=0.8,transform shape]	
   				\node at (0,0) {\includegraphics[height=4.3cm]{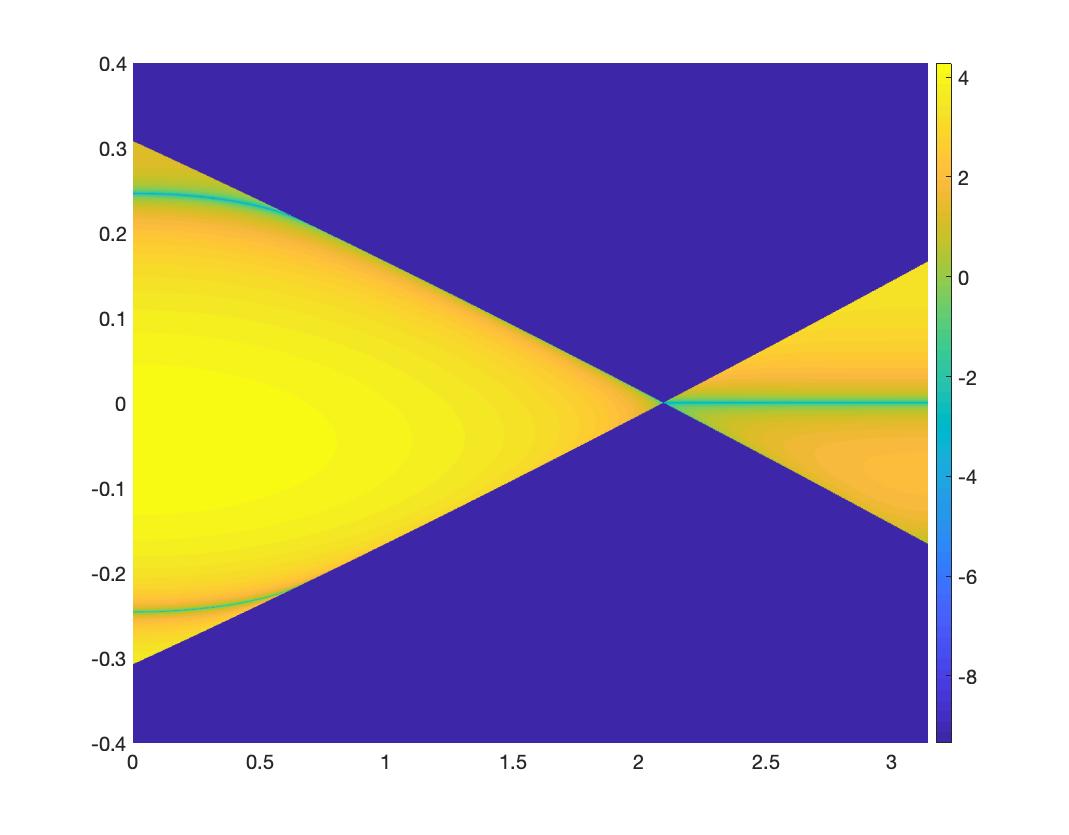}};
   				\draw (0.,2.1) node{$a_{11}=5,\;E_\text{lim}=0.4$};
  				\draw (-2.5,0) node{$E$};
  				%\draw (0,-2.5) node{$k$};
   			\end{scope}
			
   			\begin{scope}[shift={(0,6)},scale=0.8,transform shape]	
   				\node at (0,0) {\includegraphics[height=4.3cm]{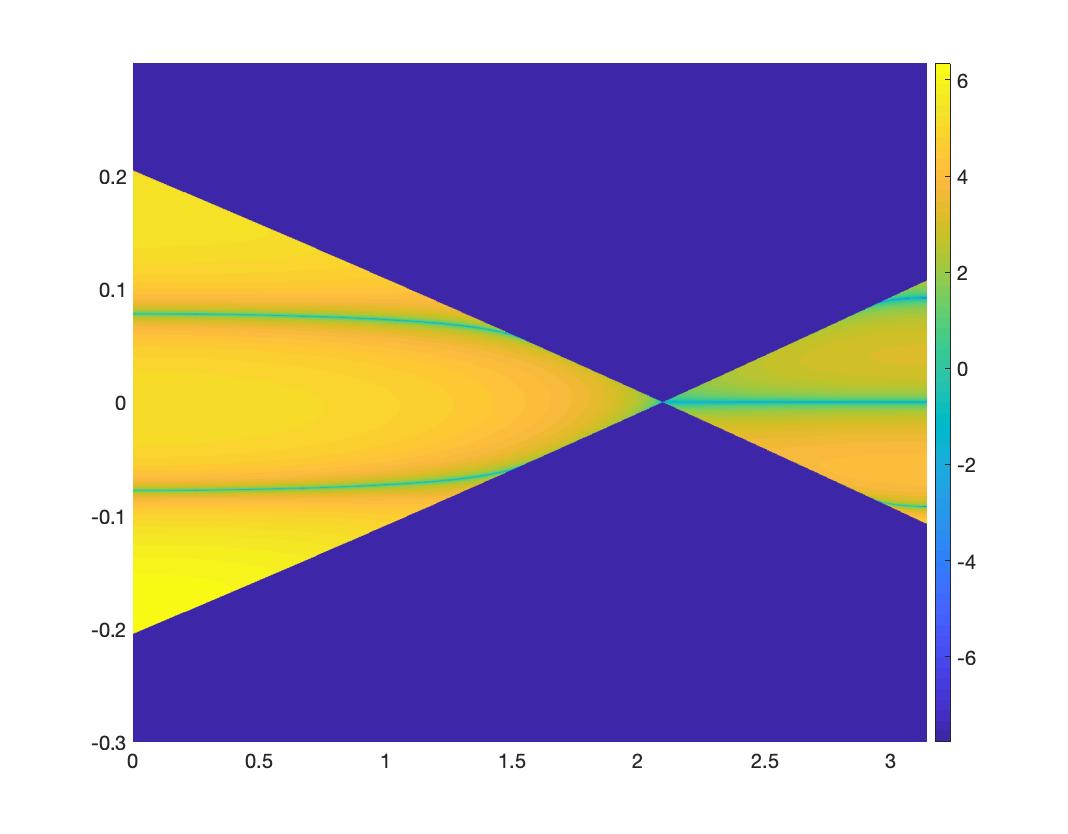}};
   				\draw (0.,2.1) node{$a_{11}=8,\;E_\text{lim}=0.3$};
  				%\draw (0,-2.5) node{$k$};
   			\end{scope}
			
   			\begin{scope}[shift={(4.0,6)},scale=0.8,transform shape]	
   				\node at (0,0) {\includegraphics[height=4.3cm]{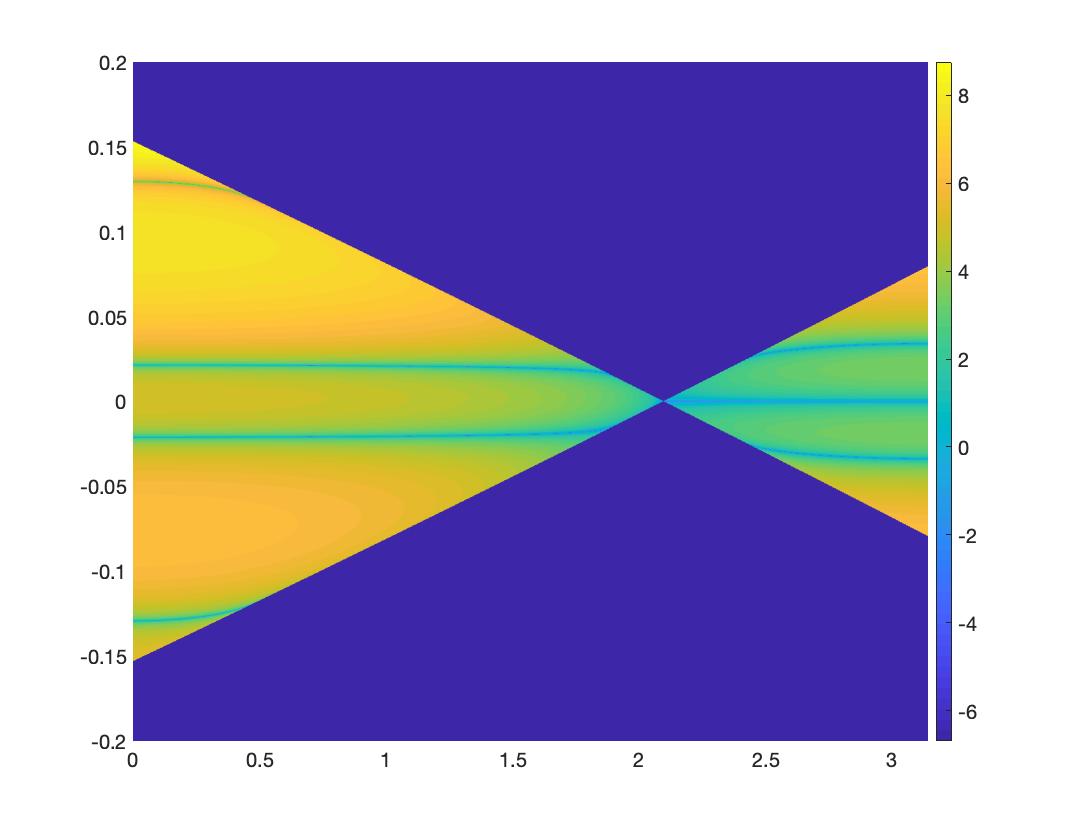}};
   				\draw (0.,2.1) node{$a_{11}=11,\;E_\text{lim}=0.2$};
  				%\draw (0,-2.5) node{$k$};
   			\end{scope}
			
   			\begin{scope}[shift={(-4.0,2.5)},scale=0.8,transform shape]
   				\node at (0,0) {\includegraphics[height=4.3cm]{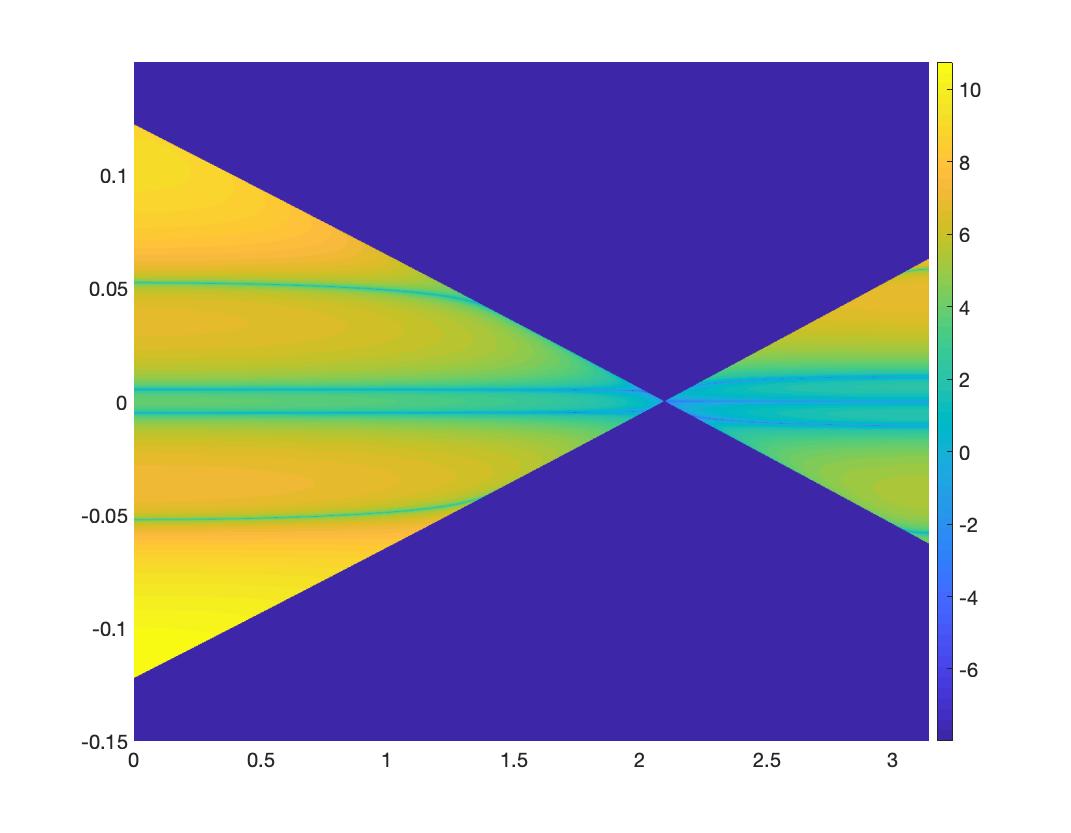}};
   				\draw (0.,2.1) node{$a_{11}=14,\;E_\text{lim}=0.2$};
           \draw (-2.5,0) node{$E$};
  				%\draw (0,-2.5) node{$k$};
   			\end{scope}
   			\begin{scope}[shift={(0,2.5)},scale=0.8,transform shape]
   				\node at (0,0) {\includegraphics[height=4.3cm]{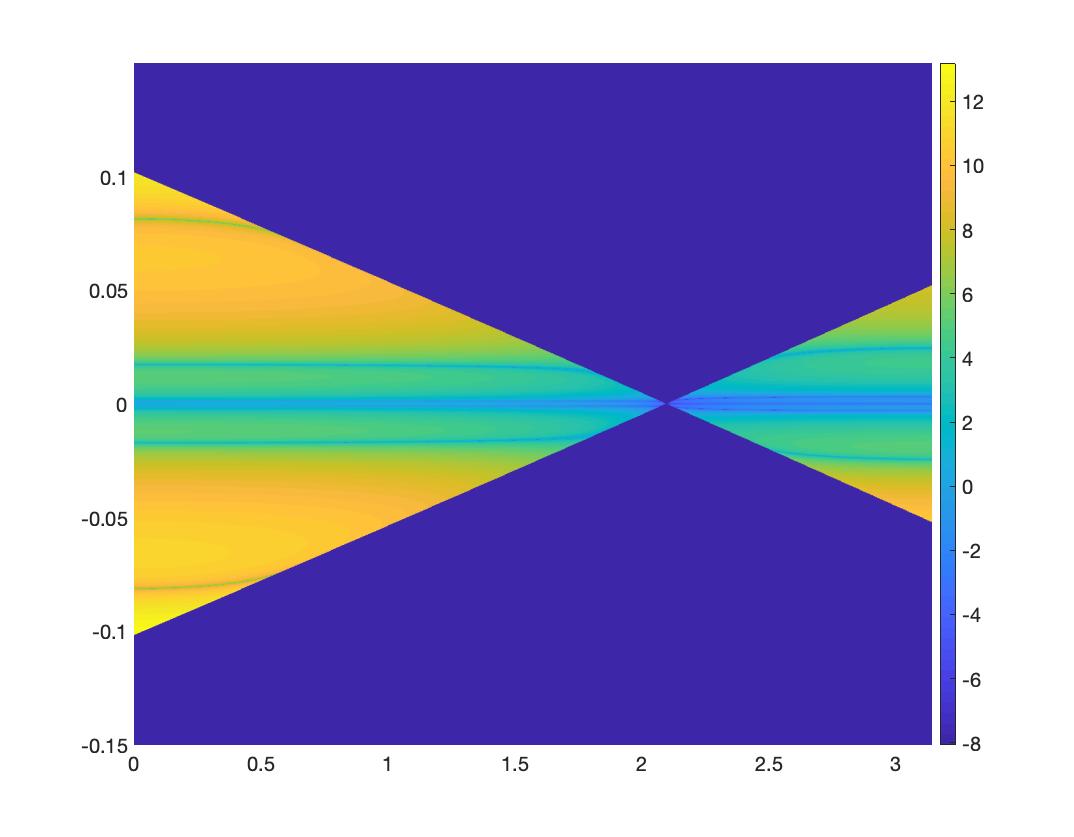}};
   				\draw (0.,2.1) node{$a_{11}=17,\;E_\text{lim}=0.15$};
  				%\draw (0,-2.5) node{$k$};
   			\end{scope}
   			\begin{scope}[shift={(4.0,2.5)},scale=0.8,transform shape]
   				\node at (0,0) {\includegraphics[height=4.3cm]{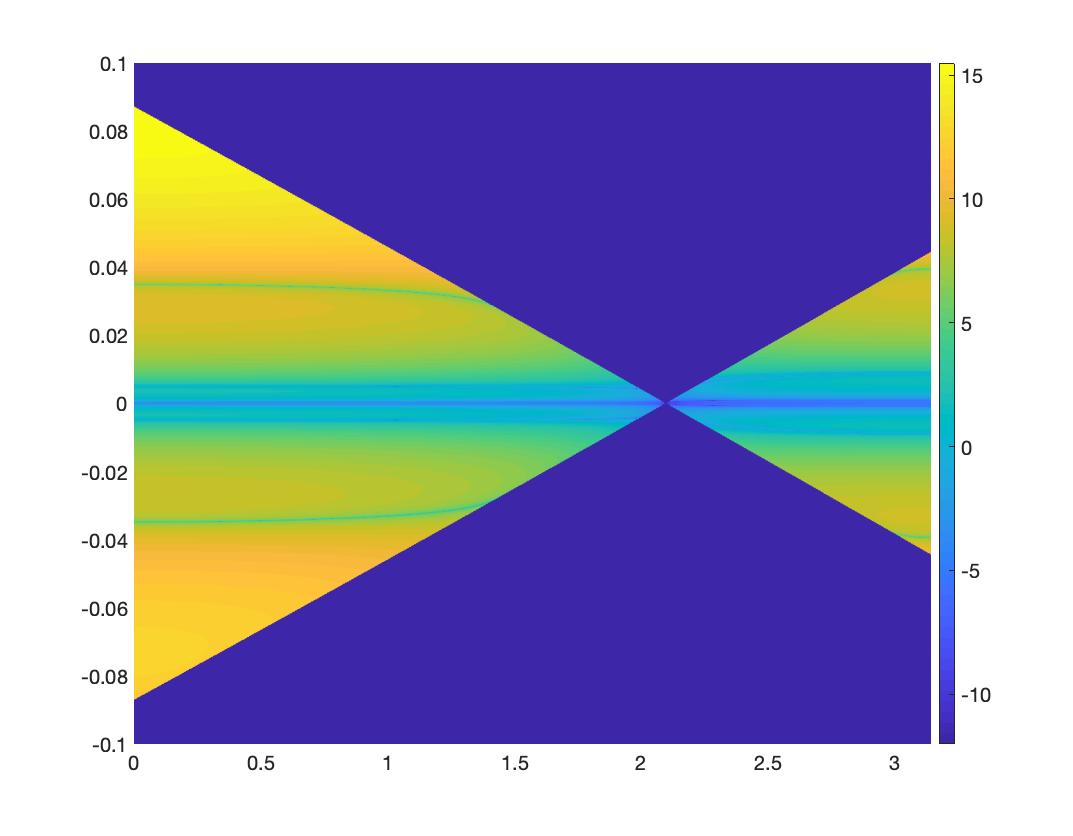}};
   				\draw (0.,2.1) node{$a_{11}=20,\;E_\text{lim}=0.1$};
  				%\draw (0,-2.5) node{$k$};
   			\end{scope}

 		     			\begin{scope}[shift={(0,-1)},scale=0.8,transform shape]
 		     				\node at (0,0) {\includegraphics[height=4.3cm]{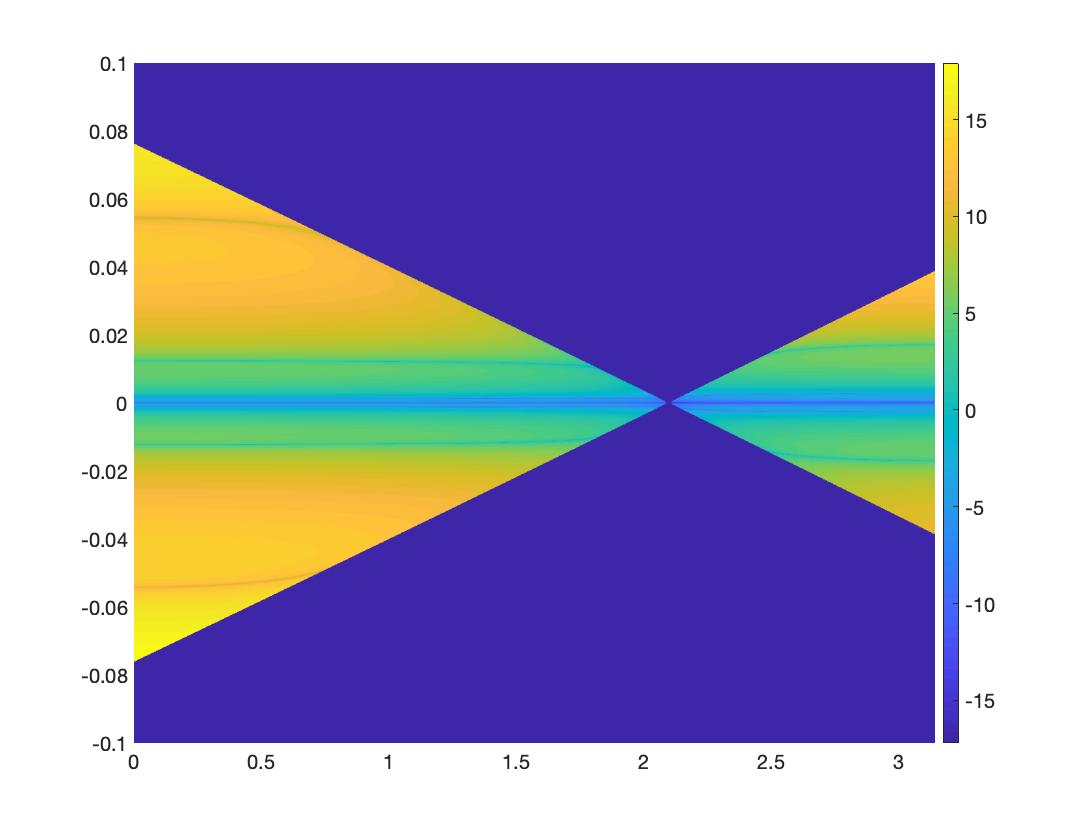}};
 		     				\draw (0.,2.1) node{$a_{11}=23,\;E_\text{lim}=0.1$};
 		    				\draw (0,-2.1) node{$k$};
 		     			\end{scope}

                \begin{scope}[shift={(-4.0,-1)},scale=0.8,transform shape]
 		     				\node at (0,0) {\includegraphics[height=4.3cm]{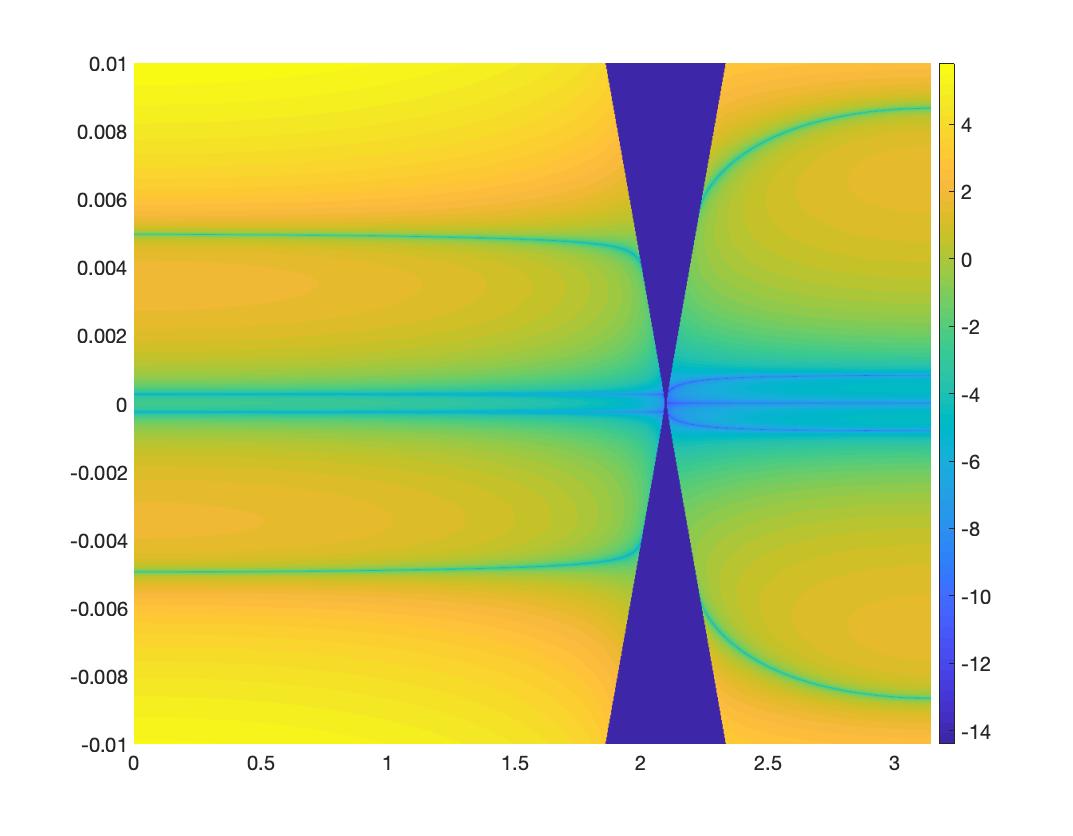}};
 		     				\draw (0.,2.1) node{$a_{11}=20,\;a_{12}^{(6)}=1,\;E_\text{lim}=0.01$};
 		    				\draw (0,-2.1) node{$k$};
                 \draw (-2.5,0) node{$E$};
 		     			\end{scope}

 		     			\begin{scope}[shift={(4.0,-1)},scale=0.8,transform shape]
 		     				\node at (0,0) {\includegraphics[height=4.3cm]{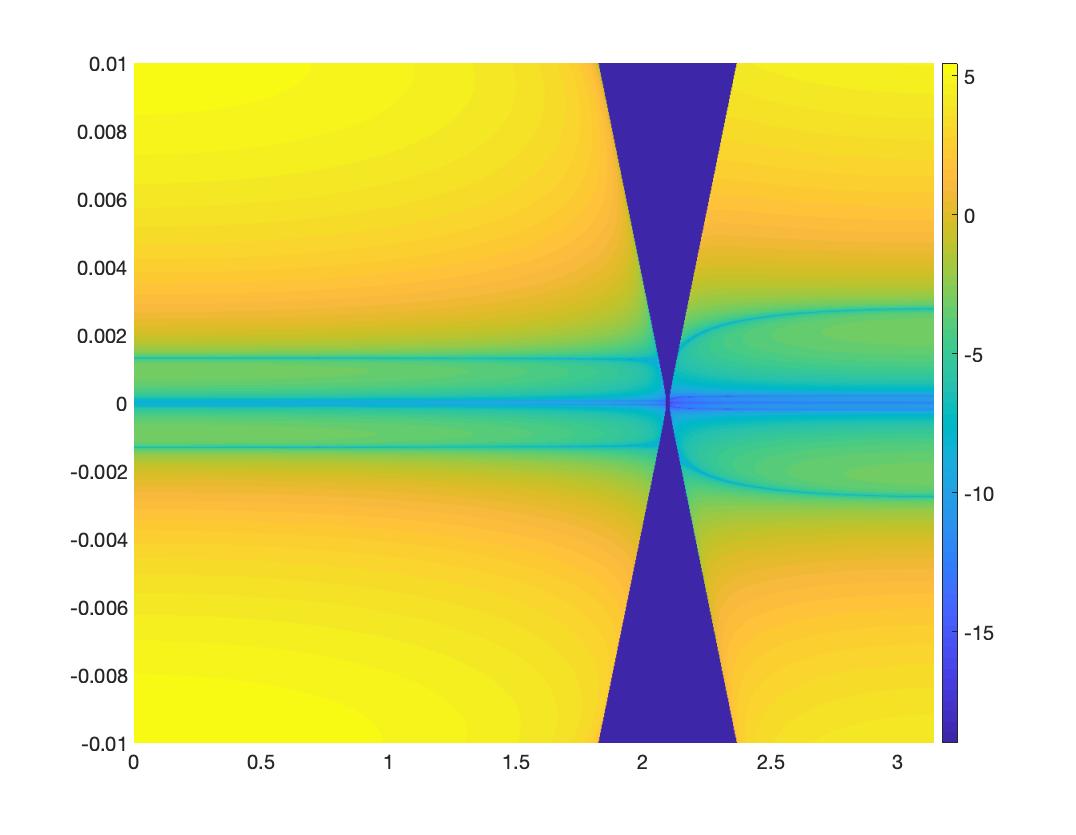}};
 		     				\draw (0.,2.1) node{$a_{11}=23,\;E_\text{lim}=0.01$};
 		    				\draw (0,-2.1) node{$k$};
 		     			\end{scope}
   		\end{tikzpicture}
   		\caption{Plots of $(\kpar,E)\mapsto\log |\Delta(E,\kpar)|$ for $k\in(0,\pi)$ ($N_k=600$ points) and $E\in(-E_\text{lim},E_\text{lim})$ ($N_E=600$ points) for various ordinary zigzag-like edges with $a_{12}=1$. }
   		\label{fig:edge_states_zigzag}
   	\end{center}
   \end{figure}
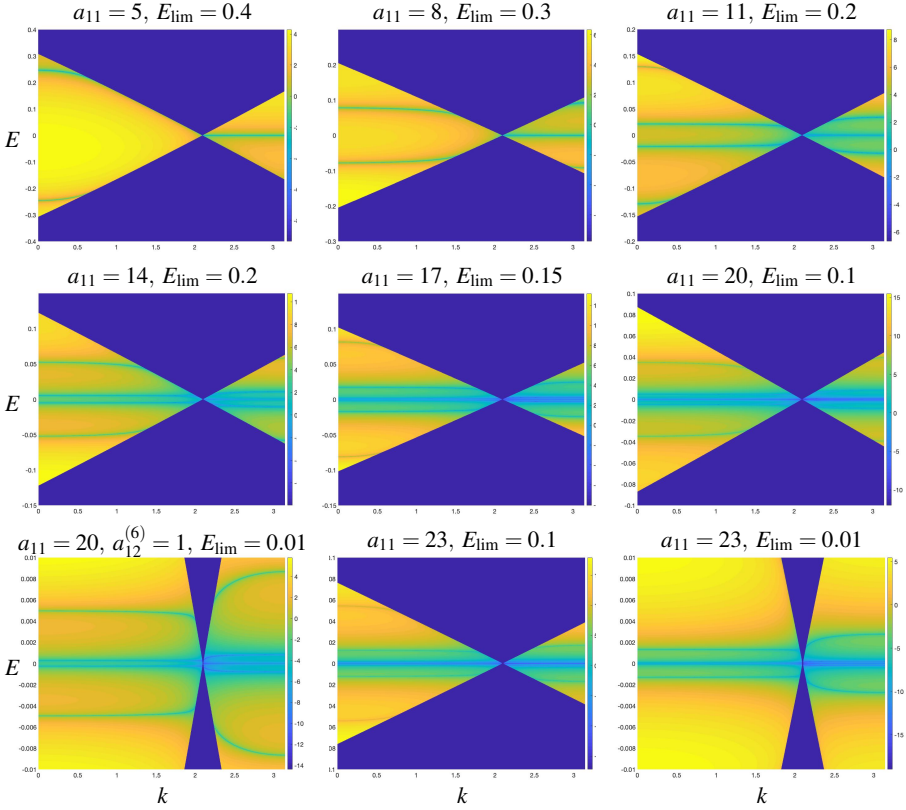
	 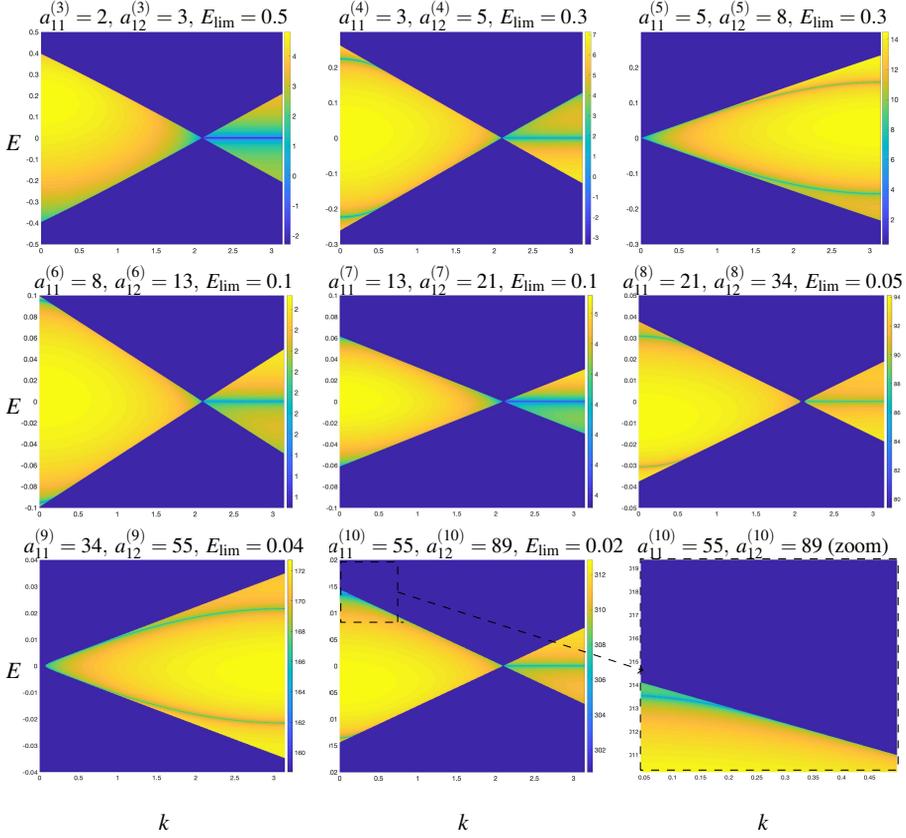
\begin{figure}[htbp]
	 	\begin{center}
	 		\begin{tikzpicture}
			
	 			\begin{scope}[shift={(-4.0,6)},scale=0.8,transform shape]	
	 				\node at (0,0) {\includegraphics[height=4.3cm]{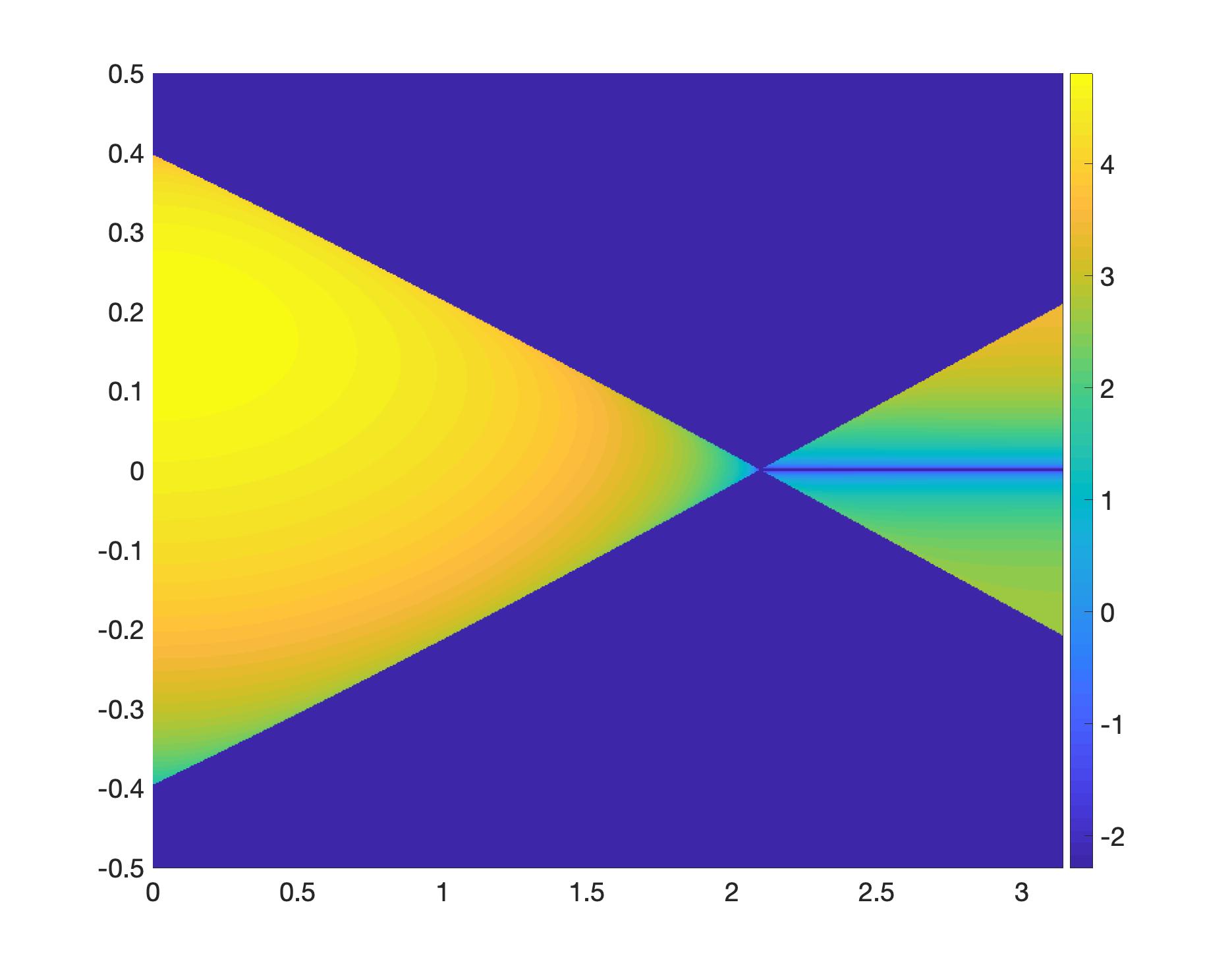}};
	 				\draw (0.,2.1) node{\small $a_{11}^{(3)}=2,\;a_{12}^{(3)}=3,\;E_\text{lim}=0.5$};
					\draw (-2.5,0) node{$E$};
					%\draw (0,-2.5) node{$k$};
	 			\end{scope}
			
	 			\begin{scope}[shift={(0,6)},scale=0.8,transform shape]	
	 				\node at (0,0) {\includegraphics[height=4.3cm]{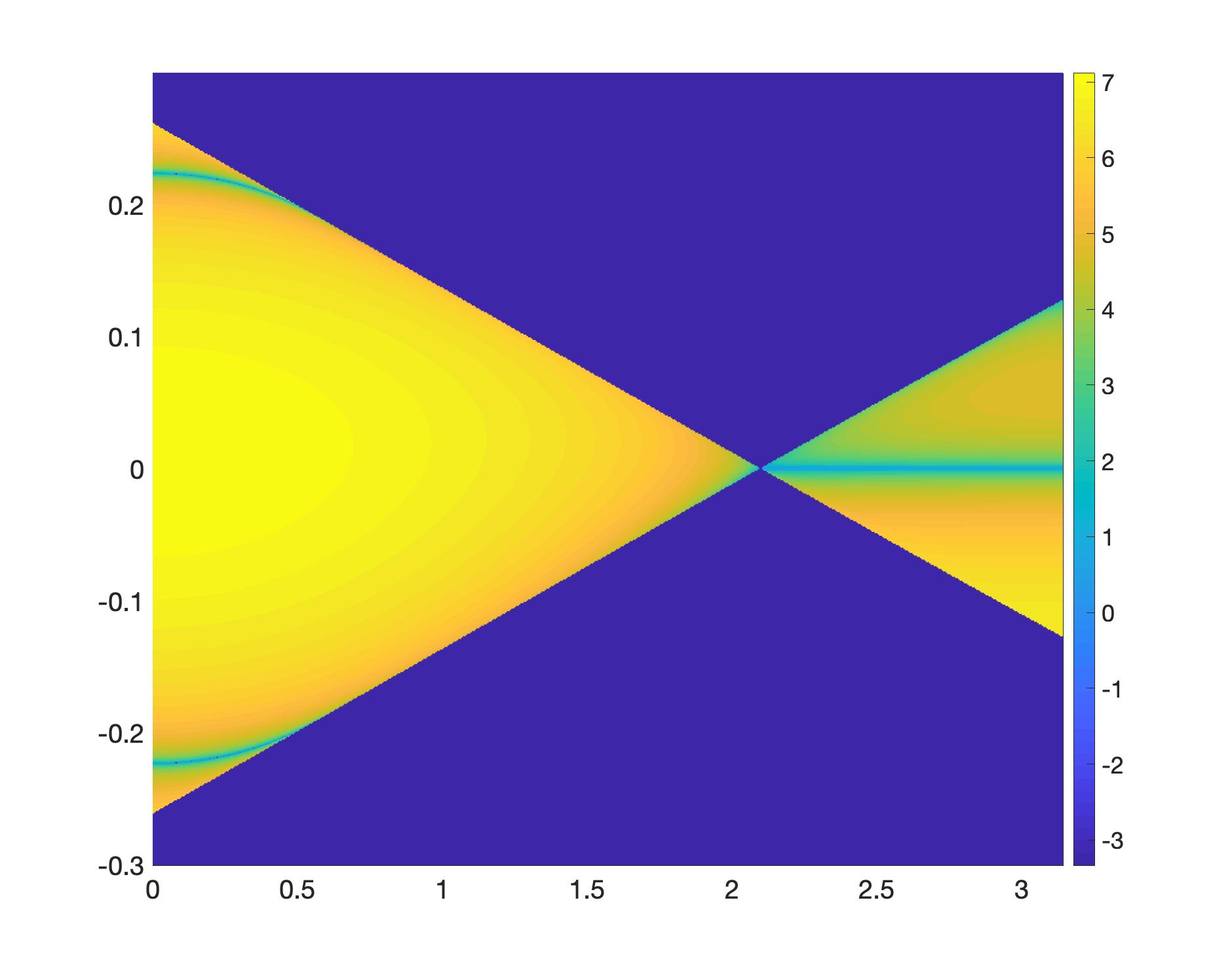}};
	 				\draw (0.,2.1) node{\small $a_{11}^{(4)}=3,\;a_{12}^{(4)}=5,\;E_\text{lim}=0.3$};
					%\draw (0,-2.5) node{$k$};
	 			\end{scope}
			
	 			\begin{scope}[shift={(4.0,6)},scale=0.8,transform shape]	
	 				\node at (0,0) {\includegraphics[height=4.3cm]{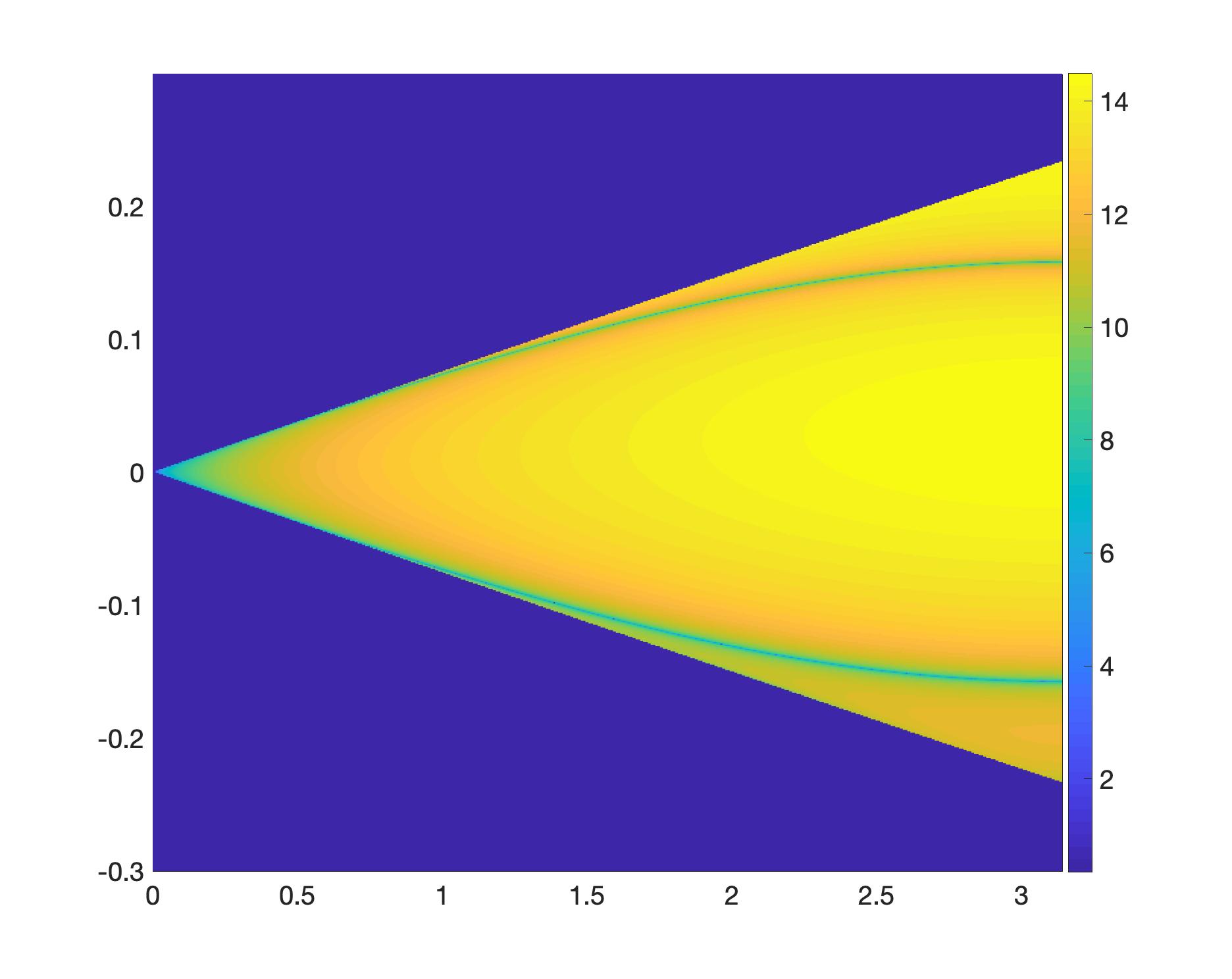}};
	 				\draw (0.,2.1) node{\small $a_{11}^{(5)}=5,\;a_{12}^{(5)}=8,\;E_\text{lim}=0.3$};
					%\draw (0,-2.5) node{$k$};
	 			\end{scope}
			
	 			\begin{scope}[shift={(-4.0,2.5)},scale=0.8,transform shape]
	 				\node at (0,0) {\includegraphics[height=4.3cm,width=5.5cm]{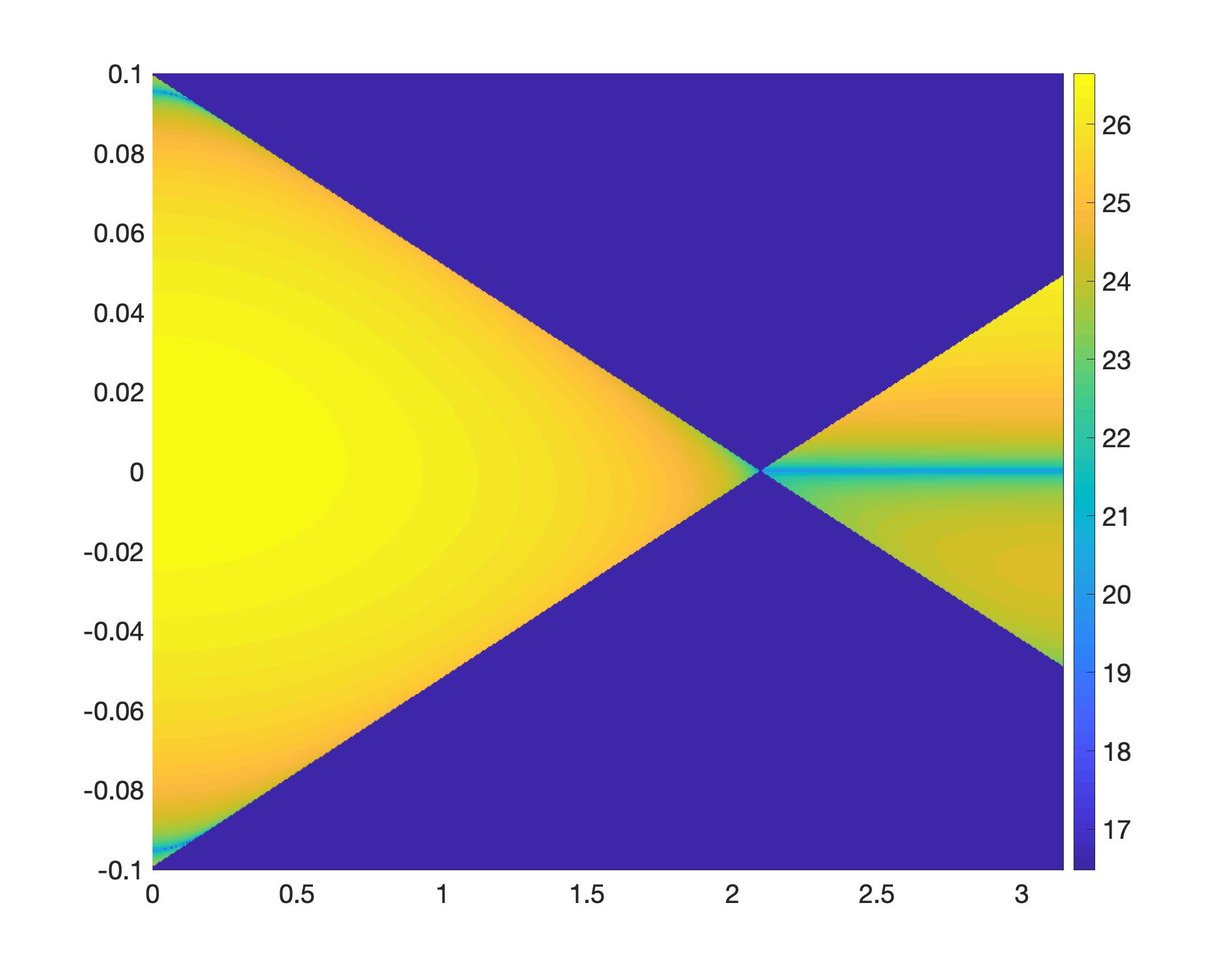}};
	 				\draw (0.,2.1) node{\small $a_{11}^{(6)}=8,\;a_{12}^{(6)}=13,\;E_\text{lim}=0.1$};
					\draw (-2.5,0) node{$E$};
					%\draw (0,-2.5) node{$k$};
	 			\end{scope}
	  			\begin{scope}[shift={(0.0,2.5)},scale=0.8,transform shape]
	  				\node at (0,0) {\includegraphics[height=4.3cm,width=5.5cm]{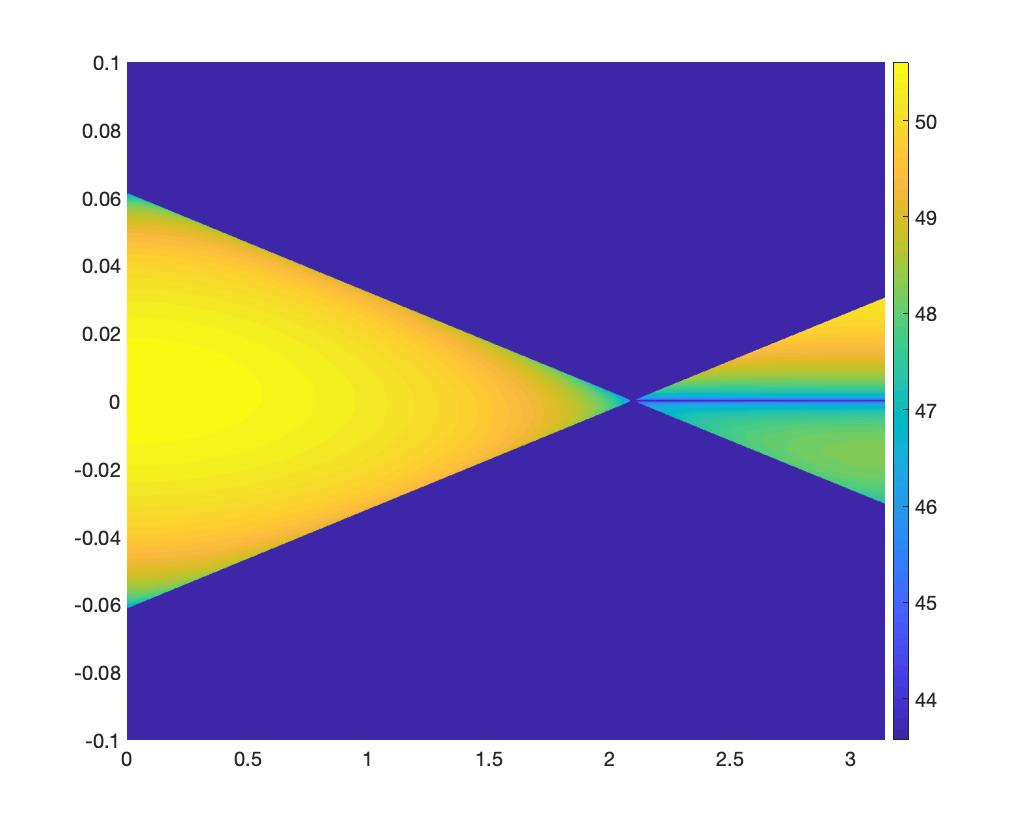}};
	  				\draw (0.,2.1) node{\small $a_{11}^{(7)}=13,\;a_{12}^{(7)}=21,\;E_\text{lim}=0.1$};
	 				%\draw (0,-2.5) node{$k$};
	  			\end{scope}
	  			\begin{scope}[shift={(4.0,2.5)},scale=0.8,transform shape]
	  				\node at (0,0) {\includegraphics[height=4.3cm,width=5.5cm]{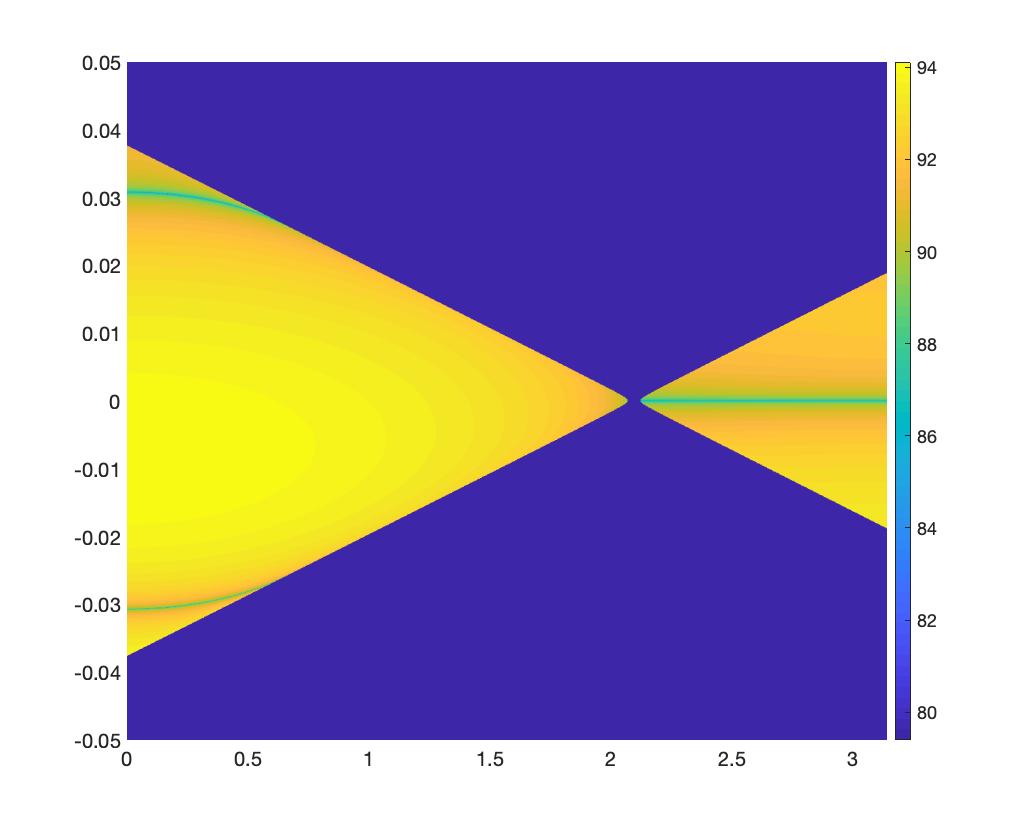}};
	  				\draw (0.,2.1) node{\small $a_{11}^{(8)}=21,\;a_{12}^{(8)}=34,\;E_\text{lim}=0.05$};
	 				%\draw (0,-2.5) node{$k$};
	  			\end{scope}
	  			\begin{scope}[shift={(4.0,-1)},scale=0.8,transform shape]
	  				\node at (0,0) {\includegraphics[height=4.3cm,width=5.5cm]{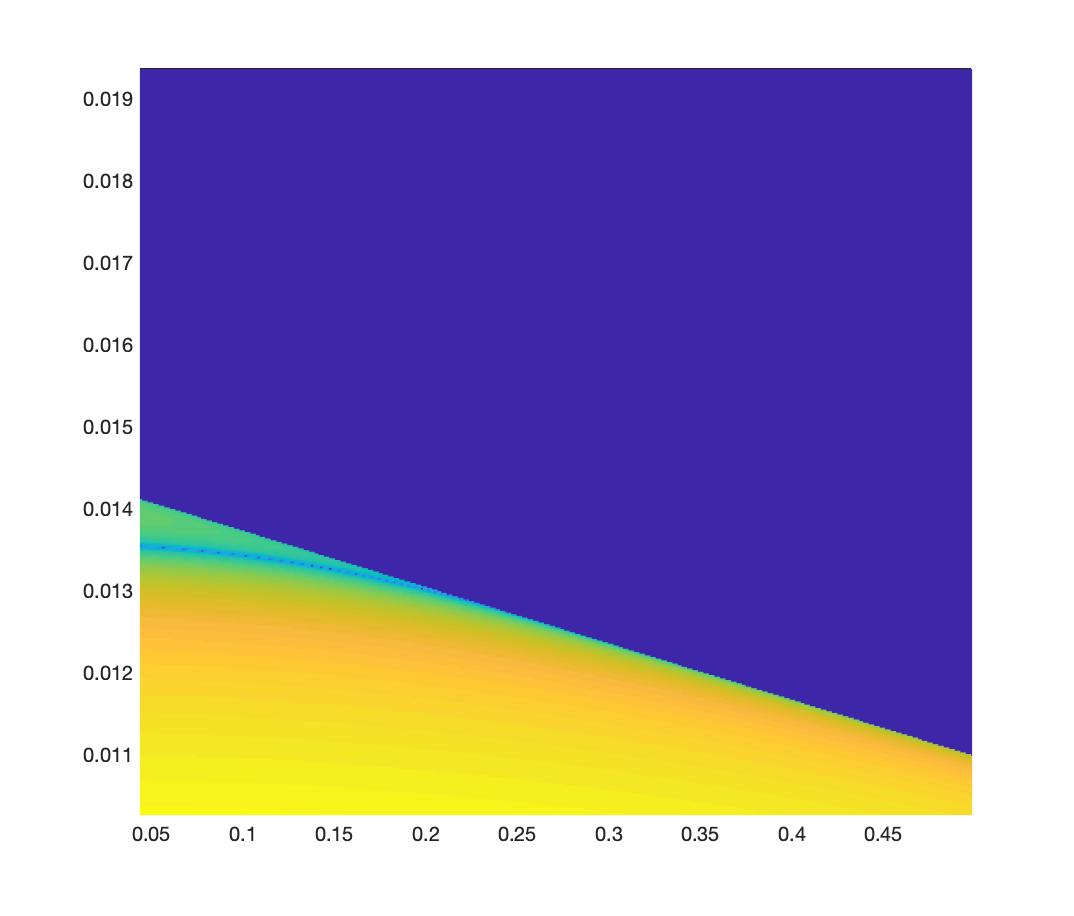}};
	  				\draw (0.,2.1) node{\small  $a_{11}^{(10)}=55,\;a_{12}^{(10)}=89$ (zoom)};
	 				\draw (0,-2.5) node{$k$};
					\draw[dashed] (-2.05,-1.65)--(-2.05,1.85)--(2.25,1.85)--(2.25,-1.65)--(-2.05,-1.65);
	  			\end{scope}
	  			\begin{scope}[shift={(0.0,-1)},scale=0.8,transform shape]
	  				\node at (0,0) {\includegraphics[height=4.3cm,width=5.5cm]{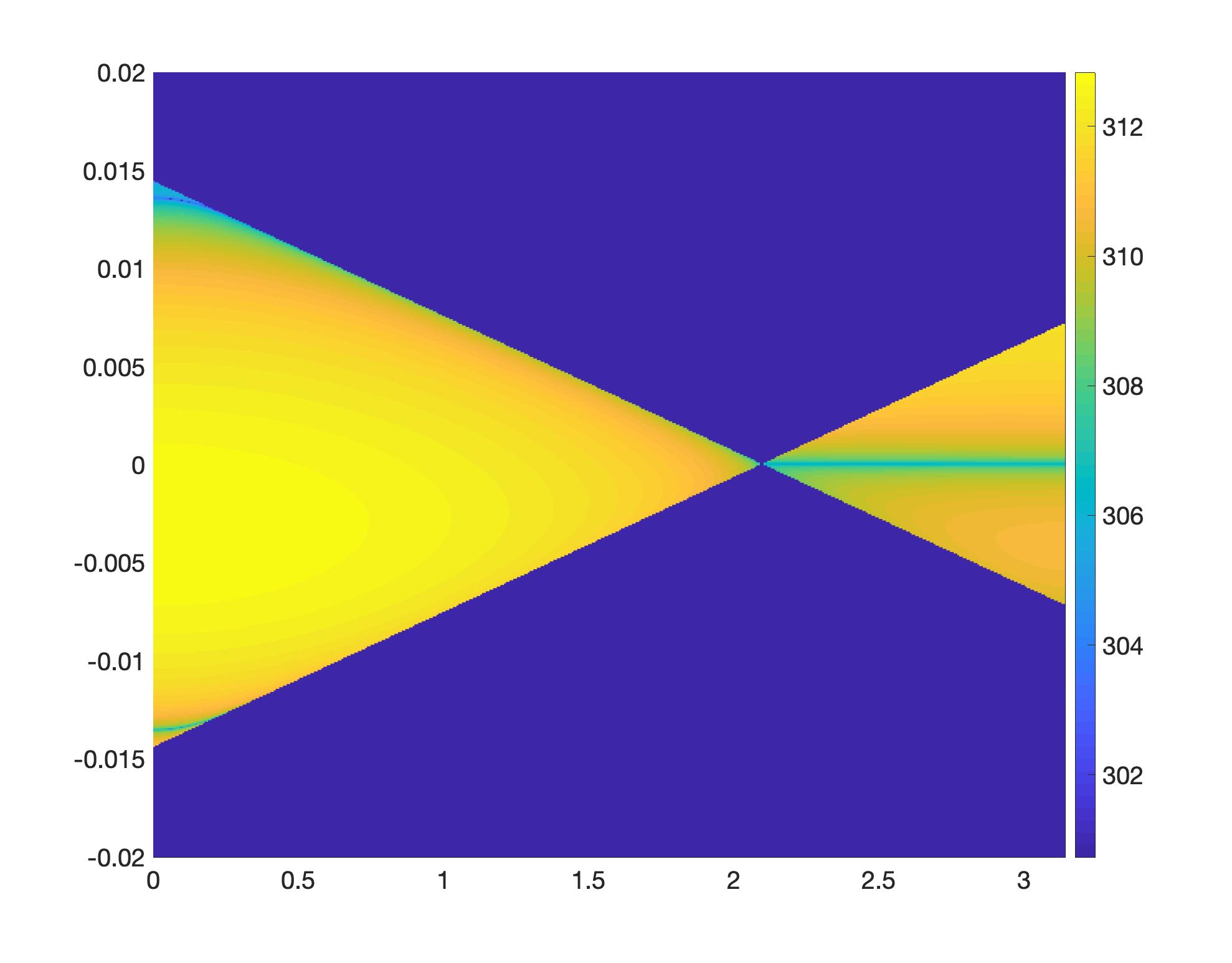}};
	  				\draw (0.2,2.1) node{ \small  $a_{11}^{(10)}=55,\;a_{12}^{(10)}=89,\;E_\text{lim}=0.02$};
	 				\draw (0,-2.5) node{$k$};
					\draw[dashed] (-1.1,0.8)--(-1.1,1.8)--(-2.05,1.8)--(-2.05,0.8)--(-1,0.8);
					\draw[dashed,->] (-1.1,1.3)--(3,0);
	  			\end{scope}
          \begin{scope}[shift={(-4.0,-1)},scale=0.8,transform shape]
	  				\node at (0,0) {\includegraphics[height=4.3cm,width=5.5cm]{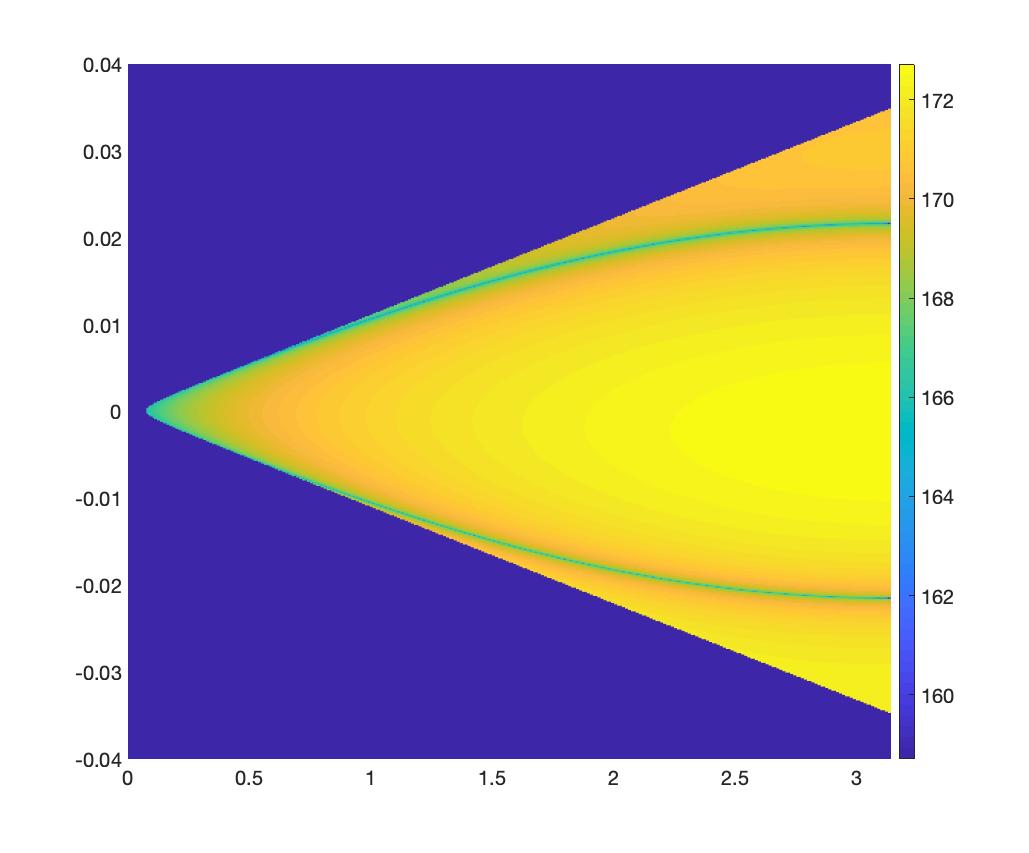}};
	  				\draw (0.,2.1) node{\small $a_{11}^{(9)}=34,\;a_{12}^{(9)}=55,\;E_\text{lim}=0.04$};
					\draw (-2.5,0) node{$E$};
	 				\draw (0,-2.5) node{$k$};
	  			\end{scope}
	 		\end{tikzpicture}
	 		\caption{Plots of $(\kpar,E)\mapsto\log |\Delta(\kpar,E)|$ for $k\in(0,\pi)$ ($N_k=1000$ points) and $E\in(-E_\text{lim},E_\text{lim})$ ($N_E=1000$ points) and with $\nA=\nB=0$ for the Fibonacci sequence. }
	 		\label{fig:edge_states_Fib}
	 	\end{center}
	 \end{figure}

{Thus, we have found strong numerical evidence for the existence edge states, whose energies lie on (non-flat) dispersion curves in the gap of the essential spectrum. In our simulations, all edge state curves (flat and non-flat) appear to emerge (bifurcate) from and terminate at band-crossings. Any edge state energies which are not in the gap of the essential spectrum must be embedded in the essential spectrum.
We have not addressed the general question of the existence of edge state energies or curves of edge state energies embedded in the essential spectrum; see Remark \ref{ou_ptspec}.}

%%      ---------------------------------------------------------------------
%%      ------------------------- APPENDIX (OPTIONAL) -----------------------
%%      ---------------------------------------------------------------------
        
%%      If you have one appendix, uncomment the line \appendix and add
%%      a \section{ *** APPENDIX TITLE ***}. If you have more than
%%      one, uncomment the line \appendices and add a \section{ ***
%%      APPENDIX TITLE ***} command for each appendix title.

\appendix
%\appendices
\section{Changing basis} \label{sec:vs_circ_to_vs}

In this section we prove \eqref{vB-1}, \eqref{ttn-ttm}, that the $\tilde{n}_\nu$'s defined in \eqref{ttn-ttm} are all distinct except in the classical zigzag case and finally the inequality \eqref{eq:n_nu_ineq}. Note that thanks to \eqref{e-nu} and \eqref{eq:change_basis}, in terms of the $\{\bv_1,\bv_2\}$ basis, the $\be^\nu$ are given by 
 \begin{equation}
 \begin{aligned}
e^1 &= \frac13\left(\ocirc{\bv}_1+\ocirc{\bv}_2\right) = 
\frac{1}{3}(a_{22}-a_{21}) \bv_1+ \frac{1}{3}(a_{11}-a_{12}) \bv_2\\
 e^2 &= \frac13\left(\ocirc{\bv}_1-2\ocirc{\bv}_2\right)= 
 \frac{1}{3}(a_{22}+2a_{21}) \bv_1- \frac{1}{3}(a_{12}+2a_{11}) \bv_2\\
e^3 &= \frac13\left(-2\ocirc{\bv}_1+\ocirc{\bv}_2\right) = 
  -\frac{1}{3}(a_{21}+2a_{22}) \bv_1 + \frac{1}{3}(a_{11}+2a_{12}) \bv_2
\end{aligned}
\label{A4}
\end{equation}

\subsection{ Points $\bv_A$ and  $\bv_B$ that lie in the fundamental cell}\label{vAvB-app}
The fundamental cell $\Gamma(0,0)$,  defined in \eqref{Gamma-mn}, contains the $A-$point given by \eqref{vA-1}. We look now for the $B$-point $\bv_B$ which lies in $\Gamma(0,0)$, i.e. we look for integers $q_1$, $q_2$ such that
 $\bv_B=q_1\bv_1+q_2\bv_2 + \ocirc{\bv}_B\in\Gamma(0,0)$, where
 $\bv_B^\circ = \be^1$. By \eqref{A4}, we have  $\bv_B\in\Gamma(0,0)$ if and only if
\begin{equation*}
q_1+\frac{1}{3}(a_{22}-a_{21}),\quad q_2+\frac{1}{3}(a_{11}-a_{12}) \in \left(-\frac12,+\frac12\right]
\label{q0-app}
\end{equation*}
By \eqref{k1k2s1s2}, this holds if and only if 
$ q_1=-k_1, \quad q_2=-k_2$
so that $\bv_B$ are given by \eqref{vB-1}.

\subsection{ $A-$points and $B-$ points  in the cell $\Gamma(m,n)$}\label{vAvB-mn-app}
 The $A-$points and $B-$points of our honeycomb are given, respectively, by $\bv_A+m\bv_1+n\bv_2$ and $\bv_B+m\bv_1+n\bv_2$,
 where $(m,n)\in\Z^2$. Since $\bv_A, \bv_B\in\Gamma(0,0)$, it follows that 
 \begin{equation}
 \begin{aligned}
 &\bv_A+m\bv_1+n\bv_2,\ \bv_B+m\bv_1+n\bv_2 \in \Gamma(m,n)=\Gamma(0,0) +m\bv_1+n\bv_2 \\
  &=\Big\{ x_1\bv_1+x_2\bv_2: x_1\in (m-\frac12,m+\frac12],\ x_2 \in (m-\frac12,m+\frac12] \Big\}
 \end{aligned}
 \label{Gam_mn_pts}
 \end{equation}
 Since $\Gamma(m,n)$ partition  $\R^2$, it follows that for any given $(m,n)\in\Z^2$, the points in \eqref{Gam_mn_pts}
  are the only points of the honeycomb that lie in $\Gamma(m,n)$.
  
  Next we represent the vectors $\be^\nu+(\bv_A-\bv_B)$ for $\nu=1,2,3$ with respect to the basis
   $\{\bv_1,\bv_2\}$. Using \eqref{A4}, \eqref{vA-1}, \eqref{vB-1} and \eqref{k1k2s1s2}, we deduce easily \eqref{tntm} where the $\ttm_\nu$'s and $\ttn_\nu$'s are given by \eqref{ttn-ttm}. 
Note, thanks to \eqref{k1k2s1s2} that
\begin{equation*}
\ttm_1+\ttm_2+\ttm_3=3k_1+a_{21}-a_{22} = - s_1\in\{-1,0,+1\}
\label{A13}
\end{equation*}
and 
\begin{equation}
\ttn_1+\ttn_2+\ttn_3=3k_2-a_{11}+a_{12} = - s_2\in\{-1,0,+1\}.
\label{A14}
\end{equation}

\subsection{Except in the zigzag case, the integers $\ttn_1$, $\ttn_2$, $\ttn_3$ are distinct}\label{distinct-n}

The integers $\ttn_1$, $\ttn_2$, $\ttn_3$ are distinct if and only if
\begin{equation}
a_{11}\ne0,\quad a_{12}\ne0,\quad a_{11}\ne-a_{12}
\label{A15}\end{equation}
We enumerate the cases where \eqref{A15} fails to hold. Recall that $a_{11}a_{22}-a_{12}a_{21}=1$.
Therefore, 
\begin{align*}
&\textrm{$a_{11}=0$ implies $-a_{12}a_{21}=1$. Hence, $a_{12}=\pm1$;}\\
&\textrm{$a_{12}=0$ implies $a_{11}a_{22}=1$. Hence, $a_{11}=\pm1$;}\\
&\textrm{$a_{11}=-a_{12}$ implies $a_{11}(a_{22}+a_{21})=1$. Hence,  $a_{11}=-a_{12}=1$ 
 or $a_{11}=-a_{12}=-1$.}
 \end{align*}
 Therefore, $\ttn_1, \ttn_2, \ttn_3$ are distinct except in the following cases:
 \begin{equation}
 \begin{aligned}
 (a_{11}, a_{12})=(0,1),\ (a_{11}, a_{12})=(0,-1)\\
 (a_{11}, a_{12})=(1,0),\ (a_{11}, a_{12})=(-1,0)\\
 (a_{11}, a_{12})=(1,-1),\ (a_{11}, a_{12})=(-1,1)
\end{aligned}
\label{A16}
 \end{equation}
 It is easily verified that all terminated honeycombs along an edge direction given in \eqref{A16} are equivalent
  to a balanced zigzag edge or an unbalanced zigzag edge by a symmetry of the honeycomb; see Definition \ref{def:ZZAC} and Figure \ref{fig:classical_edges} .
  
  \subsection{A few elementary inequalities involving $\ttn_\nu$}\label{a-few-ineq}
  
  Suppose $\ttn_1,\ttn_2,\ttn_3$ are all distinct. If all the $\ttn_\nu$ are non-negative , then $\ttn_1+\ttn_2+\ttn_3\ge0+1+2=3$, which contradicts \eqref{A14}. Similarly, if the $\ttn_\nu$ are all non-positive, then 
$\ttn_1+\ttn_2+\ttn_3\le0-1-2=-3$, again contradicting \eqref{A14}. Hence,
$
\min\{\ttn_1,\ttn_2,\ttn_3\}<0<\max\{\ttn_1,\ttn_2,\ttn_3\}.
$
 Recall (see \eqref{perm-n}) that we let $(n_1,n_2,n_3)$ denote the permutation of $\ttn_1,\ttn_2,\ttn_3$ satisfying
$
 n_1<n_2<n_3;
$
Therefore, $n_1<0< n_3$, which is \eqref{eq:n_nu_ineq}.
 Recall that by \eqref{nAnBs2}, $\nB-\nA\in\{-1,0,+1\}$. Consequently, \eqref{eq:n_nu_ineq} implies
 \begin{equation*}
 \nA + n_1\le \nB\quad\text{and}\quad  \nB- n_3\le \nA.
 \label{A19}\end{equation*} 

\section{The Wedge of the Edge Proposition}\label{app:woe-prop}

%We prove Proposition \ref{woe-prop} in Appendix \ref{app:woe}

\begin{proof}[Proof of Proposition \ref{woe-prop}]
We use here the notations:
 \[ \beta=\left(\begin{matrix}\beta_1\\ \beta_2\end{matrix}\right),\quad \gamma=\left(\begin{matrix}\gamma_1\\ \gamma_2\end{matrix}\right),\quad \Delta\beta=\beta_1-\beta_2,\quad \Delta\gamma=\gamma_1-\gamma_2.\]
We recall that $(\kpar,\hat{k}_\perp)$ is given by \eqref{h_kpar}- \eqref{h_kperp}. We consider $\kappa>0$; the other cases can be deduced by symmetry of the essential spectrum. For each fixed $\kpar$, the mapping $\kperp\mapsto \pm |h(\kpar,\kperp)|$ sweeps out the essential spectrum of $\HTB_\sharp(\kpar)$.  
Introduce $\tilde{k}_\perp(\kpar)$ such that 
\[ \eta_+(\kpar)= \min_{\kperp} |h(\kpar,\kperp)|=+|h(\kpar,\tilde{k}_\perp(\kpar))|,\qquad \eta_-(\kpar)= -\min_{\kperp} |h(\kpar,\kperp)| = - |h(\kpar,\tilde{k}_\perp(\kpar))|.\]
{Note that $\hat{k}_\perp\equiv \tilde{k}_\perp(\hat{\kpar})$ and note that $\eta_\pm(\hat{k}_\perp)=0$. }
 To find $\tilde{k}_\perp(\kpar)$ (and then $\eta_\pm(\kpar)$) for $\kpar$ near $\hat{\kpar}$ 
 we first expand  $s(\kappa,\kappa_\perp) \equiv |h(\hat{k}+\kappa,\hat{k}_\perp+\kappa_\perp)|^2$  for $|(\kappa,\kappa_\perp)|$ small, where $(\hat{k},\hat{k}_\perp)$ is given by \eqref{h_kpar}- \eqref{h_kperp}. Note that 
  \begin{align*} 
 s(\kappa,\kappa_\perp) 
 %&=\ \Big| 1 + e^{i\beta_1(\hat{\kpar}+\kappa)}\  e^{i\gamma_1(\hat{k}_\perp+\kappa_\perp)} + e^{i\beta_2(\hat{\kpar}+\kappa)}\  e^{i\gamma_2(\hat{k}_\perp+\kappa_\perp)}\Big|^2 \\
 &= \Big| 1 + \hat\omega_1e^{i(\beta_1\kappa+\gamma_1\kappa_\perp)}\  +\ \hat\omega_2e^{i(\beta_2\kappa+\gamma_2\kappa_\perp)}\ \Big|^2,
 \end{align*}
 where $\hat\omega_1=e^{i(\beta_1\hat{\kpar}+\gamma_1\hat{k}_\perp)}$  and 
  $ \hat\omega_2=  e^{i(\beta_2\hat{\kpar}+\gamma_2\hat{k}_\perp)}$ are distinct nontrivial cube
roots of unity. 
  The function $(\kappa,\kappa_\perp)\mapsto s(\kappa,\kappa_\perp) $ is smooth and its minimum value is achieved at $(\kappa,\kappa_\perp)=(0,0)$; $ s(0,0)=|1+\widehat{\omega_1}+ \widehat{\omega_2}|^2 = 0$.   Expanding for  $(\kappa,\kappa_\perp)$ small, we obtain:
%   \begin{align*}
%  s(\kappa,\kappa_\perp)
%  % &= \sum_{j=1}^2 \frac12\hat\omega_j [i(\beta_j\kappa+\gamma_j\kappa_\perp)]^2\ +\ \sum_{j=1}^2\frac12\overline{\hat\omega_j} [-i(\beta_j\kappa+\gamma_j\kappa_\perp)]^2\\
% %  &\qquad +\ \frac12\widehat{\omega_1}\ \overline{\widehat{\omega_2}}\ [i(\beta_1-\beta_2)\kappa+(\gamma_1-\gamma_2)\kappa_\perp)]^2\\
% %  &\qquad +\ \frac12\overline{\widehat{\omega_1}}\ \widehat{\omega_2}\ \ [-i(\beta_1-\beta_2)\kappa+(\gamma_1-\gamma_2)\kappa_\perp)]^2 \ +\ q_3(\kappa,\kappa_\perp)\\
%  &= -\sum_{j=1}^2 \frac12(\hat\omega_j+\overline{\hat\omega_j}) (\beta_j\kappa+\gamma_j\kappa_\perp)^2\\
%  &\qquad  - \frac12\left(\widehat{\omega_1}\ \overline{\widehat{\omega_2}}\ +\ \overline{\widehat{\omega_1}}\ \widehat{\omega_2} \right)[\Delta\beta\kappa+\Delta\gamma\kappa_\perp)]^2\ +\ q_3(\kappa,\kappa_\perp),
%  \end{align*}
\begin{multline*}
  s(\kappa,\kappa_\perp) 
  % &= \sum_{j=1}^2 \frac12\hat\omega_j [i(\beta_j\kappa+\gamma_j\kappa_\perp)]^2\ +\ \sum_{j=1}^2\frac12\overline{\hat\omega_j} [-i(\beta_j\kappa+\gamma_j\kappa_\perp)]^2\\
 %  &\qquad +\ \frac12\widehat{\omega_1}\ \overline{\widehat{\omega_2}}\ [i(\beta_1-\beta_2)\kappa+(\gamma_1-\gamma_2)\kappa_\perp)]^2\\
 %  &\qquad +\ \frac12\overline{\widehat{\omega_1}}\ \widehat{\omega_2}\ \ [-i(\beta_1-\beta_2)\kappa+(\gamma_1-\gamma_2)\kappa_\perp)]^2 \ +\ q_3(\kappa,\kappa_\perp)\\
  = -\sum_{j=1}^2 \frac12(\hat\omega_j+\overline{\hat\omega_j}) (\beta_j\kappa+\gamma_j\kappa_\perp)^2
  \\- \frac12\left(\widehat{\omega_1}\ \overline{\widehat{\omega_2}}\ +\ \overline{\widehat{\omega_1}}\ \widehat{\omega_2} \right)[\Delta\beta\kappa+\Delta\gamma\kappa_\perp)]^2\ +\ q_3(\kappa,\kappa_\perp), 
\end{multline*}
  where $q_3(\kappa,\kappa_\perp)$ is smooth and of cubic order for $|(\kappa,\kappa_\perp)|$ small. Since $\hat\omega_j$, $j=1,2$ are distinct nontrivial cube
roots of unity, 
 $ \Re\widehat{\omega_1}\ =\ \Re\widehat{\omega_2}\ =\ \Re(\widehat{\omega_1}\ \overline{\widehat{\omega_2}})\ =\ -1/2,$
  and therefore
  \begin{align}
 s(\kappa,\kappa_\perp) = 
 \frac12\left[ 
 \sum_{j=1}^2  (\beta_j\kappa+\gamma_j\kappa_\perp)^2
 + [\Delta\beta\kappa+\Delta\gamma\kappa_\perp)]^2\
  \right]  + q_3(\kappa,\kappa_\perp).
\label{s-expand}  \end{align}
 Fix $\kappa$ small. Then, from \eqref{s-expand} we find that the minimum of $\kappa_\perp\mapsto s(\kappa,\kappa_\perp)$ 
 is attained at 
 \begin{equation} \tilde\kappa_\perp(\kappa) \equiv -  \Big(\frac{\beta\cdot\gamma + \Delta\beta\Delta\gamma}{|\gamma|^2+(\Delta\gamma)^2}\Big) \kappa\ +\ \mathcal{O}(\kappa^2).\label{t-kappa}\end{equation}

 Substitution of \eqref{t-kappa} into \eqref{s-expand} yields
 \begin{equation}
 s\left(\kappa,\tilde\kappa_\perp(\kappa)\right)
  = \frac12\ \frac{\left(|\beta|^2+(\Delta\beta)^2\right)\ \left(|\gamma|^2+(\Delta\gamma)^2\right) 
   - \left(\beta\cdot\gamma+ \Delta\beta\Delta\gamma\right)^2}{|\gamma|^2+(\Delta\gamma)^2}\ \kappa^2 
    \ +\ \mathcal{O}(\kappa^3).
 \label{s-expand1}
 \end{equation}
 The leading term in \eqref{s-expand1} can be simplified using the identities
 \[
 \det[\beta\ \gamma])^2 = |\beta|^2\ |\gamma|^2 - (\beta\cdot\gamma)^2 \quad\text{and}\quad |\beta|^2(\Delta\gamma)^2+|\gamma|^2(\Delta\beta)^2-2\beta\cdot\gamma\Delta\beta\Delta\gamma=|\beta\Delta\gamma-\gamma\Delta\beta|^2
 \]
 We find, by \eqref{det-sig}
 \begin{equation}
 s\left(\kappa,\tilde\kappa_\perp(\kappa)\right) =  \frac{3}{2}\ \frac{1}{|\gamma|^2+(\Delta\gamma)^2}\ \kappa^2 + 
 \mathcal{O}(\kappa^3)\
 \label{s-expand2} \end{equation}
 We simplify the denominator using \eqref{bg-def} and \eqref{ttn-ttm}:
 \[
  |\gamma|^2+(\Delta\gamma)^2\ 
  = (n_2-n_1)^2+ (n_3-n_1)^2+(n_2-n_3)^2= 2\left( a_{11}^2 + a_{11}a_{12}+a_{12}^2\right).
  \]
  Finally, substitution into \eqref{s-expand2} yields \eqref{m-wedge1} and the proof of Proposition \ref{woe} is complete.
\end{proof}

 \section{ $p_+(\zeta,\kp)$ and $p_-(\zeta,\kp)$ are honeycomb edge polynomials}\label{sec:p_pm}

 A honeycomb edge polynomial  (see Section \ref{sec:HEPs}) is a polynomial of the form \eqref{hep-1}
 where  $\kp\in[0,2\pi]$ and  $\beta_1, \beta_2, \gamma_1, \gamma_2$ are integers 
 such that \eqref{gamma12} \eqref{det-sig} are satisfied.
In this appendix we verify that the polynomials $p_+$ and $p_-$ defined by  \eqref{B-poly},\eqref{A-poly}
 % \begin{equation}
%  p_+(\zeta,\kp) \equiv 1 + e^{i(m_2-m_1)\kp} \zeta^{n_2-n_1} + e^{i(m_3-m_1)\kp} \zeta^{n_3-n_1}\ .
%  \label{B-poly-app} %(62)
%  \end{equation}
%  and
%  \begin{equation}
%  p_-(\zeta,\kp) \equiv 1 + e^{i(m_3-m_2)\kp} \zeta^{n_3-n_2} + e^{i(m_3-m_1)\kp} \zeta^{n_3-n_1}\ .
%  \label{A-poly-app} %(63)
%  \end{equation}
 are honeycomb edge polynomials. 

 We first note using \eqref{ttn-ttm} that
  \begin{equation*}
  \det\begin{pmatrix} 1&\tm_1&\tn_1\\ 1&\tm_2&\tn_2\\ 1&\tm_3&\tn_3\end{pmatrix} 
   =  \det\begin{pmatrix} 1&k_1&k_2\\ 0&a_{21}&-a_{11}\\ 0&-a_{22}&a_{12}\end{pmatrix} = a_{21}a_{12}-a_{11}a_{22}=-1.
  \end{equation*}
 Now the matrix 
 \[ \begin{pmatrix} 1&m_1&n_1\\ 1&m_2&n_2\\ 1&m_3&n_3\end{pmatrix} \]
 is obtained from the above matrix by permutation of the rows (so that $n_1<n_2<n_3$). Hence, 
  \[ \sigma_{\rm det}\ :=\ \det\begin{pmatrix} 1&m_1&n_1\\ 1&m_2&n_2\\ 1&m_3&n_3\end{pmatrix} \in\{-1,+1\}.
  \]
 Furthermore, 
  \begin{align}
  \sigma_{\rm det}\ = \det\begin{pmatrix} 1 &m_1&n_1\\ 0&m_2-m_1&n_2-n_1\\ 0&m_3-m_1&n_3-n_1\end{pmatrix}
 = \det\begin{pmatrix}m_2-m_1&n_2-n_1\\ m_3-m_1&n_3-n_1\end{pmatrix}  \in\{-1,+1\}, \label{det1-app}
  \end{align} 
 and, on the other hand, 
 \begin{align}
  \sigma_{\rm det}\ = \det\begin{pmatrix} 0&m_1-m_3&n_1-n_3\\ 0&m_2-m_3&n_2-n_3\\ 1&m_3&n_3\end{pmatrix}
   =  \det\begin{pmatrix} m_1-m_3&n_1-n_3\\ m_2-m_3&n_2-n_3\end{pmatrix}
   \in\{-1,+1\}. \label{det2-app}\end{align}
 We can verify now that the polynomials $p_\pm(\zeta,\kp)$ are honeycomb edge polynomials.

  First consider $p_+(\zeta,\kp)$. We have $\beta_1=m_2-m_1$, $\gamma_1=n_2-n_1$, $\beta_2=m_3-m_1$
  and $\gamma_2=n_3-n_1$. Since $n_1<n_2<n_3$, the condition \eqref{gamma12} holds and furthermore
   by \eqref{det1-app} the condition \eqref{det-sig} holds. 
  
   We turn to $p_-(\zeta,\kp)$. Here, $\beta_1=m_3-m_2$, $\gamma_1=n_3-n_2$, $\beta_2=m_3-m_1$ and $\gamma_2=n_3-n_1$. Again since $n_1<n_2<n_3$, the condition \eqref{gamma12} holds and furthermore 
    by \eqref{det2-app} the condition \eqref{det-sig} holds.

  \section{Monotonicity Lemma \ref{mono-lem}}\label{app:mono-lem}
 
Let us recall that the functions $(\rho_1,\rho_2)\mapsto \alpha_j(\rho_1,\rho_2)$ for $j\in\{1,2\}$ defined in \eqref{alpha12} are  real-valued if $\rho_1+\rho_2\ge1$ and $|\rho_1-\rho_2|\le1$. Fix $\kappa\in(0,1)$ and let us consider the function $\rho\mapsto M(\rho)$ defined in \eqref{M-def}
 %
 % \begin{equation}
 %  M(\rho) = \alpha_1(\rho^\kappa,\rho)\ -\ \kappa\ \alpha_2(\rho^\kappa,\rho)
 %  \label{Mdef1}
 %  \end{equation}
 %
 % Define $\rho_c$ to be the unique solution in $(0,1)$ of: $\rho_c^\kappa + \rho_c =1$. \medskip
For $\rho_c:=\rho_{\rm critical}\in(0,1)$ defined in \eqref{rhoc2}, we have that $\rho\mapsto \rho^\kappa+\rho$ is monotone: $\rho\ge\rho_c\ \implies \rho^\kappa + \rho \ge 1$ and hence $M(\rho)$ is real-valued on $[\rho_c,1)$. We now prove Lemma \ref{mono-lem}, i. e. that $M(\rho)$ is monotone decreasing over the interval $[\rho_c,1)$. 
 \bigskip

 The proof follows from the following three claims.

 \begin{enumerate}
	 \item \underline{Claim 1}: The function $r$ defined by
     \begin{equation}
      r(\rho) := 2(1+\rho^2)\rho^{2\kappa} - \rho^{4\kappa} - (1-\rho^2)^2\ .
      \label{r_of_rho}\end{equation}
	 is such that $r(\rho)\ge0$ for all $\rho\in[\rho_c,1)$.
     \item \underline{Claim 2}:  We have
     \[  \kappa \left[ (1-\kappa) \rho^{2\kappa} - 1\right] \ +\ (\kappa-1)\rho^2 < 0\]
     for any fixed $0<\kappa<1$, and all $\rho\in(0,1)$.
 \item \underline{Claim 3}:
 \begin{equation} M^\prime(\rho) = \frac{2}{\rho}\ \frac{1}{\sqrt{r(\rho)} }\ 
 \Big[\ 
  \kappa \left[ (1-\kappa) \rho^{2\kappa} - 1\right] \ +\ (\kappa-1)\rho^2 \ 
  \Big]\ , \label{Mprime}\end{equation}
  where $r$ is defined \eqref{r_of_rho}.
 \end{enumerate}

 \subsection{Proof of Claim 1}
 Let us first show that $r(\rho_c)=0$. Since $\rho_c^\kappa=1-\rho_c$,
 \[
 r(\rho_c) = 2(1+\rho_c^2)\cdot(1-\rho_c)^2-(1-\rho_c)^4-(1-\rho_c)^2(1+\rho_c)^2
=0.
 \]
Now we show that $r^\prime(\rho)>0$ for all $\rho\in(0,1)$. Direct computation yields
 \begin{align*}
 r^\prime(\rho) &= \frac{4\kappa}{\rho}
 \left(\ 
 \rho^{2\kappa}(1-\rho^{2\kappa}) + \rho^{2\kappa+2}\ \right) + \frac{4}{\rho}\rho^{2\kappa+2}+4\rho(1-\rho^2)
 \end{align*}
 For $\rho\in(0,1)$, we have $r^\prime(\rho) >0$.
 	 
 \subsection{Proof of Claim 2}
 Claim 2 is clear since $\kappa\in(0,1)$ and $\rho\in(0,1)$.

 \subsection{Proof of Claim 3}
We rewrite
 \begin{equation*}
 M(\rho) = \ac(X(\rho)) +\kappa \left( \ac(Y(\rho) \right),\label{Mdef2}
 \end{equation*}
 where
 \begin{equation*}
 X(\rho)= \frac{\rho^2-\rho^{2\kappa}-1}{2\rho^\kappa},\quad Y(\rho)= \frac{\rho^{2\kappa}-\rho^2-1}{2\rho}
 \label{XY-def}\end{equation*}
We deduce that \begin{align*}
 M^\prime(\rho) & = \ac^\prime(X(\rho))\ \D_\rho X(\rho) + \kappa\ \ac^\prime(Y(\rho))\ \D_\rho Y(\rho) 
 \end{align*} where we recall that
$\ac^\prime(x)=-(1-x^2)^{-1/2} .$
 \subsubsection{Calculation of $\ac^\prime(X(\rho))$ and $\ac^\prime(Y(\rho))$}
We have that 
 \begin{align}
 \ac^\prime(X(\rho)) 
  = - \frac{2\rho^\kappa}{\sqrt{r(\rho)}} \label{DacX1}\\
  \ac^\prime(Y(\rho)) 
  = -\frac{2\rho}{\sqrt{r(\rho)}}\label{DacY1}
  \end{align}
  where $r(\rho)$ is defined in \eqref{r_of_rho}, by using that
  \[
 (2\rho^\kappa)^2-(\rho^2-\rho^{2\kappa}-1)^2= 4\rho^{2\kappa}- ((\rho^2-1)-\rho^{2\kappa})^2 = r(\rho), 
  \]
and
   \[
   (2\rho)^2-(\rho^{2\kappa}-\rho^2-1)^2 
   %=(2\rho)^2-(\rho^{2\kappa}-(\rho^2+1))^2\\
  %= 4\rho^2-\left[\rho^{4\kappa} + (\rho^2+1)^2 -2\rho^{2\kappa}(\rho^2+1)\right]\\
  = \underbrace{4\rho^2-(\rho^2+1)^2}_{-(1-\rho^2)^2} -\rho^{4\kappa}+2\rho^{2\kappa}(\rho^2+1) = r(\rho).
   \]
  % \footnote{previous version was missing a minus sign in front of $(1-\rho^2)^2$ on the previous line. Thanks Charlie.}

 \subsubsection{Calculation of $\D_\rho X(\rho)$ and $\D_\rho Y(\rho)$}
  
  \begin{align}
  \D_\rho X(\rho) & =  \frac{(2\rho^\kappa)\left(2\rho-2\kappa\rho^{2\kappa-1}\right)-\left(\rho^2-\rho^{2\kappa}-1\right)(2\kappa\rho^{\kappa-1})}{(2\rho^\kappa)^2} \label{DX}\\
   \D_\rho Y(\rho) & =  \frac{(2\rho)\left(2\kappa\rho^{2\kappa-1}-2\rho\right) - \left(\rho^{2\kappa}-\rho^2-1\right)(2)}{(2\rho)^2}
   \label{DY}
 \end{align}

   %!    

 \subsubsection{Calculation of $M^\prime(\rho)$}

 Multiplying \eqref{DacX1} and \eqref{DX}, we deduce

 \begin{align}
 \ac^\prime(X(\rho))\D_\rho X(\rho) 
 %&= - \frac{2\rho^\kappa}{\sqrt{r(\rho)}} \times \frac{(2\rho^\kappa)\left(2\rho-2\kappa\rho^{2\kappa-1}\right)-\left(\rho^2-\rho^{2\kappa}-1\right)(2\kappa\rho^{\kappa-1})}{(2\rho^\kappa)^2}\nn \\
% &=   \frac{1}{\sqrt{r(\rho)}} \times \frac{-4\rho^{2\kappa}\left(2\rho-2\kappa\rho^{2\kappa-1}\right)+4\kappa\rho^{2\kappa-1}\left(\rho^2-\rho^{2\kappa}-1\right)}{4\rho^{2\kappa}} \nn \\
% &=   \frac{1}{\sqrt{r(\rho)}} \times \frac{-2\rho^{2\kappa}\left(\rho-\kappa\rho^{2\kappa-1}\right)+\kappa\rho^{2\kappa-1}\left(\rho^2-\rho^{2\kappa}-1\right)}{\rho^{2\kappa}}\nn \\
% &=   \frac{1}{\sqrt{r(\rho)}} \times \frac{-2\rho^{2\kappa-1}\left(\rho^2-\kappa\rho^{2\kappa}\right)+\kappa\rho^{2\kappa-1}\left(\rho^2-\rho^{2\kappa}-1\right)}{\rho^{2\kappa}}\nn \\
% &=   \frac{1}{\rho}\frac{1}{\sqrt{r(\rho)}} \times \frac{-2\rho^{2\kappa}\left(\rho^2-\kappa\rho^{2\kappa}\right)+\kappa\rho^{2\kappa}\left(\rho^2-\rho^{2\kappa}-1\right)}{\rho^{2\kappa}}\nn \\
% &= \frac{1}{\rho}\frac{1}{\sqrt{r(\rho)}} \times \left( -2\left(\rho^2-\kappa\rho^{2\kappa}\right)+\kappa\left(\rho^2-\rho^{2\kappa}-1\right)\right)\nn
%  \\
 &= \frac{1}{\rho}\frac{1}{\sqrt{r(\rho)}} \times \left( -2\rho^2+\kappa(\rho^{2\kappa}+\rho^2-1)\right),
 \label{prodX}\end{align}
and multiplying \eqref{DacY1} and \eqref{DY} yields
 \begin{align}
 \ac^\prime(Y(\rho))\D_\rho Y(\rho) 
 %&=  -\frac{2\rho}{\sqrt{r(\rho)}}\times \frac{(2\rho)\left(2\kappa\rho^{2\kappa-1}-2\rho\right) - 2\left(\rho^{2\kappa}-\rho^2-1\right)}{(2\rho)^2}\nn\\
% &= \frac{1}{\sqrt{r(\rho)}}\times \frac{-4\rho^2\left(2\kappa\rho^{2\kappa-1}-2\rho\right) +4\rho \left(\rho^{2\kappa}-\rho^2-1\right)}{4\rho^2}\nn\\
% &=\frac{1}{\rho}\frac{1}{\sqrt{r(\rho)}}\times \left(-\rho\left(2\kappa\rho^{2\kappa-1}-2\rho\right) + \left(\rho^{2\kappa}-\rho^2-1\right)\right)\nn\\
 &=\frac{1}{\rho}\frac{1}{\sqrt{r(\rho)}} \times \left( (1-2\kappa)\rho^{2\kappa}+\rho^2-1\right)\label{prodY}
 \end{align}
Since $M^\prime(\rho)=$ \eqref{prodX}\ +\ $\kappa\ \times$\ \eqref{prodY}, we deduce \eqref{Mprime}.

\section{Independence of the choice of $(a_{21},a_{22})$}\label{app:v_2}

{As explained in Section \ref{sec:rat-edge}, while the integer vector $(a_{11},a_{12})$ is uniquely defined, the vector $(a_{21},a_{22})$ is uniquely specified up to translates by integer multiples of $(a_{11},a_{12})$. In other words, if $(a_{21},a_{22})$ is a possible choice then for all $j\in\Z$, $(a^{(j)}_{21}:=a_{21}+ja_{11},a^{(j)}_{22}:=a_{22}+ja_{12})$ is also a possible choice. We check in this section that our results are independent of this choice, i.e. of $j\in\Z$.}

{Let us add a superscript ${(j)}$ to all the objects used in the analysis. First, note that $\bv_2^{(j)}= \bv_2^{(0)}+j\bv_1$. Moreover, it is easy to see that 
\[
	a_{22}^{(j)}-a_{21}^{(j)}=3k_1^{(0)}-3jk_2+s_1^{(0)}-js_2\]
	where we have used the notation introduced in \eqref{k1k2s1s2} and where $k_2$ and $s_2$ are independent of $j$ since they depend only on $(a_{11},a_{12})$. We deduce that
\[
	s_1^{(j)}:=a_{22}^{(j)}-a_{21}^{(j)}\;\; \text{mod } 3=s_1^{(0)}-js_2\; \text{mod } 3.
	\]
	 If we denote by $k^{(j)}$ the integer so that $s_1^{(0)}-js_2=3k^{(j)}+s_1^{(j)}$ we have
	 \begin{equation}\label{kj}
		k_1^{(j)}=k_1^{(0)}-jk_2+k^{(j)}.
	 \end{equation}
	 This implies in particular that the point $\bv^{(j)}_B$ defined in \eqref{vAvB-1} as the only $B-$point in the parallelogram $\Gamma^{(j)}(0,0)$ is in the parallelogram $\Gamma^{(0)}(-k^{(j)},0)$ since
	\[
		\bv^{(j)}_B=\frac{1}3(s_1^{(j)}\bv_1+s_2(\bv_2^{(0)}+j\bv_1))=\bv^{(0)}_B-k^{(j)}\bv_1.
	\]
	This implies that for any $\omega\in\HH_\sharp$
	\[
		\begin{array}{rcl}
			\omega=\bv_A+m\bv_1+n\bv_2^{(j)}&\Rightarrow & \omega=\bv_A+(m+nj)\bv_1+n\bv_2^{(0)}\\
			\omega=\bv_B^{(j)}+m\bv_1+n\bv_2^{(j)}&\Rightarrow & \omega=\bv_B^{(0)}+(m+nj-k^{(j)})\bv_1+n\bv_2^{(0)}
		\end{array}
	\]
	Consequently, for any wave function $(\psi_\omega)_{\omega\in\HH_\sharp}$ which can be represented with respect to any basis $\{\bv_1,\bv_2^{(j)}\}$, we have
	\[
		\psi^{A,(j)}(m,n)= \psi^{A,(0)}(m+nj,n),\quad\text{and}\quad \psi^{B,(j)}(m,n)= \psi^{B,(0)}(m+nj-k^{(j)},n),\]
		where we have used the notation introduced in \eqref{psi-mn}.}

	Note finally that by \eqref{kj}, the integers introduced in \eqref{ttn-ttm} are such that
		\[
			\forall \ell\neq\ell',\quad \tilde{m}_\ell^{(j)}-\tilde{m}_{\ell'}^{(j)}= \tilde{m}_\ell^{(0)}-\tilde{m}_{\ell'}^{(0)}-j(\tilde{n}_\ell-\tilde{n}_{\ell'})\]
			and the re-ordering mentioned in \eqref{perm-n} is independent of $j$. This allows to show that 
			\begin{equation}\label{eq:roots}
				\zeta\text{ is a root of }p_\pm^{(j)}(\cdot,k)\quad\Leftrightarrow\quad e^{-\imath j k}\zeta\text{ is a root of }p_\pm^{(0)}(\cdot,k),\end{equation}
				where $p_\pm^{(j)}(\cdot,k)$ are the polynomials defined in \eqref{B-poly} and \eqref{A-poly}.

Let us now discuss the independence of our results with respect to $j$. First the condition of existence of a zero energy edge state (see Table \ref{tab:table} and Theorems \ref{th:zigzag} and \ref{th:armchair}) depends only on $(a_{11},a_{12})$ so it is of course independent of $j$. Concerning the formulas of the zero energy edge states, Theorem \ref{0en-reps} with \eqref{eq:k-QP} yields that the edge states are given by 
	 \[\psi^{(j)}(m,n)
	 =c e^{\imath k m}\sum_{\substack{l_1+\dots+l_r=n-\nbase+1 \\ l_1,\dots,l_r\ge1}}  (\zeta_1^{(j)})^{l_1-1}\cdots(\zeta_r^{(j)})^{l_r-1},
	 \quad {\rm for}\; n\geq \nbase \]
	where $c$ is a constant and $\zeta_1^{(j)},\ldots \zeta_r^{(j)}$ are the roots of $p_+^{(j)}(\cdot,k)$ if it is a $B-$edge state and $p_-^{(j)}(\cdot,k)$ if it is an $A-$edge state. By \eqref{eq:roots}, we recover (up to a change of the constant $c$) the edge states given by $j=0$.

Using similar arguments, we can show the independence of the result of Proposition \ref{es-cond} with respect to $j$.
%%      Type body of appendix/-ices here.

%%      ---------------------------------------------------------------------
%%      ---------------------------ACKNOWLEDGMENTS (OPTIONAL) ---------------
%%      ---------------------------------------------------------------------

%% ***** UNCOMMENT THE FOLLOWING LINE TO ADD ACKNOWLEDGMENTS.

 \ack 

 This research was initiated
 at a working group on "Irrational edges" at the American Institute of Mathematics (AIM) Workshop on the {\it Mathematics of Topological Insulators}, December 7-11, 2020, which was supported by the  American Institute of Mathematics, the US National Science Foundation, the Simons Foundation and Columbia University.
 C.L.F. was supported in part by National Science Foundation grant DMS-1700180.
M.I.W. was supported in part by National Science Foundation grants DMS-1620418 and DMS-1908657 as well as Simons Foundation Math + X Investigator Award \#376319. We warmly thank the participants of the AIM working group, as well as Pierre Delplace, David Gontier and Mikael Rechtsman for very stimulating discussions.

%%      ---------------------------------------------------------------------
%%      --------------------------- BIBLIOGRAPHY ----------------------------
%%      ---------------------------------------------------------------------

\bibliography{rational-ffw}
\bibliographystyle{cpam}
\end{document}